%% file: smoothed-sat-arXiv.tex
\documentclass{article}
\PassOptionsToPackage{numbers}{natbib}
\usepackage[preprint]{neurips_2021}
\input{macro.tex}

\title{The Smoothed Satisfaction of Voting Axioms
}

\author{Lirong Xia\\ RPI \\ xialirong@gmail.com} 

\begin{document}
\maketitle

\begin{abstract}
We initiate the work towards a comprehensive picture of the smoothed satisfaction of  voting axioms, to provide a finer and more realistic foundation for comparing voting rules. We adopt the {\em smoothed social choice framework}~\citep{Xia2020:The-Smoothed},  where an adversary chooses arbitrarily correlated ``ground truth'' preferences for the agents, on top of which random noises are added. We  focus on characterizing the smoothed satisfaction of two well-studied voting axioms:  {\em Condorcet criterion} and {\em participation}. We prove that  for any fixed number of alternatives, when the number of voters $n$ is sufficiently large, the smoothed satisfaction of the Condorcet criterion under a wide range of voting rules is $1$, $1-\exp(-\Theta(n))$, $\Theta(n^{-0.5})$, $ \exp(-\Theta(n))$, or being $\Theta(1)$ and $1-\Theta(1)$ at the same time; and the smoothed satisfaction of participation is  $1-\Theta(n^{-0.5})$.  
Our results  address  open questions by~\citet{Berg1994:On-probability} in 1994, and also  confirm the following high-level message: the Condorcet criterion is a bigger concern than participation under realistic models. 
\end{abstract}

\section{Introduction}
The {\em ``widespread presence of impossibility results''}~\citep{Sen1999:The-Possibility} is one of the most  fundamental and significant challenges in social choice theory. These impossibility results often assert that no ``perfect'' voting rule exists for three or more alternatives~\citep{Arrow63:Social,Gibbard73:Manipulation,Satterthwaite75:Strategy}. Nevertheless, an (imperfect) voting rule must be designed and used in practice for agents to make a collective decision. 
In the social choice literature, the dominant paradigm of doing so has been the {\em axiomatic approach}, i.e., voting rules are designed, evaluated, and compared to each other w.r.t.~their satisfaction of desirable normative properties, known as {\em (voting) axioms}. 

Most definitions of dissatisfaction of voting axioms are based on  worst-case analysis. For example, a voting rule $r$ does not satisfy {\sc Condorcet Criterion} ($\CC$ for short), if there exists a collection of votes, called a profile, where the {\em Condorcet winner} exists but is not chosen by $r$ as a winner. The Condorcet winner is the alternative who beats all other alternatives in their head-to-head competitions. As another example, a voting rule $r$ does not satisfy {\sc Participation} ($\Par$ for short), if there exist a profile and a voter who has incentive to abstain from voting. An instance of  dissatisfaction of  $\Par$ is also  known as the {\em no-show paradox}~\citep{Fishburn1983:Paradoxes}. Unfortunately,  when the number of alternatives $m$ is at least four, no irresolute voting rule satisfies  $\CC$ and  $\Par$ simultaneously~\citep{Moulin1988:Condorcets}. 

While the  classical worst-case analysis of (dis)satisfaction of axioms can be  desirable in high-stakes applications such as political elections, it is often too coarse to serve as practical criteria for comparing different voting rules in more frequent, low-stakes applications of social choice, such as business decision-making~\citep{Bhattacharjya14:Bayesian}, crowdsourcing~\citep{Mao13:Better},  informational retrieval~\citep{Liu11:Learning}, meta-search engines~\citep{Dwork01:Rank}, recommender systems~\citep{Wang2016:Ranking-Oriented}, etc.  A decision maker who desires  both axioms would find it hard to choose between a voting rule that satisfies $\CC$ but not $\Par$, such as Copeland, and a voting rule that satisfies $\Par$ but not $\CC$, such as plurality. A finer and more quantitative measure of satisfaction of axioms is therefore called for.

One natural and classical approach is to measure the likelihood of  satisfaction of  axioms under a probabilistic model of agents' preferences, in particular  the independent and identically distributed (i.i.d.) uniform distribution over all rankings, known as {\em Impartial Culture (IC)} in social choice. This line of research was initiated and established by~\citeauthor{Gehrlein1976:The-probability} in a series of work in the 1970's~\citep{Gehrlein1976:The-probability,Gehrlein1978:Probabilities,Gehrlein1978:Coincidence}, and has become a ``{\em new sub-domain of the theory of social choice}''~\citep{Diss2021:Evaluating}. Some classical results were summarized in the 2011 book by~\citet{Gehrlein2011:Voting}, and  recent progresses can be found in the  2021 book edited by~\citet{Diss2021:Evaluating}.

While this line of work is highly significant and interesting from a theoretical point of view, its practical implications may not be as strong, because most previous work focused on a few specific distributions, especially IC, which has been widely criticized to be unrealistic~(see, e.g., \citep[p.~30]{Nurmi1999:Voting}, \citep[p.~104]{Gehrlein2006:Condorcets}, and~\citep{Lehtinen2007:Unrealistic}). 
Indeed, conclusions drawn under any specific distribution may not hold in practice, as ``{\em all models are wrong}''~\citep{Box1979:Robustness}. 
Technically, characterizing the likelihood of satisfaction of $\CC$ and of $\Par$ are already highly challenging w.r.t.~IC, and despite that~\citet{Berg1994:On-probability} explicitly posed them as open questions in 1994, not much is known beyond a few voting rules. 
Therefore, the following  question largely remains open.

\vspace{2mm}
{\hfill   \textbf{How likely are voting axioms satisfied under realistic models?} \hfill}
\vspace{2mm}

The importance of successfully answering this question is two-fold. First,   it tells us whether the worst-case violation of an axiom   is a significant concern in practice. Second,   it  provides a finer and more quantitative foundation for comparing   voting rules. 

We believe that the {\em smoothed analysis} proposed by~\citet{Spielman2004:Smoothed} provides a promising framework for addressing the question. In this paper, we  adopt the 
{\em smoothed social choice} framework by~\citeauthor{Xia2020:The-Smoothed},  which models the satisfaction of a {\em per-profile} voting axiom $X$  by a function $\sat{X}(r,P)\in\{0,1\}$, where $r$ is a voting rule and $P$ is a profile, such that $r$ satisfies $X$ if $\min_{P}\sat{X}(r,P) = 1$.  Let $\Pi$ denote a set of distributions over all rankings over the $m$ alternatives (denoted by $\ml(\ma)$), which represents the ``ground truth'' preferences for a single agent that the adversary can choose from. Let $n$ denote the number of agents.  Because a higher value of $\sat{X}(r,P)$ is more desirable to the decision maker, the adversary aims at minimizing  expected $\sat{X}(r,P)$ by  choosing $\vec \pi\in\Pi^n$---the  profile $P$ is generated from $\vec \pi$. The smoothed satisfaction of $X$ under $r$ with $n$ agents, denoted by $\satmin{X}{\Pi}(r,n)$,  is defined as follows~\citep{Xia2020:The-Smoothed}:
\begin{equation}
\label{dfn:s-sat}
\satmin{X}{\Pi}(r,n) = \inf\nolimits_{\vec\pi\in\Pi^n}\Pr\nolimits_{P\sim\vec\pi}\sat{X}(r,P)
\end{equation}
Notice that agents' ground truth preferences can be arbitrarily correlated, while the noises are independent, which is a standard assumption in the literature and in practice~\citep{Xia2020:The-Smoothed}.

\begin{ex}[\bf\boldmath Smoothed $\CC$ under plurality]
\label{ex:sCC-plu}
Let $X=\CC$ and $r=\iplu$ denote the irresolute plurality rule, which chooses all alternatives that are ranked at the top most often as the (co-)winners. Suppose there are three alternatives, denoted by $\ma = \{1,2,3\}$, and suppose   $\Pi = \{\pi^1,\pi^2\}$, where $\pi^1$ and $\pi^2$ are distributions shown in Table~\ref{tab:sCC-plu}.\\
\begin{minipage}[t][][b]{0.5\textwidth}
Then, we have $\satmin{\CC}{\Pi}(\iplu,n) = \inf\nolimits_{\vec\pi\in\{\pi^1,\pi^2\}^n}\Pr\nolimits_{P\sim\vec\pi}\sat{\CC}(\iplu,P)$. When $n=2$, the adversary has four choices of $\vec \pi$, i.e., $\{(\pi^1,\pi^1),(\pi^1,\pi^2),(\pi^2,\pi^1),(\pi^2,\pi^2)\}$. 
\end{minipage}
\hfill
\begin{minipage}[t][][b]{0.45\textwidth}
\centering
\begin{tabular}{|@{\ }c@{\ }|c|c|c|c|c|c|c| }
\hline & \small $123$& \small $132$& \small $231$& \small $321$& \small $213$& \small $312$ \\

\hline $\pi^1$& $1/4$& $1/4$&$1/8$& $1/8$& $1/8$&$1/8$  \\

\hline $\pi^2$& $1/8$& $1/8$&$3/8$& $1/8$& $1/8$& $1/8$ \\

\hline
\end{tabular} 
\captionof{table}{\small $\Pi$ in Example~\ref{ex:sCC-plu}.\label{tab:sCC-plu}}
\end{minipage}

Each $\vec\pi$ leads to a distribution over the set of all  profiles of two agents, i.e., $\ml(\ma)^2$. We have $\satmin{\CC}{\Pi}(\iplu,2)=1$, because $\CC$ is satisfied at all profiles of two agents. As we will see later in Example~\ref{ex:thm-sCC-pos}, for all sufficiently large $n$,  $\satmin{\CC}{\Pi}(\iplu,n)= \exp(-\Theta(n))$.  
\end{ex}


\subsection{\bf Our Contributions}
We initiate the work towards a comprehensive picture of smoothed satisfaction of voting axioms under commonly-studied voting rules, by focusing on  $\CC$ and $\Par$ in this paper  due to their importance, popularity, and incompatibility~\citep{Moulin1988:Condorcets}.  
Recall that $m$ is the number of alternatives and  $n$ is the number of agents.   Our technical contributions are two-fold.  

\noindent{\bf First, smoothed satisfaction of $\CC$ (Theorem~\ref{thm:sCC-scoring} and~\ref{thm:sCC-MRSE}).} We prove that, under mild assumptions, for any fixed $m\ge 3$ and any sufficiently large $n$, the smoothed satisfaction of $\CC$ under a wide range of voting rules is $1$, $1-\exp(-\Theta(n))$, $\Theta(n^{-0.5})$, $ \exp(-\Theta(n))$, or being $\Theta(1)$ and $1-\Theta(1)$ at the same time (denoted by $\Theta(1)\wedge(1-\Theta(1))$).  The  $1-\exp(-\Theta(n))$ case is positive news, because it states that $\CC$ is satisfied almost surely when $n$ is large, regardless of the adversary's choice. The remaining three cases are negative news, because they state that the adversary can make $\CC$ to be violated with non-negligible probability, no matter how large $n$ is.

 \noindent{\bf Second, smoothed satisfaction of $\Par$ (Theorems~\ref{thm:sPar-mm-rp-sch},   \ref{thm:sPar-copeland}, \ref{thm:sPar-MRSE}, \ref{thm:sPar-Cond-Pos}).} We prove that, under mild assumptions, for any fixed $m\ge 3$ and any  sufficiently large $n$, the smoothed satisfaction of $\Par$ under a wide range of voting rules is $1-\Theta(n^{-0.5})$. These are positive news, because they state that $\Par$ is satisfied almost surely for large $n$, regardless of the adversary's choice. While this message  may not be surprising at a high level, as the probability for a single agent to change the winner vanishes as $n\ra\infty$,  the theorems are useful and non-trivial, as they provide asymptotically tight rates.

In particular, straightforward corollaries of our theorems to IC address  open questions posed by~\citet{Berg1994:On-probability} in 1994, and also provides a mathematical justification of two common beliefs related to $\Par$: first, IC exaggerates the likelihood for paradoxes to happen, and second, the dissatisfaction of $\Par$  is not a significant concern in practice~\citep{Lepelley2001:Scoring}, especially when it is compared to our results on smoothed $\CC$. Table~\ref{tab:summary} summarizes corollaries of our results under some commonly-studied voting rules w.r.t.~IC as well as the satisfaction of $\CC$ and $\Par$  on Preflib data~\citep{Mattei13:Preflib}.
\begin{table}[htp]
\centering
\caption{\small Satisfaction of $\CC$ and $\Par$ w.r.t.~IC and w.r.t.~315 Preflib profiles of linear orders under elections category. 
Experimental results  are presented in  Appendix~\ref{app:exp}. \label{tab:summary}}
{
\begin{tabular}{|@{\ }c@{\ }|  c @{\ } |  c@{\ } |@{\ }c@{\ } |@{\ }c@{\ }| c@{\ } |c |@{\ }c@{\ } |@{\ }c@{\ } |@{\ }c @{\ }|@{\ }c @{\ }|}
\cline{2-11}
\multicolumn{1}{@{}c@{}|}{}& \small  \bf Axiom &\small \bf  Plu. &\small \bf  Borda & \small\bf  Veto&\small\bf   STV &\small\bf  Black &\small\bf   MM&\small \bf  Sch.&\small\bf    RP&\small\bf    Copeland$_{0.5}$\\
\cline{2-11}
\hline
\multirow{ 2}{*}{\bf Theory} &  $\CC$  & \multicolumn{4}{@{}c@{}|}{   $\Theta(1)\wedge(1-\Theta(1))$} & \multicolumn{5}{@{  }c@{ }|}{  always satisfied}\\
\cline{2-11}
  & $\Par$ & \multicolumn{3}{@{}c@{}|}{ always satisfied} & \multicolumn{6}{@{}c@{}|}{  $1-\Theta\left(n^{-0.5}\right)$}\\
\hline
\hline
\multirow{ 2}{*}{\bf  Preflib} &$\CC$ &    96.8\%  &  92.4\% &  74.2\% & 99.7\% &  100\% &  100\%  &  100\% &  100\% &  100\% \\
\cline{2-11}
&$\Par$  & 100\% &  100\% &  100\% & 99.7\% & 99.4\% & 100\%  &  100\% &  100\% &  99.7\% \\
\hline
\end{tabular}
}
\end{table}

Table~\ref{tab:summary} provides a more quantitative way of comparing voting rules. Suppose the decision maker puts 50\% weight (or any fixed non-zero ratio) on both $\CC$ and $\Par$, and assume that the preferences are generated from IC. Then, when $n$ is sufficiently large, the last five voting rules in the table (that satisfy $\CC$) outperform  the first five voting rules in the table (the first four satisfies $\Par$).


 \noindent{\bf Beyond $\CC$ and $\Par$.} Theorems~\ref{thm:sCC-scoring}--\ref{thm:sPar-Cond-Pos}  are proved by (non-trivial) applications of a {\em categorization lemma} (Lemma~\ref{lem:categorization}), which characterizes smoothed satisfaction of a large class of axioms that can be represented by unions of finitely many polyhedra, including $\CC$ and $\Par$. We believe that Lemma~\ref{lem:categorization} is a promising tool for analyzing other axioms in future work. 


\subsection{Related Work and Discussions} 
\label{sec:related-work}
\noindent {\bf The  Condorcet criterion ($\CC$)} was proposed by~\citeauthor{Condorcet1785:Essai} in 1785~\citep{Condorcet1785:Essai}, has been one of the most classical and well-studied  axioms, and has {\em ``nearly universal acceptance''}~\cite[p.~46]{Saari1995:Basic}. $\CC$ is satisfied by many commonly-studied voting rules,  except positional scoring rules~\citep{Fishburn74:Paradoxes}  and multi-round-score-based elimination rules, such as STV.  
 Most previous work focused on characterizing the {\em Condorcet efficiency}, which is the probability for the Condorcet winner to win  conditioned on its existence~\citep{Fishburn1974:Simple,Fishburn1974:Aspects,Paris1975:Plurality,Gehrlein1978:Coincidence,Newenhizen1992:The-Borda}. 
 Beyond positional scoring rules, the study was mostly based  on computer simulations, see, e.g., \citep{Fishburn1976:An-analysis,Fishburn1977:An-analysis,Merrill1985:A-statistical,Nurmi1992:An-Assessment}. 
 
\noindent{\bf The participation axiom ($\Par$)}  was motivated by the {\em no-show paradox}~\citep{Fishburn1983:Paradoxes} and was proved to be incompatible with $\CC$ for every $m\ge 4$~\citep{Moulin1988:Condorcets}. The likelihood of $\Par$ under commonly studied voting rules w.r.t.~IC was posed as an open question by~\citet{Berg1994:On-probability} in 1994,  and has been investigated  in a series of works including~\citep{Lepelley1996:The-likelihood,Lepelley2001:Scoring,Wilson2007:Probability}, see~\cite[Chapter~4.2.2]{Gehrlein2011:Voting}. In particular, \citet{Lepelley2001:Scoring} analyzed the likelihood of various no-show paradoxes for three alternatives under {\em scoring runoff rules}, which includes STV, w.r.t.~IC and other distributions, and ``{\em strongly believe that the no-show paradox is not an important flaw of the scoring run-off voting systems}''. 

\noindent{\bf Our work vs.~previous work on $\CC$ and $\Par$.} Our results address open questions  by~\citet{Berg1994:On-probability} about the likelihood of satisfaction of $\CC$ and $\Par$ in two dimensions: first, we conduct smoothed analysis, which extends i.i.d.~models and is believed to be significantly more general and realistic. 
Second, our results cover a wide range of voting rules whose likelihood of satisfaction under $\CC$ or $\Par$ even w.r.t.~IC were not  mathematically characterized before, including $\CC$ under STV, and $\Par$ under maximin, Copeland, ranked pairs, Schulze, and Black's rule.   
While all results in this paper assume that the number of alternatives $m$ is fixed,  they are already more general than many previous work that focused on $m=3$.


\noindent{\bf Smoothed analysis.} There is a large body of literature on the applications of smoothed analysis to computational problems~\citep{Spielman2009:Smoothed}. Its main idea, i.e., the worst average-case analysis, has been proposed and investigated in other disciplines as well. For example, it is the central idea in frequentist statistics (as in the {\em frequentist expected loss} and {\em minimax decision rules}~\citep{Berger85:Statistical}) and is also closely related to the {\em min of means} criteria in decision theory~\citep{Gilboa1989:Maxmin}.  

Recently, \citet{Baumeister2020:Towards} and~\citet{Xia2020:The-Smoothed} independently proposed to conduct smoothed analysis in social choice. We adopt the framework in the latter work, though our motivation and goal  are quite different. We aim at providing a comprehensive picture of  smoothed  satisfaction of voting axioms, while \citep{Xia2020:The-Smoothed} focused on analyzing smoothed likelihood of Condorcet's voting paradox and the ANR impossibility on {\em anonymity} and {\em neutrality}. 
On the technical level, while  Lemma~\ref{lem:categorization} is a straightforward corollary of~\cite[Theorem~2]{Xia2021:How-Likely},  applications of results like Lemma~\ref{lem:categorization} can be highly non-trivial and problem dependent as commented in~\citep{Xia2021:How-Likely}, which is the case of this paper.  We believe that Lemma~\ref{lem:categorization}'s main merit is conceptual, as it provides a general categorization of smoothed satisfaction of a large class of per-profile axioms beyond $\CC$ and $\Par$ for future work.  

\section{Preliminaries}
\label{sec:prelim}
For any  $q\in\mathbb N$, we let $[q]=\{1,\ldots,q\}$. Let $\ma=[m]$ denote the set of $m\ge 3$ {\em alternatives}. Let $\ml(\ma)$ denote the set of all linear orders over $\ma$. Let $n\in\mathbb N$ denote the number of agents (voters). Each agent uses a linear order $R\in\ml(\ma)$ to represent his or her preferences, called a {\em vote}, where $a\succ_R b$ means that the agent prefers alternative $a$ to alternative $b$. The vector of $n$ agents' votes, denoted by $P$, is called a {\em (preference) profile}, sometimes called an $n$-profile. The set of $n$-profiles for all $n\in\mathbb N$ is denoted by $\ml(\ma)^* = \bigcup_{n =1}^{\infty} \ml(\ma)^n$.  A {\em fractional} profile is a   profile $P$ coupled with a possibly non-integer and/or negative weight vector $\vec \omega_P=(\omega_R:R\in P)\in{\mathbb R}^{n}$ for the votes in $P$. It follows that a non-fractional profile is a fractional profile with uniform weight, namely $\vec \omega_P = \vec 1$.  Sometimes  the weight vector is omitted when it is clear from the context or when $\vec\omega_P=\vec 1$. 

For any (fractional) profile $P$, let $\hist(P)\in {\mathbb R}_{\ge 0}^{m!}$ denote the anonymized profile of $P$, also called the {\em histogram} of $P$, which contains the total weight of every  linear order in $\ml(\ma)$ according to $P$.  An {\em irresolute voting rule} $\cor:\ml(\ma)^*\ra (2^{\ma}\setminus \{\emptyset\})$ maps a profile to a non-empty set of winners in $\ma$. A {\em resolute} voting rule $r$ is a special irresolute voting rule that always chooses a single alternative as the (unique) winner. 
We say that a  voting rule $r$ is a {\em refinement} of another voting rule $\cor$, if for every profile $P$, $r(P)\subseteq \cor(P)$.

\noindent{\bf (Un)weighted majority graphs and (weak) Condorcet winners.}  For any (fractional) profile $P$ and any pair of alternatives $a,b$, let $ P[a\succ b]$ denote the total weight of votes in $P$ where $a$ is preferred to $b$. Let $\wmg(P)$ denote the {\em weighted majority graph} of $P$, whose vertices are $\ma$ and whose weight on edge $a\ra b$ is $w_P(a,b) = P[a\succ b] - P[b\succ a]$. Let $\umg(P)$ denote the  {\em unweighted majority graph}, which is the unweighted directed graph that is obtained from  $\wmg(P)$ by keeping the edges with strictly positive weights. Sometimes a distribution $\pi$ over $\ml(\ma)$ is viewed as a fractional profile, where for each $R\in\ml(\ma)$ the weight on $R$ is $\pi(R)$. In such cases, we let $\wmg(\pi)$ denote the weighted majority graph of the {fractional} profile represented by $\pi$.  

The {\em Condorcet winner} of a profile $P$ is the alternative that only has outgoing edges in $\umg(P)$. A {\em weak Condorcet winner} is an alternative that does not have incoming edges in $\umg(P)$. Let $\cwinner(P)$ and $\wcw(P)$ denote the set of Condorcet winners and weak Condorcet winners in $P$, respectively. Notice that $\cwinner(P)\subseteq \wcw(P)$ and $|\cwinner(P)|\le 1$. The domain of $\cwinner(\cdot)$ and $\wcw(\cdot)$ can be naturally extended to all weighted or unweighted directed graphs.

For example, a distribution $\hat\pi$, $\wmg(\hat\pi)$, and $\umg(\hat\pi)$ for $m=3$ are illustrated in Figure~\ref{fig:ex-m3}. 
We have $\cwinner(\hat\pi) = \emptyset$ and $\wcw(\hat\pi) = \{1,2\}$. As another example, let $\piuni$ denote the uniform distribution over $\ml(\ma)$. Then, the weight on every edge in $\wmg(\piuni)$ is $0$ and $\umg(\piuni)$ does not contain any edge. 
\begin{figure}[htp]
\noindent\resizebox{1\textwidth}{!}
{
$\hat\pi = \left\{\begin{array}{ll}1\succ 2\succ 3 &\text{w.p. }1/4\\
2\succ 1\succ 3 &\text{w.p. }1/4\\
\text{other ranking} &\text{w.p. }1/8\\\end{array}\right.\Longrightarrow \wmg(\hat\pi) = $
\begin{minipage}{0.15\linewidth}
\includegraphics[width = \linewidth]{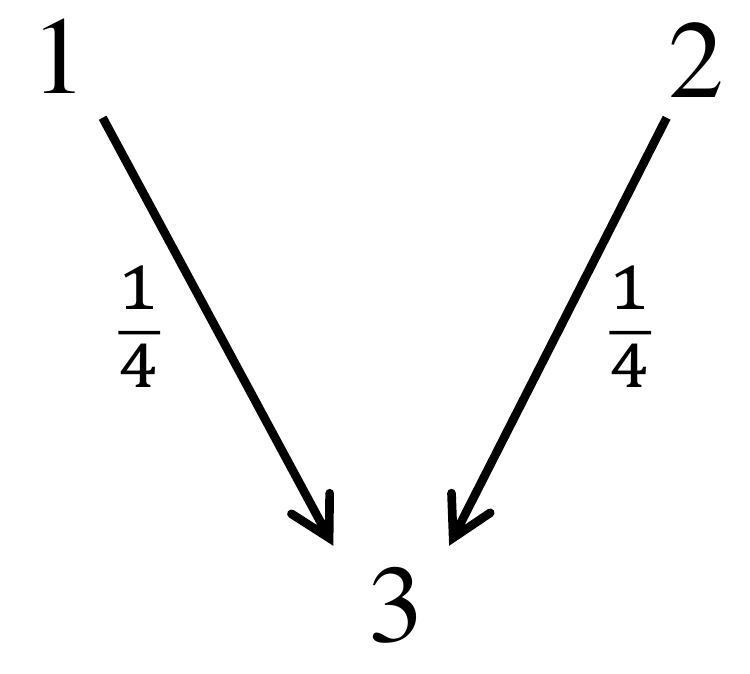}
\end{minipage}
$\Longrightarrow \umg(\hat\pi) =$
\begin{minipage}{0.15\linewidth}
\includegraphics[width = \linewidth]{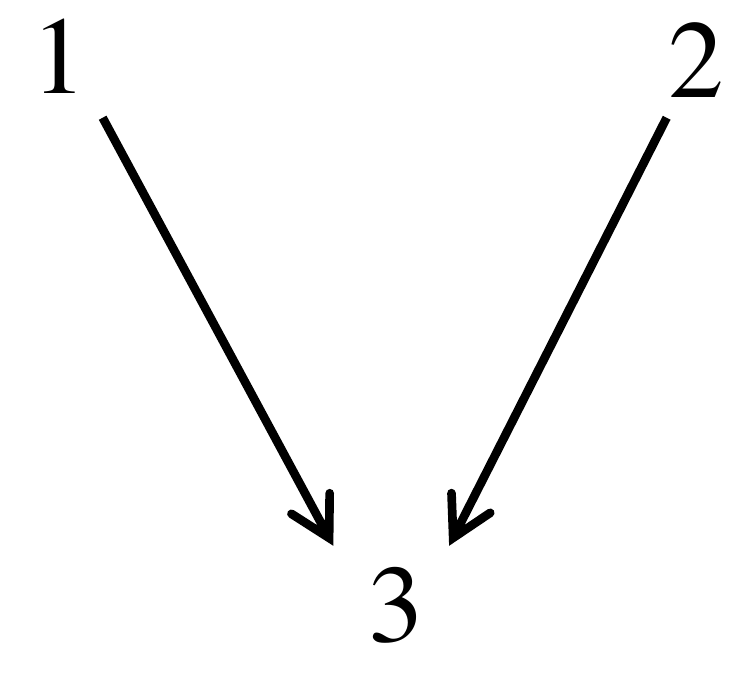}
\end{minipage}
} 
\caption{\small  $\hat\pi$, $\wmg(\hat\pi)$ (only positive edges are shown), and $\umg(\hat\pi)$.\label{fig:ex-m3}}
\end{figure}

Due to the space constraint, we  focus on presenting smoothed $\CC$ on positional scoring rules and MRSE rules in the main text, whose irresolute versions are defined below.  Their resolute versions can be obtained by applying a tie-breaking mechanism  on the co-winners. See Section~\ref{app:more-rule} for definitions of other rules studied in  Section~\ref{sec:CC} for $\Par$. 

\noindent{\bf Integer positional scoring rules.}  An {\em (integer) positional scoring rule}  $\cor_{\vec s}$  is characterized by an integer scoring vector $\vec s=(s_1,\ldots,s_m)\in{\mathbb Z}^m$ with $s_1\ge s_2\ge \cdots\ge s_m$ and $s_1>s_m$. For any alternative $a$ and any linear order $R\in\ml(\ma)$, we let $\vec s(R,a)=s_i$, where $i$ is the rank of $a$ in $R$. Given a profile $P$ with weights $\vec \omega_P$,  the  positional scoring rule $\cor_{\vec s}$ chooses all alternatives $a$ with maximum $\sum_{R\in P}\omega_R\cdot \vec s(R,a)$. For example, {\em plurality} uses the scoring vector $(1,0,\ldots,0)$, {\em Borda} uses the scoring vector $(m-1,m-2,\ldots,0)$, and {\em veto} uses the scoring vector $(1,\ldots,1,0)$. 



\noindent{\bf Multi-round score-based elimination (MRSE) rules.}  An irresolute MRSE  rule $\cor$ for $m$ alternatives is defined by a vector of  $m-1$ rules $(\cor_2,\ldots, \cor_{m})$, where for every $2\le i\le m$, $\cor_i$ is a  positional scoring  rule over $i$ alternatives that outputs a {\em total preorder} over them in the decreasing order of   their scores. Given a profile $P$, $\cor(P)$ is selected in $m-1$ rounds. For each $1\le i\le m-1$, in round $i$, a loser (an alternative with the lowest score) under $\cor_{m+1-i}$ is eliminated. 
We  use the {\em parallel-universes tie-breaking (PUT)}~\citep{Conitzer09:Preference} to select winners---an alternative $a$ is a winner if there is a way to break ties among the losers in each round, so that $a$ is the remaining alternative after $m-1$ rounds. If an MRSE rule $\cor$ only uses integer position scoring rules, then it is called an {\em int-MRSE rule}. Commonly studied int-MRSE rules include {\em STV}, which uses  plurality in each round, {\em Coombs}, which uses veto   in each round, and {\em Baldwin's rule}, which uses Borda in each round. 

\begin{ex}[\bf Irresolute STV]
\label{ex:PUTSTV}
 Figure~\ref{fig:PUfigure} illustrates the execution of irresolute STV, 
\begin{minipage}[t][][b]{0.25\textwidth}
denoted by $\istv$, under $\piuni$ (the uniform distribution) and $\hat\pi$ (the distribution in Figure~\ref{fig:ex-m3}), where each node represents the (tied) losers of the corresponding round, and each edge represents the loser to be eliminated. We have $\istv(\piuni) = \{1,2,3\}$ and $\istv(\hat\pi) = \{1,2\}$.
\end{minipage}
\ \ 
\begin{minipage}[t][][b]{0.7\textwidth}
\includegraphics[width =  \textwidth]{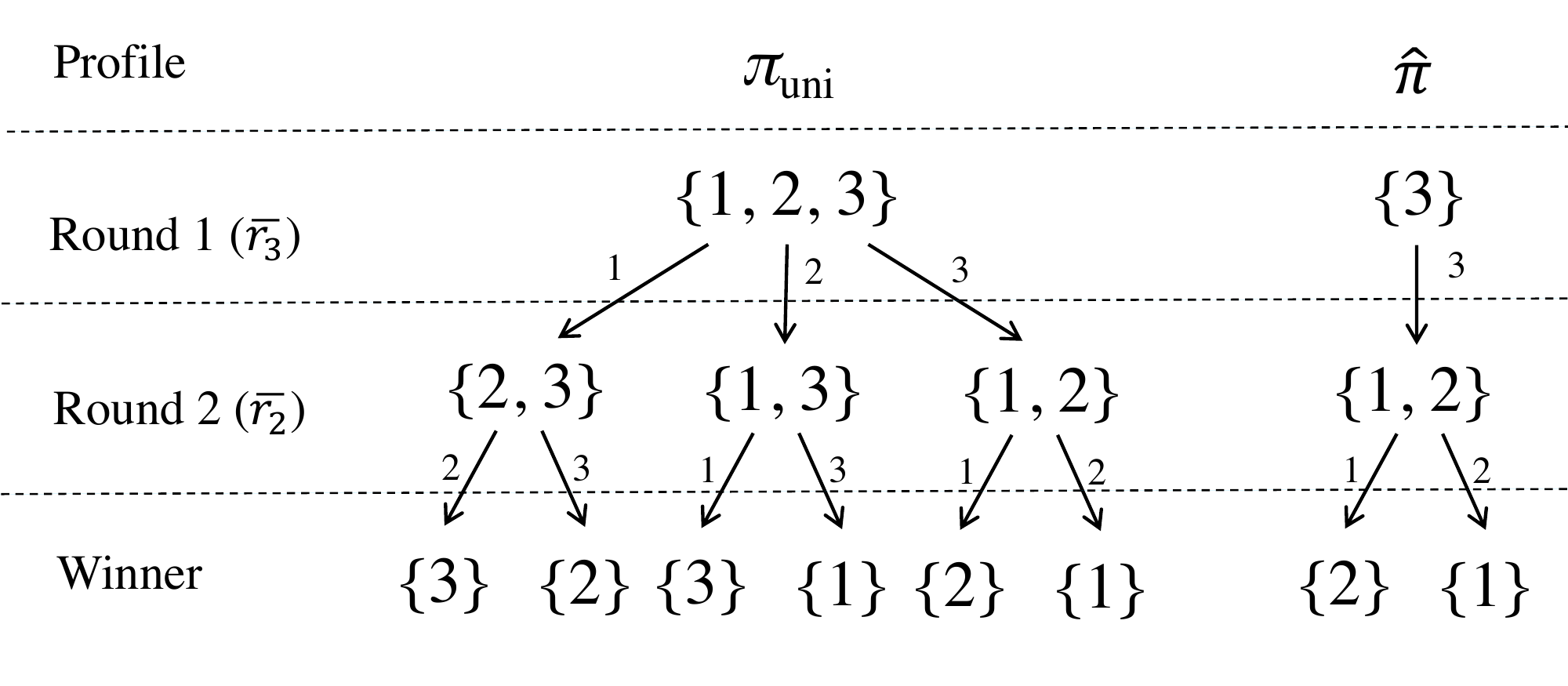}
\captionof{figure}{\small $\istv$ under $\piuni$ and $\hat\pi$ (defined in Figure~\ref{fig:ex-m3}).}
\label{fig:PUfigure}
\end{minipage}
\end{ex}

%


\noindent{\bf Axioms of voting.}  We focus on {\em per-profile axioms}~\citep{Xia2020:The-Smoothed} in this paper. A per-profile  axiom   is defined as a function $\sat{X}$ that maps a voting rule $\cor$ and a profile $P$ to $\{0,1\}$, where $0$ (respectively $1$)  means that $\cor$ dissatisfies/violates (respectively, satisfies) the axiom  at  $P$. Then, the classical (worst-case) satisfaction of the axiom under $\cor$ is defined to be $\min_{P\in \ml(\ma)^*} \sat{X}(\cor,P)$. 

For example, a (resolute or irresolute)  rule $\cor$ satisfies $\CC$, if  $\min_{P\in \ml(\ma)^*} \sat{\CC}(\cor,P)  =1$, where $\sat{\CC}(\cor,P)=1$ if and only if either (1) there is no Condorcet winner under $P$, or (2) the Condorcet winner is a co-winner of $P$ under $\cor$. A resolute  rule $r$ satisfies $\Par$, if $\min_{P\in \ml(\ma)^*} \sat{\Par}(r,P)  =1$, where 
$\sat{\Par}(r,P)=1$ if and only if no voter has incentive to abstain from voting. 
Formally, let $P=(R_1,\ldots,R_n)$, then  
$ [\sat{\Par}(r,P)=1]\Longleftrightarrow \left[\forall j\le n, r(P)\succeq_{R_j} r(P -R_j)\right],$ 
where $P -R_j$ is the $(n-1)$-profile that is obtained from $P$ by removing the $j$-th vote. 
For any pair of alternatives $a$ and $b$, we write $\{a\}\succeq _{R_j} \{b\}$ if and only if agent $j$, whose preferences are $R_j$, prefers $a$ to $b$. 
See Appendix~\ref{app:per-profile-axioms} for a  list of $13$ well-studied per-profile axioms and one non-per-profile axiom.


\noindent{\bf Smoothed satisfaction of axioms.}   Given a per-profile axiom $X$, a set $\Pi$ of distributions over $\ml(\ma)$, a voting rule $\cor$, and $n\in\mathbb N$, the {\em smoothed satisfaction of $X$} under $\cor$ with $n$ agents, denoted by $\satmin{X}{\Pi}(\cor,n)$, is defined in Equation~(\ref{dfn:s-sat}) in the Introduction. We note that the ``min'' in the superscript means that the adversary aims at minimizing the satisfaction of $X$. 
Formally, $\Pi$ is part of the single-agent preference model defined as follows. 


\begin{dfn}[\bf Single-Agent Preference Model~\citep{Xia2020:The-Smoothed}]A {\em single-agent preference model} is denoted by $\mm=(\Theta,\ml(\ma),\Pi)$, where $\Theta$ is the parameter space, $\ml(\ma)$ is the sample space, and $\Pi$ consists of distributions indexed by $\Theta$. $\mm$ is {\em strictly positive} if there exists $\epsilon>0$ such that the probability of any linear order under any distribution in $\Pi$ is at least $\epsilon$.  
$\mm$ is {\em closed} if $\Pi$ (which is a subset of the probability simplex in $\mathbb R^{m!}$) is a closed set in $\mathbb R^{m!}$. 
\end{dfn}
Example~\ref{ex:sCC-plu} illustrates a strictly positive and closed single-agent preference model for $m=3$, where $\Pi=\{\pi^1,\pi^2\}$ and $\epsilon = 1/8$. Other examples can be found in~\cite[Example~2 in the appendix]{Xia2020:The-Smoothed}. 




\section{The Smoothed Satisfaction of $\CC$ and $\Par$}
\label{sec:CC}


\noindent{\bf Smoothed $\CC$ under Integer Positional Scoring Rules.}
To present the results, we first define {\em almost Condorcet winners (ACW)} of a profile $P$, which are the two alternatives (whenever they exist) that are tied in the UMG and beat all other alternatives in head-to-head competitions. 
\begin{dfn}[\bf Almost Condorcet Winners]
\label{dfn:almost-Condorcet-winners}
For any unweighted directed graph $G$ over $\ma$, a pair of alternatives $a,b$ are {\em almost Condorcet winners (ACWs)}, denoted by $\almostCW(G)$, if (1) $a$ and $b$ are tied in $G$, and (2) for any other alternative $c\notin\{a,b\}$, $G$ has $a\ra c$ and $b\ra c$.  For any   profile $P$, let $\almostCW(P)=\almostCW(\umg(P))$.
\end{dfn}

For example, $1$ and $2$ are ACWs of $\hat\pi$ (as a fractional profile) in Figure~\ref{fig:ex-m3}.  By definition, for any profile $P$, $|\almostCW(P)|$ is either $0$ or $2$, and when it is $2$, $\wcwinner(P) = \almostCW(P)$. 

We now present a full characterization of  smoothed $\CC$ under integer positional scoring rules.

 

\begin{thm}[\bf \boldmath Smoothed $\CC$: Integer Positional Scoring Rules]
\label{thm:sCC-scoring}
For any fixed $m\ge 3$, let $\mm= (\Theta,\ml(\ma),\Pi)$ be a strictly positive and closed single-agent preference model, let $\cor_{\vec s}$ be an irresolute integer positional scoring rule, and let $r_{\vec s}$ be a refinement of $\cor_{\vec s}$. For any $n\ge 8m+49$ with $2\mid n$,  

\begin{centering}
{$\satmin{\CC}{\Pi}(r_{\vec s},n) = \left\{\begin{array}{@{}ll}
1- \exp(-\Theta(n)) &\text{if } \forall \pi\in\conv(\Pi),  |\wcw(\pi)|\times |\cor_{\vec s}(\pi)\cup \wcw(\pi)|\le 1\\
\Theta(n^{-0.5}) &\text{if}  
\begin{cases}
\text{(1) }\forall  \pi\in \conv(\Pi), \cwinner(\pi)\cap (\ma\setminus \cor_{\vec s}(\pi))=\emptyset \text{ and}\\
\text{(2) }  \exists \pi\in \conv(\Pi)\text{ s.t. } |\almostCW(\pi)\cap (\ma\setminus \cor_{\vec s}(\pi))|=2
\end{cases}
\\
\exp(-\Theta(n)) &\text{if }\exists \pi\in \conv(\Pi)\text{ s.t. } \cwinner(\pi)\cap (\ma\setminus \cor_{\vec s}(\pi))\ne \emptyset\\
\Theta(1)\wedge(1-\Theta(1)) &\text{otherwise}
\end{array}\right.$
}
\end{centering}

\noindent For any $n\ge 8m+49$ with $2\nmid n$, 

\begin{centering}
{$\satmin{\CC}{\Pi}(r_{\vec s},n) = \left\{\begin{array}{@{}ll}
1- \exp(-\Theta(n)) &\text{same as the }2\mid n\text{ case}\\
\exp(-\Theta(n)) &\text{if }\exists \pi\in \conv(\Pi)\text{ s.t.}
\begin{cases}
\text{(1) }\cwinner(\pi)\cap (\ma\setminus \cor_{\vec s}(\pi))\ne \emptyset \text{ or} \\
\text{(2) }  |\almostCW(\pi)\cap (\ma\setminus \cor_{\vec s}(\pi))|=2
\end{cases}
\\
\Theta(1)\wedge (1-\Theta(1)) &\text{otherwise}
\end{array}\right.$
}
\end{centering}
\end{thm}
\noindent{\bf Generality.}  We believe that Theorem~\ref{thm:sCC-scoring} is quite general, as it can be applied to {\em any} refinement of {\em any} irresolute integer positional scoring rule (i.e., using {\em any} tie-breaking mechanism) w.r.t.~{\em any} $\Pi$ that satisfies mild conditions.  The power of Theorem~\ref{thm:sCC-scoring} is that it converts complicated probabilistic arguments about smoothed $\CC$ to deterministic arguments about properties of (fractional) profiles in $\conv(\Pi)$, i.e., $\cor_{\vec s}(\pi)$, $\cwinner(\pi)$, $\almostCW(\pi)$, and $\wcw(\pi)$,  which are much easier to check. In particular, Theorem~\ref{thm:sCC-scoring} can be easily applied to i.i.d.~distributions (including IC) as shown in Example~\ref{ex:thm-sCC-pos} below.

\noindent{\bf Intuitive explanations of the conditions.} While the conditions for the cases in Theorem~\ref{thm:sCC-scoring} may appear technical, they have intuitive explanations. Take the $2\mid n$ case for example. 
{\bf \boldmath The $1- \exp(-\Theta(n))$ case} happens if every $\pi\in\conv(\Pi)$ is a ``robust'' instance  of $\CC$ satisfaction,
in the sense that after any small perturbation is introduced to $\pi$, it is still an instance of $\CC$ satisfaction. For {\bf \boldmath the $\Theta(n^{-0.5})$ case,} condition (1) states that every $\pi\in\conv(\Pi)$ is an instance of $\CC$ satisfaction, and condition (2) requires that some $\pi\in\conv(\Pi)$ corresponds to a ``non-robust'' instance of $\CC$ satisfaction, in the sense that after a small perturbation $\vec \eta$ is added to $\pi$, $\CC$ is violated at $\pi+\vec\eta$. 
{\bf \boldmath The $\exp(-\Theta(n))$ case} happens if there exists a ``robust'' instance of $\CC$ dissatisfaction $\pi\in\conv(\Pi)$, in the sense that after any small perturbation is introduced to $\pi$, it is still an instance of $\CC$ dissatisfaction. {\bf \boldmath The $\Theta(1)\wedge (1-\Theta(1))$} case holds if none of the other cases hold.
 
\noindent{\bf \boldmath Odd vs.~even $n$.} The $2\nmid n$ case  has similar explanations. The main difference  is that when $2\nmid n$, the UMG of any $n$-profile must be a complete graph. Therefore, when $\almostCW(\pi)\ne\emptyset$, with high probability an alternative in $\almostCW(\pi)$ is the Condorcet winner in the randomly-generated $n$-profile. Then, the $\Theta(n^{-0.5})$ case in  $2\mid n$ becomes part of the $\exp(-\Theta(n))$ case in $2\nmid n$.

\begin{ex} [\bf Applications of Theorem~\ref{thm:sCC-scoring} to plurality]
\label{ex:thm-sCC-pos}
In the setting of Example~\ref{ex:sCC-plu}, we apply Theorem~\ref{thm:sCC-scoring} to any sufficiently large $n$ with $2\mid n$ and any refinement of irresolute plurality, denoted by $\plu$, for the following sets of distributions.\\
$\bullet$ $\Pi=\{\pi^1,\pi^2\}$. We have $\satmin{\CC}{\Pi}(\plu,n)=\exp(-\Theta(n))$,   because let $\pi'  = \frac{3\pi^1+\pi^2}{4}$, we have $\cwinner(\pi')= \wcw(\pi')=\{2\}$, $\almostCW(\pi')=\emptyset$, and $\iplu(\pi' )=\{1\}$.\\
$\bullet$  $\Pi_{\text{IC}} = \{\piuni\} $, i.e.,  smoothed $\CC$ becomes likelihood of $\CC$ w.r.t.~IC.  We have $\satmin{\CC}{\Pi_{\text{IC}}}(\plu ,n)=\Theta(1)\wedge (1-\Theta(1))$, because $\cwinner(\piuni)=\emptyset$, $\wcw(\piuni)=\{1,2,3\}$, and $\almostCW(\piuni)=\emptyset$.

%

\end{ex}


\noindent{\bf Smoothed $\CC$ under int-MRSE Rules.} 
Smoothed $\CC$ under an MRSE rule $\cor$ depends on whether the positional scoring rules it uses satisfy the {\sc Condorcet loser ($\CL$)} criterion, which requires that the Condorcet loser, whenever it exists, never wins. The Condorcet loser is the alternative that loses to all head-to-head competitions. For any voting rule $\cor$, we write  $\sat{\CL}(\cor)=1$ if and only if $\cor$ satisfies {\sc Condorcet loser}. 

To present the result, we first define {\em parallel universes} under an MRSE rule $\cor$ at $\vec x\in\mathbb R^{m!}$, denoted by $\PU{\cor}{\vec x}$, to be the set of all elimination orders in the execution of $\cor$ at $\vec x$. Then, for any alternative $a$, let the {\em possible losing rounds}, denoted by $\PULR{\cor}{\vec x,a}\subseteq [m-1]$, be the set of all rounds in the parallel universes where $a$ drops out. The formal definitions can be found in Definition~\ref{dfn:PU-LR} in Appendix~\ref{app:proof-thm:sCC-MRSE}.

\begin{ex}
\label{ex:PU} In the setting of Example~\ref{ex:PUTSTV}, we let $\cor = \istv$.  $\PU{\istv}{\piuni}$ consists of   linear orders that correspond  to all paths from the root to leaves in Figure~\ref{fig:PUfigure}. Therefore,  $\PU{\istv}{\piuni} = \ml(\ma)$. For every $a\in\ma$, $\PULR{\istv}{\piuni,a}$ corresponds to the rounds where $a$ is in a node of that round in Figure~\ref{fig:PUfigure}. Therefore, for every $a\in\ma$, we have $\PULR{\istv}{\piuni,a}= \{1,2\}$.

For $\hat \pi$ in Figure~\ref{fig:ex-m3}, we have:
 $\PU{\istv}{\hat\pi} = \{[3\rhd 1\rhd 2], [3\rhd 2\rhd 1]\}\footnote{We use $\rhd$ to indicate the elimination order to avoid confusion with $\succ$.} , \PULR{\istv}{\hat\pi,1} =\PULR{\istv}{\hat\pi,2}  = \{2\},\text{ and }\PULR{\istv}{\hat\pi,3}  = \{1\}$.
\end{ex}
We are now ready to present the $2\mid n$ case of our characterization of smoothed $\CC$ under MRSE rules.  The full version  can be found in Appendix~\ref{app:proof-thm:sCC-MRSE}.

%

\begin{thm}[\bf \boldmath Smoothed $\CC$: int-MRSE rules, $2\mid n$]
\label{thm:sCC-MRSE}
For any fixed $m\ge 3$, let $\mm= (\Theta,\ml(\ma),\Pi)$ be a strictly positive and closed single-agent preference model, let $\cor =(\cor_2,\ldots,\cor_{m})$ be an  int-MRSE rule and let $r$ be a refinement of $\cor $. For any $n\in\mathbb N$ with $2\mid n$, we have
 
\begin{centering}
{$\satmin{\CC}{\Pi}(r,n) = \left\{\begin{array}{@{}ll}
1 &\text{if \ } \forall 2\le i\le m, \sat{\CL}(\cor_i)=1\\
1- \exp(-\Theta(n)) &\text{if} 
\begin{cases}
\text{(1) }\exists 2\le i\le m\text{ s.t. }\sat{\CL}(\cor_i)=0\text{ and }\\
\text{(2) } \forall \pi\in\conv(\Pi), \forall  a\in \wcw(\pi), \forall i^*\in \PULR{\cor}{\pi,a}, \\
\hfill \text{we have } \sat{\CL}(\cor_{m+1-i^*})=1 
\end{cases}
\\
\Theta(n^{-0.5}) &\text{if}  
\begin{cases}
\text{(1) }\forall  \pi\in \conv(\Pi), \cwinner(\pi)\cap (\ma\setminus \cor (\pi))=\emptyset \text{ and} \\
\text{(2) }  \exists \pi\in \conv(\Pi)\text{ s.t. } |\almostCW(\pi)\cap (\ma\setminus \cor (\pi))|=2
\end{cases}
\\
\exp(-\Theta(n)) &\text{if \ }\exists \pi\in \conv(\Pi)\text{ s.t. } \cwinner(\pi)\cap (\ma\setminus \cor (\pi))\ne \emptyset\\
\Theta(1)\wedge(1-\Theta(1)) &\text{otherwise}
\end{array}\right.$
}
\end{centering}
\end{thm}

The most interesting cases are the $1$ case and the $1- \exp(-\Theta(n))$ case. The $1$ case happens when  all positional scoring rules used in $\cor$ satisfy {\sc Condorcet loser}. In this case, if the Condorcet winner exists, then it cannot be a loser in any round, which means that it is the unique winner under $\cor$. 
 The $1- \exp(-\Theta(n))$ case happens when (1) the $1$ case does not happen, and (2) for every distribution $\pi\in\conv(\Pi)$, every weak Condorcet winner $a$, and every possible losing round $i^*$ for $a$, the positional scoring rule used in round $i^*$, i.e.~$\cor_{m+1-i^*}$, must satisfy {\sc Condorcet loser}. (2)  guarantees that  when a small permutation is added to $\pi$, if a weak Condorcet winner $a$ becomes the Condorcet winner, then it will be the unique winner under $\cor$. 

\begin{ex}[\bf Applications of Theorem~\ref{thm:sCC-MRSE} to STV]\label{ex:CC-MRSE}
In the setting of Example~\ref{ex:PU},  let  $\stv$ denote an arbitrary refinement of $\istv=(\cor_2,\cor_3)$. The 1 case does not hold for sufficiently large $n$, because $\cor_3$ (plurality) does not satisfy {\sc Condorcet loser}.

When $\Pi_{\text{IC}}= \{\piuni\}$, Theorem~\ref{thm:sCC-MRSE} implies that for any sufficiently large $n$ with $2\mid n$, the $\Theta(1)\wedge(1-\Theta(1))$ case holds. The $1- \exp(-\Theta(n))$ case does not hold, because its condition (2) fails:  $1\in \wcw(\piuni)$ and round $1$ is a possible losing round for alternative $1$ (i.e., $1\in \PULR{\istv}{\piuni,1}$), yet $\cor_3$ does not satisfy {\sc Condorcet loser}. The $\Theta(n^{-0.5})$ case does not hold, because its condition (2) fails: $\almostCW(\piuni)=\emptyset$. The $\exp(-\Theta(n))$ case does not hold because $\cwinner(\piuni) = \emptyset$.

\end{ex}


Like Theorem~\ref{thm:sCC-scoring}, Theorem~\ref{thm:sCC-MRSE} can also be easily applied to  i.i.d.~distributions. Like Example~\ref{ex:CC-MRSE}, we have the following corollary w.r.t.~IC, which corresponds to $\Pi_{\text{IC}} = \{\piuni\}$. 

\begin{coro}[\bf Likelihood of $\CC$ under int-MRSE rules w.r.t.~IC]
\label{Coro:sCC-STV-IC}
For any fixed $m\ge 3$, any refinement $r$ of any int-MRSE rule $\cor$, and any $n\in\mathbb N$, 

\begin{center}
$\Pr\nolimits_{P\sim(\piuni)^n}(\sat{\CC}(r,P)=1) = \left\{\begin{array}{ll}
1 &\text{if } \forall 2\le i\le m, \sat{\CL}(\cor_i)=1\\
\Theta(1)\wedge(1-\Theta(1)) &\text{otherwise}
\end{array}\right.$
\end{center}
\end{coro}

\noindent{\bf Proof sketches for Theorem~\ref{thm:sCC-scoring} and~\ref{thm:sCC-MRSE}.}  In light of various multivariate  central limit theorems (CLTs), 
when $n$ is large, the profile is approximately $n\cdot \pi^*$ for  $\pi^*=(\sum_{j=1}^n \pi_j)/n\in \conv(\Pi)$  with high probability. Despite this high-level intuition, the conditions of the cases are  quite differently from smoothed $\CC$ by definition. To see this, note that (i) the adversary may not be able to set any agent's ground truth preferences to be $\pi^*\in\conv(\Pi)$, because $\pi^*$ may not be in $\Pi$ as shown in Example~\ref{ex:thm-sCC-pos}, and (ii) in the definition of smoothed $\CC$, agent $j$'s vote is a random variable distributed as $\pi_j$, instead of the fractional vote $\pi_j$. Standard CLTs can probably be applied to prove the $1- \exp(-\Theta(n)) $ case and the $\Theta(1)\wedge(1-\Theta(1))$ case, but they are too coarse for other cases.

To address this challenge, we model the satisfaction of $\CC$ by the union of multiple polyhedra $\upoly$ as exemplified in Section~\ref{sec:cat-lemma}. This converts the smoothed $\CC$ problem to a {\em PMV-in-$\upoly$} problem~\citep{Xia2021:How-Likely} (Definition~\ref{dfn:PMV-in-C}). Then, we refine~\cite[Theorem~2]{Xia2021:How-Likely} to prove a categorization lemma (Lemma~\ref{lem:categorization}), and apply it to obtain Lemma~\ref{lem:sCC-GISR} that characterizes smoothed $\CC$ for a large class of voting rules called {\em generalized irresolute scoring rules (GISRs)}~\cite{Freeman2015:General,Xia15:Generalized} (Definition~\ref{dfn:GISR} in Appendix~\ref{app:dfn-GISR}). Finally, we apply Lemma~\ref{lem:sCC-GISR} to integer positional scoring rules and int-MRSE rules to obtain Theorem~\ref{thm:sCC-scoring} and Theorem~\ref{thm:sCC-MRSE}. The full proof can be found in Appendix~\ref{app:proof-thm:sCC-scoring} and~\ref{app:proof-thm:sCC-MRSE}, respectively.
$\hfill\Box$


\noindent{\bf The smoothed satisfaction of $\Par$.} Due to the space constraint, we   briefly introduce our characterizations of smoothed $\Par$ under commonly-studied voting rules defined in Appendix~\ref{app:more-rule}, 
which 
belong to a large class of voting rules called {\em generalized scoring rules (GSRs)}~\citep{Xia08:Generalized} (Definition~\ref{dfn:GISR} in Appendix~\ref{app:dfn-GISR}).  
Formal statements and  proofs of the theorems can be found in Appendix~\ref{app:proof-thm:sPar-mm-rp-sch}--\ref{app:proof-thm:sPar-Cond-Pos}.

\noindent{\bf Theorems~\ref{thm:sPar-mm-rp-sch},   \ref{thm:sPar-copeland}, \ref{thm:sPar-MRSE}, \ref{thm:sPar-Cond-Pos} (Smoothed $\Par$ under commonly-studied rules). }{\em For any fixed $m\ge 4$,  any GSR  $r$ that is a refinement of maximin, STV, Schulze,    ranked pairs, Copeland, any int-MRSE, or any Condocetified positional scoring rule, and any strictly positive and closed $\Pi$ over $\ml(\ma)$ with $\piuni\in \conv(\Pi)$, there exists   $N\in\mathbb N$ such that for every $n\ge N$,   $\satmin{ \Par}{\Pi}(r,n ) = 1-  \Theta({\frac{1}{\sqrt n}})$.}

In fact, if $\piuni\not\in \conv(\Pi)$, then smoothed $\Par$ converges to $1$ at a faster rate, which is more positive news, as shown in Lemma~\ref{lem:sPar-GSR} (Appendix~\ref{app:proof-lem:sPar-GSR}).

\section{Beyond $\CC$ and $\Par$: The Categorization Lemma}
\label{sec:cat-lemma}
In this section, we present a general   lemma that characterizes smoothed satisfaction of per-profile axioms that can be represented by unions of polyhedra, including $\CC$ and $\Par$. To develop intuition, we start with an example of modeling $\CC$ under irresolute plurality as the union of the following two types of polyhedra in $\mathbb R^{m!}$.

$\bullet$ {\bf\boldmath $\upolynoCW$}  represents that there is no Condorcet winner, which is the union of polyhedra $\ppoly{G}$, where $G$ is an unweighted graph over $\ma$ that does not have a Condorcet winner, as exemplified in Example~\ref{ex:HG}.

$\bullet$ {\bf\boldmath $\upolyCWwin$} represents that the Condorcet winner exists and also wins the plurality election, which is the union of polyhedra $\ppoly{a}$ for every $a\in \ma$, that represents $a$ being the Condorcet winner as well as a $\iplu$ co-winner, as exemplified in Example~\ref{ex:H-a}. 
\begin{ex} [\bf \boldmath $\ppoly{G}$]
\label{ex:HG}
 Let $m=3$ and let $x_{abc}$ denote the number of $[a\succ b\succ c]$ votes in a profile. The following figure shows  $G$ (left) and $\ppoly{G} $ (right). 

$G=$\begin{minipage}{0.15\linewidth}
\includegraphics[width = \textwidth]{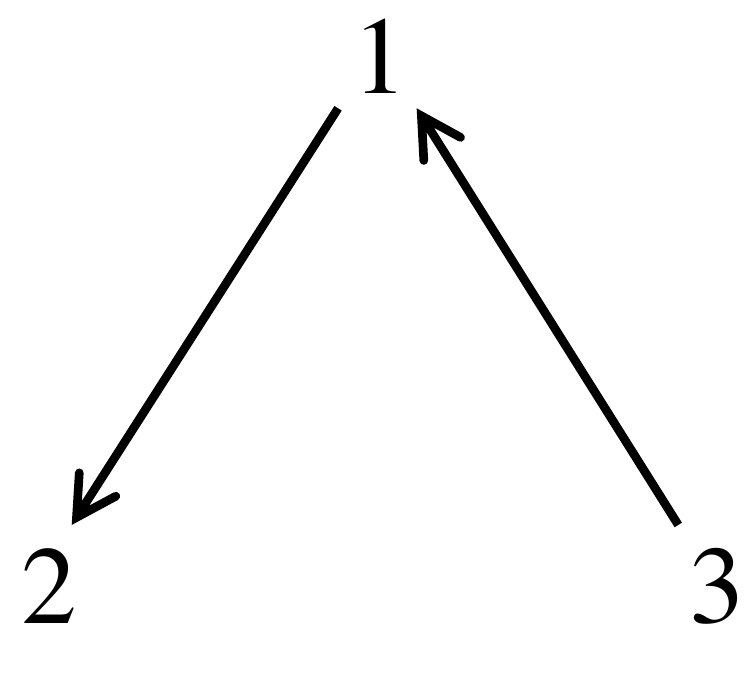}
\end{minipage}
$\Longleftrightarrow$
\begin{minipage}{0.7\linewidth}
\raggedleft
\begin{align}
& (x_{213} + x_{231} + x_{321}) - (x_{123} + x_{132} + x_{312}) \le -1 \label{ex:HG1}\\
& (x_{123} + x_{132} + x_{213}) - (x_{231} + x_{321} + x_{312}) \le -1 \label{ex:HG2}\\
& (x_{132} + x_{312} + x_{321}) - (x_{123} + x_{213} + x_{231}) \le 0 \label{ex:HG3}\\
& (x_{123} + x_{213} + x_{231}) - (x_{132} + x_{312} + x_{321}) \le  0\label{ex:HG4}
\end{align}
\end{minipage}

Among the four inequalities, (\ref{ex:HG1}) represents the $1\ra 2$ edge in $G$,  (\ref{ex:HG2}) represents the $3\ra 1$ edge in $G$, and (\ref{ex:HG3}) and (\ref{ex:HG4}) represent the tie between $2$ and $3$ in $G$.
\end{ex}
 
\begin{ex} [\bf \boldmath $\ppoly{a}$]
\label{ex:H-a} Let $m=3$.   $\ppoly{1} $ is the polyhedron represented by the following four inequalities:
\begin{align*}
  \begin{rcases}  & (x_{213} + x_{231} + x_{321}) - (x_{123} + x_{132} + x_{312}) \le -1 \\
& (x_{231} + x_{321} + x_{312}) - (x_{123} + x_{132} + x_{213})  \le -1  \end{rcases} & \ 1\text{ is the Condorcet winner } \\
 \begin{rcases}
& (x_{213} +  x_{231}  ) - (x_{123} +  x_{132}  )  \le 0  \\
& ( x_{321} + x_{312}) - (x_{123} +  x_{132}  )  \le 0
\end{rcases}& \ 1\text{ is a $\iplu$ co-winner } 
\end{align*}
\end{ex}

It is not hard to see that $\iplu$ satisfies \CC{} at a profile $P$ if and only if  $\hist(P)$ is in  $\upoly=\upolynoCW\cup \upolyCWwin$, where
$\upolynoCW = {\bigcup\nolimits_{G:\cwinner(G) = \emptyset}\ppoly{G}}  \text{ and }  \upolyCWwin= {\bigcup\nolimits_{a\in \ma} \ppoly{a}}$.  An  example of $\Par$ under Copeland can be found in Appendix~\ref{sec:modeling-par}.
 In general, the satisfaction of a wide range of axioms   can be represented by   unions  of finitely many polyhedra. 
Then, the smoothed satisfaction problem   reduces to the lower bound of the following PMV-in-$\upoly$ problem. 

\begin{dfn}[\bf\boldmath The PMV-in-$\upoly$ problem~\citep{Xia2021:How-Likely}]
\label{dfn:PMV-in-C}
Given $q,I\in\mathbb N$,  $\upoly = \bigcup_{i\le I}\cpoly{i}$, where $\forall i\le I$,  $\cpoly{i}\subseteq \mathbb R^q$ is a polyhedron, and a set $\Pi$ of distributions over $[q]$, we are interested in 
$$\text{\bf the upper bound }\sup\nolimits_{\vec\pi\in\Pi^n}\Pr(\vXp \in \upoly)\text{, and {\bf the  lower bound}}\inf\nolimits_{\vec\pi\in\Pi^n}\Pr(\vXp \in \upoly),$$
where $\vXp$ is the $(n,q)$-{\em Poisson multinomial variable (PMV)} that corresponds to the histogram of $n$ independent random variables distributed as $\vec\pi$.
\end{dfn}

%


See Example~\ref{ex:pmv} in Appendix~\ref{app:cat-lemma-formal} for an example of PMV. The following lemma provides an asymptotic characterization on the lower bound of the PMV-in-$\upoly$ problem.
\begin{lem}[\bf Categorization lemma, simplified]\label{lem:categorization} 
For any PMV-in-$\upoly$ problem and any $n\in\mathbb N$,  $\inf_{\vec\pi\in\Pi^n}\Pr (\vXp \in \upoly)$ is $0$, $\exp(-\Theta(n))$, $\text{poly}^{-1}(n)$, $\Theta(1)\wedge(1-\Theta(1))$, $1-\text{poly}^{-1}(n)$, $1- \exp(-\Theta(n))$, or $1$.
\end{lem}

The full version of Lemma~\ref{lem:categorization} (Appendix~\ref{app:cat-lemma-formal}) also  characterizes   the condition for each case, the degree of polynomial, and   $\sup_{\vec\pi\in\Pi^n}\Pr (\vXp \in \upoly)$.    
Lemma~\ref{lem:categorization}'s main merit is conceptual,  as it  categorizes the smoothed likelihood into  seven cases for quantitative comparisons, summarized in the increasing order in the table below, 
which are   {\em 0}, {\em very unlikely (VU)}, {\em unlikely (U)}, {\em medium (M)}, {\em likely (L)}, {\em very likely (VL)}, and {\em 1}. The first three cases ($0$, VU, U) are negative news, where the adversary can set the ground truth so that  the axiom is  almost surely violated  in large elections ($n\ra \infty$).  The last three cases (L, VL, and $1$) are positive news, because  the axiom is satisfied almost surely in large elections, regardless of the adversary's choice. The M case can be interpreted positively or negatively, depending on the context.
\begin{table}[htp]
\resizebox{1\textwidth}{!}{
\begin{tabular}{|@{\ }c@{\ }|@{\ }c@{\ }|@{\ }c@{\ }|@{\ }c@{\ }|@{\ }c@{\ }|@{\ }c@{\ }|@{\ }c@{\ }|@{\ }c@{\ }|}
\hline Name& \bf 0& {\bf  VU}& {\bf  U}& {\bf  M}& {\bf L}& {\bf  VL}& {\bf  1}\\
\hline Lem.~\ref{lem:categorization}& 0 & $\exp(-\Theta(n))$& $\text{poly}^{-1}(n)$& $\Theta(1)\wedge(1-\Theta(1))$& $1-\text{poly}^{-1}(n)$ & $1- \exp(-\Theta(n))$ & 1\\
\hline
\end{tabular}
}
\end{table}
 


%
%
%

\section{Future work} 
There are many  open questions for future work. What are the smoothed $\CC$ and smoothed $\Par$ for voting rules not sutdied in this paper, such as Bucklin? What is the smoothed satisfaction of $\Par$ when a group of agents can simultaneously abstain from voting~\cite{Lepelley2001:Scoring}?  More generally, we believe that drawing a comprehensive picture of smoothed satisfactions of other voting axioms and/or paradoxes, such as those described in Appendix~\ref{app:per-profile-axioms},  is an important, promising, and challenging mission, and the categorization lemma (Lemma~\ref{lem:categorization}) can be a useful conceptual and technical tool to start with.

\subsection*{Acknowledgments}
We thank   anonymous reviewers for helpful  comments. This work is supported by NSF \#1453542, ONR \#N00014-17-1-2621, and a gift fund from Google.  
\bibliographystyle{plainnat}
\bibliography{/Users/administrator/GGSDDU/references}
\newpage
\tableofcontents
\newpage

\appendix

\section{Definitions of More Voting Rules}
\label{app:more-rule}

\paragraph{\bf WMG-based rules.} A voting rule is said to be {\em weighted-majority-graph-based (WMG-based)} if its winners only depend on the WMG of the input profile. In this paper we consider the following commonly-studied WMG-based irresolute rules.
\begin{itemize}
\item {\bf Copeland.} The Copeland rule is parameterized by a number $0\le \alpha\le 1$, and is therefore denoted by  $\icopeland$. For any   profile $P$, an alternative $a$ gets $1$ point for each other alternative it beats in  head-to-head competitions, and gets $\alpha$ points for each tie.  $\icopeland$ chooses all alternatives with the highest total score as   winners. 
\item  {\bf Maximin.} For each alternative $a$, its {\em min-score}  is defined to be $\minscore_P(a)=\min_{b\in\ma}w_P(a,b)$. Maximin, denoted by $\imaximin$, chooses all alternatives with the max min-score as  winners.
\item {\bf Ranked pairs.} Given a profile $P$, an alternative $a$ is a winner under ranked pairs (denoted by $\irp$) if there exists a way to fix edges in $\wmg(P)$ one by one in a non-increasing order w.r.t.~their weights (and sometimes break ties), unless it creates a cycle with previously fixed edges, so that after all edges are considered, $a$ has no incoming edge. This is known as the {\em parallel-universes tie-breaking (PUT)}~\citep{Conitzer09:Preference}. 
\item {\bf Schulze.} The {\em strength} of any directed path in the WMG  is defined to be the minimum weight on   single edges along the path. For any pair of alternatives $a,b$, let $s[a,b]$ denote the highest weight among all paths from $a$ to $b$. Then, we write $a\succeq b$ if and only if $s[a,b]\ge s[b,a]$, and~\citet{Schulze11:New} proved that the strict version of this binary relation, denoted by $\succ$, is transitive. The Schulze rule, denoted by $\ischulze$, chooses all alternatives $a$ such that for all other alternatives $b$, we have $a\succeq b$. 
\end{itemize}
\paragraph{\bf Condorcetified (integer) positional scoring rules.}  The rule is defined by an integer scoring vector $\vec s\in{\mathbb Z}^m$ and is denoted by $\iCondorcet{ \vec s}$, which selects the Condorcet winner when it exits, and otherwise uses $\cor_{\vec s}$ to select the (co)-winners. For example, {\em Black's rule}~\citep{Black58:Theory} is the Condorcetified Borda rule.

\section{Per-Profile and Non-Per-Profile Axioms}
\label{app:per-profile-axioms}
In this section, we provide an (incomplete) list of $14$ commonly-studied  per-profile axioms and one commonly-studied non-per-profile axiom that we do not see a clear per-profile representation.

\paragraph{\bf Per-Profile Axioms.}
We present the definitions of the per-profile axioms in the alphabetical order. Their equivalent $\sat{X}$ definition is often straightforward unless explicitly discussed below.
\begin{enumerate}
\item {\sc Anonymity}  states that the winner is insensitive to the identities of the voters. It is a per-profile axiom as shown in~\citep{Xia2020:The-Smoothed}.
\item {\sc Condorcet criterion} is a per-profile axiom as discussed in the Introduction.
\item {\sc Condorcet loser} requires that a {\em Condorcet loser}, which is the alternative who {\em loses} to every head-to-head competition with other alternatives, should not be selected as the winner. It is a per-profile axiom in the same sense as $\CC$.
\item {\sc Consistency} requires that for any profile $P$ and any sub-profile $P'$ of $P$,  if $r(P') = r(P\setminus P')$, then $r(P) = r(P')$. Therefore, for any profile $P$, we can define
$$\left[\sat{\sc Consistency}(r,P)=1\right]\Longleftrightarrow \left[\forall P'\subset P,   \left[ r(P') = r(P\setminus P') \right]\Rightarrow \left[r(P) =r(P')\right] \right]$$

\item {\sc Group-Non-Manipulable} is defined similarly to {\sc Non-Manipulable} below, except that multiple voters are allowed to simultaneously change their votes, and after doing so, at least one of them strictly prefers the old winner.
\item {\sc Independent of clones} requires that the winner does not change when {\em clones} of an alternative is introduced. The clones and the original alternative must be ranked consecutively in each vote.  Let ${\sc IoC}$   denote {\sc Independent of clones}. For any profile $P$, we let $\sat{\sc IoC}(r,P) =1$ if and only if for every alternative $a$ and every profile $P'$ obtain from $P$ by introducing clones of $a$, we have $r(P) = r(P')$.


\item {\sc Majority criterion} requires that any alternative that is ranked at the top place in more than $50\%$ of the votes must be selected as the winner.  {\em Majority criterion} is stronger than {\sc Condorcet criterion}.
\item {\sc Majority loser} requires that any alternative who is ranked at the bottom place in more than $50\%$ of the votes should not be selected as the winner. {\sc Majority loser} is weaker than {\sc Condorcet loser}.
\item {\sc Monotonicity} requires raising up the position of the current winner in any vote will not cause it to lose. Let {\sc Mono} denote {\sc Monotonicity}. One way to define $\sat{\sc Mono}$ is the following.Let $\sat{\sc Mono}^1(r,P) = 1$ if and only if for every profile $P'$ that is obtained from $P$ by raising the position of $r(P)$ in one vote, we have $r(P')=r(P)$. Another definition is: $\sat{\sc Mono}^2(r,P) = 1$ if and only if for every profile $P'$ that is obtained from $P$ by raising the position of $r(P)$ in arbitrarily many votes, we have $r(P')=r(P)$. Notice that the classical (worst-case) {\sc Monotonicity} is satisfied if and only if $\min_{P}\sat{\sc Mono}^1(r,P) =1$ or equivalently, $\min_{P}\sat{\sc Mono}^2(r,P) =1$. The smoothed satisfaction of $\min_{P}\sat{\sc Mono}^1$ might be different from $\min_{P}\sat{\sc Mono}^2$, which is beyond the scope of this paper.

\item {\sc Neutrality} states that the winner is insensitive to the identities of the alternatives. It is a per-profile axiom as shown in~\citep{Xia2020:The-Smoothed}.

\item {\sc Non-Manipulable} requires that no agent has incentive to unilaterally change his/her vote to improve the winner w.r.t.~his/her true preferences. More precisely, for any profile $P=(R_1,\ldots,R_n)$,  we have
$$\left[\sat{\sc Non-Manipulable}(r,P)=1\right]\Leftrightarrow \left[\forall j\le n, \forall R_j'\in\ml(\ma), r(P)\succeq_{R_j} r(P\cup\{R_j'\}\setminus \{R_j\})\right]$$
\item {\sc Participation}  is a per-profile axiom as discussed in the Introduction.

\item {\sc Reversal symmetry} requires that the winner of any profile should not be the winner when all voters' rankings are inverted.
\end{enumerate}

\paragraph{\bf Non-Per-Profile Axiom(s).} We were not able to model {\sc Non-Dictatorship (ND)}  as a per-profile axiom studied in this paper. A voting rule is not a dictator if for each $j\le n$, there exists a profile $P$ whose winner is not ranked at the top of agent $j$'s preferences. 

\section{Materials for Section~\ref{sec:cat-lemma}: The Categorization Lemma}
\label{app:categorization-lemma}
While the categorization lemma (Lemma~\ref{lem:categorization}) was presented after Theorems~\ref{thm:sCC-scoring} through~\ref{thm:sPar-Cond-Pos} in the main text, the proofs of the theorems depend on the lemma. Therefore, we present materials for the categorization letter before the proofs for the theorems in the appendix.

 
\subsection{Modeling Satisfaction  of $\Par$ as A Union of Polyhedra}
\label{sec:modeling-par}

\begin{wrapfigure}{R}{0.35\textwidth}
\centering
\includegraphics[width = \linewidth]{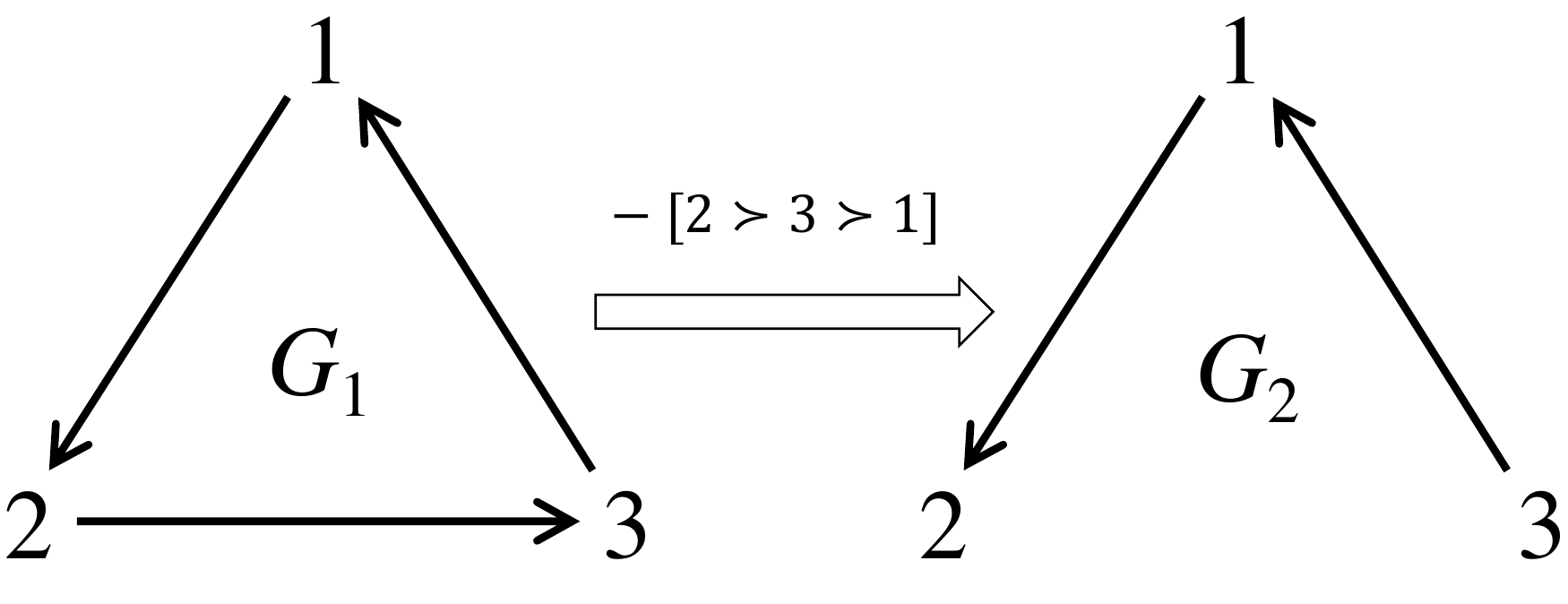}
\caption{\small $G_1$, $G_2$, and $R$. \label{fig:ex-Par-Copeland}}
\end{wrapfigure}

\paragraph{\bf \boldmath $\Par$ under Copeland$_\alpha$.} 
We now show how to approximately model the satisfaction of $\Par$ under Copeland$_\alpha$. For every pair of unweighted directed graphs $G_1,G_2$ over $\ma$ and every $R\in\ml(\ma)$, we define a polyhedron $\ppoly{G_1,R,G_2}$  to represent the histograms of profile $P$ that contains an $R$-vote,  $G_1=\umg(P)$, and $G_2=\umg(P\setminus\{R\})$. 
The linear inequalities used to specify the UMGs of $P$ and $(P\setminus\{R\})$ are similar to $\ppoly{G}$ defined above, as illustrated in the following example.

\begin{ex} Let $m=3$,  $R = [2\succ 3\succ 1]$, and let $G_1,G_2$ denote the graphs  in Figure~\ref{fig:ex-Par-Copeland}. $\ppoly{G_1,R,G_2}$ is represented by the following  inequalities.


\begin{align}
- x_{231} \le -1 \ \ \ & \label{ex:HParG3} \\
\begin{rcases}
&(x_{213} + x_{231} + x_{321}) - (x_{123} + x_{132} + x_{312})\le -1 \\
&(x_{123} + x_{132} + x_{213}) - (x_{231} + x_{321} + x_{312}) \le -1 \\
&(x_{132} + x_{312} + x_{321}) - (x_{123} + x_{213} + x_{231}) \le -1  
\end{rcases}&  \label{ex:HParG1} \\
\begin{rcases}
& (x_{213} + x_{231}-1 + x_{321}) - (x_{123} + x_{132} + x_{312}) \le -1\\
& (x_{123} + x_{132}+ x_{213}) - (x_{231}-1  + x_{321} + x_{312}) \le -1\\
& (x_{132} + x_{312} + x_{321}) - (x_{123} + x_{213} + x_{231}-1) \le 0\\
& (x_{123} + x_{213} + x_{231}-1) - (x_{132} + x_{312} + x_{321}) \le  0
\end{rcases}&  \label{ex:HParG2}
\end{align}

(\ref{ex:HParG3}) guarantees that  $P$ contains an $R$-vote.  
The three inequalities in (\ref{ex:HParG1}) represent  $\umg(P) = G_1$, and the four inequalities in (\ref{ex:HParG2}) represent  $\umg(P) = G_2$. 
\end{ex}
We do not require $x_R$'s to be non-negative, which does not affect the results of the paper, because the histograms of randomly-generated profiles are always non-negative.

By enumerating  $G_1$, $R$, and $G_2$ that correspond to a violation of $\Par$, the polyhedra that represent  satisfaction of $\Par$  under Copeland$_\alpha$ are:
$$
\upoly = \bigcup\nolimits_{G_1,R,G_2:\text{Copeland}_{\alpha}(G_1)  \succeq_R  \text{Copeland}_{\alpha}(G_2)} \ppoly{G_1,R,G_2}
$$

\subsection{Formal Statement of the Categorization Lemma and Proof}
\label{app:cat-lemma-formal}
We first introduce notation for polyhedra. Given $q\in\mathbb N, L\in\mathbb N$, an $L\times q$ integer matrix $\ba$, a  $q$-dimensional row vector $\vbb$,  we define  

$\hfill\begin{array}{ll}
\poly \triangleq \left\{\vec x\in {\mathbb R}^q: \ba\cdot \invert{\vec x}\le \invert{\vec b}\right\}, &\polyz \triangleq \left\{\vec x\in {\mathbb R}^q: \ba\cdot \invert{\vec x}\le \invert{\vec 0}\right\} \end{array}\hfill$

That is, $\poly$ is the polyhedron represented by $\ba$ and ${\vbb}$ and $\polyz$ is the {\em characteristic cone} of $\poly$.    

\begin{ex}[\bf\boldmath Poisson multinomial variable (PMV) $\vXp$]
\label{ex:pmv} In the setting of Example~\ref{ex:sCC-plu}, we have $q = m! = 6$. Let $n=2$ and $\vec \pi = (\pi^2,\pi^1)$. $\vXp$ is the histogram of two random variables $Y_1,Y_2$ over $[q]$, where $Y_1$ (respectively, $Y_2$) is distributed as $\pi^2$ (respectively, $\pi^1$).

For example, let $\vec x\in \{0,1,2\}^{\ml(\ma)}$ denote the vector whose $123$ and $231$ components are $1$ and all other components are $0$. We have $\Pr(\vXp = \vec x)=\frac 14\times \frac 38 + \frac 18\times \frac 18=\frac{7}{64}$.
\end{ex}

\begin{dfn}[\bf Almost complement] 
\label{dfn:almost-complement}
Let $\upoly$ denote a union of finitely many polyhedra. We say that a union of finitely many polyhedra $\upoly^*$ is an {\em almost complement} of $\upoly$, if (1) $\upoly\cap \upoly^*  =\emptyset$ and (2) ${\mathbb Z}^q \subseteq \upoly\cup \upoly^*$.
\end{dfn}
$\upoly^*$ is called an ``almost complement'' (instead of ``complement'') of $\upoly$ because $\aupoly\cup \upoly\ne \mathbb R^q$. Effectively, $\upolyz^*$ can be viewed as the complement of $\upoly$ when only integer vectors are concerned. It it not hard to see that $\upoly$ is an almost complement of $\upoly^*$.  The following result states that the characteristic cones of $\upoly$ and $\upoly^*$, which may overlap, cover $\mathbb R^q$.

\begin{prop}\label{prop:almost-complement-union}
For any union of finitely many polyhedra $\upoly$ and any almost complement $\upoly^*$ of $\upoly$, we have $\upolyz\cup \upolyz^*={\mathbb R}^q$.
\end{prop}
\begin{proof} Suppose for the sake of contradiction that $\upolyz\cup \upolyz^*\ne {\mathbb R}^q$. Let $\vec x\in {\mathbb R}^q\setminus (\upolyz\cup \upolyz^*)$ with $|\vec x|_1 =1$. Because $\upolyz$ and $\upolyz^*$ are unions of polyhedra, there exists an $\delta>0$ neighborhood $B_\delta = \{\vec x'\in{\mathbb R}^q: |\vec x' - \vec x|_\infty\le \delta\}$ of $\vec x$ in $\mathbb R^q$ that is $\eta>0$ away from $\upolyz\cup \upolyz^*$. Therefore, there exists $n\in \mathbb N$ with $n>\frac 1\delta$ such that $nB_\delta = \{n\vec x': \vec x'\in B_\delta\}$ do not overlap $\upoly\cup \upoly^*$. Because the radius of $nB_\delta$ is larger than 1, there exists an integer vector in $nB_\delta$, which contradicts the assumption that ${\mathbb Z}^q\subseteq \upoly\cup \upoly^*$.
\end{proof}

W.l.o.g., in this paper we assume that all polyhedra   are represented  by integer matrices $\ba$ where the entries of each row are coprimes, which means that the greatest common divisor of all entries in the row is $1$. 
For any $\upoly = \bigcup_{i\le  I}\cpoly{i}$ where $\cpoly{i}$ is the polyhedron characterized by integer matrices $\ba_i$ with coprime entries and $\vbb_i$, its almost complement always exists and is not unique. Let us define an specific almost complement of $\upoly$ that will be commonly used in this paper. 
\begin{dfn}[\bf Standard almost complement] Let $\upoly=\cup_{i\le I} \cpoly{i}$ denote a union of $I$ rational polyhedra characterized by $\ba_i$ and $\vbb_i$, we define its {\em standard almost complement}, denoted by $\saupoly$, as follows.
$$\saupoly = \bigcup\nolimits_{ \vec a_i\in \ba_i: \forall i\le I} \bigcap\nolimits_{i\le I} \left\{\vec x\in{\mathbb R}^q: -\vec a_i\cdot \vec x\le -b_i'-1\right\},$$
where $\vec a_i$ is a row in $\ba_i$ and $b_i'$ is the corresponding component in $\vbb_i$.  We write $\saupoly = \bigcup_{i^*\le \hat I}\acpoly{i^*}$, where $\hat I\in\mathbb N$ and each $\acpoly{i^*}$ is a rational polyhedron.
\end{dfn}
It is not hard to verify that $\saupoly$ is indeed an almost complement of $\upoly$. Let us take a look at a simple example for $q=2$.

\begin{ex}\label{ex:almost-complement} Let $\upoly = \cpoly{1}\cup \cpoly{2}$, where $\cpoly{1} = \left\{\vec x\in {\mathbb R}^2:\left[\begin{array}{rr}-1&0\\2&-1\end{array}\right]\cdot \invert{\vec x}\le \left[\begin{array}{r}0\\-2\end{array}\right]\right\}$ and $\cpoly{2} = \left\{\vec x\in {\mathbb R}^2:\left[\begin{array}{rr}-1&2\\1&-2\end{array}\right]\cdot \invert{\vec x}\le \left[\begin{array}{r}8\\8\end{array}\right]\right\}$. It follows that $\saupoly = \acpoly{1} \cup \acpoly{2} \cup \acpoly{3} \cup \acpoly{4}$, where 
\begin{align*}
\acpoly{1} = \left\{\vec x\in {\mathbb R}^2:\left[\begin{array}{rr}1&0\\1&-2\end{array}\right]\cdot \invert{\vec x}\le \left[\begin{array}{r}-1\\-9\end{array}\right]\right\}, 
&\acpoly{2}  = \left\{\vec x\in {\mathbb R}^2:\left[\begin{array}{rr}1&0\\-1&2\end{array}\right]\cdot \invert{\vec x}\le \left[\begin{array}{r}-1\\-9\end{array}\right]\right\}\\
\acpoly{3}  = \left\{\vec x\in {\mathbb R}^2:\left[\begin{array}{rr}-2&1\\1&-2\end{array}\right]\cdot \invert{\vec x}\le \left[\begin{array}{r}1\\-9\end{array}\right]\right\}, 
&\acpoly{4} = \left\{\vec x\in {\mathbb R}^2:\left[\begin{array}{rr}-2&1\\-1&2\end{array}\right]\cdot \invert{\vec x}\le \left[\begin{array}{r}1\\-9\end{array}\right]\right\}\\
\end{align*}

Figure~\ref{fig:ex-cstar} (a) shows $\upoly$ and $\saupoly$. Figure~\ref{fig:ex-cstar} (b) shows $\upolyz$ and  $\saupolyz$, where $\cpoly{2}$ is a one-dimensional polyhedron, i.e.,~a straight line. Note that $\upoly\cup \saupoly\ne \mathbb R^q$ and $\upolyz\cup \saupolyz = \mathbb R^q$. 

\begin{figure}[htp]
\centering
\begin{tabular}{cc }
\includegraphics[width = 0.4\textwidth]{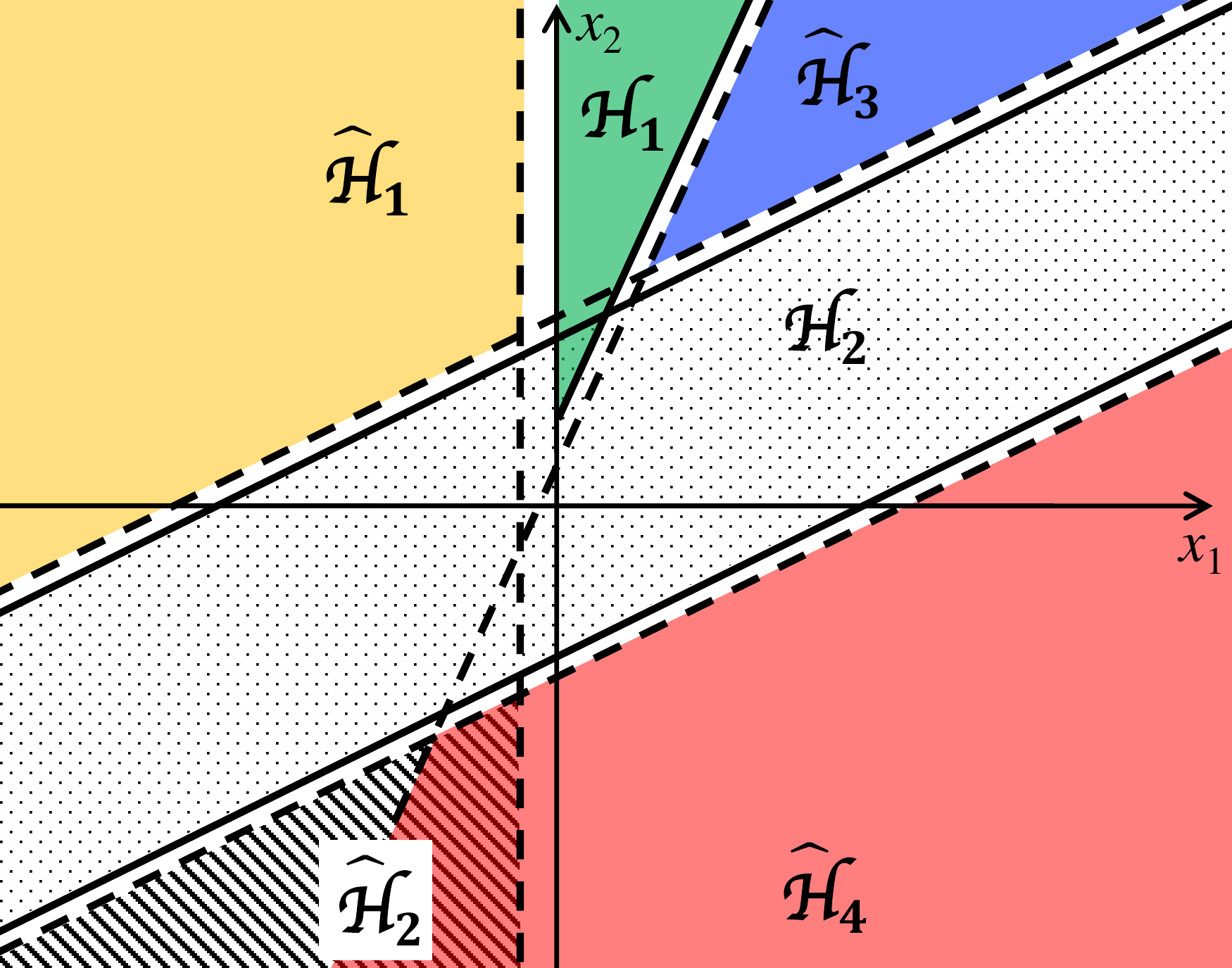}& 
\includegraphics[width = 0.4\textwidth]{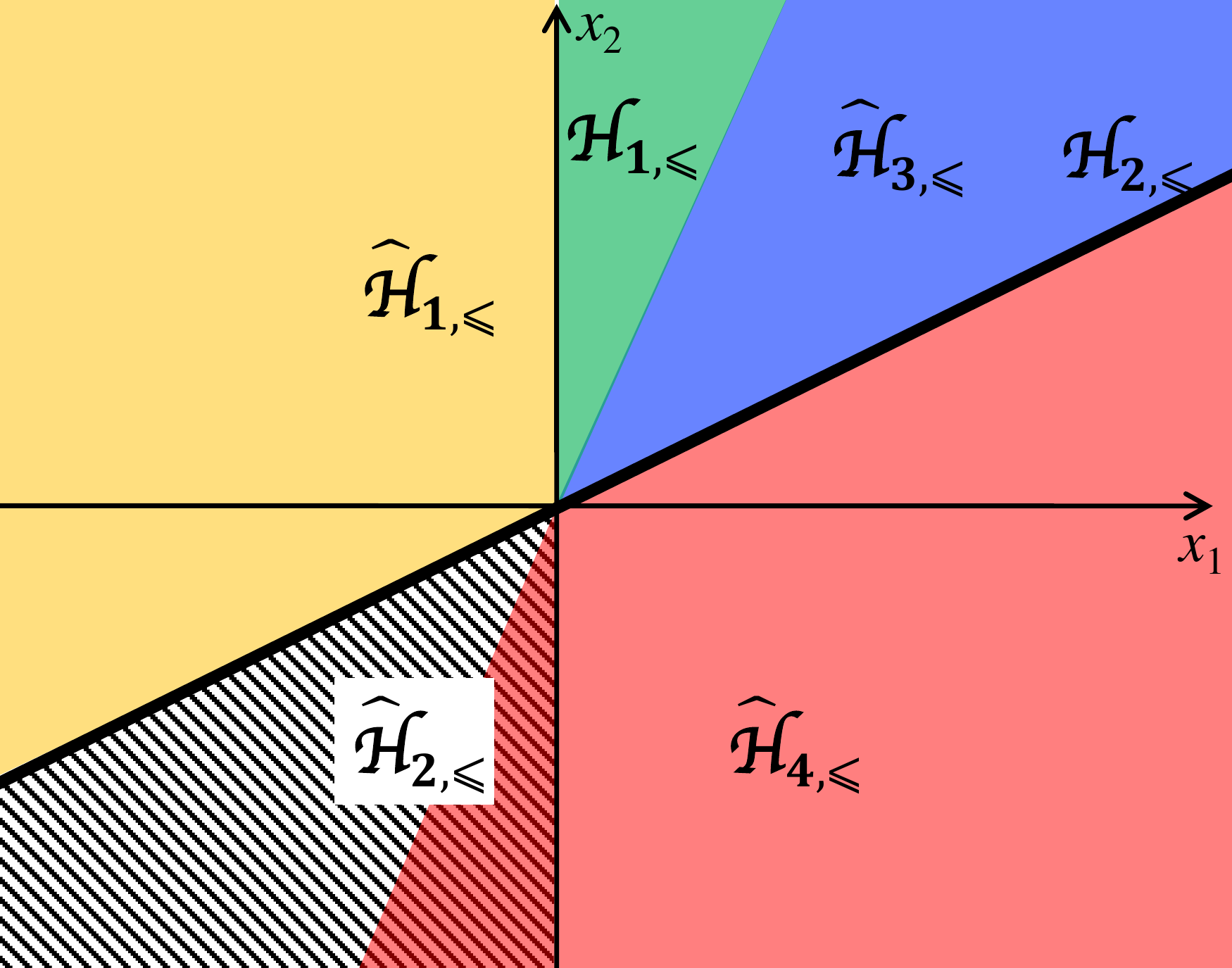}\\
(a) $\upoly$ and $\saupoly$. & (b) $\upolyz$ and $\saupolyz$. 
\end{tabular}
\caption{\small In (a), $\upoly = \cpoly{1}\cup \cpoly{2}$, where $\cpoly{1}$ is the green area and $\cpoly{2}$ is a shaded area, and $\saupoly = \acpoly{1} \cup \acpoly{2} \cup \acpoly{3} \cup \acpoly{4}$, where $\acpoly{2}$ is a shaded area, and $\acpoly{1}$, $\acpoly{3}$, and $\acpoly{4}$ are the yellow, red, and blue areas, respectively. In (b), $\upolyz \cup \saupolyz = \mathbb R^q$, where $\acpoly{2}$ is a straight line. \label{fig:ex-cstar}}
\end{figure}
\end{ex}

To present the categorization lemma, we recall the definitions of $\alpha_n$,  $\beta_n$, and Theorem~2 in~\citep{Xia2021:How-Likely}. We first define the {\em activation graph}.

\begin{dfn}[\bf Activation graph~\citep{Xia2021:How-Likely}]\label{dfn:activation-graph} For each $ \Pi $, $\cpoly{i}$, and $n\in\mathbb N$, the {\em activation graph}, denoted by $\calG_{\Pi,\upoly,n}$,  is defined to be the complete bipartite graph with two sets of vertices $\conv(\Pi)$ and $\{\cpoly{i}:i\le I\}$, and the weight on the edge $(\pi,\cpoly{i})$ is defined as follows.
 $$w_n(\pi,\cpoly{i})\triangleq\left\{\begin{array}{ll}-\infty&\text{if } \cpolynint{i} = \emptyset\\
-\frac{n}{\log n} &\text{otherwise, if }  \pi\notin \cpolyz{i} \\
\dim(\cpolyz{i})& \text{otherwise}
 \end{array}\right.,$$
where $\cpolynint{i}$ is the set of non-negative integer vectors in $\cpoly{i}$ whose $L_1$ norm is $n$. 
\end{dfn}
Definition~\ref{dfn:activation-graph} slightly abuses notation, because its vertices $\{\cpoly{i}:i\le I\}$ are not explicitly indicated in the subscript of $\calG_{\Pi,\upoly,n}$. This does not cause confusion when they are clear from the context.

When $\cpolynint{i} = \emptyset$ we say that $\cpoly{i}$ is {\em inactive (at $n$)}, and when $\cpolynint{i} \ne \emptyset$ we say that $\cpoly{i}$ is {\em active (at $n$)}. In addition, if the weight on any edge $(\pi,\cpoly{i})$ is positive, then we say that $\pi$ is {\em active} and  is {\em activated} by $\cpoly{i}$ (which must be active at $n$).
 
Roughly speaking, for any sufficiently large $n$ and  $\vec\pi = (\pi_1,\ldots,\pi_n)\in\Pi^n$, let $\pi = \frac 1n \sum_{j=1}^n \pi_j$, then \cite[Theorem~1]{Xia2021:How-Likely} implies
$$\Pr(\vXp\in \cpoly{i}) \approx n^{w_n(\pi,\cpoly{i})-q}$$
It follows that $\Pr(\vXp\in \upoly)$ is mostly determined by the heaviest weight on edges connected to $\pi$, denoted by  $\md{\upoly}{ \pi}$, which is formally  defined as follows:
$$\md{\upoly}{\pi}\triangleq\max\nolimits_{i\le I} w_n (\pi, \cpoly{i})$$
Then, a max-(respectively, min-) adversary aims to choose $\vec\pi=(\pi_1,\ldots,\pi_n)\in \Pi^n$ to maximize (respectively, minimize) $\md{\upoly}{\frac 1n\sum_{j=1}^n \pi_j}$, which are characterized by  $\alpha_n$ (respectively, $\beta_n$) defined as follows.
\begin{align*}
&\alpha_{n} \triangleq \max\nolimits_{\pi\in\conv(\Pi)} \md{\upoly}{\pi}\\
&\beta_{n} \triangleq \min\nolimits_{\pi\in\conv(\Pi)} \md{\upoly}{\pi}
\end{align*}
We further define the following notation that will be frequently used in the proofs of this paper.
 Let  $\upolynint$ denote the set of all non-negative integer vectors in $\upoly$ whose $L_1$ norm is $n$.  That is,
$$\upolynint = \bigcup\nolimits_{i\le I} \cpolynint{i}$$
By definition, $\upolynint = \emptyset$ if and only if all $\cpoly{i}$'s are inactive at $n$. Therefore, we have 
$$(\alpha_n = -\infty) \Longleftrightarrow (\beta_n = -\infty)\Longleftrightarrow (\upolynint = \emptyset) $$
For completeness, we recall~\citep[Theorem~2]{Xia2021:How-Likely} below.

\paragraph{\bf \boldmath Theorem~2 in~\citep{Xia2021:How-Likely} (Smoothed likelihood of PMV-in-$\upoly$).}{\em
Given any $q, I\in\mathbb N$, any closed and strictly positive $\Pi$ over $[q]$, and any set $\upoly = \bigcup_{i\le I}\cpoly{i}$ that is the union of finitely many polyhedra with integer matrices, for any $n\in\mathbb N$, 
\begin{align*}
&\sup_{\vec\pi\in\Pi^n}\Pr\left(\vXp \in \upoly\right)=\left\{\begin{array}{ll}0 &\text{if } \alpha_n = -\infty\\
\exp(-\Theta(n)) &\text{if } -\infty<\alpha_n<0\\
\Theta\left(n^{\frac{\alpha_n-q}{2}}\right) &\text{otherwise (i.e. }\alpha_n\ge 0\text{)}
\end{array}\right.,\\
&\inf_{\vec\pi\in\Pi^n}\Pr\left(\vXp \in \upoly\right)=\left\{\begin{array}{ll}0 &\text{if } \beta_n = -\infty\\
\exp(-\Theta(n)) &\text{if } -\infty< \beta_n <0\\
\Theta\left(n^{\frac{\beta_n-q}{2}}\right) &\text{otherwise (i.e. } \beta_n\ge 0
\text{)}\end{array}\right..
\end{align*}
}

For any almost complement $\upoly^*$ of $\upoly$, let $\alpha_n^*$ and $\beta_n^*$  denote the counterparts of $\alpha_n$ and $\beta_n$ for $\aupoly$, respectively. We note that  $\alpha_n^*$ and $\beta_n^*$ depend on the polyhedra used to representation  $\aupoly$. We are now ready to present the full version of the categorization lemma as follows.

\appLem{lem:categorization}{Categorization Lemma, Full Version}
{Given any $q, I\in\mathbb N$, any closed and strictly positive $\Pi$ over $[q]$,  any  $\upoly = \bigcup_{i\le I}\cpoly{i}$ and its almost complement $\aupoly = \bigcup_{i^*\le I^*}\cpoly{i^*}^*$,  for any $n\in\mathbb  N$, 
\begin{align*}
&\inf_{\vec\pi\in\Pi^n}\Pr\left(\vXp \in \upoly\right)=\left\{\begin{array}{ll}0 &\text{if } \beta_n = -\infty\\
\exp(-\Theta(n)) &\text{if }  -\infty< \beta_n<0\\
\Theta\left(n^{\frac{\beta_n-q}{2}}\right) &\text{if }  0\le \beta_n<q\\
\Theta(1)\wedge(1-\Theta(1)) &\text{if }\alpha^*_n=\beta_n=q\\
1-\Theta\left(n^{\frac{\alpha^*_n-q}{2}}\right) & \text{if }  0\le \alpha^*_n<q\\ 
1- \exp(-\Theta(n)) &\text{if } -\infty< \alpha^*_n <0\\
1 &\text{if } \alpha^*_n=\infty\end{array}\right.
\end{align*}
\begin{align*}
&\sup_{\vec\pi\in\Pi^n}\Pr\left(\vXp \in \upoly\right)=\left\{\begin{array}{ll}0 &\text{if } \alpha_n = -\infty\\
\exp(-\Theta(n)) &\text{if } -\infty< \alpha_n <0\\
\Theta\left(n^{\frac{\alpha_n-q}{2}}\right) &\text{if } 0\le \alpha_n<q\\
\Theta(1)\wedge(1-\Theta(1)) &\text{if } \alpha_n=\beta^*_n=q\\
1-\Theta\left(n^{\frac{\beta_n^*-q}{2}}\right) & \text{if } 0\le \beta^*_n<q\\
1- \exp(-\Theta(n)) &\text{if } -\infty< \beta^*_n<0\\
1 &\text{if }  \beta^*_n = -\infty
\end{array}\right.
\end{align*}
}

\begin{proof}
We present the proof for the $\inf$ part of Lemma~\ref{lem:categorization} and the proof for the $\sup$ part is similar. Notice that ${\mathbb Z}^q \subseteq  \upoly\cup \aupoly$, we have:
$$\inf\nolimits_{\vec\pi\in\Pi^n}\Pr\left(\vXp \in \upoly\right)= 1-\sup\nolimits_{\vec\pi\in\Pi^n}\Pr\left(\vXp \in \aupoly\right)$$
The proof is done by combining the $\inf$ part of~\cite[Theorem~2]{Xia2021:How-Likely} (applied to $\upoly$) and one minus the $\sup$ part of~\cite[Theorem~2]{Xia2021:How-Likely} (applied to $\aupoly$).
\begin{itemize}
\item {\bf \boldmath The $0$, $\exp(-\Theta(n))$ and $\Theta\left(n^{\frac{\beta_n-q}{2}}\right)$ cases}  follow after the corresponding $\inf$ part of~\cite[Theorem~2]{Xia2021:How-Likely} applied to $\upoly$.
\item {\bf \boldmath The $\Theta(1)\wedge(1-\Theta(1))$ case.} The condition of this case implies that the polynomial bounds in the $\inf$ part of~\cite[Theorem~2]{Xia2021:How-Likely} (applied to $\upoly$) hold, which means that $\inf_{\vec\pi\in\Pi^n}\Pr\left(\vXp \in \upoly\right)=\Theta(1)$, and the polynomial bounds in the $\sup$ part of~\cite[Theorem~2]{Xia2021:How-Likely} (applied to $\aupoly$) hold, which means that 
$$\inf\nolimits_{\vec\pi\in\Pi^n}\Pr\left(\vXp \in \upoly\right)= 1-\sup\nolimits_{\vec\pi\in\Pi^n}\Pr\left(\vXp \in \aupoly\right) = 1-\Theta(1)$$
\item {\bf \boldmath The $1-\Theta\left(n^{\frac{\alpha^*_n-q}{2}}\right)$,  $1-\exp(-\Theta(n))$, and $1$ cases}  follow after one minus the $\sup$ part of~\cite[Theorem~2]{Xia2021:How-Likely} (applied to $\aupoly$). 
\end{itemize}
\end{proof}
\begin{wrapfigure}{R}{0.45\textwidth}
\centering
\includegraphics[width = 0.45\textwidth]{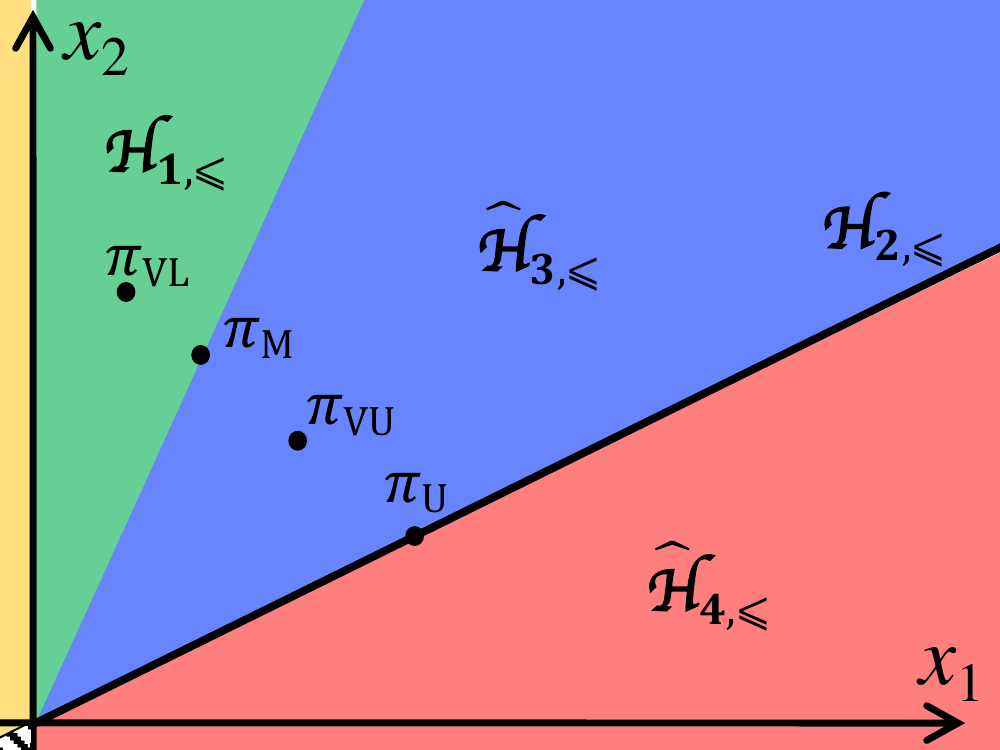}
\caption{\small An Illustration of $\pi_{\text{VU}}$, $\pi_{\text{U}}$, $\pi_{\text{M}}$, and $\pi_{\text{VL}}$ for the $\inf$ part of Lemma~\ref{lem:categorization}. \label{fig:ex-categorization-lemma}}
\end{wrapfigure}

\paragraph{\bf Remarks.} The conditions for all, except $0$ and $1$, cases are different between $\sup$ and $\inf$ parts of the lemma. Moreover, the degrees of polynomial in the L and U cases may be different between $\sup$ and $\inf$ parts.  Let us use the setting in Example~\ref{ex:almost-complement} and Figure~\ref{fig:ex-categorization-lemma}  to illustrate the conditions for the $\inf$ case. For the purpose of illustration, we assume that all polyhedra in $\upoly$ and $\upoly^*$ are active at $n$.

$\bullet$ {\bf \boldmath The $0$ (respectively, $1$) case} holds when no non-negative integer with $L_1$ norm $n$ is in $\upoly$ (respectively, in $\aupoly$). 

$\bullet$ {\bf \boldmath The VU case.} Given that the $0$ and $1$ cases do not hold, the VU case holds when $\conv(\Pi)$ contains a distribution $\pi_{\text{VU}}$ that is not in $\upolyz$. Notice that $\upolyz$ is a closed set and $\upolyz\cup \aupolyz=\mathbb R^q$. This means that $\pi_{\text{VU}}$ is an interior point of $\aupolyz$. For example, in Figure~\ref{fig:ex-categorization-lemma}, $\pi_{\text{VU}}$ is not in $\upolyz$ and is an interior point of $\acpolyz{3}$.

$\bullet$ {\bf The U case} holds when  $\conv(\Pi)\subseteq \upolyz$, and $\conv(\Pi)$ contains a distribution $\pi_{\text{U}}$ that lies on a (low-dimensional) boundary of $\upolyz$. For example, in Figure~\ref{fig:ex-categorization-lemma}, $\pi_{\text{U}}$ lies in a $1$-dimensional polyhedron $\cpolyz{2}\subseteq \upolyz$, and is not in any $2$-dimensional polyhedron in $\upolyz$.

$\bullet$ {\bf The M case} holds when the U case does not hold, and $\conv(\Pi)$ contains a distribution  $\pi_{\text{M}}$ that lies in the intersection of a $q$-dimensional subspace of $\upolyz$ and a $q$-dimensional subspace of $\aupolyz$. For example, in Figure~\ref{fig:ex-categorization-lemma}, $\pi_{\text{U}}$ lies in $\cpolyz{1}$ and $\acpolyz{3}$, both of which are $2$-dimensional.

$\bullet$ {\bf The L case holds} when every distribution in $\conv(\Pi)$ is in a $q$-dimensional subspace of $\upolyz$, and there exists $\pi_{\text{L}}\in \conv(\Pi)$ that  lies in a (low-dimensional) boundary of $\aupolyz$.  No such $\pi_{\text{L}}$ exists in Figure~\ref{fig:ex-categorization-lemma}'s example, but if we apply Lemma~\ref{lem:categorization} to $\aupoly$, then $\pi_{\text{U}}$ in Figure~\ref{fig:ex-categorization-lemma} is an example of $\pi_{\text{L}}$ for $\aupoly$.

$\bullet$ {\bf The VL case holds} when every distribution in $\conv(\Pi)$ is an inner point of $\upolyz$. For example, in Figure~\ref{fig:ex-categorization-lemma}, $\pi_{\text{VL}}$ is an inner point of $\cpolyz{1}\subseteq \upoly$.

\section{ GISRs and Their Algebraic Properties}
\label{app:algebraic-GISR}

\subsection{Definition of GISRs} 
\label{app:dfn-GISR}
All irresolute voting rules  studied in this paper are generalized irresolute scoring rules (GISR)~\citep{Freeman2015:General,Xia15:Generalized}, whose resolute versions are known as {\em generalized scoring rules (GSRs)}~\citep{Xia08:Generalized}. We recall the definition of GISRs based on separating hyperplanes~\citep{Xia09:Finite,Mossel13:Smooth}. 

For any real number $x$,  let $\sign(x)\in\{+,-,0\}$ denote the sign of $x$. Given a set of $K$ hyperplanes in the $q$-dimensional Euclidean space, denoted by $\vH = (\vec h_1,\ldots,\vec h_K)$, for any $\vec x \in {\mathbb R}^q$, we let $\sign_\vH(\vec x) = (\sign(\vec x\cdot \vec h_1),\ldots, \sign(\vec x\cdot \vec h_K))$. In other words, for any $k\le K$, the $k$-th component of $\sign_\vH(\vec x)$ equals to $0$, if $\vec p$ lies in hyperplane $\vec h_k$; and it equals to $+$ (respectively, $-$) if  $\vec p$ lies in the positive (respectively, negative) side of  $\vec h_k$.
Each element in $\{+,-,0\}^K$ is called a {\em signature}.

\begin{dfn}[\bf Generalized irresolute scoring rule (GISR)]
\label{dfn:GISR} A {\em generalized irresolute scoring rule (GISR)} $\cor$ is defined by (1) a set of $K\ge 1$ hyperplanes $\vH = (\vec h_1,\ldots,\vec h_K)\in ({\mathbb R}^{m!})^K$ and (2) a function $g:\{+,-,0\}^K\ra (2^\ma\setminus\emptyset)$. For any  profile $P$, we let $\cor(P) = g(\sign_\vH(\hist(P)))$. $\cor$ is called an {\em integer GISR (int-GISR)} if ${\vH}\in ({\mathbb Z}^{m!})^K$.  If for all profiles $P$, we have $|\cor(P)|=1$, then $\cor$ is called a generalized scoring rule (GSR). Int-GSRs are defined similarly to int-GISRs.
\end{dfn}

\begin{dfn} [\bf Feasible and atomic signatures]
Given integer $\vH$ with $K = |\vH|$, let $\sk =  \{+,-,0\}^K$. A signature $\vec t\in\sk$ is {\em  feasible}, if there exists $\vec x\in \mathbb R^{d}$  such that $\sign_\vH(\vec x) = \vec t$. 
Let $\fs \subseteq \sk$ denote the set of all feasible signatures.

A signature $\vec t$ is called an {\em atomic signature} if and only if $\vec t\in \{+,-\}^K$. Let $\fsatomic$ denote the set of all feasible atomic signatures.
\end{dfn}
The domain of any GISR $\cor$ can be naturally extended to $\mathbb R^{m!}$ and to $\fs$. Specifically, for any $\vec t\in\fs$ we let $\cor(\vec t) = g(\vec t)$. It suffices to define $g$ on the feasible signatures, i.e., $\fs$.

Notice that the same voting rule can be   represented by different combinations of $(\vH,g)$. In the following  section we recall int-GISR representations of the voting rules studied in this paper. 

%

\subsection{Commonly-Studied Voting Rules as GISRs}
\label{sec:common-GISR}
As discussed in~\citep{Xia2021:How-Likely}, the irresolute versions of  Maximin, Copeland$_\alpha$, Ranked Pairs, and Schulze belong to the class of {\em edge-order-based (\em{EO}-based)} rules, which are defined over the weak order on edges in $\wmg(P)$.  We recall its formal definition below.

\begin{dfn}[\bf Edge-order-based rules]
\label{dfn:eo-rule}
A (resolute or irresolute) voting rule $\cor$ is {\em edge-order-based (EO-based)}, if for any pair of profiles $P_1$ and $P_2$ such that for every combination of four different alternatives $\{a,b,c,d\}\subset\ma$, $\left[ w_{P_1}(a,b)\ge w_{P_1}(c,d)\right]\Leftrightarrow \left[ w_{P_2}(a,b)\ge w_{P_2}(c,d)\right]$, we have $\cor(P_1) = \cor(P_2)$.
\end{dfn}

All EO-based rules can be represented by a GISR using a set of hyperplanes that represents the orders over WMG edges. We first recall pairwise difference vectors as follows.

\begin{dfn}[\bf Pairwise difference vectors~\citep{Xia2020:The-Smoothed}]\label{dfn:pairdiff} For any pair of different alternatives $a,b$, let $\pair_{a,b}$ denote the $m!$-dimensional vector indexed by rankings in $\ml(\ma)$: for any $R\in\ml(\ma)$, the $R$-component of $\pair_{a,b}$ is $1$ if $a\succ_R b$; otherwise it is $-1$. 
\end{dfn}

We now define the hyperplanes for edge-order-based rules.
\begin{dfn}[\bf \boldmath $\vH_{\eo}$]
\label{dfn:heo}
$ \vH_{\eo}$ consists of $m(m-1)\choose 2$ hyperplanes indexed by $\vec h_{e_1,e_2}$, where $e_1=(a_1,a_2)$ and $e_2=(a_2,b_2)$ are two different pairs of alternatives, such that 
$$\vec h_{e_1,e_2} = \pair_{a_1,b_1}-\pair_{a_2,b_2}$$
\end{dfn}
That is, for any (fractional) profile $P$, $\vec h_{e_1,e_2}\cdot \hist(P)\le 0$ if and only if the weight on $e_1$ in $\wmg(P)$ is no more than the weight on $e_2$ in $\wmg(P)$. Therefore,   given $\sign_{\vH_{\eo}}(P)$, we can compare the weights on pairs of edges, which leads to the weak order on edges in $\wmg(P)$ w.r.t.~their weights. Consequently, for any profile $P$, $\sign_{\vH}(P)$ contains enough information to determine the (co-)winners under any edge-order-based rules. Formally, the GISR representations of these rules used in this paper are defined by $\vH_{\eo}$ and the following $g$ functions that mimic the procedures of choosing the winner(s). 
\begin{dfn}
\label{dfn:some-GISR}
Let $\overline{\maximin}$, $\overline{\copeland}$, $\overline{\rp}$, $\overline{\schulze}$ denote the int-GISRs defined by $ \vH_{\eo}$ and the following $g$ functions. Given a feasible signature $\vec t\in\pfs{\vH_{\eo}}$, 
\begin{itemize}
\item  {\bf \boldmath $g_{\maximin}$} first picks a representative edge $e_a$ whose weight is no more than all other outgoing edges of $a$, then compare the weights of $e_a$'s for all alternatives and choose alternatives $a$ whose $e_a$ has the highest weight as the winners.
\item {\bf \boldmath $g_{\copeland}$} compares weights on pairs of edges $a\ra b$ and $b\ra a$, and then calculate the Copeland$_\alpha$ scores accordingly. The winners are the alternatives with the highest Copeland$_\alpha$ score.
\item {\bf \boldmath $g_{\rp}$} mimics the execution of PUT-Ranked Pairs, which only requires information about the weak order over edges w.r.t.~their weights in WMG.
\item {\bf \boldmath $g_{\schulze}$} first computes an edge $e_p$ with the minimum weight on any given directed path $p$, then for each pair of alternatives $a$ and $b$, computes an edge $e_{(a,b)}$ that represents the strongest edge among all paths from $a$ to $b$. $g_{\schulze}$ then mimics Schulze to select the winner(s). 
\end{itemize}
\end{dfn}

While Copeland can be represented by $\vH_{\eo}$ and $g_{\copeland}$ as in the definition above, in this paper we use another set of hyperplanes, denoted by $\vH_{\copeland}$, that represents the UMG of the profile. The reason is that in this way any refinement of $\copeland$ would break ties according to the UMG of the profile, which is needed in the proof of Theorem~\ref{thm:sPar-copeland}. 

\begin{dfn}[\bf \boldmath $\icopeland$ as a GISR]
\label{dfn:copeland-GISR} $\icopeland$ is represented by $\vH_{\copeland}$ and $g_{\copeland}$ defined as follows. 
For every pair of different alternatives $(a,b)$, $\vH_{\copeland}$ contains a hyperplane $\vec h_{(a,b)} = \pair_{a,b}-\pair_{b,a}$. For any profile $P$, $g_{\copeland}$ first computes the outcome of each head-to-head elections between alternatives $a$ and $b$ by checking $\vec h_{(a,b)}\cdot\hist(P)$, then calculate the Copeland$_\alpha$ score, and finally choose all alternatives with the maximum score as the winners.
\end{dfn}

The GISR representation of MRSE rules is based on the fact that the winner(s) can be computed from comparing the scores between any pair of alternatives $(a,b)$ after a set of alternatives $B$ is removed. This idea is formalized in the following definition. For any $R\in\ml(\ma)$ and any $B\subset \ma$, let $R|_{\ma\setminus B}$ denote the linear order over $(\ma\setminus B)$ that is obtained from $R$ by removing alternatives in $B$.

\begin{dfn}[\bf \boldmath MRSE rules as GISRs]
\label{dfn:MRSE-GISR} Any MRSE $\cor=(\cor_2,\ldots,\cor_{m})$ is represented by $\vH$ and $g_{\cor}$ defined as follows. 
Given an int-MRSE rule $\cor = (\cor_2,\ldots,\cor_m)$, for any pair of alternatives $a,b$ and any subset of alternatives $B\subseteq(\ma\setminus\{a,b\})$, we let $\scorediff{B,a,b}$ denote the vector, where for every $R\in\ml(\ma)$, the $R$-th component of $\pair_{B,a,b}$ is $s_i^{m-|B|}-s_j^{m-|B|}$, where $i$ and $j$ are the ranks of $a$ and $b$ in $R|_{\ma\setminus B}$, respectively.

For any pair of different alternatives $\{a,b\}\subseteq (\ma\setminus B)$, $\vH$ contains a hyperplane $\scorediff{B,a,b}$. For any profile $P$, $g_{\cor}$ mimics $\cor$ to compute the PUT winners based on whether $\vec h_{(B,a,b)}\cdot\hist(P)$ is $<0$, $=0$, or $>0$.
\end{dfn}
In fact, the GISR representation of $\cor$ in Definiton~\ref{dfn:MRSE-GISR} corresponds to the {\em PUT structure}~\citep{Xia2021:How-Likely}, which we do not discuss in this paper for simplicity of presentation. Any GSR refinement of $\cor$, denoted by $r$, uses the same $\vH$ in Definiton~\ref{dfn:MRSE-GISR} and a different $g$ function that always chooses a single loser to be eliminated in each round. The constraint is, for any profile $P$, the break-tie mechanisms used in $g$  only depends on $\sign_{\vH}(P)$ (but not any other information contained in $P$).  For example, lexicographic tie-breaking w.r.t.~a fixed order over alternatives is allowed but using the first agent's vote to break ties is not allowed.

\subsection{Minimally Continuous GISRs}
\label{sec:common-min-cont}

Next, we define (minimally) continuous GISR in a similar way as~\citet{Freeman2015:General}, except that in this paper the domain of GISR is $\mathbb R^{m!}$ (in contrast to $\mathbb R^{m!}_{\ge 0}$ in~\citep{Freeman2015:General}).  
\begin{dfn}[\bf (Minimally) continuous GISR] 
\label{dfn:min-cont-gisr}
A GISR $\cor$ is {\em continuous}, if for any $\vec x\in\mathbb R^{m!}$, any alternative $a$, and any sequence of vectors $(\vec x_1,\vec x_2\ldots)$ that converges to $\vec x$,  
$$\left[\forall j\in\mathbb N, a\in \cor(\vec x_j)\right]\Longrightarrow [a\in \cor(\vec x)]$$

A GISR $\cor$ is called {\em minimally continuous}, if it is continuous and there does not exist a continuous GISR $\cor^*$ such that (1) for all $\vec x\in\mathbb R^{m!}$, $\cor^*(\vec x)\subseteq \cor(\vec x)$, and (2) the inclusion is strict for some $\vec x$.
\end{dfn}
Equivalently, a continuous GISR $\cor$ is minimally continuous if and only if the (fractional) profiles with unique winners is a dense subset of $\mathbb R^{m!}$. That is, for any vector in $\mathbb R^{m!}$, there exists a sequence of profiles with unique winners that converge to it. As commented by~\citet{Freeman2015:General}, many commonly-studied irresolute voting rules are continuous GISRs. It is not hard to verify that positional scoring rules and MRSE rules are minimally continuous GISRs, which is formally proved in the following proposition.


\begin{prop} \label{prop:common-min-cont}
Positional scoring rules and MRSE rules are minimally continuous.
\end{prop}
\begin{proof}
Let $\vec s = (s_1,\ldots,s_m)$ denote the scoring vector. We first prove that $\cor_{\vec s}$ is continuous. For any $\vec x\in\mathbb R^{m!}$, any $a\in\ma$, and any sequence $(\vec x_1,\vec x_2,\ldots)$ that converges to $\vec x$ such that for all $j\ge 1$, $a\in \cor(\vec x_j)$, we have that for every $b\in \ma$, $\vec s(\vec x_j, a)\ge \vec s(\vec x_j, b)$. Notice that $\vec s(\vec x_j, a)$ (respectively, $\vec s(\vec x_j, b)$) converges to $\vec s(\vec x, a)$  (respectively, $\vec s(\vec x, b)$). Therefore, $\vec s(\vec x, a)\ge \vec s(\vec x, b)$, which means that $a\in \cor_{\vec s}(\vec x)$, i.e., $\cor_{\vec s}$ is continuous.

To prove that $\cor_{\vec s}$ is minimally continuous, it suffices to prove that for any $\vec x\in\mathbb R^{m!}$ and any $a\in \cor_{\vec s}(\vec x)$, there exists a sequence $(\vec x_1,\vec x_2,\ldots)$ that converges to $\vec x$ such that  for all $j\ge 1$, $\cor(\vec x_j) = \{a\}$. Let $\sigma$ denote an arbitrary cyclic permutation among $\ma\setminus\{a\}$ and $P$ denote the following $(m-1)$-profile.
$$P = \left\{\sigma^i(a\succ \others):1\le i\le m-1\right\}$$
Then, for every $j\in \mathbb N$, we let  $\vec x_j = \vec x + \frac 1j \hist(P)$. It is easy to check that $\cor(\vec x_j) = \{a\}$, which proves the minimal continuity of $\cor_{\vec s}$. 

Let $\cor = (\cor_2,\ldots,\cor_m)$ denote  the MRSE rule. We will use notation in Section~\ref{app:proof-thm:sCC-MRSE} to prove the proposition for $\cor$. We first prove that $\cor $ is continuous. Let  $\vec x\in\mathbb R^{m!}$,  $a\in\ma$, and  $(\vec x_1,\vec x_2,\ldots)$ be a sequence that converges to $\vec x$ such that for all $j\ge 1$, $a\in \cor(\vec x_j)$. Because the number of different parallel universes is finite (more precisely, $m!$), there exists a subsequence of $(\vec x_1,\vec x_2,\ldots)$, denoted by $(\vec x_1',\vec x_2',\ldots)$, and a parallel universe $O\in \ml(\ma)$ where $a$ is ranked in the last position (i.e., $a$ is the winner), such that for all $j\in\mathbb N$, $O$ is a parallel universe when executing $\cor$ on $\vec x_j'$. Therefore, for all $1\le i\le m-1$, in round $i$, $O[i]$ has the lowest $\cor_{m+1-i}$ score in $\vec x_j'|_{O[i,m]}$ among alternatives in $O[i,m]$. It follows that $O[i]$ has the lowest $\cor_{m+1-i}$ score in $\vec x|_{O[i,m]}$ among alternatives in $O[i,m]$, which means that $O$ is also a parallel universe when executing $\cor$ on $\vec x$. This proves that $\cor $ is continuous.

The proof of minimal continuity of $\cor$ is similar to the proof for positional scoring rules presented above.  For any $\vec x\in\mathbb R^{m!}$ and any $a\in \cor_{\vec s}(\vec x)$, let $O$ denote a parallel universe where $a$ is ranked in the last position. Let $P$ denote the following profile of $(m-1)!+(m-2)!+\cdots+2!$ votes, where $O$ is the unique parallel universe. 
$$P = \bigcup\nolimits_{i=1}^{m-1}\left\{O[1]\succ \cdots\succ O[i]\succ R_i: \forall R_i\in \ml(O[i+1,m])\right\}$$

For any $j\in \mathbb N$, let $\vec x_j = \vec x - \frac 1j \hist(P)$. It is not hard to verify that $(\vec x_1,\vec x_2,\ldots)$ converges to $\vec x$, and for every $1\le i\le m-1$ and every $j\in \mathbb N$, alternative $O[i]$ is the unique loser in round $i$, where $- \frac 1j \hist(P)$ is used as the tie-breaker.  This means that for all $j\in \mathbb N$, $\cor(\vec x_j) = \{a\}$, which proves the minimal continuity of $\cor$.
\end{proof}

\subsection{Algebraic Properties of GISRs}
We first define the refinement relationship among (feasible or infeasible) signatures. 

\begin{dfn}[\bf \boldmath Refinement relationship  $\unlhd$]
\label{dfn:operators}
For any pair of  signatures $\vec t_1,\vec t_2\in\sk$,   we say that $\vec t_1$ {\em refines} $\vec t_2$, denoted by $\vec t_1\unlhd \vec t_2$, if for every $k\le K$, if $[\vec t_2]_k\ne 0$ then $[\vec t_1]_k = [\vec t_2]_k$. If  $\vec t_1\unlhd\vec t_2$ and $\vec t_1\ne \vec t_2$, then we say that $\vec t_1$ {\em strictly refines} $\vec t_2$, denoted by $\vec t_1\lhd \vec t_2$.
\end{dfn}

In words, $\vec t_1$ refines $\vec t_2$ if $\vec t_1$ differs from $\vec t_2$ only on the $0$ components in $\vec t_2$. By definition,  $\vec t_1$ refines itself.  Next, given $\vH$ and a feasible signature $\vec t$,  we  define a polyhedron $\ppoly{\vH,\vec t}$ to represent profiles whose signatures are $\vec t$. 
\begin{dfn}[\bf\boldmath $\ppoly{\vH,\vec t}$  ($\ppoly{\vec t}$ in short)]
\label{dfn:poly-H-t}
For any $\vH = (\vec h_1,\ldots,\vec h_K)\in (\mathbb R^{d})^K$ and any $\vec t\in \fs$, we let  $\pba{\vec t}=\left[\begin{array}{c}\pba{\vec t}_{+}\\ \pba{\vec t}_{-}\\ \pba{\vec t}_{0} \end{array}\right]$, where 
\begin{itemize}
\item $\pba{\vec t}_{+}$ consists of a row $-\vec h_i$ for each $i\le K$ with $t_i = +$.
\item $\pba{\vec t}_{-}$ consists of a row $\vec h_i$ for each $i\le K$ with $t_i = -$.
\item  $\pba{\vec t}_{0}$ consists of two rows $-\vec h_i$ and $\vec h_i$ for each $i\le K$ with $t_i = 0$.
\end{itemize}
Let $\pvbb{\vec t} = [\underbrace{-\vec 1}_{\text{for }\pba{\vec t}_{+}},\underbrace{-\vec 1}_{\text{for }\pba{\vec t}_{-}},\underbrace{\vec 0}_{\text{for }\pba{\vec t}_{0}}]$.  The corresponding polyhedron is denoted by $\ppoly{\vH,\vec t}$, or $\ppoly{\vec t}$  in short when $\vH$ is clear from the context. 
\end{dfn}

The following proposition follows immediately after the definition.
\begin{prop}
\label{prop:refine-oplus-property}
Given $\vH$, for any pair of feasible signatures $\vec t_1,\vec t_2\in \fs$, $\vec t_1\unlhd \vec t_2$ if and only if $\ppolyz{\vec t_1}\supseteq \ppolyz{\vec t_2}$.
\end{prop}

\begin{prop}[\bf Algebraic characterization of (minimal) continuity]
\label{prop:char-continuity} A GISR $\cor$ is continuous, if and only if 
$$\forall\vec t\in\fs,\text{ we have } \cor(\vec t)\supseteq\bigcup\nolimits_{{\vec t}'\in\fs: {\vec t}'\unlhd \vec t}\cor({\vec t}')$$
$\cor$ is minimally continuous, if and only if 
$$ \forall \vec t\in\fs,\text{ we have } \cor(\vec t)=\bigcup\nolimits_{{\vec t}'\in\fsatomic: {\vec t}'\unlhd \vec t}\cor(\vec t')\text{, and (2) }\forall\vec t\in \fsatomic, \text{ we have }|\cor(\vec t)|=1$$
\end{prop}
The ``continuity'' part of Proposition~\ref{prop:char-continuity} states that for any feasible signature $\vec t$ and its refinement $\vec t'$, we must have $\cor(\vec t')\subseteq \cor(\vec t)$. The ``minimal continuity'' part states that any minimally continuous GISR is uniquely determined by its winners under atomic signatures (where a single winner is chosen for any atomic signature).

\begin{proof}{\bf The ``if'' part for continuity.}  Suppose for the sake of contradiction that there exists $\vec t\in\fs$ such that  $\cor(\vec t)\supseteq\bigcup_{{\vec t}'\in\fsatomic: {\vec t}'\unlhd \vec t}\cor(\vec t')$ but $\cor$ is not continuous. This means that there exists $\vec x\in {\mathbb R}^{m!}$ with $\signH(\vec x)=\vec t$, an infinite sequence $(\vec x_1,\vec x_2,\ldots )$ that converge to $\vec x$, and an alternative $a\notin \cor(\vec x)$, such that for every $j\in\mathbb N$, $a\in \cor(\vec x_j)$. Because the total number of (feasible) signatures is finite, there exists an infinite subsequence of $(\vec x_1,\vec x_2,\ldots)$, denoted by $(\vec x_1',\vec x_2',\ldots)$, and  $\vec t'\in \fs$ such that for all $j\in\mathbb N$ we have $\signH(\vec x_j')=\vec t'$. Note that $(\vec x_1',\vec x_2',\ldots)$ also converges to $\vec x$. Therefore, the following holds for every $k\le K$.
\begin{itemize}
\item If $t_k'=0$, then for every $j\in\mathbb N$ we have $\vec h_k\cdot \vec x_j =0$, which means that  $\vec h_k\cdot \vec x= 0$,  i.e.~$t_k=0$. 
\item If $t_k'=+$, then for every $j\in\mathbb N$ we have $\vec h_k\cdot \vec x_j >0$, which means that  $\vec h_k\cdot \vec x \ge0$,  i.e.~$t_k\in \{0,+\}$. 
\item Similarly, if $t_k'=-$, then for every $j\in\mathbb N$ we have $\vec h_k\cdot \vec x_j <0$, which means that  $\vec h_k\cdot \vec x \le0$,  i.e.~$t_k\in \{0,-\}$.
\end{itemize}
This means that $\vec t'\unlhd \vec t$. Recall that we have assumed $\cor(\vec t)\supseteq\bigcup_{{\vec t}'\in\fs: {\vec t}'\unlhd \vec t}\cor(\vec t')$, which means that $a\in\cor(\vec t')\subseteq \cor(\vec t)=\cor(\vec x)$. This contradicts the assumption that $a\notin \cor(\vec x)$.

{\bf The ``only if'' part for continuity.} Suppose for the sake of  contradiction that $\cor$ is continuous but there exists $\vec t\in\fs$ such that $\bigcup_{\vec t'\in\fs: \vec t'\unlhd \vec t}\cor(\vec t')\not\subseteq \cor(\vec t)$. This means that there exist  $\vec t'\lhd \vec t$ and an alternative $a$ such that $a\in \cor(\vec t')$ but $a\notin \cor(\vec t)$. Because both $\vec t$ and $\vec t'$ are feasible, there exists $\vec x, \vec x'\in\mathbb R^{m!}$ such that $\signH(\vec x) = \vec t$ and $\signH(\vec x') = \vec t'$. It is not hard to verify that the infinite sequence $(\vec x +\vec x', \vec x +\frac 12\vec x', \vec x +\frac 13\vec x', \ldots)$ converge to $\vec x$, and for every $j\in\mathbb N$, $\signH(\vec x +\frac 1j\vec x')=\vec t'$, which means that $a\in \cor(\vec x+\frac 1j\vec x')$. By continuity of $\cor$ we have $a\in \cor(\vec x) = \cor(\vec t)$, which  contradicts the assumption that $a\notin\cor(\vec t)$.

{\bf The ``if'' part for minimal continuity.} To simplify the presentation, we formally define refinements of GISRs as follows.
\begin{dfn} [\bf Refinements of GISRs] Let $\cor^*$ and $\cor$ be a pair of GISR such that for every $\vec x\in\mathbb R^{m!}$, $\cor^*(\vec x) \subseteq \cor(\vec x)$. $\cor^*$ is called a {\em refinement} of $\cor$. If additionally there exists $\vec x\in\mathbb R^{m!}$ such that $\cor^*(\vec x) \subset \cor(\vec x)$, then $\cor^*$ is called a {\em strict refinement} of $\cor$.
\end{dfn}
Suppose for every $\vec t\in \fs$  we have $\cor(\vec t)=\bigcup_{{\vec t}'\in\fsatomic: {\vec t}'\unlhd \vec t}\cor(\vec t')$, and for every $\vec t\in \fsatomic$ we have $|\cor(\vec t)|=1$.  By the ``continuity'' part proved above, $\cor$ is continuous. To prove that $\cor$ is minimally continuous, suppose for the sake of contradiction that $\cor$ has  a strict refinement, denoted by  $\cor^*$. Clearly for every atomic feasible signature $\vec t\in\fsatomic$ we have $\cor^*(\vec t) = \cor(\vec t)$. Therefore, by the ``continuity'' part proved above, for every feasible signature $\vec t\in\fs$, we have 
$$\cor^*(\vec t)\supseteq\bigcup_{{\vec t}'\in\fs: {\vec t}'\unlhd \vec t}\cor^*(\vec t')\supseteq \bigcup_{{\vec t}'\in\fsatomic: {\vec t}'\unlhd \vec t}\cor^*(\vec t')=\bigcup_{{\vec t}'\in\fsatomic: {\vec t}'\unlhd \vec t}\cor (\vec t') =\cor(\vec t),$$
which contradicts the assumption that $\cor^*$ is a strict refinement of $\cor$.

{\bf The ``only if'' part for minimal continuity.}  Suppose $\cor$ is a minimally continuous GISR. We define another GISR $\cor^*$ as follows.
\begin{itemize}
\item For every $\vec t\in\fsatomic$ we let $\cor^*(\vec t)\subseteq \cor(\vec t)$ and $|\cor^*(\vec t)|=1$.
\item For every $\vec t \in \fs$, we let $\cor^*(\vec t) = \bigcup\nolimits_{\vec t'\in\fsatomic: \vec t'\unlhd \vec t}\cor^*(\vec t')$.
\end{itemize}
By the continuity part proved above, $\cor^*$ is continuous. It is not hard to verify that $\cor^*$ refines $\cor$. Therefore, if either condition for minimal continuity does not hold, then $\cor^*$  is a strict refinement of $\cor$, which contradicts the minimality of $\cor$.

This proves Proposition~\ref{prop:char-continuity}.
\end{proof}

Next, we prove some properties about $\ppoly{\vec t}$ that will be frequently used in the proofs of this paper. The proposition has three parts. Part (i) characterizes profiles $P$ whose histogram is in  $\ppoly{\vec t}$; part (ii) characterizes vectors in $\ppolyz{\vec t}$; and part (iii) states that for every atomic signature $\vec t$, $\ppolyz{\vec t}$ is a full dimensional cone in $\mathbb R^{m!}$.
\begin{claim}[\bf Properties of \boldmath $\ppoly{\vec t}$]\label{claim:poly-t}
Given integer $\vH$, any $\vec t\in\fs$,  
\begin{enumerate}[label=(\roman*)]
\item for any integral profile $P$,   $\hist(P)\in \ppoly{\vec t}$ if and only if  $\signH(\hist(P)) = \vec t$;
\item for any $\vec x\in\mathbb R^{m!}$,  $\hist(\vec x)\in \ppolyz{\vec t}$ if and only if $ \vec t\unlhd \signH(\vec x) $;  
\item if $\vec t\in \fsatomic$ then  $\dim(\ppolyz{\vec t})= m!$.
\end{enumerate}
\end{claim}
\begin{proof}
Part (i) follows after the definition. More precisely, $\signH(\hist(P)) = \vec t$ if and only if for every $k\le K$, (1) $t_k= +$ if and only if $\vec h_k\cdot \hist(P)>0$, which is equivalent to $-\vec h_k\cdot \hist(P)\le -1$  because $\vec h_k\in\mathbb Z^{m!}$; (2) likewise,  $t_k= -$ if and only if $\vec h_k\cdot \hist(P)\le -1$, and (3) if $t_k= 0$ if and only if  $\vec h_k\cdot \hist(P)\le 0$ and  $-\vec h_k\cdot \hist(P)\le 0$. This proves Part (i).

Part (ii) also follows after the definition. More precisely, $ \vec x\in \ppolyz{\vec t}$  if and only if for every $k\le K$, (1) $t_k= +$ if and only if $-\vec h_k\cdot \vec x\le 0$, which is equivalent to $[\signH(\vec x)]_k\in \{0,+\}$; (2) likewise,  $t_k= -$ if and only if $\vec h_k\cdot \vec x\le 0$, which is equivalent to $[\signH(\vec x)]_k\in \{0,-\}$, and  (3) if $t_k= 0$ if and only if  $\vec h_k\cdot \vec x\le 0$ and  $-\vec h_k\cdot \vec x\le 0$, which is equivalent to $[\signH(\vec x)]_k=0$. This is equivalent to  $ \vec t\unlhd \signH(\vec x) $.

We now prove Part (iii).  Suppose $\vec t\in \fsatomic$. Let $\vec x\in \ppoly{\vec t}\cap {\mathbb R}_{\ge 0}^{m!}$ denote an arbitrary non-negative vector whose existence is guaranteed by the assumption that $\vec t\in \fsatomic$. Therefore, for every $k\le K$, either $\vec h_k\cdot \vec x\le -1$ or $-\vec h_k\cdot \vec x\le -1$, which means that there exists $\delta>0$ such that any $\vec x'$ with $|\vec x'-\vec x|_\infty<\delta$, we have $\vec h_k\cdot \vec x<0$ or $-\vec h_k\cdot \vec x<0$. This means that $\vec x$ is an interior  point of $\ppolyz{\vec t}$ in $\mathbb R^{m!}$, which implies that $\dim(\ppolyz{\vec t})=m!$.
\end{proof}


\section{Materials for Section~\ref{sec:CC}: Smoothed  {\sc Condorcet Criterion}}
\subsection{Lemma~\ref{lem:sCC-GISR} and Its Proof}
\label{app:lem-GISR}
For any GISR $\cor$, we first define $\region{\CWwin}{\cor}$ (respectively, $\region{\CWlose}{\cor}$) that corresponds to fractional profiles where a Condorcet winner exists and is a co-winner  (respectively, not a co-winner) under $\cor$. $\CWwin$ (respectively, $\CWlose$) stands for ``Condorcet winner wins'' (respectively, ``Condorcet winner loses'').
\begin{align*}
&\region{\CWwin}{\cor} = \{\vec x\in{\mathbb R}^{m!}: \cwinner(\vec x)\cap  \cor(\vec x)\ne \emptyset\}\\
&\region{\CWlose}{\cor} = \{\vec x\in{\mathbb R}^{m!}: \cwinner(\vec x)\cap (\ma\setminus \cor(\vec x))\ne\emptyset\}
\end{align*}

For any set $\region{}{}\subseteq {\mathbb R}^{m!}$, let $\closure{\region{}{}}$ denote the {\em closure} of $\region{}{}$ in ${\mathbb R}^{m!}$, that is, all points in $\region{}{}$ and their limiting points.  
 Next, we introduce    four conditions to present Lemma~\ref{lem:sCC-GISR} below. 
\begin{dfn}\label{dfn:conditions-CC-theorem}
 Given a GISR $\cor$ and $n\in\mathbb N$, we define the following conditions, where $\vec x\in \mathbb R^{m!}$.
\begin{itemize} 
\item {\bf\boldmath Always satisfaction: $\condition{\text{AS}}(\cor, n)$} holds if and only if for all $P\in\ml(\ma)^n$,  $\sat{\CC}(\cor,P)=1$.
\item {\bf\boldmath Robust satisfaction:  $\condition{\text{RS}}(\cor, \vec x)$}  holds if and only if $\vec x\notin  \closure{\region{\CWlose}{\cor}}$.
\item {\bf\boldmath Robust dissatisfaction: $\condition{\text{RD}}(\cor, \vec x)$}  holds if and only if $\cwinner(\vec x)\cap (\ma\setminus \cor(\vec x))\ne \emptyset$.
\item {\bf\boldmath Non-Robust satisfaction: $\condition{\text{NRS}}(\cor, \vec x)$}  holds if and only if $\almostCW(\vec x)\ne \emptyset$ and $\vec x \notin \closure{\region{\CWwin}{\cor}}$.
\end{itemize}
\end{dfn}
In words, $\condition{\text{AS}}(\cor, n)$ means that $\cor$ always satisfies $\CC$ for $n$ agents. Robust satisfaction  $\condition{\text{RS}}(\cor, \vec x)$ states that $\vec x$ is away from the  dissatisfaction instances (i.e., $\region{\CWlose}{\cor}$) by a constant margin. Robust dissatisfaction  $\condition{\text{RD}}(\cor, \vec x)$ states that the Condorcet winner exists under $\vec x$ and is not a co-winner under $\cor$. Robust satisfaction and robust dissatisfaction are not ``symmetric'', because there are two sources of satisfaction: (1) no Condorcet winner exists and (2)  the Condorcet winner exists and is also a  winner, while there is only one source of dissatisfaction:  the Condorcet winner exists but is not a  winner.

The intuition behind Non-Robust satisfaction  $\condition{\text{NRS}}(\cor, \vec x)$ may not  be immediately clear by definition.  It is called ``satisfaction'', because  $\almostCW(\vec x)\ne\emptyset$  implies that $\cwinner(\vec x)=\emptyset$, which means that $\cor$ satisfies $\CC$ at $\vec x$. The reason behind ``non-robust'' is that when a small perturbation $\vec x'$ is introduced, $\umg(\vec x+\vec x')$ often contains a Condorcet winner that is not a co-winner under $\vec x$, because $\vec x$ is constantly far away from $\region{\CWwin}{\cor}$. 

\begin{ex}[\bf The four conditions in Definition~\ref{dfn:conditions-CC-theorem}]
\label{ex:cond-CC-thm}
Let $m=3$ and $n=14$.  Table~\ref{tab:ex-dist} illustrates four distributions, their UMG, the irresolute plurality winners, and their (dis)satisfaction of the four conditions introduced defined in Definition~\ref{dfn:conditions-CC-theorem}.  $\pi^1,\pi^2$, and $\pi'$ are the same as in Example~\ref{ex:sCC-plu} and~\ref{ex:thm-sCC-pos}.  Notice that $\pi'$ is a linear combination of $\pi^1$ and $\pi^2$.
\begin{table}[htp]
\centering
\resizebox{\textwidth}{!}{
\begin{tabular}{|@{\ }c@{\ }|@{\ }c@{\ }|@{\ }c@{\ }|@{\ }c@{\ }|@{\ }c@{\ }|@{\ }c@{\ }|@{\ }c@{\ }|@{\ }c@{\ }|@{\ }c@{\ }|@{\ }c@{\ }|@{\ }c@{\ }|@{\ }c@{\ }|@{\ }c@{\ }|}
\hline & \small $123$& \small $132$& \small $231$& \small $321$& \small $213$& \small $312$& UMG& $\iplu$  winner(s) & $\condition{\text{AS}}$& $\condition{\text{RS}}$& $\condition{\text{RD}}$& $\condition{\text{NRS}}$\\

\hline $\pi^1$& $\frac 14$& $\frac 14$& $\frac 18$& $\frac 18$& $\frac 18$& $\frac 18$ & \begin{minipage}{0.1\linewidth}
\includegraphics[width=1\textwidth]{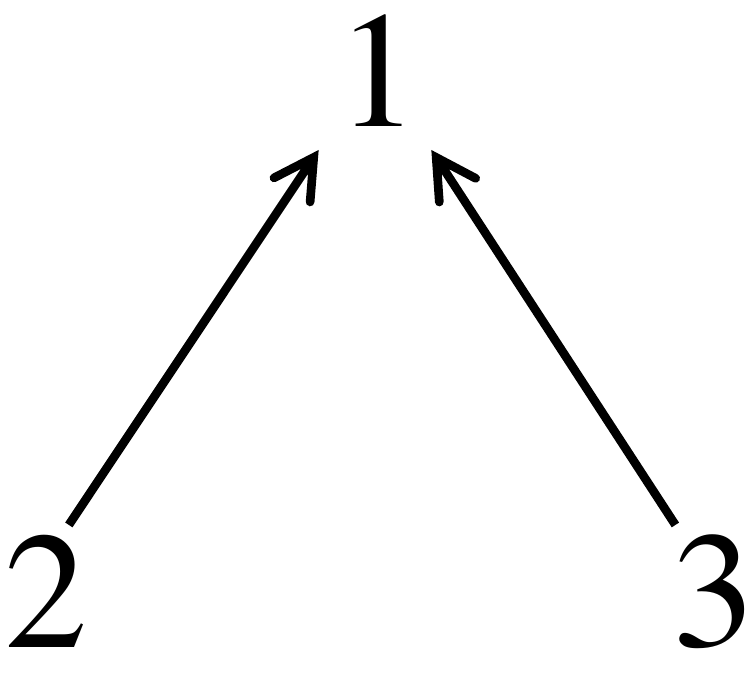}
\end{minipage} & $\{1\}$ & N& N& N&Y\\

\hline $\pi^2$& $\frac 18$& $\frac 18$& $\frac 38$& $\frac 18$& $\frac 18$& $\frac 18$ & \begin{minipage}{0.1\linewidth}
\includegraphics[width=1\textwidth]{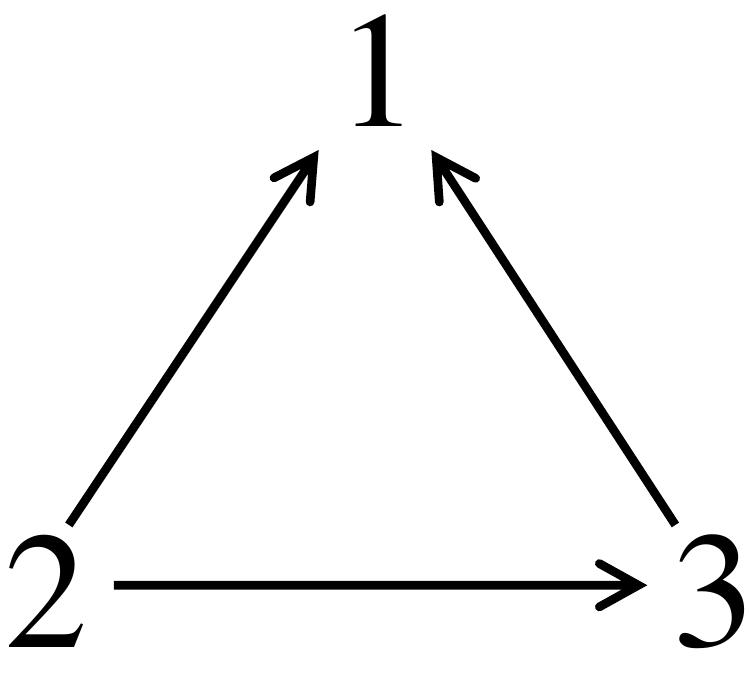}
\end{minipage} & $\{2\}$& N& Y& N&N\\

\hline $\piuni$& $\frac 16$& $\frac 16$& $\frac 16$& $\frac 16$& $\frac 16$& $\frac 16$ & \begin{minipage}{0.1\linewidth}
\includegraphics[width=1\textwidth]{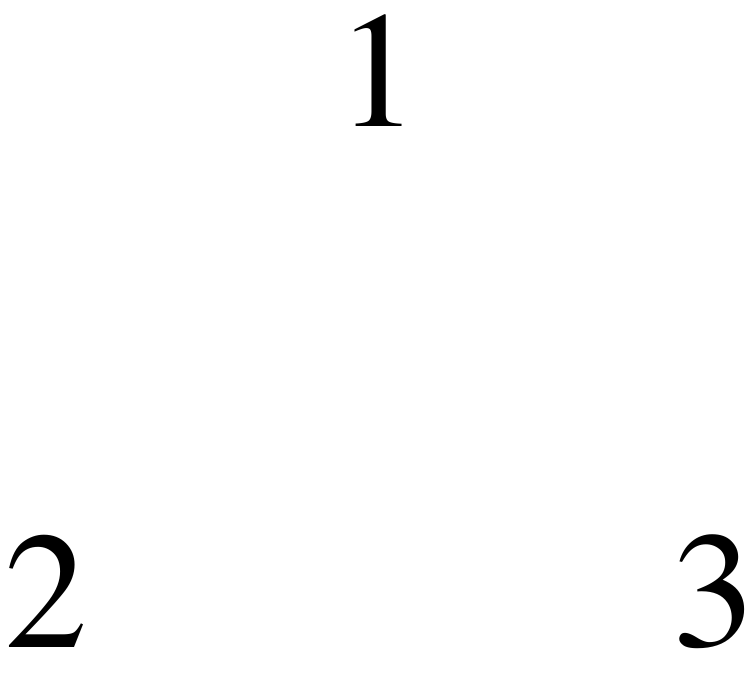}
\end{minipage} & $\{1,2,3\}$& N& N& N&N\\

\hline $\frac{3\pi^1+\pi^2}{4}$& $\frac{7}{32}$& $\frac{7}{32}$& $\frac {3}{16}$& $\frac 18$& $\frac 18$& $\frac 18$ & \begin{minipage}{0.1\linewidth}
\includegraphics[width=1\textwidth]{fig/umg-pi1.pdf}
\end{minipage} & $\{1 \}$& N& N& Y&N\\


\hline
\end{tabular}
}
\caption{Distributions and their (dis)satisfaction of conditions in Definition~\ref{dfn:conditions-CC-theorem}.\label{tab:ex-dist}}
\end{table}

Let $P_{14}$ denote the $14$-profile $\{6\times [1\succ 2\succ 3], 4\times [2\succ 3\succ 1], 4\times [2\succ 1\succ3]\}$. It is not hard to verify that alternative $2$ is the Condorcet winner under $P_{14}$ and $\iplu(P_{14}) = \{1\}$. Therefore, $\condition{\text{AS}}(\iplu, 14)=N$.

\begin{itemize}

\item {\boldmath $\pi^1$.} $\condition{\text{RS}}(\iplu, \pi^1)=N$. To see this, let $\vec x'$ denote the vector that corresponds to the single-vote profile $\{ 2\succ 3\succ 1\}$. For any sufficiently small $\delta>0$, $\pi^1+\delta\vec x'\in \region{\CWlose}{\iplu}$, because $2$ is the Condorcet winner and $1$ is the unique plurality winner.  $\condition{\text{RD}}(\iplu, \pi^1)=N$ because $\cwinner(\pi^1)=\emptyset$. $\condition{\text{NRS}}(\iplu, \pi^1)=Y$ because $\almostCW(\pi^1)=\{2,3\}$, and for any $\vec x'\in\mathbb R^{6}$ and any $\delta>0$ that is sufficiently small, in $\pi^1+\delta\vec x'$ we have that $2$ or $3$ is Condorcet winner and $1$ is the unique plurality winner, which means that $\pi^1+\delta\vec x'\not\in \region{\CWwin}{\cor}$.

\item {\boldmath $\pi^2$.} $\condition{\text{RS}}(\iplu, \pi^2)=Y$ because the plurality score of $2$ is strictly higher than the plurality score of any other alternative, which means that  for any $\vec x'\in\mathbb R^{m!}$, for any  $\delta>0$ that is sufficiently small, $2$ is  the Condorcet winner as well as the unique plurality winner in $\pi^2+ \delta\vec x'$. This means that $\pi^2$ is not in the closure of vectors where $\CC$ is violated.   $\condition{\text{RD}}(\iplu, \pi^2)=N$ because $\cwinner(\pi^2)\cap (\ma\setminus \iplu(\pi^2)) = \{2\}\cap \{1,3\}=\emptyset$. $\condition{\text{NRS}}(\iplu, \pi^2)=N$ because $\almostCW(\pi^2)=\emptyset$.

\item {\boldmath $\piuni$.} $\condition{\text{RS}}(\iplu, \piuni)=N$. To see this, let $\vec x'$ denote the vector that corresponds to the $14$-profile $P_{14}$ defined earlier in this example to prove $\condition{\text{AS}}(\iplu, 14)=N$. For any $\delta>0$ that is sufficiently small, we have $\piuni+\delta\vec x'\in \region{\CWlose}{\iplu}$, because $2$ is the Condorcet winner  and $1$ is the unique plurality winner.  $\condition{\text{RD}}(\iplu, \piuni)=N$ because $\cwinner(\piuni)=\emptyset$. $\condition{\text{NRS}}(\iplu, \piuni)=N$ because $\almostCW(\piuni)=\emptyset$.

\item {\boldmath $\frac{3\pi^1+\pi^2}{4}$.} Let $\pi' = \frac{3\pi^1+\pi^2}{4}$. $\condition{\text{RS}}(\iplu, \pi')=N$ because  $\pi'\in \region{\CWlose}{\iplu}$.  $\condition{\text{RD}}(\iplu, \pi')=Y$ because $\cwinner(\pi')\cap (\ma\setminus \iplu(\pi')) = \{2\}\cap \{2,3\}\ne \emptyset$. $\condition{\text{NRS}}(\iplu, \pi')=N$ because $\almostCW(\pi')=\emptyset$.

\end{itemize}
 
\end{ex}

For any condition $Y$, we use $\neg Y$ to indicate that $Y$ does not hold.  For example, $\neg \condition{\text{AS}}(\cor, n)$ means that $\condition{\text{AS}}(\cor, n)$ does not hold, i.e., there exists $P\in\ml(\ma)^n$ with $\sat{\CC}(\cor,P)=0$. A  GISR rule $r_1$  is a {\em refinement} of another voting rule $r_2$, if for all $\vec x\in\mathbb R^{m!}$, we have $r_1(\vec x) \subseteq r_2(\vec x)$. We note that while the four conditions in Definition~\ref{dfn:conditions-CC-theorem} are not mutually exclusive by definition, they provide a complete characterization of smoothed $\CC$ under any refinement of any minimally continuous int-GISR as shown in the lemma below.

\begin{lem}[\bf Smoothed $\CC$: Minimally Continuous Int-GISRs]
\label{lem:sCC-GISR}{
For any fixed $m\ge 3$, let $\mm= (\Theta,\ml(\ma),\Pi)$ be a strictly positive and closed single-agent preference model, let $\cor$ be a minimally continuous int-GISR and let $r$ be a refinement of $\cor$. 
For any $n\in\mathbb N$ with $2\mid n$, we have
$$\satmin{\CC}{\Pi}(r,n) = \left\{\begin{array}{ll}
1&\text{if }\condition{\text{AS}}(\cor, n)\\
1- \exp(-\Theta(n)) &\text{if } \neg\condition{\text{AS}}(\cor, n)\text{ and } \forall \pi\in \conv(\Pi),  \condition{\text{RS}}(\cor, \pi)\\
\Theta(n^{-0.5}) &\text{if } 
\left\{\begin{array}{l}   \text{(1) }  \forall \pi\in \conv(\Pi),  \neg \condition{\text{RD}}(\cor, \pi) \text{ and} \\
\text{(2) }  \exists \pi\in \conv(\Pi) \text{ s.t. }  \condition{\text{NRS}}(\cor, \pi)\end{array}\right.\\
\exp(-\Theta(n)) &\text{if }\exists \pi\in \conv(\Pi)\text{ s.t. } \condition{\text{RD}}(\cor, \pi)\\
\Theta(1)\wedge(1-\Theta(1)) &\text{otherwise}
\end{array}\right.$$
For any $n\in\mathbb N$ with $2\nmid n$, we have
$$\satmin{\CC}{\Pi}(r,n) = \left\{\begin{array}{ll}
1&\text{same as the }2\mid n\text{ case}\\
1- \exp(-\Theta(n)) &\text{same as the }2\mid n\text{ case}\\
\exp(-\Theta(n)) &\text{if }\exists \pi\in \conv(\Pi)\text{ s.t. } \condition{\text{RD}}(\cor, \pi)\text{ or } \condition{\text{NRS}}(\cor, \pi)\\
\Theta(1)\wedge(1-\Theta(1)) &\text{otherwise}
\end{array}\right.$$
}
\end{lem}

Lemma~\ref{lem:sCC-GISR} can be applied to a wide range of  resolute voting rules because it works for any refinement $r$ (i.e., using any tie-breaking mechanism) of any minimally continuous GISR (which include all voting rules discussed in this paper).  Notice that $r$ is not required to be a GISR,   the L case and the $0$ case never happen, and the conditions of all cases depend on $\cor$ but not $r$.

\begin{ex}[\bf Applications of Lemma~\ref{lem:sCC-GISR} to plurality]
\label{ex:thm1}
Continuing the setting of Example~\ref{ex:cond-CC-thm}, we let $\plu$ denote any  refinement of $\iplu$. We first apply the $2\mid n$ part of Lemma~\ref{lem:sCC-GISR} to the following four cases of $\Pi$ for sufficiently large $n$ using Table~\ref{tab:ex-dist}. The first three cases correspond  to i.i.d.~distributions, i.e., $|\Pi|=1$. In particular, $\Pi=\{\piuni\}$ corresponds to IC. 
\begin{itemize}
\item $\Pi=\{\pi^1,\pi^2\}$. We have $\satmin{\CC}{\Pi}(\plu ,n)=\exp(-\Theta(n))$, that is, the VU case holds. This is because let $\pi' = \frac{3\pi^1+\pi^2}{4}$, we have $\pi'\in \conv(\Pi)$ and $\condition{\text{RS}}(\iplu, \pi')=N$ according to Table~\ref{tab:ex-dist}.
\item $\Pi_1=\{\pi^1\}$. We have $\satmin{\CC}{\Pi_1}(\plu ,n)=\Theta(n^{-0.5})$, that is, the U case holds.

\item $\Pi_2=\{\pi^2\}$. We have $\satmin{\CC}{\Pi_2}(\plu ,n)=1- \exp(-\Theta(n))$, that is, the VL case holds.
\item $\Pi_{\text{IC}}=\{\piuni\}$. We have $\satmin{\CC}{\Pi_{\text{IC}}}(\plu ,n)=\Theta(1)\wedge (1-\Theta(1))$, that is, the M case holds.

\end{itemize}
When $2\nmid n$ and $\Pi_1=\{\pi^1\}$, we have $\satmin{\CC}{\Pi_1}(\plu ,n)=\exp(-\Theta(n))$, that is, the VU case holds.
\end{ex}

\paragraph{\bf Intuitive explanations.} The conditions   in Lemma~\ref{lem:sCC-GISR} can be explained as follows. Take the $2\mid n$ case for example. In light of various multivariate  central limit theorems, the histogram of the randomly-generated profile when the adversary chooses $\vec \pi = (\pi_1,\ldots,\pi_n)$ is concentrated in a $\Theta(n^{-0.5})$ neighborhood of $\sum_{j=1}^n \pi_j$, denoted by $B_{\vec\pi}$. Let $\avg{\vec\pi} = \frac 1n \sum_{j=1}^n \pi_j$, which means that $\avg{\vec\pi}\in\conv(\Pi)$. The condition for the $1$ case is straightforward. Suppose the $1$ case does not happen, then the VL case happens if all distributions in $\conv(\Pi)$, which includes $\avg{\vec\pi}$, are far from instances of dissatisfaction, so that no instance of dissatisfaction is in $B_{\vec\pi}$.  Suppose the VL case does not happen. The U case happens if the min-adversary can find a non-robust satisfaction instance ($\condition{\text{NRS}}(\cor, \pi)$) but cannot find a robust dissatisfaction instance ($\neg\condition{\text{RD}}(\cor, \pi)$). And if the min-adversary can find a robust dissatisfaction instance ($\condition{\text{RD}}(\cor, \pi)$), then $B_{\vec\pi}$ does not contain any instance of satisfaction, which means that the VU case happens. All remaining cases are M cases.

\paragraph{\bf \boldmath Odd vs.~even $n$.} The $2\nmid n$ case also admits a similar explanation. The main difference  is that when $2\nmid n$, the UMG of any $n$-profile must be a complete graph, i.e., no alternatives are tied in the UMG. Therefore, when $\condition{\text{NRS}}(\cor, \pi)$ is satisfied, a Condorcet winner (who is one of the two ACWs in $\pi$) must exist  and constitutes an instance of robust dissatisfaction when $2\nmid n$. On the other hand, it is possible that the two ACWs in $\pi$ are tied in an $n$-profile when $2\mid n$, which constitutes a case where $\CC$ is satisfied because the Condorcet winner does not exist. This 
 happens with probability $\Theta(n^{-0.5})$. This difference leads to the $\Theta(n^{-0.5})$ case  when $2\mid n$, and it becomes part of the $\exp(-\Theta(n))$ case when $2\nmid n$ .

\paragraph{\bf Proof sketch.} Before presenting the formal proof in the following subsection, we present a proof sketch here.  

We first  prove  the  special case $r=\cor$, which is done by applying Lemma~\ref{lem:categorization} in the following three steps.  {\bf Step 1.} Define $\upoly$ that characterizes the satisfaction of $\CC$ under $\cor$, and an almost complement $\aupoly$ of $\upoly$. In fact, we will let $\upoly = \upolynoCW\cup \upolyCWwin$ as in Section~\ref{sec:cat-lemma} and Section~\ref{sec:modeling-par}, and prove that one choice of $\aupoly$ is   the union of polyhedra that represent profiles where the Condorcet winner exists but is not an $\cor$ co-winner.    {\bf Step 2.}  Characterize  $\alpha^*_n$ and $\beta_n$, which is technically the most involved part due to the generality of the theorem. {\bf Step 3.} Formally apply Lemma~\ref{lem:categorization}.

Then,  let $r$ denote an arbitrary refinement of $\cor$. We define a slightly different version of $\CC$, denoted by $\sat{\CC^*}$, whose satisfaction under $\cor $  will be used as a lower bound on the satisfaction of $\CC$ under $r$. For any GISR $\cor$  and any profile $P$, we define
$$\sat{\CC^*}(\cor,P)=\left\{\begin{array}{ll}
1&\text{if } \cwinner(P)=\emptyset \text{ or } \cwinner(P) = \cor(P)\\
0& \text{otherwise}
\end{array}\right.$$
Compared to $\sat{\CC}$,   $\sat{\CC^*}$ rules out profiles $P$ where a Condorcet winner exists and is not the unique winner under $\cor$. Therefore,  for any $\vec \pi\in\Pi^n$, we have 
$$\Pr\nolimits_{P\sim\vec \pi}(\sat{\CC^*}(\cor,P)=1)\le \Pr\nolimits_{P\sim\vec \pi}(\sat{\CC}(r,P)=1) \le \Pr\nolimits_{P\sim\vec \pi}(\sat{\CC}(\cor,P)=1)$$
Then, we prove that smoothed $\CC^*$, i.e.,  $\satmin{\CC^*}{\Pi}(\cor,n)$, asymptotically matches  $\satmin{\CC}{\Pi}(\cor,n)$, which concludes the proof of Lemma~\ref{lem:sCC-GISR}. 

\subsubsection{Proof of Lemma~\ref{lem:sCC-GISR}}
\label{app:proof-lem:sCC-GISR}
 
\begin{proof} The $1$ cases of the theorem is trivial. {\bf \boldmath In the rest of the proof, we assume that the $1$ case does not hold.} That is, there exists an $n$-profile $P$ such that $\cwinner(P)$ exists but is not in $\cor(P)$. We will prove that the theorem holds for any $n>N_{\cor}$, where $N_{\cor}\in\mathbb N$  is a constant that only depends on $\cor$ that will be defined later (in Definition~\ref{dfn:N-cor}). This is without loss of generality, because when $n$ is bounded above by a constant, the $1$ case belongs to the U case (i.e., $\Theta(n^{-0.5}) $) and the VU case (i.e., $\exp(-\Theta(n)) $).

Let $\cor$ be defined by $\vH$ and $g$. We first  prove the theorem for the special case where $r=\cor$, and then show how to modify the proof for general $r$.  For any irresolute voting rule $\cor$, we recall that $\sat{\CC}(\cor,P) = 1$ if and only if either $P$ does not have a Condorcet winner, or the Condorcet winner is a co-winner under $\cor$.

\paragraph{\bf\boldmath Proof for the special case $r=\cor$.} Recall that in this case $\cor$ is a minimally continuous GISR. In light of Lemma~\ref{lem:categorization}, the proof proceeds in the following three steps. {\bf Step 1.} Define $\upoly$ that characterizes the satisfaction of {\sc Condorcet Criterion} of $\cor$ and an almost complement $\aupoly$ of $\upoly$. {\bf Step 2.}  Characterize $\Pi_{\upoly,n}$,   $\Pi_{\aupoly,n}$, $\beta_n$, and $\alpha^*_n$. {\bf Step 3.} Apply Lemma~\ref{lem:categorization}.

\paragraph{\bf \boldmath Step 1: Define $\upoly$ and $\aupoly$.} The definition is similar to the ones presented in Section~\ref{sec:cat-lemma} for plurality. We will define $\upoly = \upolynoCW\cup \upolyCWwin$, where $\upolynoCW$ represents the histograms of profiles that do not have a Condorcet winner, and $\upolyCWwin$ represents histograms of profiles where a Condorcet winner exists and is a co-winner under $\cor$. $\upolynoCW$ is similar to the set defined in~\cite[Proposition~5 in the Appendix]{Xia2020:The-Smoothed}. For completeness we recall its definition using the notation of this paper.  

Recall that $\pair_{a,b}$ is the pairwise difference vector defined in Definition~\ref{dfn:pairdiff}. It follows that for any profile $P$ and any pair of alternatives $a,b$, $\pair_{a,b}\cdot \hist(P)>0$ if and only if there is an edge $a\ra b$ in $\umg(P)$; $\pair_{a,b}\cdot \hist(P)=0$ if and only if $a$ and $b$ are tied in $\umg(P)$.
Then, we use $\pair_{a,b}$'s to define polyhedra that characterize histograms of profiles whose UMGs equal to a given graph $G$.

\begin{dfn}[\bf \boldmath $\ppoly{G}$]
Given an unweighted directed graph $G$ over $\ma$, let  $\pba{G} = \left[\begin{array}{c} \pba{G}_{\text{edge}}\\ \pba{G}_\text{tie}\end{array}\right]$, where  $\pba{G}_{\text{edge}}$ consists of rows $\pair_{b,a}$ for all edges $a\ra b\in G$, and $\pba{G}_{\text{edge}}$ consists of two rows $\pair_{b,a}$ and $\pair_{a,b}$ for each tie $\{a,b\}$ in $G$. Let $\pvbb{G} = [\underbrace{-\vec 1}_{\text{for }\pba{G}_{\text{edge}}},\underbrace{\vec 0}_{\text{for }\pba{G}_\text{tie}}]$ and 
$$\ppoly{G} = \left\{\vec x\in {\mathbb R}^{m!}: \pba{G}\cdot \invert{\vec x}\le \invert{\pvbb{G}} \right\}$$\end{dfn}

Next, we define polyhedra indexed by an alternative $a$ and a feasible signature $\vec t\in \fs$ that characterize the histograms of profiles $P$ where $a$ is the  Condorcet winner and $\signH(P) = \vec t$.
\begin{dfn}[\boldmath $\ppoly{a, \vec t}$]
\label{dfn:poly-a-t}
Given $\vH = (\vec h_1,\ldots,\vec h_K)\in (\mathbb R^{d})^K$,  $a\in\ma$, and $\vec t\in \fs$, we let  $\pba{a,\vec t}=\left[\begin{array}{l}\pba{\text{CW} = a}\\ \pba{\vec t}\end{array}\right]$, where  $\pba{\text{CW} = a}$ consists of pairwise difference vectors $\pair_{b,a}$ for each alternative $b\ne a$, and  $\pba{\vec t}$ is the matrix used to define  $\ppoly{\vec t}$ in Definition~\ref{dfn:poly-H-t}.
Let $\pvbb{a,\vec t} = [\underbrace{-\vec 1}_{\text{for }\pba{\text{CW}=a}},\underbrace{\pvbb{\vec t}}_{\text{for }\pba{\vec t}} ]$ and
$$\ppoly{a, \vec t} = \{\vec x\in {\mathbb R}^{m!}: \pba{a,\vec t}\cdot \invert{\vec x}\le \invert{\pvbb{a,\vec t}} \}$$
\end{dfn}

Next, we use $\ppoly{G}$ and $\ppoly{a,\vec t}$ as building blocks to define $\upoly = \upolynoCW\cup\upolyCWwin$  and an almost complement of $\upoly$, denoted by $\upolyCWlose$. At a high level,  $\upolynoCW$ corresponds to the profiles where no Condorcet winner exists ($\noCW$ represents ``no Condorcet winner''), $\upolyCWwin$ corresponds to profiles where the Condorcet winner exists and is also an $\cor$ co-winner ($\CWwin$ represents ``Condorcet winner wins''), and $\upolyCWlose$ corresponds to profiles where the Condorcet winner exists and is not an $\cor$ co-winner ($\CWlose$ represents ``Condorcet winner loses''). 
\begin{dfn}[\bf \boldmath $\upoly$ and $\upolyCWlose$]\label{dfn:C-CWlose-GISR}
Given an int-GISR characterized by $\vH$ and $g$, we define
\begin{align*}
&\upoly = \upolynoCW\cup\upolyCWwin, \hspace{3mm}\text{where }\upolynoCW = \bigcup\nolimits_{G: \cwinner(G)=\emptyset}\ppoly{G} \text{ and }\upolyCWwin = \bigcup\nolimits_{a\in\ma, \vec t\in \fs: a\in\cor(\vec t)}\ppoly{a,\vec t}\\
&\upolyCWlose = \bigcup\nolimits_{a\in\ma, \vec t\in \fs: a\notin\cor(\vec t)}\ppoly{a,\vec t}
\end{align*}
\end{dfn}
We note that  some $\ppoly{a,\vec t}$ can be empty.  To see  that $\upolyCWlose$ is indeed an almost complement of $\upoly = \upolynoCW\cup \upolyCWwin$, we note that $\upoly\cap \upolyCWlose = \emptyset$, and for any integer vector $\vec x$, 
\begin{itemize}
\item  if $\vec x$ does not have a Condorcet winner then $\vec x\in \upolynoCW\subseteq \upoly$; 
\item if $\vec x$ has a Condorcet winner $a$, which is also an $\cor$ co-winner, then $\vec x\in \ppoly{a,\signH(\vec x)}\subseteq \upolyCWwin\subseteq \upoly$;
\item otherwise $\vec x$ has a Condorcet winner $a$, which is not an $\cor$ co-winner. Then $\vec x\in \ppoly{a,\signH(\vec x)}\subseteq   \upolyCWlose$.
\end{itemize}
Therefore, ${\mathbb Z}^q \subseteq \upoly\cup \upolyCWlose$. 

\paragraph{\bf \boldmath Step 2: Characterize $\Pi_{\upoly,n}$,   $\Pi_{\upolyCWlose,n}$, $\beta_n$, and $\alpha^*_n$.} Recall that  $\beta_n$ and $\alpha_n^*$ are defined by   $\md{\upoly}{\pi}$ and $\md{\upolyCWlose}{\pi}$ for   $\pi\in\conv(\Pi)$ as follows:
\begin{align*}
&\beta_n = \min\nolimits_{\pi\in \conv(\Pi)} \md{\upoly}{\pi} =\min\nolimits_{\pi\in \conv(\Pi)}\max\left( \md{\upolynoCW}{\pi},\md{\upolyCWwin}{\pi}\right)\\
&\alpha_n^* = \max\nolimits_{\pi\in\conv(\Pi)} \md{\upolyCWlose}{\pi}
\end{align*} 
For convenience, we let $\Pi_{\upoly,n}$  denote the distributions in $\conv(\Pi)$, each of which is connected to an edge with positive weight in the activation graph (Definition~\ref{dfn:activation-graph}). Formally, we have the following definition. 
\begin{dfn}[\boldmath $\Pi_{\upoly,n}$]
\label{dfn:Pi-C-n}
Given a set of distributions $\Pi$ over $q$, $\upoly = \bigcup_{i\le I}\cpoly{i}$,  and $n\in\mathbb N$, let  
$$\Pi_{\upoly,n} = \{\pi\in\conv(\Pi): \exists  i\le I\text{ s.t. } \cpolynint{i}\ne\emptyset \text{ and } \pi\in\cpolyz{i}\}$$
\end{dfn}

Table~\ref{tab:dim-max-summary} gives an overview  of the rest of the proof in Step 2, which characterizes $\md{\upoly}{\pi}$ and $\md{\upolyCWlose}{\pi}$ by the membership of $\pi\in\conv(\Pi)$ in $\Pi_{\upolynoCW,n},\Pi_{\upolyCWwin,n}$, and $\Pi_{\upolyCWlose,n}$, respectively, where $n\ge N_\cor$ for a constant $N_\cor$ that will be defined momentarily (in Definition~\ref{dfn:N-cor}). 

\begin{table}[htp]
\centering
\begin{tabular}{|@{\ }l@{\ }|@{\ }c@{\ }|@{\ }c@{\ }|@{\ }c@{\ }|@{\ }c@{\ }|@{\ }c@{\ }|@{\ }c@{\ }|}
\hline
$\pi\in \Pi_{\upolynoCW,n}$  & $\ast$& $\ast$ & N& Y& Y& N\\
\hline $\pi\in \Pi_{\upolyCWwin,n}$  & Y & Y&N& N& N& N\\
\hline $\pi\in \Pi_{\upolyCWlose,n}$   & Y & N&Y& Y& N& N\\
\hline $\md{\upolynoCW}{\pi}$ (Claim~\ref{claim:upoly1})& $\ast$ &$\ast$& $-\frac{n}{\log n}$& $m!$ or $m!-1$& $m!$  &\multirow{4}{*}{N/A} \\
\cline{1-6} $\md{\upolyCWwin}{\pi}$ (Claim~\ref{claim:chara-Pi-CWW-CWL-n}) & $m!$& $m!$& $\le -\frac{n}{\log n}$ &$<0$& $<0$ & \\
\cline{1-6} \begin{tabular}{@{}l@{}}$\md{\upoly}{\pi}=$\\ $\max\left( \md{\upolynoCW}{\pi},\md{\upolyCWwin}{\pi}\right)$\end{tabular} & $m!$& $m!$& $-\frac{n}{\log n}$&  $\md{\upolynoCW}{\pi}$ & $m!$& \\
\cline{1-6} $\md{\upolyCWlose}{\pi}$ (Claim~\ref{claim:chara-Pi-CWW-CWL-n})& $m!$&$-\frac{n}{\log n}$& $m!$ & $m!$ & $-\frac{n}{\log n}$&  \\
\hline
\end{tabular}
\caption{\small $\md{\upoly}{\pi}$ and $\md{\upolyCWlose}{\pi}$ for $\CC$ for  $\pi\in\conv(\Pi)$ and sufficiently large $n$.\label{tab:dim-max-summary}}
\end{table}
We will first specify $N_\cor$ in Step 2.1. Then in Step 2.2, we will characterize $\Pi_{\upolynoCW,n}$ and $\md{\upolynoCW}{\pi}$ in Claim~\ref{claim:upoly1}, and characterize  $\Pi_{\upolyCWwin,n}$, $\md{\upolyCWwin}{\pi}$, $\Pi_{\upolyCWlose,n}$, and $\md{\upolyCWlose}{\pi}$ in Claim~\ref{claim:chara-Pi-CWW-CWL-n}. Finally, in Step 2.3 we will verify $\md{\upoly}{\pi}$ and $\md{\upolyCWlose}{\pi}$ in Table~\ref{tab:dim-max-summary}.

\paragraph{\bf \boldmath Step 2.1.~Specify $N_\cor$.} We first prove the following claim, which provides a sufficient condition for a polyhedron to be active for sufficiently large $N$.

\begin{claim}\label{claim:poly-active-at-all-n} For any polyhedron $\poly$ characterized by integer matrix $\ba$ and $\vbb \le \vec 0$, if $\dim(\polyz)=m!$ and $\poly\cap {\mathbb R}_{>0}^{m!}\ne\emptyset$, then there exists $N\in\mathbb N$ such that for all $n\ge N$, $\poly$ is active at $n$.
\end{claim}
\begin{proof} By Minkowski-Weyl theorem (see e.g., \citep[p.~100]{Schrijver1998:Theory}),  $\poly = \mV+\polyz$, where $\mV$ is a finitely generated polyhedron. Therefore, any affine space containing $\poly$ can be shifted to contain $\polyz$, which means that $\dim(\poly)\ge \dim(\polyz) = m!$. Because $\poly\cap {\mathbb R}_{>0}^{m!}\ne\emptyset$, it contains an interior point (inner point with an full dimensional neighborhood), denoted by $\vec x$, whose $\delta$ neighborhood (for some $0<\delta<1$) in $L_\infty$ is contained in $\poly\cap {\mathbb R}_{>0}^{m!}$. Let $B$ denote the $\delta$ neighborhood of $\vec x$.   Let $N=\frac{m!|\vec x|_1}{\delta }$. Then, because $\vbb\le \vec 0$ and $\frac{N}{|\vec x|_1}\ge 1$, for every $n> N$ and every $\vec x'\in B$ we have 
$$\ba \cdot \invert{\frac{n}{|\vec x|_1}\vec x'}<\frac{n}{|\vec x|_1}\invert{\vbb}\le \invert{\vbb}$$
This means that  $\frac{n}{|\vec x|_1}B\subseteq \poly\cap {\mathbb R}_{>0}^{m!}$. Moreover, it is not hard to verify that $\frac{n}{|\vec x|_1}B$ contains the following non-negative integer $n$ vector 
$$\left(\left\lfloor \frac{n}{|\vec x|_1} x_1 \right\rfloor, \ldots, \left \lfloor \frac{n}{|\vec x|_1} x_{m!-1} \right\rfloor, n-\sum_{i=1}^{m!-1} \left\lfloor \frac{n}{|\vec x|_1} x_{i} \right\rfloor\right)$$
This proves Claim~\ref{claim:poly-active-at-all-n}.
\end{proof}

We now define the constant $N_\cor$ used throughout the proof.
\begin{dfn}[{\bf \boldmath $N_{\cor}$}]\label{dfn:N-cor}
Let $N_{\cor}$ denote a number that is larger than $m^4$ and the maximum $N$ obtain from applying Claim~\ref{claim:poly-active-at-all-n} to all polyhedra $\poly$ in $\upolynoCW$, $\upolyCWwin$, or $\upolyCWlose$ where $\dim(\polyz)=m!$ and $\poly\cap {\mathbb R}_{>0}^{m!}\ne\emptyset$. 
\end{dfn}
\paragraph{\bf \boldmath Step 2.2.~Characterize $\Pi_{\upolynoCW,n}$, $\Pi_{\upolyCWwin,n}$,  and $\Pi_{\upolyCWlose,n}$.} 
\begin{claim}[\bf \boldmath Characterizations of $\Pi_{\upolynoCW,n}$ and $\md{\upolynoCW}{\pi}$]
\label{claim:upoly1}
For any $n\ge m^4$ such that $\neg\condition{\text{AS}}(\cor, n)$ and any distribution $\pi$ over $\ma$, we have 
\begin{itemize}
\item{if $2\mid n$,} then $\pi\in \Pi_{\upolynoCW,n}$ if and only if $\cwinner(\pi)=\emptyset$, and
$$\md{\upolynoCW}{\pi} = \left\{\begin{array}{ll}-\frac{n}{\log n} &\text{if }\cwinner(\pi)\ne \emptyset\\
m! - 1 & \text{if }\almostCW(\pi)\ne \emptyset\\
m! & \text{otherwise (i.e. }\cwinner(\pi)\cup \almostCW(\pi)=\emptyset\text{)}\\
\end{array}\right.$$
\item{if $2\nmid n$,} then $\pi\in \Pi_{\upolynoCW,n}$ if and only if $\cwinner(\pi)\cup \almostCW(\pi)=\emptyset$, and
$$\md{\upolynoCW}{\pi} = \left\{\begin{array}{ll}-\frac{n}{\log n} &\text{if }\cwinner(\pi)\cup \almostCW(\pi)\ne \emptyset\\
m! & \text{otherwise (i.e. }\cwinner(\pi)\cup \almostCW(\pi)=\emptyset\text{)}\\
\end{array}\right.$$
\end{itemize}
\end{claim}
\begin{proof}In the proof we assume that $n\ge m^4$.  We first recall the following characterization of $\ppoly{G}$, where part (i)-(iii) are due to~\cite[Claim 3 in the Appendix]{Xia2020:The-Smoothed} and part (iv) follows after~\cite[Claim 6 in the Appendix]{Xia2020:The-Smoothed}.
\begin{claim}[\bf \boldmath Properties of $\ppoly{G}$~\citep{Xia2020:The-Smoothed}]
\label{claim:h-umg}
For any UMG $G$, 
\begin{enumerate}[label=(\roman*)]
\item for any integral profile $P$,  $\hist(P)\in \ppoly{G}$ if and only if $G = \umg(P)$;
\item for any $\vec x\in\mathbb R^{m!}$,  $\vec x\in \ppolyz{G}$ if and only if $\umg(\vec x)$ is a subgraph of $G$.  
\item $\dim(\ppolyz{G})= m! - \ties(G)$.
\item For any $n\ge m^4$, $\ppoly{G}$ is active at $n$ if (1) $n$ is even, or (2) $n$ is odd and $G$ is a complete graph.
\end{enumerate}
\end{claim}
\paragraph{\bf \boldmath The $2\mid n$ case.} By Claim~\ref{claim:h-umg} (iv), when $n\ge m^4$ and $2\mid n$, every $\ppoly{G}$ is active. This means that $\pi\in \Pi_{\upolynoCW,n}$ if and only if $\pi\in\ppolyz{G}$ for some graph $G$ that does not have a Condorcet winner.  According to Claim~\ref{claim:h-umg} (ii),  this holds if and only if there exists  a supergraph of $\umg(\pi)$ (which can be $\umg(\pi)$ itself) that not have a Condorcet winner, which is equivalent to $\umg(\pi)$ does not have a Condorcet winner, i.e.~$\cwinner(\pi)=\emptyset$. It follows that  $\md{\upolynoCW}{\pi}=-\frac{n}{\log n}$ if and only if $\cwinner(\pi)\ne \emptyset$.

To characterize the $m!-1$ case and the $m!$ case for $\md{\upolynoCW}{\pi}$, we first prove the following claim to characterize graphs whose complete supergraphs all have Condorcet winners.
 
\begin{claim}
\label{claim:comp-sup-has-noCW}
For any unweighted directed graph $G$ over $\ma$, the following conditions are equivalent. (1) Every complete supergraph of  $G$ has a Condorcet winner. (2) $\cwinner(G)\cup \almostCW(G)\ne \emptyset$.
\end{claim}
\begin{proof}
We first prove (1)$\Rightarrow$(2) in the following three cases.
\begin{itemize}
\item {\bf\boldmath Case 1: $|\wcwinner(G)|=1$.} In this case we must have $\cwinner(G)=\wcwinner(G)$, otherwise  there exists an alternative $b$ that is different from the weak Condorcet winner, denoted by $a$, such that $a$ and $b$ are tied in $G$. Notice that $b$ is not a weak Condorcet winner. Therefore, we can complete $G$ by adding $b\ra a$ and breaking other ties arbitrarily, and it is not hard to see that the resulting graph does not have a Condorcet winner, which is a contradiction. 

\item {\bf\boldmath Case 2: $|\wcwinner(G)|=2$.}  Let $\wcwinner(G) = \{a,b\}$. We note that $a$ and $b$ are not tied with any other alternative. Otherwise for the sake of contradiction suppose $a$ is tied with $c\ne b$. Then, we can extend $G$ to a complete graph by assigning $c\ra a$ and $a\ra b$. The resulting complete graph does not have a Condorcet winner, which is a contradiction. This means that $a$ and $b$ are the almost Condorcet winners, and hence (2) holds.

\item {\bf\boldmath Case 3: $|\wcwinner(G)|\ge 3$.}  In this case, we can assign directions of edges between $\wcwinner(G)$ to form a cycle, and then assign arbitrary direction to other missing edges in $G$ to form a complete graph, which does not have a Condorcet winner and is thus a contradiction. 
\end{itemize} 
(2)$\Rightarrow$(1) is straightforward.  If $\cwinner(G)\ne \emptyset$, then any supergraph of $G$ has the same Condorcet winner. If $\almostCW(G) = \{a,b\}\ne \emptyset$, then any complete supergraph of $G$ either has $a$ as the Condorcet winner or has $b$ as the Condorcet winner. This proves Claim~\ref{claim:comp-sup-has-noCW}.
\end{proof}

\noindent {\bf \boldmath The $\md{\upolynoCW}{\pi}=m!-1$ case when $2\mid n$.}   Suppose $\almostCW(\pi)=\{a,b\}$. Let $G^*$ denote a supergraph of $\umg(\pi)$ where ties in $\umg(\pi)$ except $\{a,b\}$ are broken arbitrarily. By Claim~\ref{claim:h-umg} (ii), $\pi\in \ppolyz{G^*}$ and by Claim~\ref{claim:h-umg} (iii), $\ppolyz{G^*} = m!-1$. Recall from Claim~\ref{claim:h-umg} (iv) that $\ppoly{G^*}$ is active at $n$ because we assumed that $n>m^4$. Therefore, $\md{\upolynoCW}{\pi}\ge m!-1$. To see that $\md{\upolynoCW}{\pi}\le m!-1$, we note that for every graph $G$ that does not have a Condorcet winner such that $\pi\in \ppolyz{G}$. By Claim~\ref{claim:h-umg} (ii), $G$ is a supergraph of $\umg(\pi)$. This means that $G$ is not a complete graph,  because by Claim~\ref{claim:comp-sup-has-noCW}, any complete supergraph of $\umg(\pi)$ must have a Condorcet winner. It follows that $\ties(G)\ge 1$ and by Claim~\ref{claim:h-umg} (iii), $\ppolyz{G} \le m!-1$. Therefore, $\md{\upolynoCW}{\pi}=m!-1$.

\noindent {\bf \boldmath The $\md{\upolynoCW}{\pi}=m!$ case  when $2\mid n$.} Suppose $\cwinner(\pi)\cup \almostCW(\pi)=\emptyset$. By Claim~\ref{claim:comp-sup-has-noCW} there exists a complete supergraph $G$ of $\umg(\pi)$ that does not have a Condorcet winner, which means that $\ppoly{G}\subseteq\upolynoCW\subseteq \upoly$. We have $\pi\in \ppolyz{G}$ (Claim~\ref{claim:h-umg} (ii)),  $\dim( \ppolyz{G}) = m!$ (Claim~\ref{claim:h-umg} (iii)), and $\ppoly{G}$ is active at $n$ (Claim~\ref{claim:h-umg} (iv)). Therefore, $\md{\upolynoCW}{\pi}=m!$.

\paragraph{\bf \boldmath The $2\nmid n$ case.}   By Claim~\ref{claim:h-umg} (iv), when $n\ge m^4$ and $2\nmid n$, $\ppoly{G}$ is active if and only if $G$ is a complete graph.  It follows from Claim~\ref{claim:h-umg} (ii) that $\pi\in \Pi_{\upolynoCW,n}$ if and only if $\pi\in\ppolyz{G}$, where $G$ is complete supergraph of $\umg(\pi)$ that does not have a Condorcet winner.  By  Claim~\ref{claim:h-umg} (iii), $\dim(\ppolyz{G})=m!$. Therefore, by Claim~\ref{claim:comp-sup-has-noCW}, $\pi\in \Pi_{\upolynoCW,n}$ if and only if  $\cwinner(\pi)\cup \almostCW(\pi)=\emptyset$. Moreover, whenever $\pi\in \Pi_{\upolynoCW,n}$  we have  {$\md{\upolynoCW}{\pi}=m!$}.

This proves Claim~\ref{claim:upoly1}.
\end{proof}

Recall that we have assumed the $1$ case of the theorem does not hold, that is, $\neg\condition{\text{AS}}(\cor, n)$. The following claim characterizes $\Pi_{\upolyCWwin,n}$, $\md{\upolyCWwin}{\pi}$,   $\Pi_{\upolyCWlose,n}$, and $\md{\upolyCWlose}{\pi}$, when $\neg\condition{\text{AS}}(\cor, n)$.

\begin{claim}[\bf \boldmath Characterizations of  $\Pi_{\upolyCWwin,n}$, $\md{\upolyCWwin}{\pi}$,   $\Pi_{\upolyCWlose,n}$, and $\md{\upolyCWlose}{\pi}$]
\label{claim:chara-Pi-CWW-CWL-n}
Given any strictly positive $\Pi$ and any minimally continuous int-GISR $\cor$,  for any $n\ge N_\cor$ (see Definition~\ref{dfn:N-cor}) such that $\neg\condition{\text{AS}}(\cor, n)$ and any $\pi\in\conv(\Pi)$, 
\begin{align*}
\left[\pi\in  \Pi_{\upolyCWwin,n}\right]&\Leftrightarrow \left[\pi\in\closure{\region{\CWwin}{\cor}}\right]\Leftrightarrow \left[\md{\upolyCWwin}{\pi} = m!\right], \text{ and}\\
\left[\pi\in  \Pi_{\upolyCWlose,n}\right]&\Leftrightarrow \left[\pi\in\closure{\region{\CWlose}{\cor}}\right]\Leftrightarrow \left[\md{\upolyCWlose}{\pi} = m!\right]
\end{align*}
\end{claim}
\begin{proof} We first prove properties of $\ppoly{a, \vec t}$ in the following claim, which has three parts. Part (i) states that $\ppoly{a, \vec t}$ characterizes histograms of the profiles whose signature is $\vec t$ and where alternative $a$ is the Condorcet winner. Part (ii) characterizes the characteristic cone of $\ppoly{a, \vec t}$. Part (iii) characterizes the dimension of the characteristic cone for some cases.

\begin{claim}[\bf \boldmath Properties of $\ppoly{a,\vec t}$]
\label{claim:poly-a-t}
Given  $\vH$, for any $a\in\ma$ and any $\vec t\in \fs$,  
\begin{enumerate}[label=(\roman*)]
\item for any integral profile $P$,  $\hist(P)\in \poly^{a,\vec t}$ if and only if $a$ is the Condorcet winner under $P$ and $\signH(P)=\vec t$;
\item for any  $\vec x\in\mathbb R^{m!}$,  $\vec x\in {\ppolyz{a,\vec t}}$ if and only if $a$ is a weak Condorcet winner under $\vec x$ and $\vec t\unlhd \signH(\vec x)$; 
\item if $\vec t\in \fsatomic$ and $\ppoly{a,\vec t} \ne\emptyset$, then $\dim(\ppolyz{a,\vec t})= m!$.
\end{enumerate}
\end{claim}
\begin{proof}
Part (i) follows after the definition. More precisely, $\pba{\text{CW}=a}\cdot \invert{\hist(P)}\le \invert{- \vec 1}$ if and only if $a$ is the Condorcet winner under $P$, and by Claim~\ref{claim:poly-t} (i), $\pba{\vec t}\cdot \invert{\hist(P)}\le \invert{\pvbb{\vec t}}$ if and only if $\signH(\hist(P))=\vec t$.

Part (ii) also follows after the definition. $\pba{\text{CW}=a}\cdot \invert{\vec x}\le \invert{\vec 0}$ if and only if $a$ is a weak  Condorcet winner under $P$, and by Claim~\ref{claim:poly-t} (ii), $\pba{\vec t}\cdot \invert{\vec x}\le \invert{\vec 0}$ if and only if $\vec t\unlhd \signH(\vec x)$.

To prove Part (iii), suppose $\vec x\in \ppoly{a,\vec t}$.  Because $\vec t\in\fsatomic$, we have $\pvbb{a,\vec t}=-\vec 1$ (Definition~\ref{dfn:poly-a-t}). Therefore,  there exists $\delta>0$ such that for all vector $\vec x'$ such that $|\vec x'-\vec x|_1<\delta$, $\pba{a,\vec t}\cdot\invert{\vec x'}<\invert{\vec 0}$, which means that $\vec x'\in \ppolyz{a,\vec t}$. Therefore, $\ppolyz{a,\vec t}$ contains the $\delta$ neighborhood of $\vec x$, whose dimension is $m!$. This means that $\dim(\ppolyz{a,\vec t})=m!$.
\end{proof}

\paragraph{\bf \boldmath $\left[\pi\in  \Pi_{\upolyCWwin,n}\right]\Leftarrow \left[\pi\in\closure{\region{\CWwin}{\cor}}\right]$.} Suppose $\pi\in\closure{\region{\CWwin}{\cor}}$ and let $(\vec x_1,\vec x_2,\ldots)$ denote an infinite sequence in $\region{\CWwin}{\cor}$ that converges to $\pi$. Because the number of alternatives and the number of feasible signatures are finite, there exists an infinite subsequence $(\vec x_1',\vec x_2',\ldots)$ such that (1) there exists $a\in\ma$ such that for all $j\in\mathbb N$, $\cwinner(\vec x_j') =\{a\}$, and (2) there exists $\vec t\in \fs$ such that $a\in\cor(\vec t)$ and for all $j\in\mathbb N$, $\signH(\vec x_j') =\vec t$. Because $\cor$ is minimally continuous, by Proposition~\ref{prop:char-continuity}, there exists a feasible atomic refinement of $\vec t$, denoted by $\vec t_a\in\fsatomic$, such that $\cor(\vec t_a)=\{a\}$. Therefore, to prove that $\pi\in  \Pi_{\upolyCWwin,n}$, it suffices to prove that (i) for every $n>N_\cor$, $\ppoly{a,\vec t_a}$ is active, and (ii) $\pi\in \ppolyz{a,\vec t_a}$, which will be done as follows. 

{\bf \boldmath (i) $\ppoly{a,\vec t_a}$ is active.} By Claim~\ref{claim:poly-active-at-all-n}, it suffices to prove that $\ppoly{a,\vec t_a}\cap \mathbb R_{> 0}^{m!}\ne \emptyset$. This  is proved by explicitly constructing a vector in $\ppoly{a,\vec t_a}\cap \mathbb R_{\ge 0}^{m!}$ as follows. Because $\vec t_a$ is feasible, there exists $\vec x^a\in\mathbb R^{m!}$ such that $\signH(\vec x^a) = \vec t_a$. Recall that $\pi$ is strictly positive and $(\vec x_1',\vec x_2',\ldots)$ converges to $\pi$, there exists $j\in\mathbb N$ such that $\vec x_j'>\vec 0$. For any $\delta>0$, let $\vec x_\delta = \vec x_j'+\delta\vec x^a$. We let $\delta>0$ denote a sufficiently small number  such that the following two conditions hold.
\begin{itemize}
\item $\vec x_\delta>\vec 0$. The existence of such $\delta$ follows after noticing that $\vec x_j'>\vec 0$.
\item $\cwinner(\vec x_\delta) = \{a\}$. The existence of such $\delta$ is due to the assumption that $\cwinner(\vec x_j') =\{a\}$, which means that $\pba{\text{CW}=a}\cdot\invert{\vec x_j'}<\invert{\vec 0}$, where $\pba{\text{CW}=a}$ is defined in Definition~\ref{dfn:poly-a-t}. Therefore, for any sufficiently small $\delta>0$ we have $\pba{\text{CW}=a}\cdot\invert{\vec x_\delta}<\invert{\vec 0}$, which means that $a$ is the Condorcet winner under $\vec x_\delta$.
\end{itemize}
Because $\vec t_a$ is a refinement of $\vec t$, we have $\signH(\vec x_\delta) = \vec t_a$. Therefore, $\vec x_\delta\in \ppoly{a,\vec t_a}\cap{\mathbb R}_{> 0}^{m!}$. Following  Claim~\ref{claim:poly-active-at-all-n} and the definition of $N_\cor$ (Definition~\ref{dfn:N-cor}), we have that $\ppoly{a,\vec t_a}$ is active for all $n> N_\cor$.

{\bf \boldmath (ii) $\pi\in \ppolyz{a,\vec t_a}$.} Because for all $j\in\mathbb N$, $\pba{\text{CW}=a}\cdot\invert{\vec x_j'}<\invert{\vec 0}$ and $(\vec x_1',\vec x_2',\ldots)$ converge to $\pi$, we have $\pba{\text{CW}=a}\cdot\invert{\pi}\le \invert{\vec 0}$, which means that $a$ is a weak Condorcet winner under $\pi$. It is not hard to verify that for every $k\le K$, if $t_k=+$ (respectively, $-$  and  $0$), then we have $[\signH(\pi)]_k \in \{0,+\}$ (respectively, $\{0,-\}$  and  $\{0\}$). Therefore, $\vec t\unlhd \signH(\pi)$, which means that $\vec t_a\unlhd \signH(\pi)$ because $\vec t_a\unlhd \vec t$. By Claim~\ref{claim:poly-a-t} (ii), we have $\pi\in\ppolyz{a,\vec t_a}$. 

\paragraph{\bf \boldmath $\left[\pi\in  \Pi_{\upolyCWwin,n}\right]\Rightarrow \left[\pi\in\closure{\region{\CWwin}{\cor}}\right]$.} Suppose $\pi\in  \Pi_{\upolyCWwin,n}$, which means that there exists $a\in\ma$ and $\vec t\in \fs$ such that $ \pi\in \ppolyz{a,\vec t}$, $a\in \cor(\vec t)$, $\cwinner(\vec t)=\{a\}$, and $\ppoly{a,\vec t}$ contains a non-negative integer $n$-vector, denoted by $\vec x'$. By Proposition~\ref{prop:char-continuity}, because $\cor$ is minimally continuous, there exists $\vec t_a\in \fsatomic$ such that $\vec t_a\unlhd \vec t$ and $\cor(\vec t_a) = \{a\}$.  Let $\vec x^*\in \ppoly{\vec t_a}$ denote an arbitrary vector, which is guaranteed to exist because $\vec t_a\in \fsatomic$. Because $\vec x'\in \ppoly{a,\vec t}$, we have $\pba{\text{CW}=a}\cdot\invert{\vec x'}  \le \invert{-\vec 1}$. Therefore, there exists $\delta_a$ such that $\pba{\text{CW}=a}\cdot\invert{\vec x'+\delta_a\vec x^*}  < \invert{\vec 0}$. Let $\vec x = \vec x' + \delta_a\vec x^*$. Recall that $\pi\in\ppolyz{a,\vec t}$, which means that $\pba{\text{CW}=a}\cdot\invert{\pi}  \le \invert{\vec 0}$. Therefore, for all $\delta>0$ we have 
$$\pba{\text{CW}=a}\cdot\invert{\pi+\delta\vec x} =\pba{\text{CW}=a}\cdot\invert{\pi}+\delta\pba{\text{CW}=a}\cdot\invert{\vec x}  < \invert{\vec 0},$$
which means that $\cwinner(\pi+\delta\vec x) = \{a\}$. It is not hard to verify that $\signH(\pi+\delta\vec x) = \vec t_a$, which means that $\cor(\pi+\delta\vec x) = \{a\}$. Consequently, for every $\delta>0$ we have $\pi+\delta\vec x\in \region{\CWwin}{\cor}$. Notice that the sequence $(\pi+\vec x, \pi+\frac12 \vec x,\ldots)$ converges to $\pi$. Therefore, $\pi\in \closure{\region{\CWwin}{\cor}}$.

\paragraph{\bf \boldmath  $\left[\pi\in\closure{\region{\CWwin}{\cor}}\right]\Rightarrow \left[\md{\upolyCWwin}{\pi} = m!\right]$.} Continuing the proof of the $\left[\pi\in  \Pi_{\upolyCWwin,n}\right]\Rightarrow \left[\pi\in\closure{\region{\CWwin}{\cor}}\right]$ part, because $\pi$ is strictly positive and $(\pi+\vec x, \pi+\frac12 \vec x,\ldots)$ converges to $\pi$, there exists $j\in \mathbb N$ such that $\pi+\frac1j \vec x >\vec 0$. Recall that $\cwinner(\pi+\frac1j \vec x)=\{a\}$, $\signH(\pi+\frac1j \vec x)=\vec t_a$, and $\vec t_a$ is atomic, we have
$$\pba{\text{CW}=a}\cdot\invert{ \pi+\frac1j \vec x}<\invert{\vec 0}\text{ and }\pba{\vec t_a}\cdot\invert{ \pi+\frac1j \vec x}<\invert{\vec 0}$$ 
Therefore, there exists $\ell>0$ such that  
$$\pba{\text{CW}=a}\cdot\invert{\ell(\pi+\frac1j \vec x)}\le \invert{-\vec 1}\text{ and }\pba{\vec t_a}\cdot\invert{\ell(\pi+\frac1j \vec x)}\le\invert{-\vec 1},$$
which means that $\ell(\pi+\frac1j \vec x)\in \ppoly{a,\vec t_a}\cap \mathbb R_{> 0}^{m!}\ne \emptyset$. by Claim~\ref{claim:poly-a-t} (iii), we have $\md{\upolyCWwin}{\pi} = m!$.

\paragraph{\bf \boldmath  $\left[\md{\upolyCWwin}{\pi} = m!\right]\Rightarrow \left[\pi\in  \Pi_{\upolyCWwin,n}\right]$}  follows after the definition of $\Pi_{\upolyCWwin,n}$. More concretely, $\md{\upolyCWwin}{\pi} = m!$ means that there exists a polyhedron $\poly\subseteq \upolyCWwin$ such that the weight on the edge $(\pi,\poly)$ in the activation graph is $m!$, which implies that $\pi\in  \Pi_{\upolyCWwin,n}$.

The proofs for $\Pi_{\upolyCWlose,n}$ and $\md{\upolyCWlose}{\pi}$  are similar to the proofs for $\Pi_{\upolyCWwin,n}$ and $\md{\upolyCWwin}{\pi}$. For completeness, we include the full proofs below.

\paragraph{\bf \boldmath $\left[\pi\in  \Pi_{\upolyCWlose,n}\right]\Leftarrow \left[\pi\in\closure{\region{\CWlose}{\cor}}\right]$.} Suppose $\pi\in\closure{\region{\CWlose}{\cor}}$ and let $(\vec x_1,\vec x_2,\ldots)$ denote an infinite sequence in $\region{\CWlose}{\cor}$ that converges to $\pi$. Because the number of alternatives and the number of feasible signatures are finite, there exists an infinite subsequence $(\vec x_1',\vec x_2',\ldots)$ such that (1) there exists $a\in\ma$ such that for all $j\in\mathbb N$, $\cwinner(\vec x_j') =\{a\}$, and (2) there exists $\vec t\in \fs$ such that $a\notin\cor(\vec t)$ and for all $j\in\mathbb N$, $\signH(\vec x_j') =\vec t$. Let $b\in \cor(\vec t)$ denote an arbitrary winner. Because $\cor$ is minimally continuous, by Proposition~\ref{prop:char-continuity}, there exists a feasible atomic refinement of $\vec t$, denoted by $\vec t_b$, such that $\cor(\vec t_b)=\{b\}$. Therefore, to prove that $\pi\in  \Pi_{\upolyCWlose,n}$, it suffices to show that (i) for every $n>N$, $\ppoly{a,\vec t_b}$ is active, and (ii) $\pi\in \ppolyz{a,\vec t_b}$. 

{\bf \boldmath (i) $\ppoly{a,\vec t_b}$ is active.} We will apply Claim~\ref{claim:poly-active-at-all-n} to prove that $\ppoly{a,\vec t_b}$ is active at every $n>N$. In fact, it suffices to prove that $\ppoly{a,\vec t_b}\cap \mathbb R_{>0}^{m!}\ne \emptyset$. This will be proved by explicitly constructing a vector in $\ppoly{a,\vec t_b}\cap \mathbb R_{>0}^{m!}$ as follows. Because $\vec t_b$ is feasible, there exists $\vec x^b\in\mathbb R^{m!}$ such that $\signH(\vec x^b) = \vec t_b$. Recall that $\pi$ is strictly positive and $(\vec x_1',\vec x_2',\ldots)$ converges to $\pi$, there exists $j\in\mathbb N$ such that $\vec x_j'>\vec 0$. For any $\delta>0$, let $\vec x_\delta = \vec x_j'+\delta\vec x^b$. We let $\delta>0$ denote a sufficiently small number  such that the following two conditions hold.
\begin{itemize}
\item $\vec x_\delta>\vec 0$. The existence of such $\delta$ follows after noticing that $\vec x_j'>\vec 0$.
\item $\cwinner(\vec x_\delta) = \{a\}$. The existence of such $\delta$ is due to the assumption that $\cwinner(\vec x_j') =\{a\}$, which means that $\pba{\text{CW}=a}\cdot\invert{\vec x_j'}<\invert{\vec 0}$, where $\pba{\text{CW}=a}$ is defined in Definition~\ref{dfn:poly-a-t}. Therefore, for any sufficiently small $\delta>0$ we have $\pba{\text{CW}=a}\cdot\invert{\vec x_\delta}<\invert{\vec 0}$, which means that $a$ is the Condorcet winner under $\vec x_\delta$.
\end{itemize}
Because $\vec t_b$ is a refinement of $\vec t$, we have $\signH(\vec x_\delta) = \vec t_b$. Therefore, $\vec x_\delta\in \ppoly{a,\vec t_b}\cap{\mathbb R}_{> 0}^{m!}$. Following  Claim~\ref{claim:poly-active-at-all-n} and the definition of $N_\cor$ (Definition~\ref{dfn:N-cor}), we have that $\ppoly{a,\vec t_a}$ is active for all $n> N_\cor$.

{\bf \boldmath (ii) $\pi\in \ppolyz{a,\vec t_b}$.} Because for all $j\in\mathbb N$, $\pba{\text{CW}=a}\cdot\invert{\vec x_j'}<\invert{\vec 0}$ and $(\vec x_1',\vec x_2',\ldots)$ converge to $\pi$, we have $\pba{\text{CW}=a}\cdot\invert{\pi}\le \invert{\vec 0}$, which means that $\pi$ is a weak Condorcet winner. It is not hard to verify that for every $k\le K$, if $t_k=+$ (respectively, $-$  and  $0$), then we have $[\signH(\pi)]_k \in \{0,+\}$ (respectively, $\{0,-\}$  and  $\{0\}$). Therefore, $\vec t\unlhd \signH(\pi)$, which means that $\vec t_b\unlhd \signH(\pi)$ because $\vec t_b\unlhd \vec t$. It follows that $\pba{\vec t_b}\cdot\invert{\pi}\le \invert{\vec 0}$. This means that $\pi\in\ppolyz{a,\vec t_b}$. 

\paragraph{\bf \boldmath $\left[\pi\in  \Pi_{\upolyCWlose,n}\right]\Rightarrow \left[\pi\in\closure{\region{\CWlose}{\cor}}\right]$.} Suppose $\pi\in  \Pi_{\upolyCWlose,n}$, which means that there exists $a\in\ma$ and $\vec t\in \fs$ such that $ \pi\in \ppolyz{a,\vec t}\subseteq \upolyCWlose$, $a\notin \cor(\vec t)$, $\cwinner(\pi)=\{a\}$, and $\ppoly{a,\vec t}$ contains a non-negative integer $n$-vector, denoted by $\vec x'$. Let $b\in \cor(\vec t)$ denote an arbitrary co-winner. By Proposition~\ref{prop:char-continuity}, because $\cor$ is minimally continuous, there exists $\vec t_b\in \fsatomic$ such that $\vec t_b\unlhd \vec t$ and $\cor(\vec t_b) = \{b\}$.  Let $\vec x^*\in \ppoly{\vec t_b}$ denote an arbitrary vector whose existence is guaranteed by the assumption that $\vec t_b\in \fsatomic$. Because $\vec x'\in \ppoly{a,\vec t}$, we have $\pba{\text{CW}=a}\cdot\invert{\vec x'}  \le \invert{-\vec 1}$. Therefore, there exists $\delta_a$ such that $\pba{\text{CW}=a}\cdot\invert{\vec x'+\delta_a\vec x^*}  < \invert{\vec 0}$. Let $\vec x = \vec x' + \delta_a\vec x^*$. Recall that $\pi\in\ppolyz{a,\vec t}$, which means that $\pba{\text{CW}=a}\cdot\invert{\pi}  \le \invert{\vec 0}$. Therefore, for all $\delta>0$ we have 
$\pba{\text{CW}=a}\cdot\invert{\pi+\delta\vec x}  < \invert{\vec 0}$, which means that $\cwinner(\pi+\delta\vec x) = \{a\}$. It is not hard to verify that $\signH(\pi+\delta\vec x) = \vec t_b$, which means that $\cor(\pi+\delta\vec x) = \{b\}$. This means that  for every $\delta>0$ we have $\pi+\delta\vec x\in \region{\CWlose}{\cor}$. Notice that $\pi$ is the limit of the sequence $(\pi+\vec x, \pi+\frac12 \vec x,\ldots)$. Therefore, $\pi\in \closure{\region{\CWlose}{\cor}}$.

\paragraph{\bf \boldmath  $\left[\pi\in\closure{\region{\CWlose}{\cor}}\right]\Rightarrow \left[\md{\upolyCWlose}{\pi} = m!\right]$.} Continuing the proof of the $\left[\pi\in  \Pi_{\upolyCWlose,n}\right]\Rightarrow \left[\pi\in\closure{\region{\CWlose}{\cor}}\right]$ part, because $\pi$ is strictly positive and $(\pi+\vec x, \pi+\frac12 \vec x,\ldots)$ converges to $\pi$, there exists $j\in \mathbb N$ such that $\pi+\frac1j \vec x >\vec 0$. Recall that $\cwinner(\pi+\frac1j \vec x)=\{a\}$, $\signH(\pi+\frac1j \vec x)=\vec t_b$, and $\vec t_b$ is atomic, which means that $\pba{\text{CW}=a}\cdot\invert{ \pi+\frac1j \vec x}<\invert{\vec 0}$ and $\pba{\vec t_b}\cdot\invert{ \pi+\frac1j \vec x}<\invert{\vec 0}$. Therefore, there exists $\ell>0$ such that  
$$\pba{\text{CW}=a}\cdot\invert{\ell(\pi+\frac1j \vec x)}\le \invert{-\vec 1}\text{ and }\pba{\vec t_b}\cdot\invert{\ell(\pi+\frac1j \vec x)}\le\invert{-\vec 1},$$
which means that $\ell(\pi+\frac1j \vec x)\in \ppoly{a,\vec t_b}\cap \mathbb R_{> 0}^{m!}\ne \emptyset$. by Claim~\ref{claim:poly-a-t} (iii), we have $\md{\upolyCWlose}{\pi} = m!$.

\paragraph{\bf \boldmath  $\left[\md{\upolyCWlose}{\pi} = m!\right]\Rightarrow \left[\pi\in  \Pi_{\upolyCWlose,n}\right]$}  follows after the definition.

This proves Claim~\ref{claim:chara-Pi-CWW-CWL-n}.
\end{proof}

We are now ready to verify Table~\ref{tab:dim-max-summary} column by column as follows.
\begin{itemize}
\item {\bf \boldmath $\ast$YY:} $\md{\upoly}{\pi}=\max(\md{\upolynoCW}{\pi},\md{\upolyCWwin}{\pi})$, and by  Claim~\ref{claim:chara-Pi-CWW-CWL-n} we have $\md{\upolyCWwin}{\pi}=m!$. The 
$\md{\upolyCWlose}{\pi}$ part also follows after  Claim~\ref{claim:chara-Pi-CWW-CWL-n}.
\item {\bf \boldmath $\ast$YN:} The $\md{\upoly}{\pi}$ part follows after  Claim~\ref{claim:chara-Pi-CWW-CWL-n}. Recall that we have assumed $\neg\condition{\text{AS}}(\cor, n)$. This means that there exists an $n$-profile $P$ such that $\cwinner(P)\ne\emptyset$ and $\cwinner(P)\not\subseteq \cor(P)$. Let $\{ a\}=\cwinner(P)$ and $\vec t = \signH(P)$. It follows that $\hist(P)\in \calH_n^{a,\vec t,\mathbb Z}\ne \emptyset$ and $\ppoly{a,\vec t}\subseteq \upolyCWlose$. Because $\pi\not\in \Pi_{\upolyCWlose,n}$, according to the definition of the activation graph (Definition~\ref{dfn:activation-graph}),   the weight on the edge $(\pi, \ppoly{a,\vec t})$ is $-\frac{n}{\log n}$, and the weight on any edge connected to $\pi$ is not positive. Therefore, $\md{\upolyCWlose}{\pi}=-\frac{n}{\log n}$.
\item {\bf \boldmath NNY:} The $\md{\upoly}{\pi}$ part follows after  the definition. The $\md{\upolyCWlose}{\pi}$ part follows after  Claim~\ref{claim:chara-Pi-CWW-CWL-n}. 
\item {\bf \boldmath YNY:} Recall that the ``N'' means that  $\pi\notin \Pi_{\upolyCWwin,n}$, which implies that   $\md{\upolyCWwin}{\pi}<0$. Therefore, $\md{\upoly}{\pi} = \max(\md{\upolynoCW}{\pi},\md{\upolyCWwin}{\pi})$, which means that $\md{\upoly}{\pi}=\md{\upolynoCW}{\pi}$. The $\md{\upolyCWlose}{\pi}$ part follows after Claim~\ref{claim:chara-Pi-CWW-CWL-n}.
\item {\bf \boldmath YNN:} We first prove the $\md{\upoly}{\pi}$ part. Because in this case $\pi\in \Pi_{\upolynoCW,n}$ and $\pi\notin \Pi_{\upolyCWwin,n}$, by the definition of $\Pi_{\upolynoCW,n}$ and $\Pi_{\upolyCWwin,n}$,  we have $\md{\upolynoCW}{\pi}\ge 0$ and $\md{\upolyCWwin}{\pi}\le -\frac{n}{\log n}$.  Therefore, $\md{\upoly}{\pi}=\md{\upolynoCW}{\pi}$. It suffices to prove that $\md{\upolynoCW}{\pi}=m!$. Recall from Proposition~\ref{prop:almost-complement-union} that 
$$ \upolyznoCW \cup \upolyzCWwin\cup \upolyzCWlose=\mathbb R^{m!}$$
Therefore, there exists a polyhedron $\poly$ in $\upolynoCW$, $\upolyCWwin$, or $\upolyCWlose$ such that $\pi\in \polyz$ and $\dim(\polyz)=m!$.  We now prove that  $\poly$ is indeed active.  Because $\pi$ is strictly positive and $\polyz$ is convex, $\polyz$ contains an interior point in $\mathbb R_{>0}^{m!}$, denoted by $\vec x$.  Formally, let $\vec x' $ denote an arbitrary interior point of $\polyz$. It is not hard to verify that for some sufficiently small $\delta>0$, $\vec x =\dfrac{\pi+\delta\vec x'}{1+\delta}\in \mathbb R_{>0}^{m!}$ is an interior point of $\polyz$.

Suppose $\poly$ is characterized by $\ba$ and $\vbb$. Then, we have $\ba\cdot\invert{\vec x}<\invert{\vec 0}$. Therefore, there exists $\ell>0$ such that  $\ba\cdot\invert{\ell \vec x}\le \invert{\vbb}$, which means that $\ell \vec x\in \poly\cap {\mathbb R}_{>0}^{m!}\ne\emptyset$. By Claim~\ref{claim:poly-active-at-all-n} and the definition of $N_\cor$ (Definition~\ref{dfn:N-cor}), $\poly$ is active at every $n>N_\cor$. 

Recall that in the YNN  case we have $\pi\notin \Pi_{\upolyCWwin,n} $ and $\pi\notin \Pi_{\upolyCWlose,n} $. Therefore, $\poly\subseteq \upolynoCW$, which means that $\md{\upolynoCW}{\pi}=m!= \md{\upoly}{\pi}$. Following a similar reasoning as in the ``$\ast$YN'' case, we have $\md{\upolyCWlose}{\pi} = -\frac{\log n}{n}$.

\item {\bf \boldmath NNN:} This case is impossible because as proved in the ``YNN'' case, for all $n>N_\cor$,  $\pi\notin \Pi_{\upolyCWwin,n}$ and  $\pi\notin \Pi_{\upolyCWlose,n}$ implies that $\pi\in \Pi_{\upolynoCW,n}$.
\end{itemize}

\paragraph{\bf \boldmath Step 3: Apply Lemma~\ref{lem:categorization}.}  In this step, we apply the $\inf$ part of Lemma~\ref{lem:categorization} by combining and simplifying conditions in Table~\ref{tab:dim-max-summary}.

\begin{itemize}
\item {\bf \boldmath The $0$  case} never holds when $n\ge m^4$, because any complete graph is the UMG of some $n$-profile~\cite[Claim 6 in the Appendix]{Xia2020:The-Smoothed}. In particular, any complete graph where there is no Condorcet winner is the UMG of an $n$-profile.
\item {\bf \boldmath The $1$  case} holds if and only if $\cor$ satisfies $\CC$ for all $n$ profile $P$, i.e.~$\condition{\text{AS}}(\cor,n)$.
\item {\bf \boldmath The VU case.} According to the $\inf$ part of Lemma~\ref{lem:categorization}, the VU case holds if and only if $\beta_n = -\frac{n}{\log n}$. Note that  we do not need to assume $\condition{\text{AS}}(\cor,n)$ in the VU case. According to Table~\ref{tab:dim-max-summary}, $\beta_n = -\frac{n}{\log n}$ if and only if there exists $\pi\in\conv(\Pi)$ such that  the ``NNY'' column holds. Recall that the ``NNN'' column is impossible for any $n>N_\cor$. Therefore, the ``NNY'' column holds for $\pi\in\conv(\Pi)$ if and only if  $\pi\notin \Pi_{\upolynoCW,n}$ and  $\pi\notin \Pi_{\upolyCWwin,n}$, which is equivalent to the following condition by Claim~\ref{claim:chara-Pi-CWW-CWL-n} 
\begin{equation}
\label{eq:VU-case-Ctwo-Cstar}
\pi\notin \Pi_{\upolynoCW,n}\text{ and }\pi\notin  \closure{\region{\CWwin}{\cor}}
\end{equation}
Next, we simplify (\ref{eq:VU-case-Ctwo-Cstar})  for $2\mid n$ and $2\nmid n$,  respectively.
\begin{itemize}
\item {\bf \boldmath $2\mid n$.} By the $2\mid n$ part of Claim~\ref{claim:upoly1},  $\pi\notin \Pi_{\upolynoCW,n}$ if and only if  $\pi$ has a Condorcet winner. We   prove that in this case (\ref{eq:VU-case-Ctwo-Cstar}) is equivalent to:
\begin{equation}\label{eq:VU-simplified}
\cwinner(\pi)\cap (\ma\setminus \cor(\pi))\ne \emptyset
\end{equation}
{\bf \boldmath (\ref{eq:VU-case-Ctwo-Cstar})$\Rightarrow$(\ref{eq:VU-simplified}).}  Suppose $\pi$ has a Condorcet winner, denoted by $a$, and (\ref{eq:VU-case-Ctwo-Cstar}) holds. For the sake of contradiction suppose that (\ref{eq:VU-simplified}) does not hold, which means that $a\in\cor(\pi)$. Then, following a similar construction as in the proof of Claim~\ref{claim:chara-Pi-CWW-CWL-n}, the minimal continuity of $\cor$ implies that there exist  $\vec t_a\in \fsatomic$ with $\vec t_a\unlhd \signH(\pi)$ and $\cor(\vec t_a)=\{a\}$, and $\vec x\in \ppoly{\vec t_a}$ such that for every $\delta>0$ we have $\pi+\delta\vec x\in \region{\CWwin}{\cor}$. Then  $(\pi+\vec x, \pi+\frac12 \vec x,\ldots)$ converges to $\pi$, which contradicts the assumption that $\pi\notin  \closure{\region{\CWwin}{\cor}}$.

{\bf \boldmath (\ref{eq:VU-simplified})$\Rightarrow$(\ref{eq:VU-case-Ctwo-Cstar}).} Let $a\in \cwinner(\pi)\cap (\ma\setminus \cor(\pi))$, which means that $\{a\}=\cwinner(\pi)$ and $a\notin\cor(\pi)$. Suppose for the sake of contradiction that (\ref{eq:VU-case-Ctwo-Cstar}) does not hold.  Due to Claim~\ref{claim:upoly1}, we have  $\pi\notin \Pi_{\upolynoCW,n}$.  Therefore, $\pi\in \closure{\region{\CWwin}{\cor}}$. This means that there exists a sequence $(\vec x_1,\vec x_2,\ldots)$ in $\region{\CWwin}{\cor}$ that converge to $\pi$. It follows that  there exists $j^*\in\mathbb N$ such that for all $j>j^*$, $a$ is the Condorcet winner under $\vec x_j$, which means that $a\in \cor(\vec x_j)$ because $\vec x_j\in \region{\CWwin}{\cor}$. Therefore, by the continuity of $\cor$, we have $a\in \cor(\pi)$, which means that $\cwinner(\pi)\cap (\ma\setminus \cor(\pi))= \emptyset$. This is a contradiction to (\ref{eq:VU-simplified}).

Therefore, when $2\mid n$, the VU case holds if and only if  there exists $\pi\in\conv(\Pi)$ such that (\ref{eq:VU-simplified})  holds, which is as described in the statement of the theorem, i.e.
$$\exists \pi\in \conv(\Pi)\text{ s.t. } \condition{\text{RD}}(\cor, \pi)$$

\item {\bf \boldmath $2\nmid n$.} By the $2\nmid n$ part of Claim~\ref{claim:upoly1},  $\pi\notin \Pi_{\upolynoCW,n}$ is equivalent to $\cwinner(\pi)\cup \almostCW(\pi)\ne \emptyset$, i.e.~either $\cwinner(\pi)\ne\emptyset$ or $\almostCW(\pi)\ne \emptyset$. When $\cwinner(\pi)\ne\emptyset$, as in the $2\mid n$ case, (\ref{eq:VU-case-Ctwo-Cstar}) becomes (\ref{eq:VU-simplified}).  When $\almostCW(\pi)\ne \emptyset$, (\ref{eq:VU-case-Ctwo-Cstar}) becomes $\condition{\text{NRS}}(\cor,\pi)=1$. Therefore, when $2\nmid n$ the VU case holds if and only if the condition in the statement of the theorem holds, i.e.
$$\exists \pi\in \conv(\Pi)\text{ s.t. } \condition{\text{RD}}(\cor, \pi)\text{ or }\condition{\text{NRS}}(\cor, \pi)$$
\end{itemize}
\item {\bf \boldmath The U case.}  According to the $\inf$ part of Lemma~\ref{lem:categorization}, the U case holds if and only if $0\le \beta_n <m!$. According to Table~\ref{tab:dim-max-summary}, $0\le \beta_n <m!$ if and only if 
\begin{itemize}
\item []{\bf (i)}  for every $\pi\in\conv(\Pi)$ the NNY column of Table~\ref{tab:dim-max-summary} does not hold, and 
\item []{\bf (ii)}  there exists $\pi\in \conv(\Pi)$ such that  the  YNY  column of Table~\ref{tab:dim-max-summary} holds  and $\md{\upolynoCW}{\pi}<m!$. 
\end{itemize}
Part (ii) can be simplified as follows. By Claim~\ref{claim:upoly1},  $\md{\upolynoCW}{\pi}<m!$  if and only if $2\mid n$ and $\almostCW(\pi)\ne \emptyset$, and in this case $\md{\upolynoCW}{\pi}=m!-1$. We show that it suffices to additionally require that $\pi\notin \Pi_{\upolyCWwin,n}$ (i.e.~the ``N''), or in other words, given $\md{\upolynoCW}{\pi}=m!-1$, $\pi\notin \Pi_{\upolyCWwin,n}$ implies $\pi\in \Pi_{\upolyCWlose,n}$ (i.e.~the second ``Y''). Suppose for the sake of contradiction that $\md{\upolynoCW}{\pi}=m!-1$, $\pi\notin \Pi_{\upolyCWwin,n}$, and $\pi\notin \Pi_{\upolyCWlose,n}$. Notice that this corresponds to the  ``YNN'' column in Table~\ref{tab:dim-max-summary}, which means that $\md{\upolynoCW}{\pi}=m!$, which is a contradiction. By Claim~\ref{claim:chara-Pi-CWW-CWL-n}, $\pi\notin \Pi_{\upolyCWwin,n}$ if and only if $\pi\notin\closure{\region{\CWwin}{\cor}}$. Therefore, part (ii) is equivalent to
$$\exists \pi\in \conv(\Pi) \text{ s.t. }  \condition{\text{NRS}}(\cor, \pi) $$
Summing up, the U case holds if and only if the condition in the statement of the theorem holds, i.e.
$$ 2\mid n,  \text{ and (1) }  \forall \pi\in \conv(\Pi),  \neg \condition{\text{RD}}(\cor, \pi), \text{ and (2) }  \exists \pi\in \conv(\Pi) \text{ s.t. }  \condition{\text{NRS}}(\cor, \pi) $$

\item {\bf \boldmath The L case}  never holds when $n\ge m^4$, because according to Table~\ref{tab:dim-max-summary}, $\alpha_n^* = \max_{\pi\in\conv(\Pi)} \md{\upolyCWlose}{\pi}$ is either $-\frac{n}{\log n}$ or $m!$, which means that it is never in $[0,m!)$.

\item {\bf \boldmath The VL case.} According to the $\inf$ part of Lemma~\ref{lem:categorization}, the VL case holds if and only if the $1$ case does not hold and $\alpha_n^* = -\frac{n}{\log n}$. According to Table~\ref{tab:dim-max-summary}, this happens in the ``$\ast$YN'' column or  the ``YNN'' column, which is equivalent to only requiring that the last ``N'' holds (because ``NNN'' is impossible), i.e.~for all  $\pi\in\conv(\Pi)$, $\pi\notin \Pi_{\upolyCWlose,n}$. By Claim~\ref{claim:chara-Pi-CWW-CWL-n},  the VL case holds if and only if
if and only if the condition in the statement of the theorem holds, i.e.
$$\neg \condition{\text{AS}}(\cor,n) \text{ and }\forall \pi\in\conv(\Pi),  \condition{\text{RS}}(\cor,\pi)$$

\item {\bf \boldmath The M case} corresponds to the remaining cases.
\end{itemize}

\paragraph{\bf \boldmath Proof for general refinement $r$ of $\cor$.} We now turn to the proof of the theorem for an arbitrary refinement of $\cor$, denoted by $r$. We first define a slightly different version of $\CC$, denoted by $\sat{\CC^*}$, which will be used as the lower bound on the (smoothed) satisfaction of the regular $\CC$. For any GISR $\cor$  and any profile $P$, we define
$$\sat{\CC^*}(\cor,P)=\left\{\begin{array}{ll}
1&\text{if } \cwinner(P)=\emptyset \text{ or } \cwinner(P) = \cor(P)\\
0& \text{otherwise}
\end{array}\right.$$
In words, $\sat{\CC^*}(\cor,P)=$ if and only if (1) the Condorcet winner does not exist, or (2) the Condorcet winner exists and is the {\em unique} winner under $P$ according to $\cor$. Compared to the standard Condorcet criterion $\sat{\CC}$,   $\sat{\CC^*}$ rules out profiles $P$ where a Condorcet winner exists and is not the unique   winner. $\sat{\CC^*}$ and $\sat{\CC}$ coincide with each other when $\cor$ is a resolute rule.   Because for any profile $P$ we have $r(P)\subseteq \cor(P)$, for any $\vec \pi\in\Pi^n$ we have 
$$\Pr\nolimits_{P\sim\vec \pi}(\sat{\CC^*}(\cor,P)=1)\le \Pr\nolimits_{P\sim\vec \pi}(\sat{\CC}(r,P)=1) \le \Pr\nolimits_{P\sim\vec \pi}(\sat{\CC}(\cor,P)=1)$$
Therefore,
\begin{equation}
\label{eq:scc-r-bounds}
\satmin{\CC^*}{\Pi}(\cor,n)\le \satmin{\CC}{\Pi}(r,n)\le \satmin{\CC}{\Pi}(\cor,n)
\end{equation}
n order to prove the theorem, it suffices to prove that the lower  bound in (\ref{eq:scc-r-bounds}), i.e., $\satmin{\CC^*}{\Pi}(\cor,n)$, has the same dichotomous characterization as $\satmin{\CC}{\Pi}(\cor,n)$. To this end, we first define a union of polyhedra, denoted by $\upoly'$, and its almost complement $\upolyCWlose'$ that are similar to Definition~\ref{dfn:C-CWlose-GISR} as follows.

\begin{dfn}[\bf \boldmath $\upoly'$ and $\upolyCWlose'$]\label{dfn:C-prime-GISR}
Given an int-GISR characterized by $\vH$ and $g$, we define
\begin{align*}
&\upoly' = \upolynoCW\cup\upolyCWwin', \hspace{3mm}\text{where } \upolyCWwin' = \bigcup\nolimits_{a\in\ma, \vec t\in \fs: \cor(\vec t)=\{a\}}\ppoly{a,\vec t}\\
&\upolyCWlose' = \bigcup\nolimits_{a\in\ma, \vec t\in \fs:\cor(\vec t)\ne\{a\}}\ppoly{a,\vec t}
\end{align*}
\end{dfn}
Notice that $\upolynoCW$ used in Definition~\ref{dfn:C-prime-GISR}  was define in Definition~\ref{dfn:C-CWlose-GISR}. Just  like $\upolyCWlose$ is an almost complement of $\upoly$,  $\upolyCWlose'$ is   an almost complement of $\upoly'$. Formally, we first note that $\upoly'\cap \upolyCWlose' = \emptyset$, and for any integer vector $\vec x$, 
\begin{itemize}
\item  if $\vec x$ does not have a Condorcet winner then $\vec x\in \upolynoCW\subseteq \upoly'$; 
\item if $\vec x$ has a Condorcet winner $a$, which is the unique $\cor$ winner, then $\vec x\in \ppoly{a,\signH(\vec x)}\subseteq \upolyCWwin'\subseteq \upoly$;
\item otherwise $\vec x$ has a Condorcet winner $a$, which is either not a $\cor$ co-winner or $|\cor(\vec x)|\ge 2$. In both cases  $\vec x\in \ppoly{a,\signH(\vec x)}\subseteq   \upolyCWlose'$.
\end{itemize}
Therefore, ${\mathbb Z}^q \subseteq \upoly'\cup \upolyCWlose'$. The proof for  $\satmin{\CC^*}{\Pi}(\cor,n)$ is similar to the proof for $\satmin{\CC}{\Pi}(\cor,n)$ presented earlier. The main difference is that $\upoly$, $\upolyCWwin$, and $\upolyCWlose$ are replaced by $\upoly'$, $\upolyCWwin'$, and $\upolyCWlose'$, respectively. The key part is to prove a counterpart to Table~\ref{tab:dim-max-summary}, which follows after proving $\Pi_{\upolyCWwin',n} = \Pi_{\upolyCWwin,n}$ and $\Pi_{\upolyCWlose',n} =\Pi_{\upolyCWlose,n}$ for every $n>N_{\cor}$, as formally shown in the following claim.

\begin{claim}
\label{claim:Pi-prime-equivalence}
For any $n>N_{\cor}$, we have $\Pi_{\upolyCWwin',n} = \Pi_{\upolyCWwin,n}$ and $\Pi_{\upolyCWlose',n} =\Pi_{\upolyCWlose,n}$.
\end{claim}
\begin{proof}
The main difference between $\upolyCWwin' $ (respectively, $\upolyCWlose'$) and $\upolyCWwin$ (respectively, $ \upolyCWlose$) is the memberships of polyhedra $\ppoly{a,\vec t}$, where $a\in \cor(\vec t)$ and $\cor(\vec t)\ge 2$. Therefore, to prove the claim, it suffices to show that the membership of $\ppoly{a,\vec t}$ does not affect $\Pi_{\upolyCWwin',n}$ (respectively, $\Pi_{\upolyCWlose',n}$) compared to $\Pi_{\upolyCWwin,n}$ (respectively, $\Pi_{\upolyCWlose,n}$). 


It suffices to show that for any polyhedron $\ppoly{a,\vec t}$, where $a\in \cor(\vec t)$ and $\cor(\vec t)\ge 2$, for any $\pi\in \conv(\Pi)$ and any $n>N_{\cor}$, if $\ppoly{a,\vec t}$ is active and $\pi\in \ppolyz{a,\vec t}$, then there exist $\ppolyz{a,\vec t_a}\subseteq \upolyCWwin\cap \upolyCWwin'$ and $\ppolyz{a,\vec t_b}\subseteq \upolyCWlose\cap \upolyCWlose'$ such that (1) $\ppolyz{a,\vec t_a}$ and $\ppolyz{a,\vec t_b}$ are active at $n$, and (2) $\pi\in \ppolyz{a,\vec t_a}\cap \ppolyz{a,\vec t_b}$. In other words, if a distribution $\pi\in\conv(\Pi)$ is in $\upolyCWwin' $, $\upolyCWlose'$, $\upolyCWwin$, or $ \upolyCWlose$ due to $\ppoly{a,\vec t}$, then it is also in the same set without considering its edge to $\ppoly{a,\vec t}$ in the activation graph.  As we will see soon, (1) follows after the assumption that $n>N_{\cor}$ and (2) follows after the minimal continuity of $\cor$. Formally, the proof proceeds in the following three steps.
\begin{itemize}
\item [\bf (i)] {\bf \boldmath Define $\vec t_a$ and $\vec t_b$.} Let $b\ne a$ denote a co-winner under $\pi$, i.e.,~$\{a,b\}\subseteq \cor(\pi)$.  Because $\cor$ is minimally continuous, by Proposition~\ref{prop:char-continuity}, there exists a feasible atomic signature $\vec t_a\in\fsatomic$ (respectively,  $\vec t_b\in\fsatomic$) such that $\vec t_a\unlhd \vec t$ (respectively,  $\vec t_b\unlhd \vec t$) and  $\cor(\vec t_a) = \{a\}$ (respectively,  $\cor(\vec t_b) = \{b\}$).

\item [\bf (ii)] {\bf \boldmath Prove that $\ppolyz{a,\vec t_a}$ and $\ppolyz{a,\vec t_b}$ are active at any $n>N_{\cor}$.} 
Because $\vec t_a$ is feasible, there exists $\vec x\in\mathbb R^{m!}$ such that $\signH(\vec x) = \vec t_a$. Therefore, recall that $\pi$ is strictly positive (by $\epsilon$), for some  sufficiently small $\delta>0$, we have $\pi+\delta\vec x\in \mathbb R_{>0}^{m!}$, $\cwinner(\pi+\delta\vec x) = \{a\}$, and $\signH(\pi+\delta\vec x) = \vec t_a$. This means that $\pi+\delta\vec x$ is an interior point of $\ppoly{a,\vec t_a}$ (which also means that $\dim(\ppoly{a,\vec t_a})=m!$). Recall that the $\vbb$ part of $\ppoly{a,\vec t_a}$ (Definition~\ref{dfn:poly-H-t} and~\ref{dfn:poly-a-t}) is non-positive, we have $\ppoly{a,\vec t_a}\subseteq \ppolyz{a,\vec t_a}$, which means that $\dim(\ppolyz{a,\vec t_a})=m!$ as well.  Therefore, according to Claim~\ref{claim:poly-active-at-all-n} and the definition of $N_{\cor}$ (Definition~\ref{dfn:N-cor}), $\ppoly{a,\vec t_a}$ is active at any $n>N_{\cor}$. Similarly, we have that $\ppoly{a,\vec t_b}$ is active at any $n>N_{\cor}$. 

\item [\bf (iii)] {\bf \boldmath Prove that $\pi\in \ppolyz{a,\vec t_a}\cap \ppolyz{a,\vec t_b}$.}  Recall that $\pi\in \ppolyz{a,\vec t}$. Therefore, according to Claim~\ref{claim:poly-a-t} (ii), we have $\vec t\unlhd \signH(\pi)$, which means that $\vec t_a\unlhd \signH(\pi)$, because $\vec t_a\unlhd \vec t$. By Claim~\ref{claim:poly-a-t} (ii) again, we have $\pi\in \ppolyz{a,\vec t_a}$. Similarly, we can prove that $\pi\in \ppolyz{a,\vec t_b}$.
\end{itemize}
This completes the proof of Claim~\ref{claim:Pi-prime-equivalence}.
\end{proof}
 
Therefore, $\satmin{\CC^*}{\Pi}(\cor,n)$ has the same characterization as $\satmin{\CC}{\Pi}(\cor,n)$, which concludes the proof of Lemma~\ref{lem:sCC-GISR} due to (\ref{eq:scc-r-bounds}).
\end{proof}

\subsection{Proof of Theorem~\ref{thm:sCC-scoring}}
\label{app:proof-thm:sCC-scoring}

\appThm{thm:sCC-scoring}
{Smoothed $\CC$: Integer Positional Scoring Rules}{Let $\mm= (\Theta,\ml(\ma),\Pi)$ be a strictly positive and closed single-agent preference model, let $\cor_{\vec s}$ be a  minimally continuous int-GISR and let $r_{\vec s}$ be a refinement of $\cor_{\vec s}$. For any $n\ge 8m+49$ with $2\mid n$, we have
$$\satmin{\CC}{\Pi}(r_{\vec s},n) = \left\{\begin{array}{ll}
1- \exp(-\Theta(n)) &\text{if } \forall \pi\in\conv(\Pi),  |\wcw(\pi)|\times |\cor(\pi)\cup \wcw(\pi)|\le 1\\
\Theta(n^{-0.5}) &\text{if }  
\left\{\begin{array}{l}   \text{(1) }\forall  \pi\in \conv(\Pi), \cwinner(\pi)\cap (\ma\setminus \cor_{\vec s}(\pi))=\emptyset \text{ and} \\
\text{(2) }  \exists \pi\in \conv(\Pi)\text{ s.t. } |\almostCW(\pi)\cap (\ma\setminus \cor_{\vec s}(\pi))|=2\end{array}\right.
\\
\exp(-\Theta(n)) &\text{if }\exists \pi\in \conv(\Pi)\text{ s.t. } \cwinner(\pi)\cap (\ma\setminus \cor_{\vec s}(\pi))\ne \emptyset\\
\Theta(1)\text{ and }1-\Theta(1) &\text{otherwise}
\end{array}\right.$$
For any $n\ge 8m+49$ with $2\nmid n$, we have
$$\satmin{\CC}{\Pi}(r_{\vec s},n) = \left\{\begin{array}{ll}
1- \exp(-\Theta(n)) &\text{same as the }2\mid n\text{ case}\\
\exp(-\Theta(n)) &\text{if }\exists \pi\in \conv(\Pi)\text{ s.t. } \left\{\begin{array}{l}   \text{(1) }\cwinner(\pi)\cap (\ma\setminus \cor_{\vec s}(\pi))\ne \emptyset \text{ or} \\
\text{(2) }  |\almostCW(\pi)\cap (\ma\setminus \cor_{\vec s}(\pi))|=2\end{array}\right.
\\
\Theta(1)\text{ and }1-\Theta(1) &\text{otherwise}
\end{array}\right.$$
}
\begin{proof}We apply Lemma~\ref{lem:sCC-GISR} to prove the theorem. For any integer irresolute positional scoring rule $\cor_{\vec s}$, we prove the following claim to simplify $\closure{\region{\CWwin}{\cor_{\vec s}}}$ and $\closure{\region{\CWlose}{\cor_{\vec s}}}$.
\begin{claim}\label{claim:closure-CWW-CWL-scoring}
For any $\pi\in\conv(\Pi)$,  
\begin{align*}
\left[\pi\in \closure{\region{\CWwin}{\cor_{\vec s}}} \right]&\Leftrightarrow \left[\wcw(\pi)\cap  \cor_{\vec s}(\pi)\ne \emptyset\right]\\
\left[\pi\in \closure{\region{\CWlose}{\cor_{\vec s}}}\right] &\Leftrightarrow \left[\exists a\ne b \text{ s.t. } a\in \wcw(\pi)\text{ and }b\in \cor_{\vec s}(\pi)\right]
\end{align*}
\end{claim}
\begin{proof} The proof is done in the following steps.
\paragraph{\bf \boldmath $\left[\pi\in \closure{\region{\CWwin}{\cor_{\vec s}}} \right]   \Rightarrow \left[\wcw(\pi)\cap  \cor_{\vec s}(\pi)\ne \emptyset\right]$.} Suppose $\pi\in \closure{\region{\CWwin}{\cor_{\vec s}}} $, which means that there exists a sequence $(\vec x_1,\vec x_2,\ldots)$ in $\region{\CWwin}{\cor_{\vec s}}$ that converges to $\pi$. It follows that there exists an alternative $a\in\mathbb \ma$ and a subsequence of $(\vec x_1,\vec x_2,\ldots)$, denoted by $(\vec x_1',\vec x_2', \ldots)$ such that for every $j\in\mathbb N$, $\cwinner(\vec x_j')=\{a\}$ and $a\in \cor_{\vec s}(\vec x_j')$. This means that the following holds. 
\begin{itemize}
\item $a$ is a weak Condorcet winner under $\pi$. Notice that for any $b\ne a$ and any $j\in\mathbb N$, we have $\pair_{b,a}\cdot \vec x_j'<0$, which means that $\pair_{b,a}\cdot \pi\le 0$.
\item $a\in \cor_{\vec s}(\pi)$. Notice that for any $b\ne a$ and any $j\in\mathbb N$, the total score of $a$ is higher than or equal to the total score of $b$ in $\vec x_j'$. Therefore, the same holds for $\pi$, which means that $a\in \cor_{\vec s}(\pi)$.
\end{itemize}
Therefore, $a$ is a weak Condorcet winner as well as a $\cor_{\vec s}$ co-winner, which implies $\wcw(\pi)\cap  \cor_{\vec s}(\pi)\ne \emptyset$.

\paragraph{\bf \boldmath $\left[\pi\in \closure{\region{\CWwin}{\cor_{\vec s}}}  \right] \Leftarrow \left[ \wcw(\pi)\cap  \cor_{\vec s}(\pi)\ne \emptyset \right] $.} Suppose $\wcw(\pi)\cap  \cor_{\vec s}(\pi)\ne \emptyset$ and let $a\in \wcw(\pi)\cap  \cor_{\vec s}(\pi)$. We will explicitly construct a sequence of vectors in $\region{\CWwin}{\cor_{\vec s}}$ that converges to $\pi$. Let $\sigma$ denote a cyclic permutation among $\ma\setminus\{a\}$ and let $P$ denote the following $(m-1)$-profile
\begin{equation}\label{eq:P-cyclic}
P = \{\sigma^i( a\succ \others ): 1\le i\le m-1\}
\end{equation}
It is not hard to verify that $\cwinner(P)=\cor_{\vec s}(P)=\{a\}$. Therefore, for any $\delta>0$ we have 
$$\cwinner(\pi+\delta\cdot \hist(P))=\cor_{\vec s}(\pi+\delta\cdot \hist(P))=\{a\},$$
which means that $\pi+\delta \cdot \hist(P)\in \closure{\region{\CWwin}{\cor_{\vec s}}}$. It follows that $(\pi+\frac1j  \hist(P): j\in\mathbb N)$ is a sequence in $\closure{\region{\CWwin}{\cor_{\vec s}}}$ that converges to $\pi$, which means that $\pi\in \closure{\region{\CWwin}{\cor_{\vec s}}} $.

\paragraph{\bf \boldmath $\left[\pi\in \closure{\region{\CWlose}{\cor_{\vec s}}}   \right] \Rightarrow \left[\exists a\ne b \text{ s.t. } a\in \wcw(\pi)\text{ and }b\in \cor_{\vec s}(\pi) \right] $.} Suppose $\pi\in \closure{\region{\CWlose}{\cor_{\vec s}}} $, which means that there exists a sequence $(\vec x_1,\vec x_2,\ldots)$ in $\region{\CWlose}{\cor_{\vec s}}$ that converges to $\pi$. It follows that there exists a pair of different alternatives $a,b\in\mathbb \ma$ and a subsequence of $(\vec x_1,\vec x_2,\ldots)$, denoted by $(\vec x_1',\vec x_2',\ldots)$ such that for every $j\in\mathbb N$, $\cwinner(\vec x_j')=\{a\}$ and $b\in \cor_{\vec s}(\vec x_j')$. Following a similar proof as in the $\region{\CWlose}{\cor_{\vec s}}$ part, we have that $a$ is a weak Condorcet winner under $\pi$ and $b\in \cor_{\vec s}(\pi)$. 

\paragraph{\bf \boldmath $\left[\pi\in \closure{\region{\CWlose}{\cor_{\vec s}}}   \right] \Leftarrow \left[\exists a\ne b \text{ s.t. } a\in \wcw(\pi)\text{ and }b\in \cor_{\vec s}(\pi) \right] $.}  Let $a\ne b$ be two alternatives such that $a\in \wcw(\pi)\text{ and }b\in \cor_{\vec s}(\pi)$. We define a profile $P$ where $\cwinner(P)=\{a\}$ and $\cor_{\vec s}(P)=\{b\}$, whose existence is guaranteed by the following claim, which is slightly different from~\cite[Theorem~6]{Fishburn74:Paradoxes}   for scoring vectors $\vec s = (s_1,\ldots,s_m)$ with $s_1>s_2>\cdots>s_m$. 
\begin{claim}
\label{claim:CW-score-winner-diff}
For any $m\ge 3$, any positional scoring rule with scoring vector $\vec s = (s_1,\ldots,s_m)$ where $s_1>s_m$,  any $n\ge 8m+ 49$, and any pair of different alternatives $a\ne b$, there exists an $n$-profile $P$ such that $\cwinner(P)=\{a\}$ and $\cor_{\vec s}(P) = \{b\}$.
\end{claim}
\begin{proof} We explicitly construct an $n$-profile $P$ where the Condorcet winner exists and is different from the unique $\cor_{\vec s}$ winner. Then, we apply a permutation over $\ma$ to $P$ to  make $a$ the Condorcet and $b$ the unique $\cor_{\vec s}$ winner.
The construction is done in two cases: $s_2 = s_m$ and $s_2>s_m$. 
\begin{itemize}
\item {\bf Case 1: \boldmath $s_2 = s_m$.} In this case  $\cor_{\vec s}$ corresponds to the plurality rule. We let 
\begin{align*}
P = \left\lfloor \frac{n-1}{2}\right\rfloor\times [2\succ 1\succ 3\succ \others]+ \left\lfloor \frac{n-3}{2}\right\rfloor\times [3\succ 1\succ 2\succ \others]&\\
+ \left(n+1-2\left\lfloor \frac{n-1}{2}\right\rfloor\right)\times [1\succ 2\succ 3\succ \others]&
\end{align*}
It is not hard to verify that the alternative $1$ is the Condorcet winner and $2$ is the unique plurality winner.

\item {\bf Case 2: \boldmath $s_2 >s_m$.} Let $2\le k\le m-1$ denote the smallest number such that $s_k>s_{k+1}$. Let $A_1 = [4\succ \cdots\succ k+1]$ and $A_2 = [k+2\succ \cdots\succ m]$, and let $P^*$ denote the following $7$-profile.
\begin{align*}
P^* = \{3\times [1\succ 2\succ A_1\succ 3\succ A_2] +2\times [2\succ 3\succ A_1\succ1\succ A_2] &\\
 +   [3\succ 1\succ A_1\succ 2\succ A_2]+[2\succ 1 \succ A_1\succ 3\succ A_2]&\}
\end{align*}
It is not hard to verify that $1$ is the Condorcet winner under $P^*$, and the total score of $1$ is $3s_1+2s_2+2s_{k+1}<3s_1+3s_2+s_{k+1}$, which is the total score of $2$. Note that the total score of any alternative in $A_1$ is $7s_k$, which might be larger than the score of $2$. If $3s_1+3s_2+s_{k+1}\ge 7s_k$, then we let $b=2$; otherwise we let $b=4$.  Let $P_b$ denote the following $(m-1)$-profile  that will be used as a tie-breaker. Let $\sigma$ denote an arbitrary cyclic permutation among $\ma\setminus\{b\}$.
$$P_b = \{\sigma^i([b\succ \others]): 1\le i\le m-1\}$$
Let 
$$P = \left\lfloor \frac{n-m+1}{7}\right\rfloor\times P^* + P_b+\left(n-m+1-7\left\lfloor \frac{n-m+1}{7}\right\rfloor\right)\times [b\succ \others]$$
It is not hard to verify that when $n\ge 8m+49 $, $\cwinner(P)=\{1\}$, $\cor_{\vec s}(P) = \{b\}$,  and $b\ne 1$. 
\end{itemize}
This proves Claim~\ref{claim:CW-score-winner-diff}.
\end{proof}
Let $P$ denote the profile guaranteed by Claim~\ref{claim:CW-score-winner-diff}. For any $\delta>0$ we have 
$$\cwinner(\pi+\delta \cdot \hist(P))=\{a\}\text{ and }\cor(\pi+\delta \cdot \hist(P))=\{b\},$$
which means that $\pi+\delta \cdot \hist(P)\in \closure{\region{\CWlose}{\cor_{\vec s}}}$. It follows that $(\pi+\frac1j  \hist(P): j\in\mathbb N)$ is a sequence in $ {\region{\CWlose}{\cor_{\vec s}}}$ that converges to $\pi$, which means that $\pi\in \closure{\region{\CWlose}{\cor_{\vec s}}} $. This proves Claim~\ref{claim:closure-CWW-CWL-scoring}.
\end{proof}

Claim~\ref{claim:closure-CWW-CWL-scoring} implies that for all $n\ge 8m+49$, the $1$ case doe not hold, i.e., $\condition{\text{AS}}(\cor_{\vec s}, n)=0$. We now apply Claim~\ref{claim:closure-CWW-CWL-scoring} to simplify the conditions in Lemma~\ref{lem:sCC-GISR}. 
\begin{itemize}
\item $\condition{\text{RS}}(\cor_{\vec s}, \pi)$. By definition, this holds if and only if $\pi\notin  \closure{\region{\CWlose}{\cor_{\vec s}}}$, which is equivalent to $\nexists a\ne b \text{ s.t. } a\in \wcw(\pi)\text{ and }b\in \cor_{\vec s}(\pi)$. In other words, either $\wcw(\pi)=\emptyset$ or  ($\wcw(\pi) = \cor_{\vec s}(\pi)$ and $|\wcw(\pi)| =1$). Notice that $\cor_{\vec s}(\pi)\ne\emptyset$. Therefore, $\condition{\text{RS}}(\cor_{\vec s}, \pi)$ is equivalent to $|\wcw(\pi)|\times |\cor_{\vec s}(\pi)\cup \wcw(\pi)|\le 1$.
\item $\condition{\text{NRS}}(\cor_{\vec s}, \pi)$. By definition, this holds if and only if  $\almostCW(\pi)\ne \emptyset$ and $\pi \notin \closure{\region{\CWwin}{\cor_{\vec s}}}$, which is equivalent to $\almostCW(\pi)\ne \emptyset$ and $\wcw(\pi)\cap  \cor_{\vec s}(\pi)= \emptyset$. The latter is equivalent to $\wcw(\pi)\cap (\ma\setminus \cor_{\vec s}(\pi))=\wcw(\pi)$. We note that when $\almostCW(\pi)\ne \emptyset$, we have $\wcw(\pi)=\almostCW(\pi)$. Therefore, $\condition{\text{NRS}}(\cor_{\vec s}, \pi)$ is equivalent to $|\almostCW(\pi)\cap (\ma\setminus \cor_{\vec s}(\pi))|=2$.
\end{itemize}
Theorem~\ref{thm:sCC-scoring} follows after Lemma~\ref{lem:sCC-GISR} with the simplified conditions discussed above.
\end{proof}

\subsection{Definitions, Full Statement, and Proof for Theorem~\ref{thm:sCC-MRSE}}
\label{app:proof-thm:sCC-MRSE}

For any $O\in \ml(\ma)$, any $1\le i<i'\le m$, and any $a\in\ma$, let $O[i]$ denote the alternative ranked at the $i$-th place in $O$, let $O[i,i']$ denote the set of alternatives ranked from the $i$-th place to the $i'$-th place in $O$, and let $O^{-1}[a]$ denote the rank of $a$ in $O$. For any $A\subseteq \ma$ and any $\vec x\in \mathbb R^{m!}$ that represents the histogram of a profile, let $\vec x|_{A}\in \mathbb R^{|A|!}$ denote the histogram of the profile restricted to alternatives in $A$. 
\begin{ex} Let $O = [3\rhd 1\rhd 2]$.\footnote{Again, we use $\rhd$ in contrast to $\succ$ to indicate that $O$ is a parallel universe instead of an agent's preferences.} We have $O[2] = 1$, $O^{-1}(2) = 3$, and $O[2,3] = \{1,2\}$.  Let $\hat\pi$ denote the (fractional) profile in Figure~\ref{fig:ex-m3}. We have $\hat\pi|_{O[2,3]} = (\underbrace{0.5}_{1\succ 2},\underbrace{0.5}_{2\succ 1})$.
\end{ex}

\begin{dfn}[\bf Parallel universes  and possible losing rounds under MRSE rules] 
\label{dfn:PU-LR}
For any MRSE rule $\cor = (\cor_2,\ldots,\cor_{m})$ and any $\vec x\in\mathbb R^{m!}$, the set of {\em parallel universes under $\cor$ at $\vec x$}, denoted by $\PU{\cor}{\vec x}\subseteq \ml(\ma)$, is the set of all elimination orders under PUT. Formally,
$$\PU{\cor}{\vec x} = \{O\in \ml(\ma):\forall 1\le i\le m-1, O[i]\in\arg\min\nolimits_{a}\score_{\cor_{m+1-i}}(\vec x|_{O[i,m]},a)\},$$
where $\score_{\cor_{m+1-i}}(\vec x|_{O[i,m]},a)$ is the total score of $a$ under the positional scoring rule $\cor_{m+1-i}$, where the profile is $\vec x|_{O[i,m]}$. 

For any alternative $a$, let the {\em possible losing rounds}, denoted by $\PULR{\cor}{\vec x,a}\subseteq [m-1]$, be the set of all rounds in the parallel universes where $a$ drops out. Formally,
$$\PULR{\cor}{\vec x,a} = \{O^{-1}[a]:O\in \PU{\cor}{\vec x}\}$$
\end{dfn}
See Example~\ref{ex:PU} for examples of parallel universes  and possible losing rounds under STV.

\appThm{thm:sCC-MRSE}{\bf Smoothed $\CC$: int-MRSE rules}{
Let $\mm= (\Theta,\ml(\ma),\Pi)$ be a strictly positive and closed single-agent preference model, let $\cor =(\cor_2,\ldots,\cor_{m})$ be an int-MRSE and let $r$ be a refinement of $\cor $. For any $n\in\mathbb N$ with $2\mid n$, we have
$$\satmin{\CC}{\Pi}(r,n) = \left\{\begin{array}{ll}
1 &\text{if } \forall 2\le i\le m, \sat{\CL}(\cor_i)=1\\
1- \exp(-\Theta(n)) &\text{if } \left\{\begin{array}{l}\text{(1) }\exists 2\le i\le m\text{ s.t. }\sat{\CL}(\cor_i)=0\text{ and }\\
\text{(2) } \forall \pi\in\conv(\Pi), \forall  a\in \wcw(\pi)\text{ and }\forall i^*\in \PULR{\cor}{\pi,a}, \\
\hfill \text{we have } \sat{\CL}(\cor_{m+1-i^*})=1  \end{array}\right.\\
\Theta(n^{-0.5}) &\text{if }  
\left\{\begin{array}{l}   \text{(1) }\forall  \pi\in \conv(\Pi), \cwinner(\pi)\cap (\ma\setminus \cor (\pi))=\emptyset \text{ and} \\
\text{(2) }  \exists \pi\in \conv(\Pi)\text{ s.t. } |\almostCW(\pi)\cap (\ma\setminus \cor (\pi))|=2\end{array}\right.
\\
\exp(-\Theta(n)) &\text{if }\exists \pi\in \conv(\Pi)\text{ s.t. } \cwinner(\pi)\cap (\ma\setminus \cor (\pi))\ne \emptyset\\
\Theta(1)\text{ and }1-\Theta(1) &\text{otherwise}
\end{array}\right.$$
For any $n\in\mathbb N$ with $2\nmid n$, we have
$$\satmin{\CC}{\Pi}(r,n) = \left\{\begin{array}{ll}
1 &\text{same as the }2\mid n\text{ case}\\
1- \exp(-\Theta(n)) &\text{same as the }2\mid n\text{ case}\\
\exp(-\Theta(n)) &\text{if }\exists \pi\in \conv(\Pi)\text{ s.t. } \left\{\begin{array}{l}   \text{(1) }\cwinner(\pi)\cap (\ma\setminus \cor(\pi))\ne \emptyset \text{ or} \\
\text{(2) }  |\almostCW(\pi)\cap (\ma\setminus \cor(\pi))|=2\end{array}\right.
\\
\Theta(1)\text{ and }1-\Theta(1) &\text{otherwise}
\end{array}\right.$$
}

\paragraph{\bf Intuitive explanations.}  The conditions for U, VU, and M cases are the same as their counterparts in Theorem~\ref{thm:sCC-scoring}. The most interesting cases are the $1$ case and the VL case. The $1$ case happens when  all positional scoring rule used in $\cor$ satisfy {\sc Condorcet loser}. This is true because for any positional scoring rule that satisfies {\sc Condorcet loser}, the Condorcet winner, when it exists, cannot have the lowest score among all alternatives. Therefore, like in Baldwin's rule, the Condorcet winner never loses in any round, which means that it must be the unique winner under $\cor$. 

The VL case happens when (1) the $1$ case does not happen, and (2) for every distribution $\pi\in\conv(\Pi)$, every weak Condorcet winner $a$, and every round $i^*$ where $a$ is eliminated in a parallel universe, the positional scoring rule used in round $i^*$, i.e.~$\cor_{m+1-i^*}$ for $m+1-i^*$ alternatives, must satisfy {\sc Condorcet loser}. (2) makes sense because it guarantees that  when a small permutation is added to $\pi$, if a weak Condorcet winner $a$ becomes the Condorcet winner, then it will be the unique winner under $\cor$, because in every round $i^*$ where $a$ can possibly be eliminated before the perturbation (i.e.~$i^*$ is a possible losing round), the voting rule used in that round, i.e.~$\cor_{m+1-i^*}$, will not eliminate $a$ after $a$ has become a Condorcet winner. The following example shows the VL case under $\istv$.

\begin{proof} We apply Lemma~\ref{lem:sCC-GISR} to prove the theorem.  We first prove the following claim, which states that when $n$ is sufficiently large, $\condition{\text{AS}}(\cor, n)=1$ if and only if all scoring rules used in $\cor$ satisfy the Condorcet loser criterion.
\begin{claim}
\label{claim:MRSE-AS}
For int-MRSE $\cor$, there exists $N\in n$ such that for  every $n>N$, $\condition{\text{AS}}(\cor, n)$ holds if and only if for all $2\le i\le m$, $\sat{\CL}(\cor_i)=1$.
\end{claim}
\begin{proof}
{\bf The \boldmath $\Leftarrow$ direction.}  Suppose for all $2\le i\le m$, $\sat{\CL}(\cor_i)=1$ and for the sake of contradiction, suppose $\condition{\text{AS}}(\cor, n)=0$, which means that there exists an $n$-profile $P$ such that $\cwinner(P)=\{a\}$ and $a\notin \cor(P)$. This means that $\PULR{\cor}{\pi,a}\ne\emptyset$. Let $O\in \PULR{\cor}{\pi,a}$ denote an arbitrary  possible losing round of $a$ and let $i^*=O^{-1}[a]$, which means that $a$ has the lowest total score in the restriction of $P$ on the remaining alternatives (i.e.~$O[i^*,m]$), when $\cor_{m+1-i^*}$ is used. In other words,
$$a\in\arg\min\nolimits_{b}\score_{\cor_{m+1-i^*}}(P|_{O[i^*,m]},b)$$
Notice that $a$ is the Condorcet winner under $P$, which means that $a$ is also the Condorcet winner under $P|_{O[i^*,m]}$. We now obtain a profile $P_{i^*}$ over $O[i^*,m]$ from $P|_{O[i^*,m]}$, which constitutes  a violation of  {\sc Condorcet loser} for $\cor_{m+1-i^*}$. Let $n'=|P|$.
$$P_{i^*} = (n'+1)\times\ml(O[i^*,m]) - P$$
That is, $P_{i^*}$ is obtained from $(n'+1)$ copies of all linear orders over $O[i^*,m])$ by subtracting linear orders in $P$. It is not hard to verify that $a$ is the Condorcet loser as well as an $\cor_{m+1-i^*}$ co-winner in $P_{i^*} $, because all alternatives are tied in the WMG of $(n'+1)\times\ml(O[i^*,m])$ and are tied w.r.t.~their total $\cor_{m+1-i^*}$ scores under $(n'+1)\ml(O[i^*,m])$. This is a contradiction to the assumption that all $\cor_i$'s satisfies the Condorcet loser criterion.

{\bf The \boldmath $\Rightarrow$ direction.}  For the sake of contradiction, suppose $\sat{\CL}(\cor_{i^*})=1$ for some $2\le i^*\le m$, which means that there exist a profile $P_1$ over $m+1-i^*$ alternatives $\{i^*,\ldots, m\}$, such that alternative $i^*$ is the Condorcet loser and a co-winner of $\cor_{m+1-i^*}$ under $P_1$. We will construct a profile $P$ over $\ma$  to show that $\condition{\text{AS}}(\cor, n)=0$ for every sufficiently large $n$. We will show that alternatives in $O[1,i^*-1]$ are eliminated in the first $i^*-1$ round of executing $\cor$ on $P$. Then $i^*$ will be eliminated in the next round.  

First, we define a  profile $P'$ over $O[i^*,m]$ where  $i^*$ is the Condorcet winner as well as the unique  $\cor_{m+1-i^*}$ loser. Let $\sigma$ denote an arbitrary cyclic permutation among $O[i^*+1,m]$, and let 
$$P_2 = \{\sigma^i( a\succ O[i^*+1,m]): 1\le i\le m-i^*\},$$
where alternatives in $O[i^*+1,m]$ are ranked alphabetically. Let $n_1=|P_1|$  and 
$$P'=m(n_1+1)\times\ml(O[i^*,m])-m\times P_1-P_2$$
It is not hard to verify that $P'$ is indeed a profile, i.e., the weight on each ranking is a non-negative integer. $i^*$ is the Condorcet winner under $P'$ because $i^*$ is the Condorcet loser in $P_1$, and $|P_2|<m$. $i^*$ is the unique loser under $P'$ because for any other alternative $a\in O[i^*,m]$, we have  
\begin{align*}
&\score_{\cor_{m+1-i^*}}(m(n'+1)\times\ml(O[i^*,m]),i^*) = \score_{\cor_{m+1-i^*}}(m(n'+1)\times\ml(O[i^*,m]),a),\\
&\score_{\cor_{m+1-i^*}}(P_1,i^*) \ge \score_{\cor_{m+1-i^*}}(P_1,a), \text{ and}\\
&\score_{\cor_{m+1-i^*}}(P_2,i^*) > \score_{\cor_{m+1-i^*}}(P_2,a).
\end{align*}

Next, we let $P^*$ denote the profile obtained from $P'$ by appending $O[1]\succ O[2]\succ \cdots\succ O[i^*-1]$ in the bottom. More precisely, we let
$$P^* = \{R\succ O[1]\succ O[2]\succ \cdots\succ O[i^*-1]: R\in P'\}$$

Finally, we are ready to define $P$. Let $\sigma_1$ denote an arbitrary cyclic permutation among alternatives in $O[1,i^*-1]$. Let $n'=|P'|$ and $P=P^1\cup P^2\cup P^3$, defined as follows.
\begin{itemize}
\item  $P^1$ consists of $n'$ copies of $\{\sigma_1^i(P^*):1\le i\le i^*-1\}$. This part has  $(n')^2(i^*-1)$ rankings and is mainly used to guarantee that $O[1,i^*-1]$ are removed in the first $i^*-1$ rounds.
\item  $P^2$ consists of $\left\lfloor \frac{n-(n')^2(i^*-1)}{n'}\right\rfloor$ copies of $P^*$. This part guarantees that $i^*$ is the Condorcet winner. We require  $n$ to be sufficiently large so that $\lfloor \frac{n-(n')^2(i^*-1)}{n'}\rfloor>n'$.
\item  $P^3$  consists of  $n-|P_1|-|P_2|$ copies of $[O[m]\succ O[m-1]\succ \cdots \succ O[1]]$, which guarantees that $|P|=n$. Note that the number of rankings in this part is no more than $n'$.
\end{itemize}
 Let $N=(n')^2$. For any $n>N$, notice that  the second part has at least $n'$ copies of $P^*$, where $i^*$ is the Condorcet winner. Therefore, $i^*$ is the Condorcet winner under $P$. It is not hard to verify that $O[1,i^*-1]$ are removed in the first $i^*-1$ rounds under $\cor$, and in the $i^*$-th round alternative $i^*$ is unique $\cor_{m+1-i^*}$ loser, which means that $i^*\notin\cor(P)$.   This concludes the proof of Claim~\ref{claim:MRSE-AS}.
\end{proof}

We prove the following claim to simplify $\closure{\region{\CWwin}{\cor }}$ and $\closure{\region{\CWlose}{\cor }}$.
\begin{claim}\label{claim:closure-CWW-CWL-MRSE}
For any int-MRSE $\cor$  and any $\pi\in\conv(\Pi)$,  
\begin{align*}
\left[\pi\in \closure{\region{\CWwin}{\cor }} \right]&\Leftrightarrow \left[\wcw(\pi)\cap  \cor (\pi)\ne \emptyset\right]\\
\left[\pi\in \closure{\region{\CWlose}{\cor }}\right] &\Leftrightarrow \left[\exists  a\in \wcw(\pi)\text{ and }i^*\in \PULR{\cor}{\pi,a}\text{ s.t. }\sat{\CL}(\cor_{m+1-i^*})=0 \right]
\end{align*}
\end{claim}
\begin{proof} The proof for the $\region{\CWwin}{\cor }$ part is similar to the proof of Claim~\ref{claim:closure-CWW-CWL-scoring}. We present the formal proof below for completeness. 

\paragraph{\bf \boldmath $\left[\pi\in \closure{\region{\CWwin}{\cor}} \right] \Rightarrow \left[\wcw(\pi)\cap  \cor(\pi)\ne \emptyset\right]$.} Suppose $\pi\in \closure{\region{\CWwin}{\cor}} $, which means that exists a sequence $(\vec x_1,x_2,\ldots)$ in $\region{\CWwin}{\cor}$ that converges to $\pi$. It follows that there exists an alternative $a\in\mathbb \ma$ and a subsequence of $(\vec x_1,\vec x_2,\ldots)$, denoted by $(\vec x_1',x_2',\ldots)$, and $O\in\ml(\ma)$ where $O[m]=a$, such that for every $j\in\mathbb N$, $\cwinner(\vec x_j')=\{a\}$ and $O\in\PU{\cor}{\vec x_j'}$. This means that the following holds. 
\begin{itemize}
\item $a$ is a weak Condorcet winner under $\pi$. 
\item $a\in \cor(\pi)$. More precisely, $O\in \PU{\cor}{\pi}$. To see this, recall that $O\in\PU{\cor}{\vec x_j'}$ is equivalent to
$$\forall 2\le i\le m, O[i]\in\arg\min\nolimits_{b}\score_{\cor_i}(\vec x_j'|_{O[i,m]},b)$$
Therefore, the same relationship holds for $\pi$, namely
$$\forall 2\le i\le m, O[i]\in\arg\min\nolimits_{b}\score_{\cor_i}(\pi|_{O[i,m]},b),$$
which means that $O\in \PU{\cor}{\pi}$.
\end{itemize}
Therefore, $a$ is a weak Condorcet winner as well as a $\cor$ co-winner, which implies that $\wcw(\pi)\cap  \cor(\pi)\ne \emptyset$.

\paragraph{\bf \boldmath $\left[\pi\in \closure{\region{\CWwin}{\cor }}\right]  \Leftarrow \left[\wcw(\pi)\cap  \cor(\pi)\ne \emptyset\right]$.} Suppose $\wcw(\pi)\cap  \cor (\pi)\ne \emptyset$ and let $a\in \wcw(\pi)\cap  \cor (\pi)$.  We will explicitly construct a sequence of vectors in $\region{\CWwin}{\cor }$ that converges to $\pi$. 
Because $a\in \cor (\pi)$, there exists a parallel universe $O\in\PU{\cor}{\pi}$ such that $O[m]=a$. Let $\vec x = -\hist(\{O\})$, i.e.~we will use  ``negative''  $O$ to break ties, so that for every $1\le i\le m-1$, $O[i]$ is eliminated in round $i$.  For any $\delta>0$, it is not hard to verify that $O\in\PU{\cor}{\pi+\delta \vec x}$. In fact, $\PU{\cor}{\pi+\delta \vec x} = \{O\}$, i.e.
$$\forall 2\le i\le m, \{O[i]\}=\arg\min\nolimits_{b}\score_{\cor_i}((\pi+\delta \vec x)|_{O[i,m]},b),$$
which means that   $\{a\}=\cor(\pi+\delta \vec x)$. Notice that $a$ is the Condorcet winner under $\pi+\delta \vec x$ for any sufficiently small $\delta>0$. Therefore, for any sufficiently small $\delta>0$ we have $\pi+\delta \vec x\in \region{\CWwin}{\cor }$. Because the sequence $(\pi+\vec x, \pi+\frac 12\vec x,\ldots)$ in $\region{\CWwin}{\cor }$ converges to $\pi$, we have $\pi\in \closure{\region{\CWwin}{\cor }}$. 

\paragraph{\bf \boldmath $\left[\pi\in \closure{\region{\CWlose}{\cor }}\right] \Rightarrow \left[\exists  a\in \wcw(\pi)\text{ and }i^*\in \PULR{\cor}{\pi,a}\text{ s.t. }\sat{\CL}(\cor_{m+1-i^*})=0\right]$.} Suppose $\pi\in \closure{\region{\CWlose}{\cor }} $, which means that there exists a sequence $(\vec x_1,\vec x_2,\ldots)$ in $\region{\CWlose}{\cor }$ that converges to $\pi$. It follows that there exists $a\in\mathbb \ma$, $O\in\ml(\ma)$ with $O[m]\ne a$, and a subsequence of $(\vec x_1,\vec x_2,\ldots)$, denoted by $(\vec x_1',\vec x_2',\ldots)$ such that for every $j\in\mathbb N$, $\cwinner(\vec x_j')=\{a\}$ and $O\in\PU{\cor}{\vec x_j'}$. Let $i^*=O^{-1}[a]$, i.e.~$i^*$ is the round where $a$ loses in the parallel universe $O$, which means that for every $j\in\mathbb N$,
$$ a\in \arg\min\nolimits_{b}\score_{\cor_{m+1-i^*}}(\vec x_j'|_{O[i^*,m]},b).$$
Notice that $a$ is the Condorcet winner among $O[i^*,m]$. This means that $\cor_{m+1-i^*}$ does not satisfy the Condorcet loser criterion, because for any sufficiently large $\psi>0$, $a$ is the Condorcet loser as well as a co-winner in  $\psi\cdot \hist(O[i^*,m])-\vec x_j'|_{O[i^*,m]}$.
Because $(\vec x_1',\vec x_2',\ldots)$ converges to $\pi$, it is not hard to verify that $a\in \wcw(\pi)$ and $O\in \PU{\cor}{\pi}$. Therefore, we have $a\in \wcw(\pi)$,  $i^*\in \PULR{\cor}{\pi,a}$, and $\sat{\CL}(\cor_{m+1-i^*})=0$.

\paragraph{\bf \boldmath $\left[\pi\in \closure{\region{\CWlose}{\cor }}\right] \Leftarrow \left[\exists  a\in \wcw(\pi)\text{ and }i^*\in \PULR{\cor}{\pi,a}\text{ s.t. }\sat{\CL}(\cor_{m+1-i^*})=0\right]$.}   Let $a\in \wcw(\pi)$ and $i^*\in \PULR{\cor}{\pi,a}$ such that $\sat{\CL}(\cor )=0$. Furthermore, we let $O^*\in \PU{\cor}{\pi}$ denote the parallel universe such that $O[i^*]=a$. Because $\cor_{m+1-i^*}$ does not satisfy the Condorcet loser criterion, there exists profile $P_a$ over $O[i^*,m]$ where $a$ is the Condorcet loser but $a\in \cor_{m+1-i^*}(P_a)$. In fact, there exists a profile $P_a^*$ where $a$ is the Condorcet loser but $\{a\}=\cor_{m+1-i^*}(P^*)$, i.e.~$a$ is the unique winner under $P_a^*$. To see this, let $\sigma$ denote an arbitrary cyclic permutation among $O[i^*+1,m]$, and let 
$$P = \{\sigma^i( a\succ O[i^*+1,m]): 1\le i\le m-i^*\}$$
It is not hard to verify that the score of $a$ is strictly larger than the score of any other alternative under $P$. Therefore, when $\delta>0$ is sufficiently small, $a$ is the Condorcet loser as well as the unique winner under $P_a^*=P_a+\delta\cdot P$.  Now, we define a profile $P'$ over $\ma$ by stacking $O[1,i^*-1]$ on top of each (fractional) ranking in $P_a^*$. In other words, a ranking $[O[1]\succ \cdots\succ O[i^*-1]\succ R^*]$ is in $P'$ if and only if $R^*\in P_a^*$, and the two rankings have the same weights (in $P'$ and $P_a^*$, respectively).

Let $\vec x = -\hist(P')$. It is not hard to verify that for any $\delta>0$, $a$ is the Condorcet winner under $\pi+\delta\vec x$ and in the first $i^*$ rounds of the execution of $\cor$, $O[1],O[2],\ldots,O[i^*]$ are eliminated in order. In particular, $O[i^*] =a$ is eliminated in the $i^*$-th round, which means that $a\notin \cor(\pi+\delta\vec x)$. Consequently, $\pi+\delta\vec x\in  {\region{\CWlose}{\cor }}$. Notice that $(\pi+\frac1j \vec x: j\in\mathbb N)$ is a sequence in $ \region{\CWlose}{\cor }$ that converges to $\pi$, which means that $\pi\in \closure{\region{\CWlose}{\cor }} $. This proves Claim~\ref{claim:closure-CWW-CWL-MRSE}.
\end{proof}

We now apply Claim~\ref{claim:closure-CWW-CWL-MRSE} to simplify the conditions in Lemma~\ref{lem:sCC-GISR}. 
\begin{itemize}
\item $\condition{\text{RS}}(\cor, \pi)$. By definition, this holds if and only if $\pi\notin  \closure{\region{\CWlose}{\cor}}$, which is equivalent to $\nexists  a\in \wcw(\pi)\text{ and }i^*\in \PULR{\cor}{\pi,a}\text{ s.t. }\sat{\CL}(\cor_{m+1-i^*})=0$. In other words,  for all $a\in\wcw(\pi)$ and all $i^*\in\PULR{\cor}{\pi,a}$, $\cor_{m+1-i^*}$ satisfies {\sc Condorcet loser}, or equivalently, $\forall  a\in \wcw(\pi)\text{ and }\forall i^*\in \PULR{\cor}{\pi,a}, \sat{\CL}(\cor_{m+1-i^*})=1$.  
\item $\condition{\text{NRS}}(\cor, \pi)$. By definition, this holds if and only if  $\almostCW(\pi)\ne \emptyset$ and $\pi \notin \closure{\region{\CWwin}{\cor}}$, which is equivalent to $\almostCW(\pi)\ne \emptyset$ and $\wcw(\pi)\cap  \cor(\pi)= \emptyset$. The latter is equivalent to $\wcw(\pi)\cap (\ma\setminus \cor(\pi))=\wcw(\pi)$. We note that when $\almostCW(\pi)\ne \emptyset$, we have $\wcw(\pi)=\almostCW(\pi)$. Therefore, $\condition{\text{NRS}}(\cor, \pi)$ is equivalent to $|\almostCW(\pi)\cap (\ma\setminus \cor(\pi))|=2$.
\end{itemize}
Theorem~\ref{thm:sCC-MRSE} follows after Lemma~\ref{lem:sCC-GISR} with the simplified conditions discussed above.
\end{proof}

\section{Materials for Section~\ref{sec:CC}: Smoothed  {\sc Participation}}

\subsection{Lemma~\ref{lem:sPar-GSR} and Its Proof}
\label{app:proof-lem:sPar-GSR}
We first introduce some notation to present the theorem. 
\begin{dfn}[\bf \boldmath $\oplus$ operator]
For any pair of signatures $\vec t_1,\vec t_2\in \sk$, we define $\vec t_1\oplus\vec t_2$ to be the following signature:
$$\forall k\le K, [\vec t_1\oplus\vec t_2]_k = \left\{\begin{array}{ll} 
[\vec t_1]_k &\text{if }[\vec t_1]_k = [\vec t_2]_k\\
0 &\text{otherwise}
\end{array}
\right.
$$
\end{dfn}
For example, when $K=3$, $\vec t_1 = (+,-,0)$, and $\vec t_2 = (+,0,0)$, we have $\vec t_1\oplus\vec t_2 = (+,0,0)$. By definition, we have $\vec t_1\unlhd \vec t_1\oplus\vec t_2$ and  $\vec t_2\unlhd \vec t_1\oplus\vec t_2$.
 
\begin{dfn}[\bf \boldmath $\violation{\Par}{r}{n}$ and $\ell_n$]
\label{dfn:vio}
For any GSR $r$ and any $n\in\mathbb N$, we define
$$\violation{\Par}{r}{n} = \left\{\signH(P)\oplus\signH(P\setminus\{R\}):P\in\ml(\ma)^n, R\in \ml(\ma), r(P\setminus\{R\}) \succ_R r(P)\right\}$$
$$\ell_n = m!-\max\nolimits_{\vec t\in \violation{\Par}{r}{n}: \exists \pi\in\conv(\Pi),\text{ s.t. } \vec t \unlhd \signH(\pi)}\dim(\ppolyz{\vec t})$$
\end{dfn}
In words, $\violation{\Par}{r}{n}$ consists of all signatures $\vec t$ that is obtained by combining two feasible signatures, i.e.,  $\signH(P)$ and $\signH(P\setminus\{R\})$, by the $\oplus$ operator, where $P$ and $R$ constitutes an violation of $\Par$.  Notice that $r(P\setminus\{R\}) \succ_R r(P)$ implicitly assumes that $P$ contains an $R$ vote. Then, $\ell_n$ is defined to be $m!$ minus the maximum dimension of polyhedron $\ppoly{\vec t}$, among all $\vec t$  in $\violation{\Par}{r}{n}$ that refines $\signH(\pi)$ for some $\pi\in\conv(\Pi)$.

\begin{lem}[\bf Smoothed $\Par$:  Int-GSR]
\label{lem:sPar-GSR}
{
Let $\mm= (\Theta,\ml(\ma),\Pi)$ be a strictly positive and closed single-agent preference model, let $r$ be an int-GSR. For any $n\in\mathbb N$, 
$$\satmin{\Par}{\Pi}(r,n) = \left\{\begin{array}{ll}
1&\text{if } \violation{\Par}{r}{n} =\emptyset\\
1- \exp(-\Theta(n)) &\text{otherwise if }\forall \pi\in\conv(\Pi)\text{ and } \vec t\in \violation{\Par}{r}{n}, \vec t \ntrianglelefteq \signH(\pi)\\ 
1- \Theta(n^{-\ell_n/2}) &\text{otherwise, i.e. }\exists \pi\in\conv(\Pi)\text{ and } \vec t\in \violation{\Par}{r}{n}\text{ s.t. } \vec t \unlhd \signH(\pi)
\end{array}\right.$$
}
\end{lem}

Applying Lemma~\ref{lem:sPar-GSR} to a voting rule $r$ often involves the following steps. First, we choose an GSR representation of $r$ by specifying the $\vH$ and $g$, though according to Lemma~\ref{lem:sPar-GSR} the asymptotic bound does not depend on such choice. Second, we characterize  $\violation{\Par}{r}{n}$ and verify whether it is empty. If $ \violation{\Par}{r}{n}$ is empty then the $1$ case holds. Third, if $\violation{\Par}{r}{n}$ is non-empty but none of $\vec t\in\violation{\Par}{r}{n}$ refines $\signH(\pi)$ for any $\pi\in\conv(\Pi)$, then the VL case holds. Finally, if neither $1$ nor VL case holds, then the L case holds, where the degree of polynomial depends on $\ell_n$. Characterizing $\violation{\Par}{r}{n}$ and $\ell_n$ can be highly challenging, as it aims at summarizing all violations of $\Par$ for $n$-profiles (using signatures under $\vH$). 

\begin{proof} The high-level idea of the proof is similar to the proof of Lemma~\ref{lem:sCC-GISR}. In light of Lemma~\ref{lem:categorization}, the proof proceeds in the following three steps. {\bf Step 1.} Define $\upoly$ that characterizes the satisfaction of  {\sc Participation} of $r$, and an almost complement $\aupoly$ of $\upoly$. {\bf Step 2.}  Characterize possible values of  $\alpha^*_n$ and their conditions, and then notice that $\alpha^*_n$ is at most $m!-1$, which means that only the $1$, VL, or L case in Lemma~\ref{lem:categorization} hold. This means that the value of $\beta_n$  does not matter. {\bf Step 3.} Apply Lemma~\ref{lem:categorization}.

\paragraph{\bf Step 1.} Given two feasible signatures $\vec t_1,\vec t_2\in \fs$ and a ranking $R\in\ml(\ma)$, we first formally define a polyhedron $\ppoly{\vec t_1,R,\vec t_2}$  
to characterize the profiles whose signature is $\vec t_1$ and after removing a voter whose preferences are $R$, the signature of the new profile becomes $\vec t_2$.
\begin{dfn}[\bf \boldmath $\ppoly{\vec t_1,R,\vec t_2}$]
\label{dfn:H-t-R-t}
Given $\vH = (\vec h_1,\ldots,\vec h_K)\in (\mathbb Z^{d})^K$,  $\vec t_1,\vec t_2\in \fs$, and $R\in\ml(\ma)$, we let 
 $\pba{\vec t_1,R,\vec t_2}=\left[\begin{array}{l}-\hist(R)\\ \pba{\vec t_1}\\ \pba{\vec t_2}\end{array}\right]$,   $\pvbb{\vec t_1,R,\vec t_2} = [-1, \underbrace{\pvbb{\vec t_1}}_{\text{for }\pba{\vec t_1}},\underbrace{\pvbb{\vec t_2}+\hist(R)\cdot \pba{\vec t_2}}_{\text{for }\pba{\vec t_2}} ]$ and
$$\ppoly{\vec t_1,R,\vec t_2} = \{\vec x\in {\mathbb R}^{m!}: \pba{\vec t_1,R,\vec t_2}\cdot \invert{\vec x}\le \invert{\pvbb{\vec t_1,R,\vec t_2}} \}$$
\end{dfn}
Notice that $\hist(R)\in \{0,1\}^{m!}$ is the vector whose $R$-component is $1$ and all other components are $0$'s. The $\pba{\vec t_2}$ part  in Definition~\ref{dfn:H-t-R-t} is equivalent to $\pba{ \vec t_2}\cdot \invert{\vec x-\hist(R)}\le \invert{\pvbb{\vec t_2}}$.  We prove properties of $\ppoly{\vec t_1,R,\vec t_2}$  in the following claim.

\begin{claim}[\bf \boldmath Properties of $\ppoly{\vec t_1,R,\vec t_2}$]
\label{claim:H-t-R-t}
Given integer $\vH$. For  any $\vec t_1,\vec t_2\in \fs$, any $R\in\ml(\ma)$, 
\begin{enumerate}[label=(\roman*)]
\item for any integral profile $P$,  $\hist(P)\in \ppoly{\vec t_1,R,\vec t_2}$ if and only if $\signH(P)=\vec t_1$ and  $\signH(P\setminus\{R\})=\vec t_2$;
\item for any $\vec x\in\mathbb R_{\ge 0}^{m!}$, $\vec x\in {\ppolyz{\vec t_1,R,\vec t_2}}$ if and only if $ \vec t_1\oplus \vec t_2\unlhd \signH(\vec x)$;
\item If there exists $\vec x\in \ppolyz{\vec t_1,R,\vec t_2}$ such that $[\vec x]_R>0$, then $\dim(\ppolyz{\vec t_1,R,\vec t_2}) = \dim(\ppolyz{\vec t_1\oplus\vec t_2})$. Moreover, if $\vec t_1\ne \vec t_2$ and $\ppolyz{\vec t_1,R,\vec t_2}\ne\emptyset$, then $\dim(\ppolyz{\vec t_1,R,\vec t_2})\le m!-1$.
\end{enumerate}
\end{claim}
\begin{proof} 
Part (i) follows after the definition.  
Part (ii) also follows after the definition. Recall that $\vec x\in \ppolyz{\vec t_1,R,\vec t_2}$ if and only if  $\pba{\vec t_1}\cdot \invert{\vec x}\le \invert{\vec 0}$,  $\pba{\vec t_2}\cdot \invert{\vec x}\le \invert{\vec 0}$, and the $R$ component of $\vec x$ is non-negative, which is automatically satisfied for every $\vec x\in\mathbb R_{\ge 0}^{m!}$. The first sets of inequalities holds if and only if  $\pba{\vec t_1\oplus \vec t_2}\cdot \invert{\vec x}\le \invert{\vec 0}$.

To prove the first part of Part (iii), let $\ba_1^+$ and $\ba_2^+$ denote the essential equalities of $\pba{\vec t_1,R,\vec t_2} $ and $\pba{\vec t_1\oplus \vec t_2}$, respectively. We show that $\ba_1^+$ and $\ba_2^+$ contains the same set of row vectors (while some rows may appear different number of times in $\ba_1^+$ and $\ba_2^+$). As noted in the proof of Part (ii), the set of row vectors in $\pba{\vec t_1,R,\vec t_2} $ is the same as the set of row vectors in $\pba{\vec t_1\oplus \vec t_2}$, except that the former contains $-\hist(R)$. Recall that we have assumed that there exists  $\vec x\in \ppolyz{\vec t_1,R,\vec t_2}$ such that $[\vec x]_R>0$, which means that $-\hist(R)\cdot\invert{\vec x}$ does not hold for every vector in $\ppolyz{\vec t_1,R,\vec t_2}$. Therefore, $-\hist(R)$ is not a row in $\ba_1^+$, which means that $\ba_1^+$ and $\ba_2^+$ contains the same set of row vectors. Then, we have 
$$\dim(\ppolyz{\vec t_1,R,\vec t_2}) = m!-\rank(\ba_1^+)= m!-\rank(\ba_2^+)=\dim(\ppolyz{\vec t_1\oplus\vec t_2})$$

The second part of Part (iii) is proved by noticing that when $\vec t_1\ne \vec t_2$, $\vec t_1\oplus \vec t_2$ contains at least one $0$. Suppose $[\vec t_1\oplus \vec t_2]_k=0$.  This means that for all $\vec x\in \ppolyz{\vec t_1,R,\vec t_2}$, we have $\vec h_k\cdot \vec x =0$, which means that $\dim(\ppolyz{\vec t_1,R,\vec t_2})\le m!-1$.
\end{proof}

We now use $\ppoly{\vec t_1,R,\vec t_2}$ to define $\upoly$ and $\aupoly$.
\begin{dfn}[\bf \boldmath $\upoly$ and  $\aupoly$ for {\sc Participation}]
Given an int-GSR $r$ characterized by $\vH$ and $g$, we define
\begin{align*}
&\upoly =  \bigcup\nolimits_{\vec t_1,\vec t_2\in\fs, R\in\ml(\ma): r(\vec t_1)\succeq_R r(\vec t_2)}\ppoly{\vec t_1,R,\vec t_2}\\
&\aupoly =  \bigcup\nolimits_{\vec t_1,\vec t_2\in\fs, R\in\ml(\ma): r(\vec t_1)\prec_R r(\vec t_2)}\ppoly{\vec t_1,R,\vec t_2}
\end{align*}
\end{dfn}
In words, $\upoly$ consists of polyhedra $\ppoly{\vec t_1,R,\vec t_2}$ that characterize the histograms of profiles such that after any $R$-vote is removed, the winner under $r$ is not improved w.r.t.~$R$. $\aupoly$ consists of polyhedra $\ppoly{\vec t_1,R,\vec t_2}$ that characterize the histograms of profiles such that after removing an $R$-vote, the winner under $r$ is strictly improved w.r.t.~$R$. It is not hard to see that $\aupoly$ is an almost complement of $\upoly$. 

It follows from Claim~\ref{claim:H-t-R-t} (i) that for any $n$-profile $P$, $\Par$ is satisfied (respectively, dissatisfied) at $P$ if and only if $\hist(P)\in \upoly$ (respectively, $\hist(P)\in \aupoly$).

\paragraph{\bf \boldmath Step 2: Characterize  $\alpha^*_n$.} In this step we discuss the values and conditions for $\alpha_n^*$ (for $\aupoly$) in the following three cases.

\paragraph{\bf \boldmath $\alpha_n^* = -\infty$.} This case holds if and only if $\Par$ holds for all $n$-profiles, which is equivalent to $\violation{\Par}{r}{n} = \emptyset$.

\paragraph{\bf \boldmath $\alpha_n^* = -\frac{n}{\log n}$.} This case holds if and only if (1) $\Par$ is not satisfied at all $n$-profiles, which is equivalent to $\violation{\Par}{r}{n} \ne \emptyset$, and (2) the activation graph $\calG_{\Pi,\aupoly,n}$ does not contain any non-negative edges, which is equivalent to $\forall \pi\in\conv(\Pi)$ and $\forall \ppoly{\vec t_1,R,\vec t_2}\subseteq \aupoly$ that is active at $n$, we have $\pi\notin \ppolyz{\vec t_1,R,\vec t_2}$. We will prove that part (2) is equivalent to the following:
\begin{equation}\label{eq:alpha-n-star-cond-two}
\text{(2)} \Longleftrightarrow \left[\forall \pi\in\conv(\Pi)\text{ and } \vec t\in \violation{\Par}{r}{n}, \vec t \ntrianglelefteq \signH(\pi)\right]
\end{equation}

We first prove the ``$\Rightarrow$'' direction of (\ref{eq:alpha-n-star-cond-two}). Suppose for the sake of contradiction that this is not true. That is, $\calG_{\Pi,\aupoly,n}$ does not contain any non-negative edges and there exist $ \pi\in\conv(\Pi)$ and $\vec t\in \violation{\Par}{r}{n}$ such that $\vec t \ntrianglelefteq \signH(\pi)$. Let $P$ denote the $n$-profile such that $\signH(P)=\vec t_1$, $\signH(P\setminus\{R\})=\vec t_2$,  $r(P\setminus\{R\})\succ_R r(P)$, and $\vec t = \vec t_1\oplus \vec t_2$.  By Claim~\ref{claim:H-t-R-t} (i), $\hist(P)\in \ppoly{\vec t_1,R,\vec t_2}$, which means that $\ppoly{\vec t_1,R,\vec t_2}$ is active at $n$. By Claim~\ref{claim:H-t-R-t} (ii), $\hist(P)\in \ppolyz{\vec t_1,R,\vec t_2}$. These imply that the weight on the edge $(\pi, \ppoly{\vec t_1,R,\vec t_2})$ in $\calG_{\Pi,\aupoly,n}$ is non-negative (whose weight is $\dim(\ppolyz{\vec t_1,R,\vec t_2})$), which contradicts the assumption that (2) holds. 

Next, we prove the ``$\Leftarrow$'' direction of (\ref{eq:alpha-n-star-cond-two}). Suppose for the sake of contradiction that (2) does not hold, which means that there exists an edge $(\pi,\ppoly{\vec t_1,R,\vec t_2})$ in $\calG_{\Pi,\aupoly,n}$ whose weight is non-negative. Equivalently, $\ppoly{\vec t_1,R,\vec t_2}$ is active at $n$ and $\pi\in \ppolyz{\vec t_1,R,\vec t_2}$. By Claim~\ref{claim:H-t-R-t} (ii), $\vec t_1\oplus \vec t_2\in  \violation{\Par}{r}{n}$. Recall that  $\pi$ is strictly positive, and then
by Claim~\ref{claim:H-t-R-t} (ii), we have $t_1\oplus \vec t_2\unlhd\signH(\pi)$. However, this contradict the assumption. 

These prove (\ref{eq:alpha-n-star-cond-two}).

\paragraph{\bf \boldmath $\alpha_n^* >0$.}  For this case, we   prove
\begin{equation}
\label{eq:alpha-n-star-par}
\alpha_n^*  = \max\nolimits_{\vec t\in \violation{\Par}{r}{n}: \exists \pi\in\conv(\Pi),\text{ s.t. } \vec t \unlhd \signH(\pi)}\dim(\ppolyz{\vec t}),
\end{equation}
We first prove the ``$\le$'' direction in (\ref{eq:alpha-n-star-par}). For any  edge $(\pi,\ppoly{\vec t_1,R,\vec t_2})$ in $\calG_{\Pi,\aupoly,n}$ whose weight is non-negative, $\ppoly{\vec t_1,R,\vec t_2}$ is active at $n$. Therefore, there exists an $n$-profile $P$ such that $\hist(P)\in \ppoly{\vec t_1,R,\vec t_2}$. Let $\vec t = \vec t_1\oplus \vec t_2$. We have $\vec t \in \violation{\Par}{r}{n}$. By Claim~\ref{claim:H-t-R-t} (ii), we have $\vec t\unlhd \signH(\pi)$. By Claim~\ref{claim:H-t-R-t} (iii), we have $\dim(\ppolyz{\vec t_1,R,\vec t_2}) = \dim(\ppolyz{\vec t})$. Therefore,  
the ``$\le$'' direction in (\ref{eq:alpha-n-star-par}) holds.

Next, we prove the $\ge$ direction of (\ref{eq:alpha-n-star-par}). For any $\vec t\in  \violation{\Par}{r}{n}$ and $\pi\in\conv(\Pi)$ such that $\vec t \unlhd \signH(\pi)$, let $P$ denote an $n$-profile and let $R$ denote a ranking that justify $\vec t$'s membership in $\violation{\Par}{r}{n}$, and let $\vec t_1 = \signH(P)$ and $\vec t_2 = \signH(P\setminus\{R\})$, which means that $\vec t = \vec t_1\oplus\vec t_2$. By Claim~\ref{claim:H-t-R-t} (i), $\hist(P)\in \ppoly{\vec t_1,R,\vec t_2}\subseteq \aupoly$, which means that $\ppoly{\vec t_1,R,\vec t_2}$ is active at $n$. By Claim~\ref{claim:H-t-R-t} (ii), $\pi\in \ppolyz{\vec t_1,R,\vec t_2}$. By Claim~\ref{claim:H-t-R-t} (iii),  $\dim(\ppolyz{\vec t_1,R,\vec t_2}) = \dim(\ppolyz{\vec t})$. This means that the weight on the edge $(\pi,\ppoly{\vec t_1,R,\vec t_2})$ in $\calG_{\Pi,\aupoly,n}$ is $\dim(\ppolyz{\vec t})$, which implies    
the ``$\ge$'' direction in (\ref{eq:alpha-n-star-par}) holds.

Therefore,  (\ref{eq:alpha-n-star-par}) holds. Notice that by Claim~\ref{claim:H-t-R-t} (iii), $\alpha_n^*\le m!-1$.

\paragraph{\bf Step 3: Applying Lemma~\ref{lem:categorization}.}  Lemma~\ref{lem:sPar-GSR} follows after a straightforward application of Lemma~\ref{lem:categorization} and Step 2. Notice that $ \Pi_{\upoly,n}$ and $\beta_n$ are irrelevant in this proof because only the $1$, $1-\exp(n)$, and $1-\poly(n)$ cases will happen. This completes the proof of Lemma~\ref{lem:sPar-GSR}.\end{proof}

\subsection{Proof of Theorem~\ref{thm:sPar-mm-rp-sch}}
\label{app:proof-thm:sPar-mm-rp-sch}
Recall from Definition~\ref{dfn:eo-rule} that an EO-based rule is determined by the total preorder over edges in WMG w.r.t.~their weights. Theorem~\ref{thm:sPar-mm-rp-sch} characterizes smoothed $\Par$ for any EO-based int-GSR refinements of maximin, Ranked Pairs,  and Schulze. 

\begin{thm}[Smoothed $\Par$: maximin, Ranked Pairs,  Schulze]\label{thm:sPar-mm-rp-sch}
{
For any $m\ge 4$, any EO-based  int-GSR $r$ 
that is a refinement of maximin, STV, Schulze, or ranked Pairs,  and any strictly positive and closed $\Pi$ over $\ml(\ma)$ with $\piuni\in \conv(\Pi)$, there exists   $N\in\mathbb N$ such that for every $n\ge N$,  
$$\satmin{ \Par}{\Pi}(r,n ) = 1-  \Theta({\frac{1}{\sqrt n}})$$
}
\end{thm}

\begin{proof} Because $r$ is EO-based, w.l.o.g., we assume that its int-GSR representation uses $\vH_{\eo}$ (Definition~\ref{dfn:heo}). 

\paragraph{\bf Overview.}  The proof is done by applying Lemma~\ref{lem:sPar-GSR} to show that for  any sufficiently large $n$, the $1$ case and the VL case do not happen, and $\ell_n=1$ in the L case. This is done by explicitly constructing an $n$-profile $P$, under which $\Par$ is violated when an $R$-vote is removed (which means that $\vec t= \signs{\vH_{\eo}}(P)\oplus \signs{\vH_{\eo}}(P\setminus\{R\})\in \violation{\Par}{r}{n}$ and therefore the $1$ case does not hold), then show that $\vec t\unlhd\piuni$, or more generally, any signature refines $\signs{\vH_{\eo}}(\piuni)$   (which means that the VL case does not hold),  and finally prove that $\dim(\ppolyz{\vec t}) = m!-1$, which means that $\ell_n = 1$.

\paragraph{\bf Maximin: $r$ refines $\overline{\maximin}$.}  We first prove the proposition for $2\nmid n$, then show how to modify the proof for $2\mid n$. As mentioned in the  overview, the proof proceeds in the following steps.

\paragraph{\bf \boldmath Constructing $P_{\maximin}$ and $R_{\maximin}$ that violates $\Par$.}  Let $G_{\maximin}$ denote the following  weighted directed graph with weights $w_{\maximin}$, where the weights on all edges are odd and different, except on $4\ra 1$ and $3\ra 2$.
\begin{itemize} 
\item $w_{\maximin}(4,1) = w_{\maximin}(3,2) = 5$, $w_{\maximin}(1,2)=1$, $w_{\maximin}(1,3)=9$,  $w_{\maximin}(2,4) = 13$, and  $w_{\maximin}(3,4) = 17$;
\item for every $5\le i\le m$, $w_{\maximin}(1,i)\ge 21$, $w_{\maximin}(2,i)\ge 21$, $w_{\maximin}(3,i)\ge 21$, and $w_{\maximin}(4,i)\ge 21$;
\item the weights on other edges are assigned arbitrarily. Moreover, the difference between any pair of edges is at least $4$, except that the weights on $4\ra 1$ and $3\ra 2$ are the same.
\end{itemize}
See the middle graph in Figure~\ref{fig:Par-MM}  for an example of $m=5$.

\begin{figure}[htp]\centering
\includegraphics[width=1\textwidth]{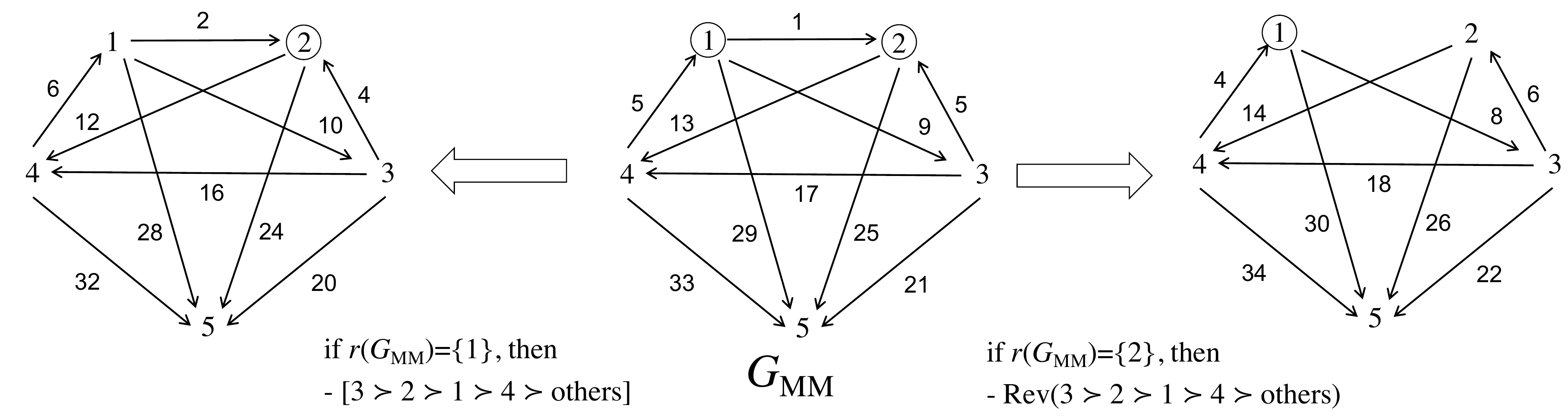} 
   \caption{WMGs for minimax. $\imaximin$ (co)-winners are circled.\label{fig:Par-MM}}
\end{figure}


It follows from McGarvey's theorem~\citep{McGarvey53:Theorem} that for any $n>m^4$ and $2\nmid n$, there exists an $n$-profile $P_{\maximin}$ whose WMG is $G_{\maximin}$. Therefore, for any $n>m^4+2$ and $2\nmid n$, there exists an $n$-profile $P_{\maximin}$ whose WMG is $G_{\maximin}$, and $P_{\maximin}$ includes the following two rankings:
$$[3\succ 2\succ 1\succ 4\succ \others], \rev{3\succ 2\succ 1\succ 4\succ \others},$$
where for any ranking $R$, $\rev{R}$ denotes its reverse ranking. 
We now show that $\sat{\Par}(r,P_{\maximin})=0$, which implies that the $1$ case does not happen. Notice that the min-score of alternatives $1$ and $2$ are the highest, which means that $r(P_{\maximin})\subseteq \{1,2\}$. 
\begin{itemize}
\item If $r(P_{\maximin})=\{1\}$, then we let $R_{\maximin}=[3\succ 2\succ 1\succ 4\succ \others]$. It follows that in $P_{\maximin}-R_{\maximin}$, the min-score of $2$ is strictly higher than   the min-score of any other alternative, which means that $r(P_{\maximin}\setminus\{R_{\maximin}\})=\{2\}$. Notice that $2\succ_{R_{\maximin}} 1$, which means that $\sat{\Par}(r,P_{\maximin})=0$.   See the left graph in Figure~\ref{fig:Par-MM} for an illustration.
\item If $r(P_{\maximin})=\{2\}$, then we let $R_{\maximin}=\rev{3\succ 2\succ 1\succ 4\succ \others}$. It follows that in $P_{\maximin}-R_{\maximin}$, the min-score of $1$ is strictly higher than any the min-score of other alternatives, which mean that $r(P_{\maximin}\setminus\{R_{\maximin}\})=\{1\}$. Notice that $1\succ_{R_{\maximin}} 2$, which again means that $\sat{\Par}(r,P_{\maximin})=0$. See the right graph in Figure~\ref{fig:Par-MM} for an illustration.
\end{itemize}

Let $\vec t_1=\signs{\vH_\eo}(P_{\maximin})$, $R=R_{\maximin}$ and $\vec t_2=\signs{\vH_\eo}(P_{\maximin}\setminus\{R_{\maximin}\})$ .  We have  $\vec t_1\oplus\vec t_2\in \violation{\Par}{r}{n}\ne\emptyset$, which means that the $1$ case of Lemma~\ref{lem:sPar-GSR} does not hold.  
The VL case of Lemma~\ref{lem:sPar-GSR} does not hold because $\vec t_1\oplus\vec t_2\unlhd \signs{\vH_{\eo}}(\piuni)$ and $\piuni\in\conv(\Pi)$.

\paragraph{\bf\boldmath Prove $\dim(\ppolyz{\vec t_{\maximin}})=m!-1$.} Let $e_1=(4,1)$ and $e_2=(3,2)$.
Notice $[\vec t_1]_{(e_1,e_2)}=[\vec t_1]_{(e_2,e_1)}=0$, where  $[\vec t_1]_{(e_1,e_2)}$ is the  $(e_1,e_2)$ component of $\vec t_1$, and all other components of $\vec t_1$ are non-zero.  Also notice that $\vec t_2$ is a refinement of $\vec t_1$. This means that $\vec t_1 \oplus \vec t_2 = \vec t_1$. Notice that $\hist(P_{\maximin})$ is an inner point of $\ppolyz{\vec t_1}$, such that all inequalities are strict except the two inequalities about $e_1$ and $e_2$. This means that the essential equalities of $\pba{\vec t_1\oplus\vec t_2}$ are equivalent to 
$$(\pair_{4,1} - \pair_{3,2})\cdot \vec x = \vec 0$$ 
Therefore, $ \dim(\ppolyz{\vec t_1\oplus \vec t_2}) = m!-1$. 

The maximin part of the proposition when $2\nmid n$ then follows after Lemma~\ref{lem:sPar-GSR}. When $2\mid n$, we only need to modify $G_{\maximin}$ in Figure~\ref{fig:Par-MM} by increasing all positive weights by $1$. 

\paragraph{\bf Ranked Pairs: $r$ refines $\overline{\rp}$.} The proof is similar to the proof of the maximin part, except that a different graph $G_{\rp}$ (with weight $w_{\rp}$) is used, as shown in the middle graph in Figure~\ref{fig:Par-RP}. Formally, when $2\nmid n$,  let $G_{\rp}$ denote the following  weighted directed graph, where the weights on all edges are odd and different, except on $4\ra 1$ and $3\ra 4$.
\begin{itemize} 
\item $w_{\rp}(4,1) = w_{\rp}(3,4) = 9$, $w_{\rp}(1,2)=5$, $w_{\rp}(1,3)=13$,  $w_{\rp}(2,4) = 17$, and  $w_{\rp}(2,3) = 21$;
\item for any $5\le i\le m$, $w_{\rp}(1,i)\ge 25$,  $w_{\rp}(2,i)\ge 25$, $w_{\rp}(3,i)\ge 25$, and $w_{\rp}(4,i)\ge 25$;
\item the weights on other edges are assigned arbitrarily. Moreover, the difference between any pair of edges is at least $4$, except that the weights on $4\ra 1$ and $3\ra 4$ are the same.
\end{itemize}

See the middle graph in Figure~\ref{fig:Par-RP}  for an example of $m=5$.

\begin{figure}[htp]\centering
\includegraphics[width=1\textwidth]{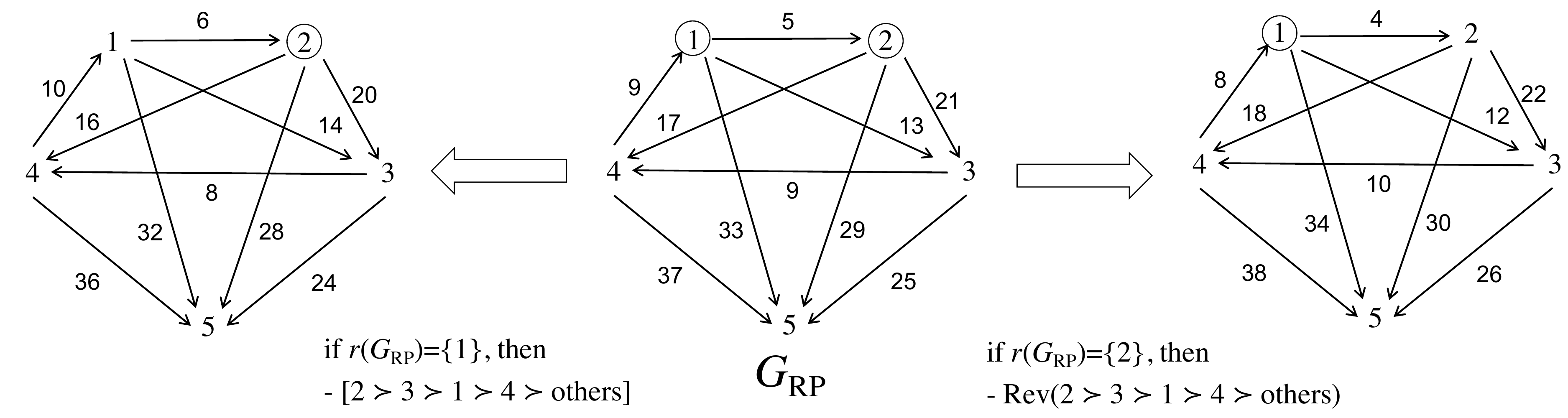} 
   \caption{WMGs for ranked pairs. $\irp$ (co)-winners are circled. \label{fig:Par-RP}}
\end{figure}

Again, according to McGarvey's theorem~\citep{McGarvey53:Theorem} that for any $n>m^4$ and $2\nmid n$, there exists an $n$-profile $P_{\rp}$ whose WMG is $G_{\rp}$. Therefore, for any $n>m^4+2$ and $2\nmid n$, there exists an $n$-profile $P_{\rp}$ whose WMG is $G_{\rp}$, and $P_{\rp}$ includes the following two rankings:
$$[2\succ 3\succ 1\succ 4\succ \others], \rev{3\succ 2\succ 1\succ 4\succ \others}$$
We now show that $\sat{\Par}(r,P_{\rp})=0$, which implies that the $1$ case does not happen. Notice that depending on how the tie between $3\ra 4$ and $4\ra 1$ are broken, the $\overline{\rp}$ winner can be $1$ or $2$, which means that $\overline{\rp}(P_{\rp}) = \{1,2\}$. 
\begin{itemize}
\item If $r(P_{\rp})=\{1\}$, then we let $R_{\rp}=[2\succ 3\succ 1\succ 4\succ \others]$. It follows that in $\wmg(P_{\rp}-R_{\rp})$, $4\ra 1$ has higher weight than $3\ra 4$, which means that $4\ra 1$ is fixed before $3\ra 4$, and therefore  $r(P_{\rp}\setminus\{R_{\rp}\})=\{2\}$. Notice that $2\succ_{R_{\rp}} 1$, which means that $\sat{\Par}(r,P_{\rp})=0$.  See the left graph in Figure~\ref{fig:Par-RP} for an illustration.
\item If $r(P_{\rp})=\{2\}$, then we let $R_{\rp}=\rev{2\succ 3\succ 1\succ 4\succ \others}$. It follows that in $\wmg(P_{\rp}\setminus\{R_{\rp}\})$, $3\ra 4$ has higher weight than $4\ra 1$, which means  $r(P_{\rp}-R_{\rp})=\{1\}$. Notice that $1\succ_{R_{\rp}} 2$, which means that $\sat{\Par}(r,P_{\rp})=0$.   See the right graph in Figure~\ref{fig:Par-RP} for an illustration.
\end{itemize}
The proof for $\ell_n=1$ is similar to the proof for the maximin part. The only difference is that now let $e_1=(4,1)$, $e_2=(3,4)$, $\vec t_1=\sign_{\vH_\eo}(P_{\rp})$, and $\vec t_2=\sign_{\vH_\eo}(P_{\rp}\setminus\{R_{\rp}\})$. When $2\mid n$, we only need to modify $G$ in Figure~\ref{fig:Par-MM} (b) such that all positive weights are increased by $1$.

\paragraph{\bf Schulze: $r$ refines $\overline{\schulze}$.} The proof is similar to the proof of the maximin part, except that a different graph $G_{\schulze}$ is used, as shown in the  middle graph in Figure~\ref{fig:Par-Schulze}. Formally, when $2\nmid n$,  let $G_{\schulze}$ denote the following  weighted directed graph, where the weights on all edges are odd and different, except on $4\ra 1$ and $2\ra 3$.
\begin{itemize} 
\item $w_{\schulze}(4,1) = w_{\schulze}(2,3) = 9$, $w_{\schulze}(1,2)=13$, $w_{\schulze}(1,3)=5$,  $w_{\schulze}(2,4) = 1$, and  $w_{\schulze}(3,4) = 17$;
\item for any $5\le i\le m$, $w_{\schulze}(1,i)\ge 21$, $w_{\schulze}(2,i)\ge 21$, $w_{\schulze}(3,i)\ge 21$, and $w_{\schulze}(4,i)\ge 21$;
\item the weights on other edges are assigned arbitrarily. Moreover, the difference between any pair of edges is at least $4$, except that the weights on $4\ra 1$ and $3\ra 4$ are the same.
\end{itemize}
See the middle graph in Figure~\ref{fig:Par-Schulze}  for an example of $m=5$.

\begin{figure}[htp]\centering
\includegraphics[width=1\textwidth]{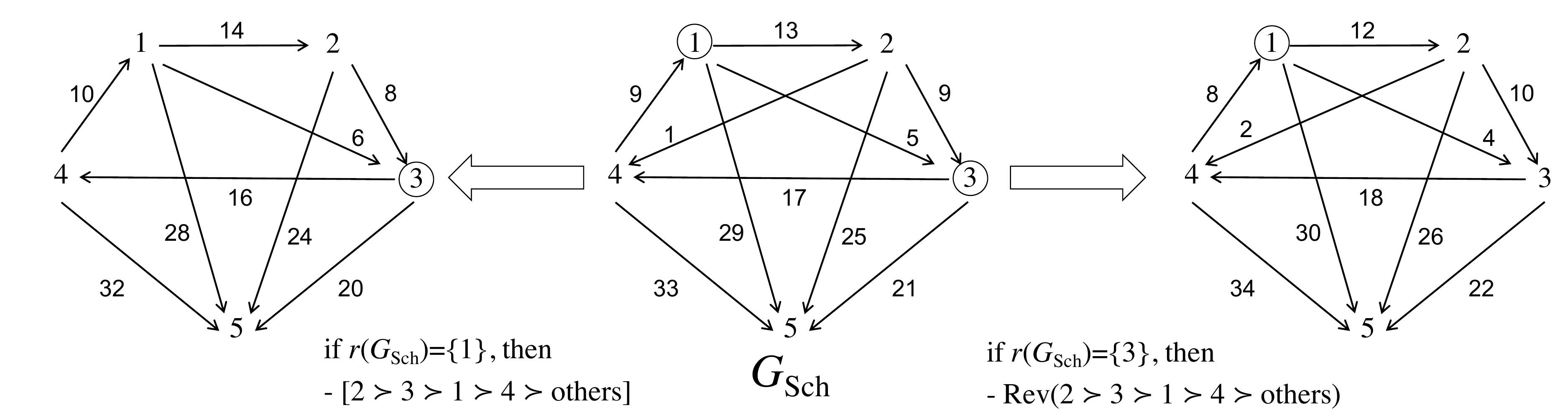} 
   \caption{WMGs for Schulze. $\ischulze$ (co)-winners are circled.\label{fig:Par-Schulze}}
\end{figure}

Again, according to McGarvey's theorem~\citep{McGarvey53:Theorem} that for any $n>m^4$ and $2\nmid n$, there exists an $n$-profile $P_{\schulze}$ whose WMG is $G_{\schulze}$. Therefore, for any $n>m^4+2$ and $2\nmid n$, there exists an $n$-profile $P_{\schulze}$ whose WMG is $G_{\schulze}$ and $P_{\schulze}$ includes the following two rankings:
$$[2\succ 3\succ 1\succ 4\succ \others], \rev{3\succ 2\succ 1\succ 4\succ \others}$$
We now show that $\sat{\Par}(r,P_{\schulze})=0$, which implies that the $1$ case does not happen. Notice that $s[1,3]=s[3,1]=9$, and for any alternative $a\in\ma\setminus\{1,3\}$ we have $s[1,a]>s[a,1]$. Therefore,  $\overline{\schulze}(P_{\schulze}) = \{1,3\}$. 
\begin{itemize}
\item If $r(P_{\schulze})=\{1\}$, then we let $R_{\schulze}=[2\succ 3\succ 1\succ 4\succ \others]$. It follows that in $ P_{\schulze}-R_{\schulze}$ we have $s[1,3] = 8<10 = s[3,1]$, which means that  $r(P_{\schulze}\setminus\{R_{\schulze}\})=\{3\}$. Notice that $3\succ_{R_{\schulze}} 1$, which means that $\sat{\Par}(r,P_{\schulze})=0$.  See the left graph in Figure~\ref{fig:Par-Schulze} for an illustration.
\item If $r(P_{\schulze})=\{3\}$, then we let $R_{\schulze}=\rev{2\succ 3\succ 1\succ 4\succ \others}$. It follows that in $ P_{\schulze}\setminus\{R_{\schulze}\}$, we have $s[1,3] = 10>9 = s[3,1]$, which means that  $r(P_{\schulze}-R_{\schulze})=\{1\}$.   Notice that $1\succ_{R_{\schulze}} 3$, which means that $\sat{\Par}(r,P_{\schulze})=0$.    See the right graph in Figure~\ref{fig:Par-Schulze} for an illustration.
\end{itemize}
The proof for $\ell_n=1$ is similar to the proof for the maximin part. The only difference is that now let $e_1=(4,1)$, $e_2=(2,3)$, $\vec t_1=\sign_{\vH_\eo}(P_{\schulze})$, and $\vec t_2=\sign_{\vH_\eo}(P_{\schulze}\setminus\{R_{\schulze}\})$. When $2\mid n$, we only need to modify $G_{\schulze}$ in  Figure~\ref{fig:Par-Schulze}  such that all positive weights are increased by $1$. 

This completes the proof of Theorem~\ref{thm:sPar-mm-rp-sch}.
\end{proof}

\subsection{Proof of Theorem~\ref{thm:sPar-copeland}}
\label{app:proof-thm:sPar-copeland}
A voting rule $r$ is said to be {\em UMG-based}, if the winner only depends on UMG of the profile. Formally, $r$ is UMG-based if for all pairs of profiles $P_1$ and $P_2$ such that $\umg(P_1)=\umg(P_2)$, we have $r(P_1)=r(P_2)$.

\begin{thm}[\bf Smoothed $\Par$: Copeland$_\alpha$]
\label{thm:sPar-copeland}
{
For any $m\ge 4$, any UMG-based int-GSR refinement of $\overline{\copeland}$, denoted by  $\copeland$, and any strictly positive and closed $\Pi$ over $\ml(\ma)$ with $\piuni\in \conv(\Pi)$, there exists $N\in\mathbb N$ such that for every $n\ge N$,  
$$\satmin{ \Par}{\Pi}(\copeland,n ) = 1-  \Theta({\frac{1}{\sqrt n}})$$
}
\end{thm}

\begin{proof} 
Because $\copeland$ is UMG-based, we can represent  $\copeland$ as a GSR with  the $\vH_{\copeland}$ defined in Definition~\ref{dfn:copeland-GISR}, which consists of $m\choose 2$ hyperplanes that represents the UMG of the profile. The high-level idea behind the proof is similar to the proof of Theorem~\ref{thm:sPar-mm-rp-sch}: we first explicitly construct a violation of $\Par$ under $\copeland$, then show that the dimension of the characteristic cone of the corresponding polyhedron is $m!-1$. 

Let $G^*$ denote the complete unweighted directed graph over $\ma$ that consists of the following edges.
\begin{itemize}
\item $1\ra 2$, $2\ra 3$, $3\ra 1$.
\item For any $i\in \{4,\ldots,m\}$, there are three edges $1\ra i$, $2\ra i$, $3\ra i$.
\item The edges among alternatives in $i\in \{4,\ldots,m\}$ are assigned arbitrarily. 
\end{itemize}
For example, Figure~\ref{fig:Par-Copeland} (a) illustrates $G^*$ for $m=4$.
\begin{figure}[htp]\centering
\includegraphics[width=.2\textwidth]{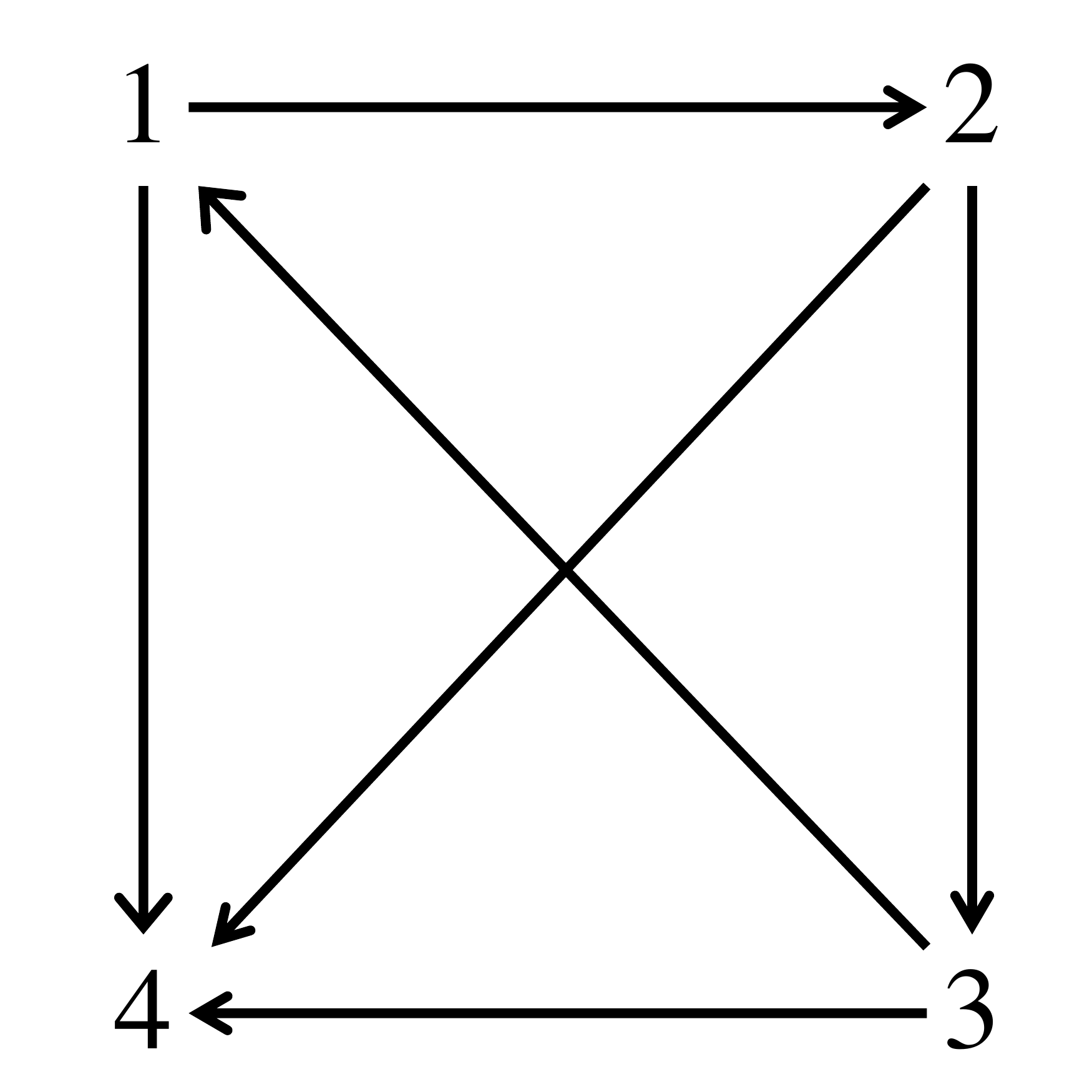}
   \caption{$G^*$ for Copeland with $m=4$.\label{fig:Par-Copeland}}
\end{figure}
Let $P$ denote any profile whose UMG is $G^*$. It is not hard to verify that $\overline{\copeland}(P) = \{1,2,3\}$. 
W.l.o.g.~let $\copeland(P)=\{1\}$. 

\paragraph{\bf\boldmath $2\nmid n$ case.} The proof is done  for the following two sub-cases: $\alpha>0$ and $\alpha=0$. 
\paragraph{\bf \boldmath $2\nmid n$ and $\alpha>0$.} Let $G_{\copeland}$ (with weights $w_{\copeland}$) denote the following weighted directed graph over $\ma$ whose UMG is $G^*$,  the weight on $2\ra 3$ is $1$, and the weights on other edges are $3$ or $-3$.
\begin{itemize}
\item $w_{\copeland}(2,3)=1$ and $w_{\copeland}(3,1)= w_{\copeland}(1,2)=3$.
\item For any $4\le i\le m$, $w_{\copeland}(1,i)=w_{\copeland}(2,i)=w_{\copeland}(3,i)=3$.
\item The weights on other edges are $3$ or $-3$. 
\end{itemize}
\begin{figure}[htp]\centering 
\begin{tabular}{ccc} 
  \includegraphics[width=.2 \textwidth]{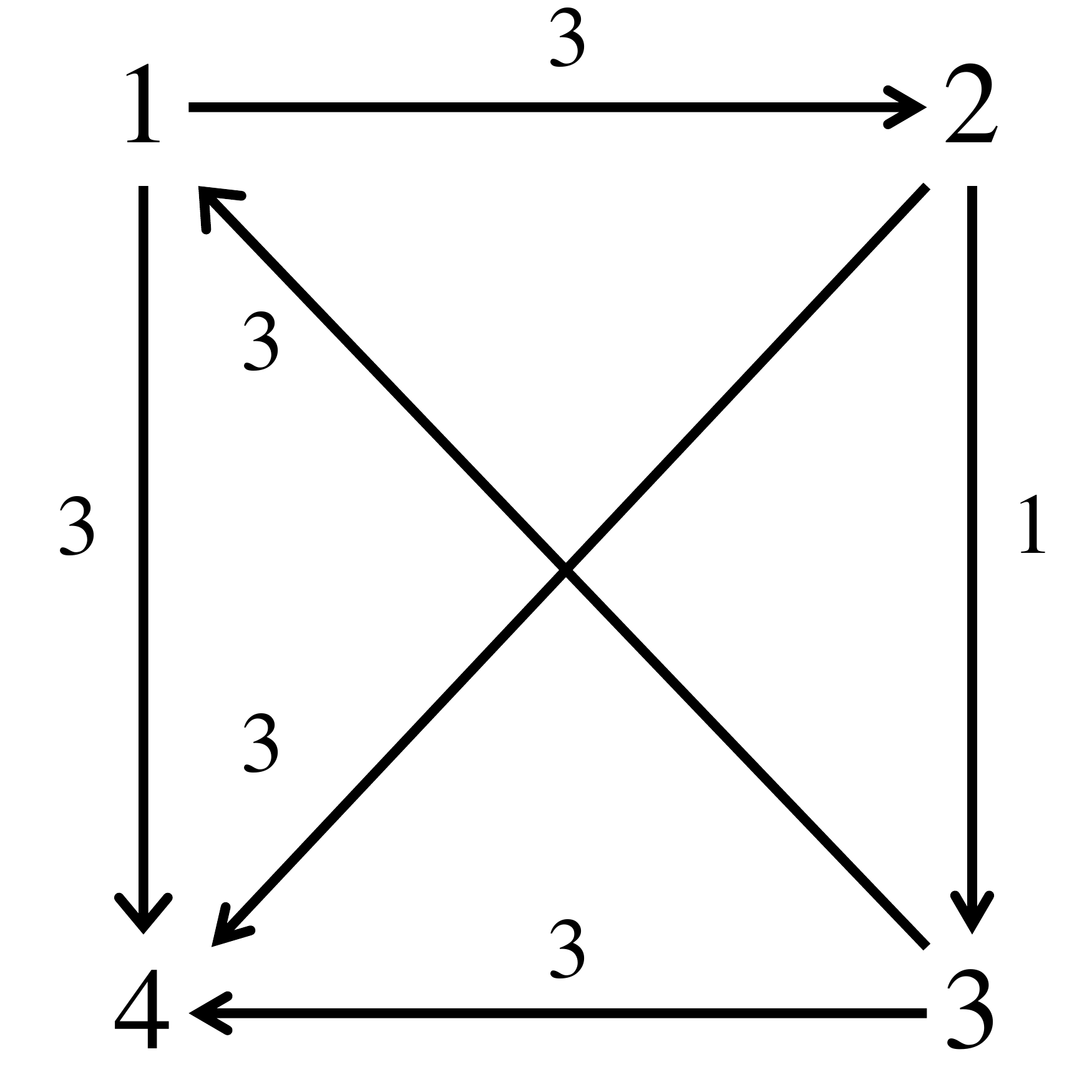} & 
  \includegraphics[width=.2 \textwidth]{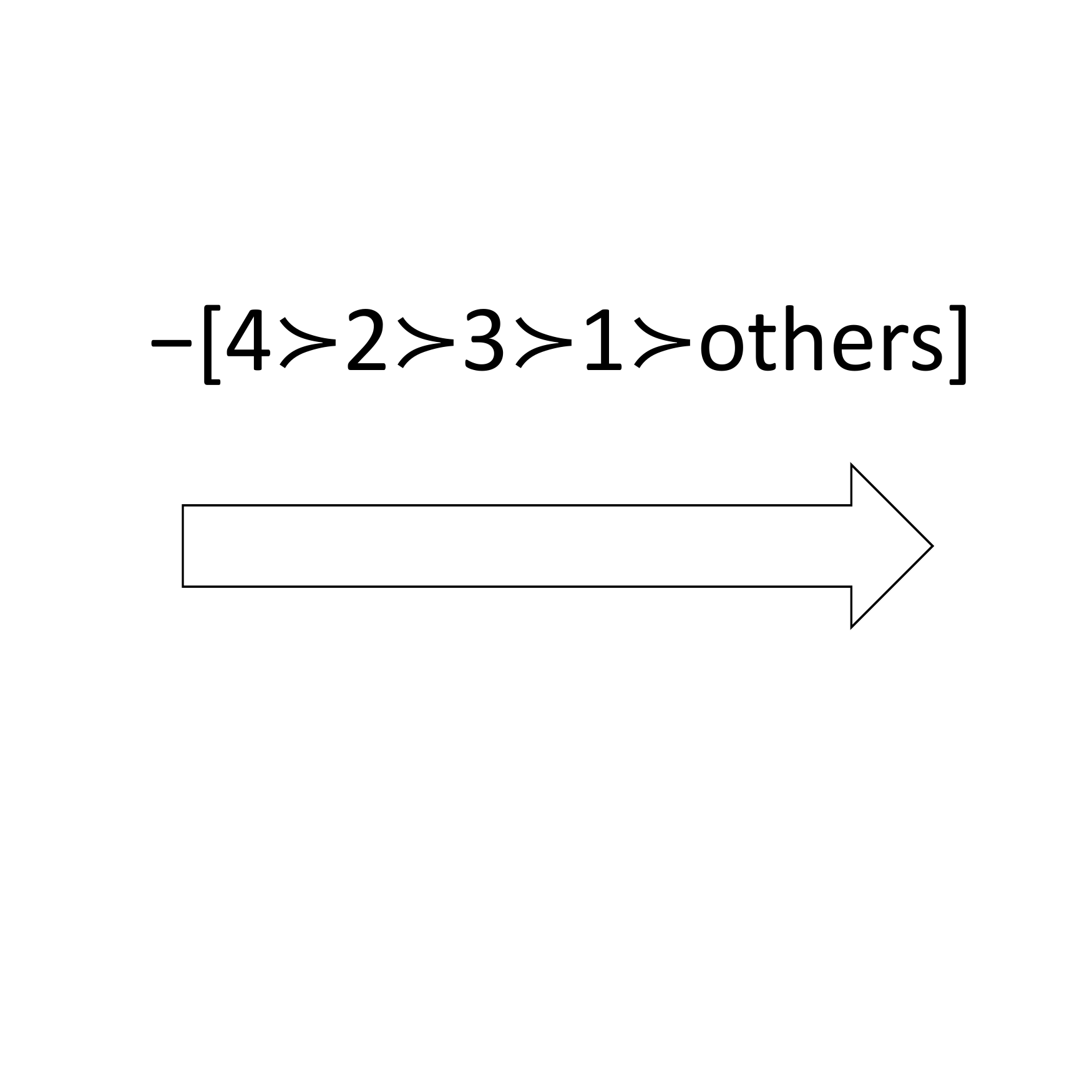} & 
  \includegraphics[width=.2 \textwidth]{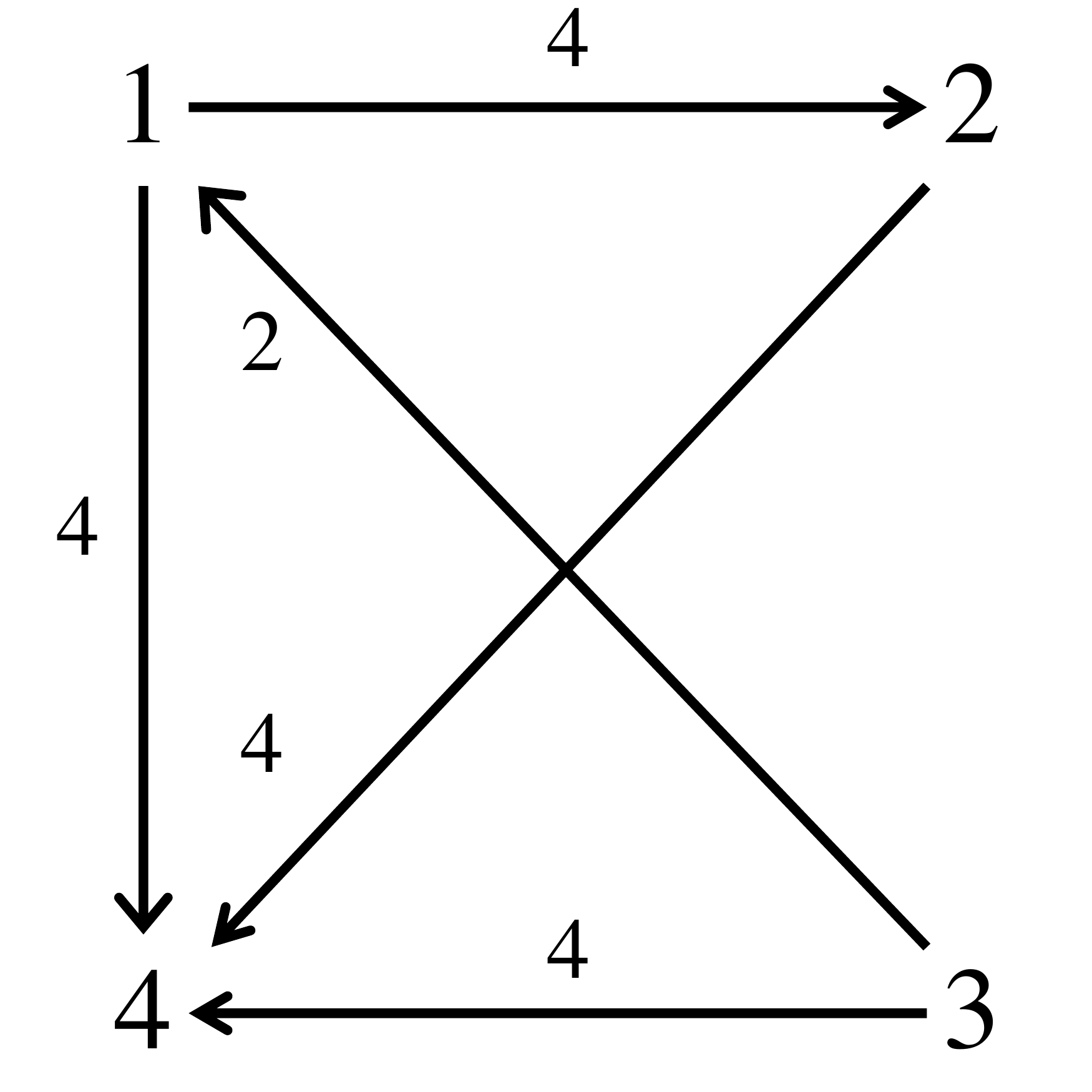} 
  \\
  (a) $G_{\copeland} = \wmg(P_{\copeland})$. & & (b) $\wmg(P_{\copeland}\setminus\{R_{\copeland}\})$.
  \end{tabular}
   \caption{$G_{\copeland}$ and $\wmg(P_{\copeland}\setminus\{P_{\copeland}\})$ for $2\nmid n$ and $\alpha>0$.\label{fig:Par-Copeland2}}
\end{figure}

See  Figure~\ref{fig:Par-Copeland2} (a) for an example of $G_{\copeland}$. According to McGarvey's theorem~\citep{McGarvey53:Theorem} that for any $n>m^4$ and $2\nmid n$, there exists an $n$-profile $P_{\copeland}$ whose WMG is $G_{\copeland}$. Therefore, for any $n>m^4+2$ and $2\nmid n$, there exists an $n$-profile $P_{\copeland}$ whose WMG is $G_{\copeland}$, and $P_{\copeland}$ includes the following two rankings.
$$[4\succ 2\succ 3\succ 1\succ \others],\rev{4\succ 2\succ 3\succ 1\succ \others}$$
We now show that $\sat{\Par}(r,P_{\copeland})=0$, which implies that the $1$ case Lemma~\ref{lem:sPar-GSR} does not hold. Let $R_{\copeland} = [4\succ 2\succ 3\succ 1\succ \others]$. Notice that in the profile $P_{\copeland}-R_{\copeland}$, the Copeland$_\alpha$ score of alternative $3$ is $m-2+\alpha$, which is strictly higher than the Copeland$_\alpha$ score of alternative $1$, which is $m-2$. Therefore,  $\copeland(P_{\copeland}\setminus\{R_{\copeland}\})=\{3\}$. See  Figure~\ref{fig:Par-Copeland2} (b) for   $\wmg(P_{\copeland}\setminus\{R_{\copeland}\})$.  Notice that $3\succ_{R_{\copeland}} 1$, which means that the $\sat{\Par}(r,P_{\copeland})=0$. 

Therefore, the $1$ case of Lemma~\ref{lem:sPar-GSR}  does not hold. Let $\vec t_1 =\signs{\vH_{\copeland}}(P_{\copeland})$ and $\vec t_2 =\signs{\vH_{\copeland}}(P_{\copeland}\setminus\{R_{\copeland}\})$. The VL case of Lemma~\ref{lem:sPar-GSR} does not hold because $\vec t_1\oplus\vec t_2\unlhd \signs{\vH_{\copeland}}(\piuni)$ and $\piuni\in \conv(\Pi)$. 

Next, we prove that    $\dim(\ppolyz{\vec t_1\oplus\vec t_2})=m!-1$.  
Notice that $[\vec t_1]_{(2,3)}=+$ and $[\vec t_2]_{(2,3)}=0$, and all other components of $\vec t_1$ and $\vec t_2$ are the same and are non-zero. Therefore, $\vec t_1$ is a refinement of $\vec t_2$, which means that $\vec t_1 \oplus \vec t_2 = \vec t_2$. Notice that $\hist(P_{\copeland})$ is an inner point of $\ppolyz{\vec t_2}$, in the sense that all inequalities are strict except the inequalities about $(2,3)$. This means that the essential equalities of $\pba{\vec t_1\oplus\vec t_2}$ are equivalent to  $\pair_{2,3}  \cdot \vec x = \vec 0$.
Therefore, $\dim(\ppolyz{\vec t_2})=\dim(\ppolyz{\vec t_1\oplus \vec t_2}) = m!-1$.  This proves the proposition when  $2\nmid n$, $\alpha>0$, and $\copeland(P)=\{1\}$.  

If $\copeland(P)=\{2\}$ (respectively, $\copeland(P)=\{3\}$), then we simply switch the weights on $2\ra 3$ and $3\ra 1$ (respectively, $2\ra 3$ and $1\ra 2$) in Figure~\ref{fig:Par-Copeland} (b), and the rest of the proof is similar to the $\copeland(P)=\{1\}$ case. This proves Theorem~\ref{thm:sPar-copeland} for $2\nmid n$ and $\alpha>0$.

\paragraph{\bf \boldmath $2\nmid n$ and $\alpha=0$.} Let $G_{\copeland}$ (with weights $w_{\copeland}$) denote the following weighted directed graph over $\ma$ whose UMG is $G^*$ as illustrated in Figure~\ref{fig:Par-Copeland} (a).
\begin{itemize}
\item $w_{\copeland}(2,3)= w_{\copeland}(3,1)= w_{\copeland}(1,2)=3$. 
\item For any $4\le i\le m$, $w_{\copeland}(1,i)=w_{\copeland}(2,i)=w_{\copeland}(3,i)=3$, except $w_{\copeland}(4,1)=1$.
\item The weights on edge between $\{4,\ldots,m\}$ are $3$ or $-3$. 
\end{itemize}
\begin{figure}[htp]\centering 
\begin{tabular}{ccc} 
  \includegraphics[width=.2 \textwidth]{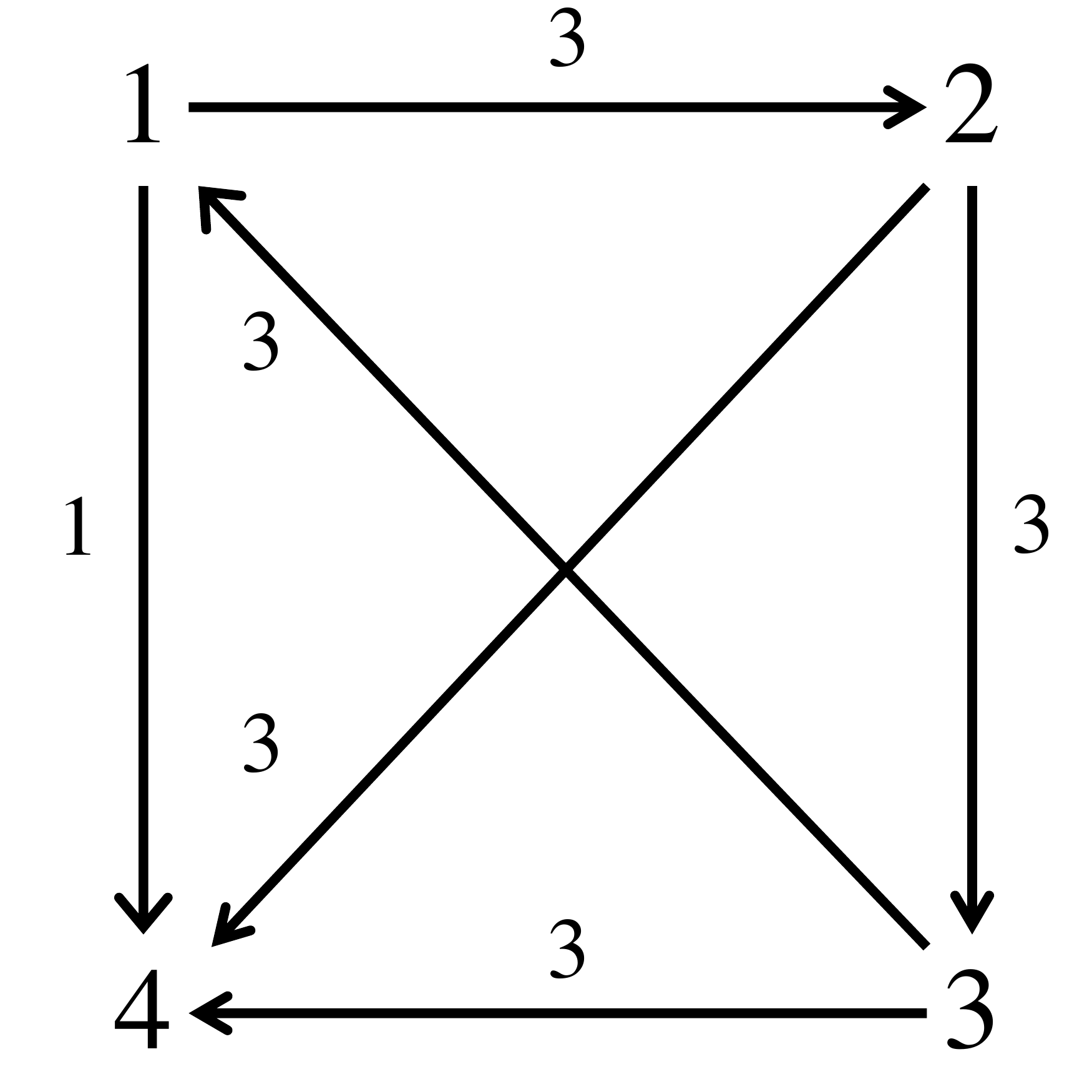} & 
   \includegraphics[width=.2 \textwidth]{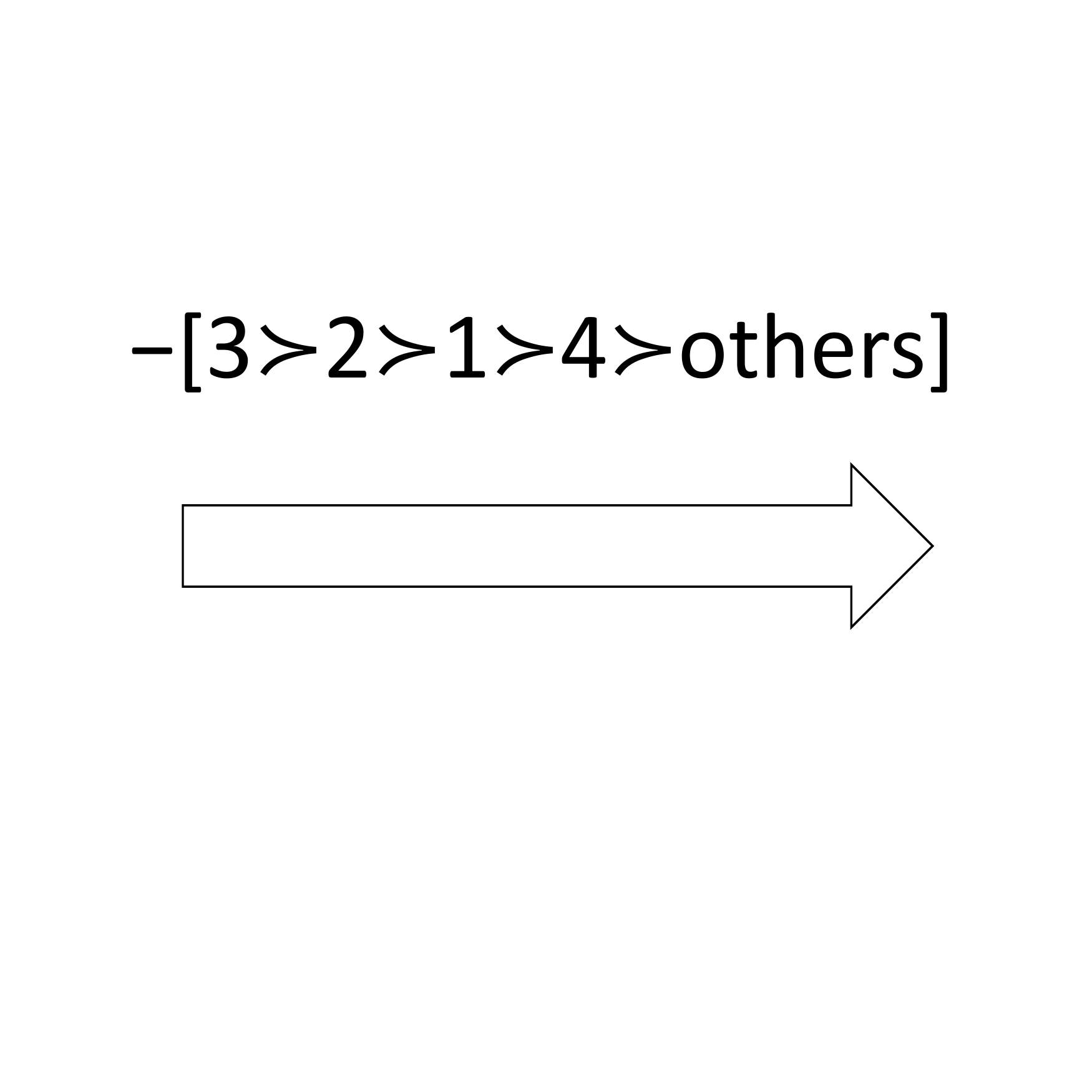} & 
     \includegraphics[width=.2 \textwidth]{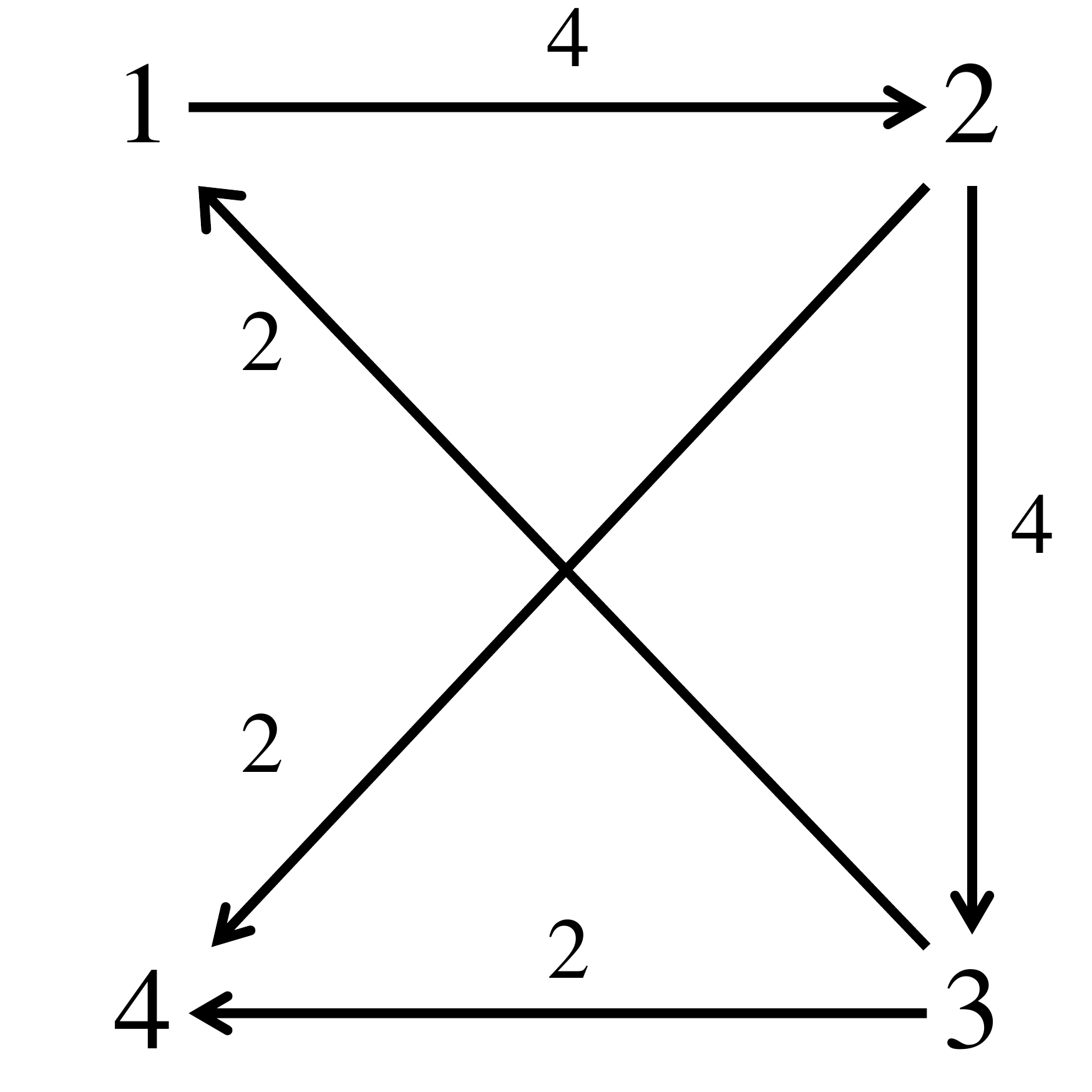} 
  \\
  (a) $G_{\copeland} = \wmg(P_{\copeland})$. & & (b) $\wmg(P_{\copeland}\setminus\{R_{\copeland}\})$.
  \end{tabular}
   \caption{$G_{\copeland}$ and $\wmg(P_{\copeland}\setminus\{P_{\copeland}\})$ for $2\nmid n$ and $\alpha=0$.\label{fig:Par-Copeland3}}
\end{figure}

See  Figure~\ref{fig:Par-Copeland3} (a) for an example of $G_{\copeland}$. According to McGarvey's theorem~\citep{McGarvey53:Theorem} that for any $n>m^4$ and $2\nmid n$, there exists an $n$-profile $P_{\copeland}$ whose WMG is $G_{\copeland}$. Therefore, for any $n>m^4+2$ and $2\nmid n$, there exists an $n$-profile $P_{\copeland}$ whose WMG is $G_{\copeland}$ and $P_{\copeland}$ includes the following two rankings.
$$[3\succ 2\succ 1\succ 4\succ \others],\rev{3\succ 2\succ 1\succ 4\succ \others}$$
We now show that $\sat{\Par}(\copeland,P_{\copeland})=0$, which implies that the $1$ case Lemma~\ref{lem:sPar-GSR} does not hold. Let $R_{\copeland} = [3\succ 2\succ 1\succ 4\succ \others]$. Notice that in the profile $P_{\copeland}\setminus\{R_{\copeland}\}$, the Copeland$_\alpha$ score of alternative $1$ is $m-3+\alpha=m-3$, which is strictly higher than the Copeland$_\alpha$ score of alternative $2$ and $3$, which means that  $\copeland(P_{\copeland}-R_{\copeland})\subseteq\{2,3\}$. See  Figure~\ref{fig:Par-Copeland3} (b) for an example of $\wmg(P_{\copeland}\setminus\{R_{\copeland}\})$.   Notice that $2\succ_{R_{\copeland}} 1$ and $3\succ_{R_{\copeland}} 1$, which means that $\sat{\Par}(\copeland,P_{\copeland})=0$.

The proofs for  $\ell_n=1$, the $\copeland(P)=\{2\}$ case, and the $\copeland(P)=\{3\}$ case are similar to their counterparts for the ``$2\nmid n$ and $\alpha=0$'' case above.

\paragraph{\bf\boldmath $2\mid n$.} The proof for the $2\mid n$ case is similar to the proof of the $2\nmid n$ case with the following modifications. The $n$-profile $P_{\copeland}$ where $\Par$ is violated is obtained from the profile in the $2\nmid n$ plus $\rev{R_{\copeland}}$. Below we present the full proof for the case of $2\mid n$ and $\alpha>0$ for example. The other cases can be proved similarly.
\paragraph{\bf \boldmath $2\mid n$ and $\alpha>0$.} W.l.o.g.~suppose $\copeland(G^*)=\{1\}$. Let $G_{\copeland}$ (with weights $w_{\copeland}$) denote the weighted directed graph in Figure~\ref{fig:Par-Copeland2} (a). According to McGarvey's theorem~\citep{McGarvey53:Theorem} that for any $n>m^4$ and $2\mid n$, there exists an $(n-1)$-profile $P_{\copeland}'$ whose WMG is $G_{\copeland}$. Let
$$P_{\copeland} = P_{\copeland}'+  \rev{4\succ 2\succ 3\succ 1\succ \others}$$
It is not hard to verify that in  $P_{\copeland}$, the Copeland$_\alpha$ score of alternative $3$ is $m-2+\alpha$, which is strictly higher than the Copeland$_\alpha$ score of alternative $1$, which is $m-2$. Therefore, $\copeland(P_{\copeland}) = \{3\}$. Let $R_{\copeland} =   \rev{4\succ 2\succ 3\succ 1\succ \others}$. Notice that $\copeland(P_{\copeland}\setminus\{R_{\copeland}\}) =\copeland(G^*)=\{1\}$ and $1\succ_{R_{\copeland}} 3$, which means that $\sat{\Par}({\copeland},P_{\copeland})=0$. Therefore, the $1$ case in Lemma~\ref{lem:sPar-GSR} does not hold. Let $\vec t_1=\sign_{\vH_{\copeland}}(P_{\copeland})$ and $\vec t_2=\sign_{\vH_{\copeland}}(P_{\copeland}\setminus\{R_{\copeland}\})$.  Like in other cases, the VL case of Lemma~\ref{lem:sPar-GSR} does not holds because $\vec t_1\oplus\vec t_2\unlhd \signs{\vH_{\copeland}}(\piuni)$.

Next, we prove that   $\dim(\ppolyz{\vec t_1\oplus \vec t_2})=m!-1$. 
Notice that $[\vec t_1]_{(2,3)}=0$ and $[\vec t_2]_{(2,3)}=+$, and all other components of $\vec t_1$ and $\vec t_2$ are the same and are non-zero. Therefore, $\vec t_1$ is a refinement of $\vec t_2$, which means that $\vec t_1 \oplus \vec t_2 = \vec t_1$. Notice that $\hist(P_{\copeland})$ is an inner point of $\ppolyz{\vec t_1}$, in the sense that all inequalities are strict except the inequalities about $(2,3)$. This means that the essential equalities of $\pba{\vec t_1\oplus\vec t_2}$ are equivalent to 
$$\pair_{2,3}  \cdot \vec x = \vec 0\text{ and }-\pair_{2,3}  \cdot \vec x = \vec 0$$

Therefore, $ \dim(\ppolyz{\vec t_1\oplus \vec t_2}) = m!-1$, which means that $\ell_n=-(m!-(m!-1))=1$. The $2\mid n$ and $\alpha>0$ case follows after Lemma~\ref{lem:sPar-GSR}.

The proof for other subcases of $2\mid n$ are similar to the proof of $2\mid n$ and $\alpha>0$ case above. This completes the proof of Theorem~\ref{thm:sPar-copeland}.
\end{proof}

\subsection{Proof of Theorem~\ref{thm:sPar-MRSE}}
\label{app:proof-thm:sPar-MRSE}
\begin{thm}[\bf Smoothed $\Par$: int-MRSE]
\label{thm:sPar-MRSE}
Given $m\ge 4$, any int-MRSE $\cor$,  any int-GSR $r$ that is a refinement of $\cor=(\cor_2,\ldots,\cor_m)$, and any strictly positive and closed $\Pi$ over $\ml(\ma)$ with $\piuni\in \conv(\Pi)$, there exists $N\in\mathbb N$ such that for every $n\ge N$,  
$$\satmin{ \Par}{\Pi}(r,n ) = 1-  \Theta({\frac{1}{\sqrt n}})$$
\end{thm}
\begin{proof}
The intuition behind the proof is similar to the proof of Theorem~\ref{thm:sPar-mm-rp-sch}. Indeed, Lemma~\ref{lem:sPar-GSR} can be applied to $r$, but it is unclear how to characterize $\ell_n$. Therefore, in this proof we do not directly characterize $dim(\ppolyz{\vec t})$ as in the proof of Theorem~\ref{thm:sPar-mm-rp-sch}, but will instead define another polyhedron $\ppoly{r}$ to characterize a set of sufficient conditions for  $\Par$ to be violated---and the dimension of the new polyhedron is easy to analyze.  Let us start with defining sufficient conditions on a profile $P$ for $\Par$ to be violated under any refinement of $\cor$. 
\begin{cond}[\bf Sufficient conditions: violation of $\Par$ under an MRSE rule]
\label{cond:Par-MRSE}
Given an MRSE $\cor$, a profile $P$ satisfies the following conditions during the execution of $\cor$.
\begin{enumerate}[label=(\arabic*)]
\item For every $1\le i\le m-4$, in the $i$-th round, alternative $i+4$ drops out. 
\item In round $m-3$, $1$ has the highest score, $2$ has the second highest score, and $3$ and $4$ are tied for the last place.
\item If $3$ is eliminated in round $m-3$, then $2$ and $4$ are eliminated in round $m-2$ and $m-1$, respectively, which means that the winner is $1$.
\item If $4$ is eliminated in round $m-3$, then $1$ and $3$ are eliminated in round $m-2$ and $m-1$, respectively, which means that the winner is $2$.
\item $P$ contains at least one vote $[4\succ 2\succ 1\succ 3\succ \others]$ and at least one vote $[3\succ 1\succ 2\succ 4\succ \others]$, where ``$\others$'' represents $5\succ \cdots\succ m$.
\item All losers described above, except in (2), are ``robust'' , in the sense that after removing any vote from $P$,  they are still the unique losers.
\end{enumerate}
\end{cond}
Let us verify that for any profile $P$ that satisfies Condition~\ref{cond:Par-MRSE}, $\sat{\Par}(r,P)=0$. It is not hard to see that $\cor(P)=\{1,2\}$. If $r(P)=\{1\}$, then let $R_r=[4\succ 2\succ 1\succ 3\succ \others]$. This means that when any voter whose preferences are $R_r$ abstain  from voting, alternative $4$ drops out in round $m-3$ of $(P\setminus\{R_r\})$, and consequently $2$ becomes the winner. Notice that $2\succ _{R_r}1$, which means that $\sat{\Par}(r,P)=0$. Similarly, if $r(P)=\{2\}$, then let $R_r=[3\succ 1\succ 2\succ 4\succ \others]$, which means that  $3$ drops out in round $m-3$ of $(P\setminus\{R_r\})$, and $1$ becomes the winner. Notice that $1\succ _{R_r}2$. Again, we have $\sat{\Par}(r,P)=0$. The procedures of executing $\cor$ under $P$ and $(P\setminus\{R_r\})$ are represented in Figure~\ref{fig:par-mrse-P}.

\begin{figure}[htp]
\centering
\includegraphics[width = 0.8\textwidth]{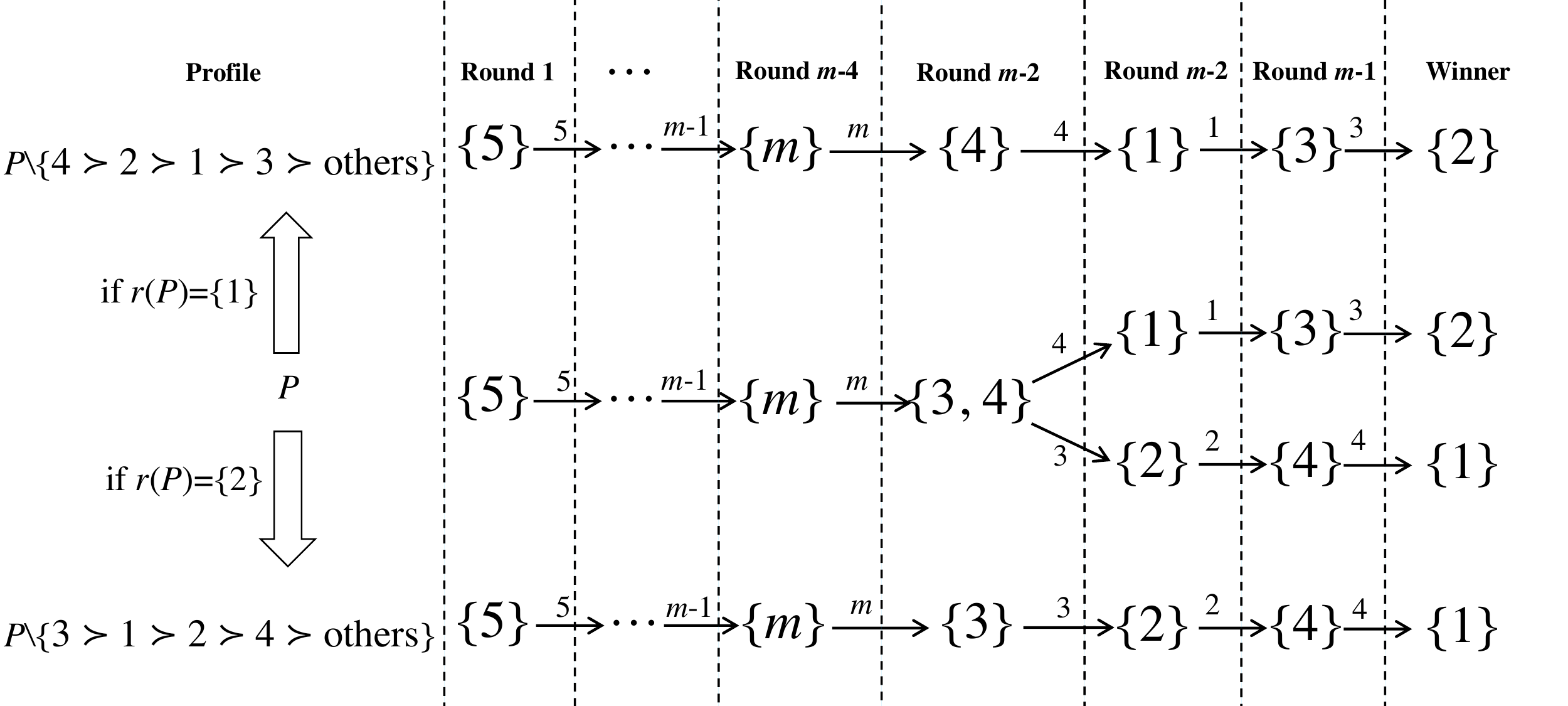}
\caption{Executing $\cor$ for a profile that satisfies Condition~\ref{cond:Par-MRSE}.\label{fig:par-mrse-P}}
\end{figure} 

The rest of the proof proceeds as follows. In Step 1 below, We will  prove by construction that for every sufficiently large $n$, there exists an $n$-profile $P_{r}$ that satisfies Condition~\ref{cond:Par-MRSE}. Then in Step 2, we formally  define $\ppoly{\cor}$ to represent profiles that satisfy Condition~\ref{cond:Par-MRSE}. Finally, in Step 3, we show that $\dim(\ppolyz{\cor})=m!-1$ because there is essentially only one equality  (in Condition~\ref{cond:Par-MRSE} (2)). Theorem~\ref{thm:sPar-MRSE} then follows after $1$ minus the polynomial case of the inf part of~\cite[Theorem~2]{Xia2021:How-Likely}.


\paragraph{\bf \boldmath Step 1: define $P_{r}$.} Before defining $P_r$, we first define a profile $P^*$ that consists of a constant and odd number of votes in Steps 1.1--1.3. We then prove that $\Par$ is violated at $P^*$ in Step 1.4 and 1.5, where in Step 1.4 we show that $\cor(P)=\{1,2\}$ and in Step 1.5 we point out a violation of $\Par$ depending on $r(P^*)$. Then in Step 1.6, we show how to expand $P^*$ to an  $n$-profile $P_r$ for any sufficiently large $n$. 

Let $P^*=P_1+P_2+P_3$, where $P_1$ consists of even number of votes and is designed to guarantee Condition~\ref{cond:Par-MRSE} (1), i.e., $ 5,\ldots, m$ are eliminated in the first $m-4$ rounds, respectively. This means that in the beginning of round $m-3$, the remaining alternatives are $\{1,2,3,4\}$. $P_2$ consists of an odd number of votes and is designed   to guarantee Condition~\ref{cond:Par-MRSE} (2), i.e.,  in round $m-3$, $\cor_{4}$ outputs the weak order  $[1\succ 2\succ 3=4]$.  $P_3$ consists of an even number of votes and is designed to guarantee Condition~\ref{cond:Par-MRSE} (3) and (4), i.e.,  if $3$ (respectively, $4$) is eliminated then $1$ (respectively, $2$) wins. 

\paragraph{\bf \boldmath Step 1.1: define $P_1$.} 
Let $P_1^1$ denote the following profile of $(24m(m-4)!+\frac{(m+5)(m-4)}{2}(m-1)!)$ votes.
$$P_1^1= m\times \{[R_1\succ R_2: \forall R_1\in \ml(\{1,2,3,4\}), R_2\in \ml(\{5,\ldots,m\})\}\cup \bigcup_{i=5}^m i\times \{[i\succ R_2]: \forall R_2\in\ml(\ma\setminus\{i\})\}$$
For every $2\le i\le m$, let the scoring vector of $\cor_i$ be $(s_1^i,\ldots,s_i^i)$. For example, the scoring vector of $\cor_{4}$ is   $(s_1^{4},s_2^{4},s_3^{4},s_4^{4})$. We let $P_1= (s_1^{4}-s_4^{4}+1)|P_2|\times P_1^1$, where $|P_2|$ is the  number of votes in  $P_2$, which is a constant and will become clear after Step 1.2.

\paragraph{\bf \boldmath Step 1.2: define $P_2$.} The main challenge in this step is to use an odd number of votes to define  $P_2$ such that in round $m-3$, the score of $1$ is strictly higher than  the score of $2$, which is strictly higher than the score of $3$ and $4$. 
We first define the following $8$-profile, denoted by $P_2^1$.
\begin{align*}
P_2^1= \{&[1\succ \others\succ 3\succ 4\succ 2], [1\succ \others\succ 4\succ 3\succ 2], \\
&3\times [1\succ \others\succ 2\succ 4\succ 3], 3\times [2\succ \others\succ 1\succ 3\succ 4]\}
\end{align*}
The numbers of times alternatives $\{1,2,3,4\}$ are ranked in each position in $ P_2^1|_{\{1,2,3,4\}}$ are indicated in Table~\ref{tab:Par-P21}. 
\begin{table}[htp]
\centering  
\begin{tabular}{|c|c|c|c|c|}
\hline Alternative & 1st & 2nd & 3rd& 4th\\
 \hline 1& 5 & 3&0 & 0\\
\hline 2 & 3& 3 & 0&2\\
\hline 3 & 0 & 1 & 4& 3\\
\hline 4 & 0 & 1 & 4& 3\\
\hline 
\end{tabular}  
\caption{Number of times each alternative is ranked in each position in $P_2^1|_{\{1,2,3,4\}}$. \label{tab:Par-P21}}
\end{table}

Next, we define a profile $P_2^2$ that consists of an odd number of votes where the scores of $3$ and $4$ are equal.  Let $d_1 = s_1^{4}-s_2^{4}$ and $d_2=s_2^{4}-s_3^{4}$. The construction is done in the following three cases.
\begin{itemize}
\item  If $d_1=0$, then we let $P_2^2$ consist of a single vote $[3\succ 4\succ 1\succ 2\succ \others]$.
\item If $d_1\ne 0$ and $d_2=0$, then we let $P_2^2$ consist of a single vote $[1\succ 3\succ 4\succ 2\succ \others]$. 
\item If $d_1\ne 0$ and $d_2\ne 0$, then we let $d_1' = d_1/\gcd(d_1,d_2)$ and $d_2' = d_2/\gcd(d_1,d_2)$, where $\gcd(d_1,d_2)$ is the greatest common divisor of $d_1$ and $d_2$. It follows that at least one of $d_1'$ and $d_2'$ is an odd number. 
\begin{itemize}
\item If $d_1'$ is odd, then we let
$$P_2^2 = (d_1'+d_2')\times [1\succ 3\succ 4\succ 2\succ \others] +  d_2'\times [4\succ 1\succ 3\succ 2\succ \others ]$$

\item Otherwise, we must have $d_1'$ is even and  $d_2'$ is odd. Then, we let
$$P_2^2 = (d_1'+d_2')\times [3\succ 4\succ 1\succ 2\succ \others] +  d_1'\times [4\succ 1\succ 3\succ 2\succ \others ]$$
\end{itemize}
\end{itemize}
It is not hard to verify that in either case $P_2^2$ consists of an odd number of votes, and the score of $3$ and $4$ are equal under $P_2^2$.
To guarantee that $3$ and $4$ have the lowest $\cor_4$ scores in $P_2|_{\{1,2,3,4\}}$, we include  sufficiently many copies of $P_2^1$ in $P_2$. Formally, let 
$$P_2 = (|P_2^2|+1)\times P_2^1 + P_2^2$$

\paragraph{\bf \boldmath Step 1.3: define $P_3$.}  We let $P_3= ((s_1-s_3)|P_2| +1)\times P_3^*$, where $P_3^*= P_3^{*1}+P_3^{*2}$ is the $36$-profile defined as follows. $P_3^{*1}$ consists of $12$ votes, where each alternative in $\{1,2,3,4\}$ is ranked in the top in three votes, followed by the remaining three alternatives in a cyclic order.
\begin{align*}
P_3^{*1}= \{&[1\succ 2\succ 3\succ 4\succ \others],  [1\succ 3\succ 4\succ 2\succ \others],[1\succ 4\succ 2\succ 3\succ \others],\\
&[2\succ 1\succ 4\succ 3\succ \others],  [2\succ 4\succ 3\succ 1\succ \others],[2\succ 3\succ 1\succ 4\succ \others],\\
&[3\succ 1\succ 4\succ 2\succ \others],  [3\succ 4\succ 2\succ 1\succ \others],[3\succ 2\succ 1\succ 4\succ \others],\\
&[4\succ 1\succ 2\succ 3\succ \others],  [4\succ 2\succ 3\succ 1\succ \others],[4\succ 3\succ 1\succ 2\succ \others]\}
\end{align*} 
$P_3^{*2}$ consists of $24$ votes that are defined in  the following three steps. First, we start with $\ml(\{1,2,3,4\})$, which consists of $24$ votes. Second, we replace $[3 \succ 2\succ 4\succ 1]$ and $[4 \succ 1\succ 3\succ 2]$ by $[3 \succ 1\succ 4\succ 2]$ and $[4 \succ 2\succ 3\succ 1]$, respectively. That is, the locations of $1$ and $2$ are exchanged in the two votes. This is designed to guarantee that the $\cor_4$ scores of all alternative are the same in $P_3^{*2}|_{\{1,2,3,4\}}$, and after $3$ is removed, $1$'s $\cor_3$ score is higher than $2$'s $\cor_3$ score; and after $4$ is removed, $2$'s $\cor_3$ score is higher than $1$'s $\cor_3$score. Third, we append the lexicographic order of $\{5,\ldots,m\}$ to the end of each of the $24$ rankings. Formally, we define
\begin{align*}
P_3^{*2}= \{&R_4\succ 5\succ \cdots \succ m: R_4\in\ml(\{1,2,3,4\})\}-[3 \succ 2\succ 4\succ 1\succ \others]\\
&- [4 \succ 1\succ 3\succ 2\succ \others]+[3 \succ 1\succ 4\succ 2\succ \others]+[4 \succ 2\succ 3\succ 1\succ \others]
\end{align*}

\paragraph{\bf \boldmath Step 1.4: Prove  $\cor(P^*)=\{1,2\}$.}   Recall that $P^*=P_1+P_2+P_3$.   Notice that the $P_1$ part guarantees that $\{5,\ldots,m\}$ are dropped out in the first $m-4$ rounds, and the scores of all alternatives in $\{1,2,3,4\}$ are the same under $P_1$ no matter what alternatives are dropped out. Therefore, it suffices to calculate the results of the last three rounds based on $P_2+P_3$, which is done as follows.

In round $m-3$, it is not hard to check that every alternative in $\{1,2,3,4\}$ gets the same total score under $P_3$, where each of them is ranked at each position for $9$ times. Therefore, due to   $P_2$, alternative $3$ and $4$ are tied for the last place in round $m-3$.

\paragraph{\bf \boldmath If $3$ is eliminated in round $m-3$,} then $P_3^*|_{\{1,2,4\}}=P_3^{*1}|_{\{1,2,4\}}+P_3^{*2}|_{\{1,2,4\}}$ becomes the following.
\begin{align*}
P_3^{*1}|_{\{1,2,4\}}= &\{2\times  [1\succ 4 \succ 2 ],  [1 \succ 2\succ 4],  2\times [2 \succ 1\succ 4],  [2\succ 4 \succ 1 ] ,\\
&  [1\succ 4\succ  2], [4\succ 2\succ  1] , [2\succ 1\succ  4],  2\times [4 \succ 1\succ 2],  [4\succ 2 \succ 1 ]\}\\
P_3^{*2}|_{\{1,2,4\}}=  &4\times \ml(\{1,2,4\}) - [2\succ 4\succ 1]- [4\succ 1\succ 2] +[1\succ 4\succ 2]+ [4\succ 2\succ 1] 
\end{align*}

It is not hard to verify that the numbers of times alternatives $\{1,2,4\}$ are ranked in each position in $ P_3^*|_{\{1,2,4\}}$ are as indicated in Table~\ref{tab:MRSE-Par} (a). 

\begin{table}[htp]
\centering 
\begin{tabular}{cc}
\begin{tabular}{|c|c|c|c|}
\hline Alternative & 1st & 2nd & 3rd\\
 \hline 1 & 13&12 & 11\\
\hline 2 & 11& 12 & 13\\
\hline 4 & 12& 12 & 12\\
\hline 
\end{tabular}
&
\begin{tabular}{|c|c|c|c|}
\hline Alternative  & 1st & 2nd & 3rd\\
 \hline 1 & 11& 12 & 13\\
\hline 2 & 13& 12 & 11\\
\hline 3 & 12& 12 & 12\\
\hline 
\end{tabular}\\
(a) $3$ is removed. & (b) $4$ is removed.
\end{tabular}
\caption{Number of times each alternative is ranked in each position in round $m-2$. \label{tab:MRSE-Par}}
\end{table}

This means that the score of alternative $2$  is strictly lower than the score of $1$ or $3$, because $s^{3}_1-s^{3}_3\ge 1$, where the score vector for $\cor_3$ is $(s^{3}_1,s^{3}_2,s^{3}_3)$. Recall that $P_3$ consists of sufficiently large number of copies of $P_3^*$. Therefore, even considering the score difference between alternatives in $P_2$, the score of $2$ is still the  strictly  lowest among $\{1,2,4\}$ in $P^*$ in round $m-2$. This means that  alternative $2$ drops in round $m-2$, and it is easy to check that $1\succ 4$ in $20$  votes in $P_3^*$, which is strictly more than half ($=16$). This means that  $1$ is the $r$ winner if $3$ is eliminated in round $m-3$.

\paragraph{\bf \boldmath If $4$ is eliminated in round $m-3$,} then  $P_3^*|_{\{1,2,3\}}=P_3^{*1}|_{\{1,2,3\}}+P_3^{*2}|_{\{1,2,3\}}$ becomes the following.\begin{align*}
P_3^{*1}|_{\{1,2,3\}}= &\{2\times  [1\succ 2 \succ 3 ], [1 \succ 3\succ 2],   2\times [2 \succ 3\succ 1],  [2\succ 1 \succ 3 ] ,\\
& 2\times [3 \succ 2\succ 1],  [3\succ 1 \succ 2 ],  [1\succ 2\succ  3], [2\succ 3\succ  1] , [3\succ 1\succ  2]\}\\
P_3^{*2}|_{\{1,2,3\}}=  &4\times \ml(\{1,2,3\}) - [3\succ 2\succ 1]- [1\succ 3\succ 2] +[3\succ 1\succ 2]+ [2\succ 3\succ 1] 
\end{align*}
The numbers of times alternatives $\{1,2,3\}$ are ranked in each position in $ P_3^*|_{\{1,2,3\}}$ are as indicated in Table~\ref{tab:MRSE-Par} (b). Again, it is not hard to verify that  alternative $1$ drops in round $m-2$, and  $2$ beats $3$ in the last round to become the $r$ winner in this case.

\paragraph{\bf \boldmath Step 1.5: Prove that $\Par$ is violated at $P^*$.} At a high-level the proof is similar to Step 1.4, and the absent vote is effectively used as a tie breaker between alternatives $3$ and $4$. Recall that $r$ is a refinement of $\cor$ and  it was shown in Step 1.4 that $\cor(P^*)=\{1,2\}$. Therefore, either $r(P^*)=\{1\}$ or $r(P^*)=\{2\}$. The proof is done in the follow two cases.
\begin{itemize}
\item If $r(P^*)=\{1\}$, then we let 
$$R_{r} = [4\succ 2\succ 1\succ 3\succ\others],$$
which is a vote in $P_3^2$. Then in $(P^*\setminus\{R_r\})$, alternative $4$ is eliminated in round $m-3$, and following a similar reasoning as in Step 1.4, we have $r(P^*\setminus\{R_r\})=\{2\}$. Notice that $2\succ_{R_r} 1$, which means that $\Par$ is violated at $P^*$.

\item If $r(P^*)=\{2\}$, then we let 
$$R_{r} = [3\succ 1\succ 2 \succ 4\succ\others],$$
which is a vote in $P_3^2$. Then in $(P^*\setminus\{R_r\})$, alternative $3$ is eliminated in round $m-3$, and following a similar reasoning as in Step 1.4, we have $r(P^*\setminus\{R_r\})=\{1\}$. Notice that $1\succ_{R_r} 2$, which means that $\Par$ is violated at $P^*$.
\end{itemize}

\paragraph{\bf \boldmath Step 1.6: Construct an $n$-profile $P_r$.}  The intuition behind the construction is the following. $P_r$ consists of three parts: $P_r^1$, $P_r^2$, and $P_r^3$. $P_r^1$ consists of multiple copies of $P^*$ defined in Steps 1.1-1.3 above, which is used to guarantee that 
$\Par$ is violated at $P_r$ and the score difference between any pair of alternatives is sufficiently large so that votes in $P_r^3$ does not affect the execution of $r$. $P_r^2$ consists of multiple copies of $\ml(\ma)$. $P_r^3$ consists of no more than $m!-1$ votes, and $|P_r^3|$ is an even number.

\paragraph{\bf \boldmath Define $P_r^1$.} To guarantee that $|P_r^3|$ is even, the definition of $P_r^1$ depends on the parity of $n$. Recall that $P^*$ consists of an odd number of votes. When $2\mid n$, we let 
$$P_r^1 = m!\left(s^{3}_1-s^{3}_3\right)\times P^*$$
When $2\nmid n$, we let 
$$P_r^1 = \left(m!\left(s^{3}_1-s^{3}_3\right)+1\right)\times P^*$$
\paragraph{\bf \boldmath Define $P_r^2$.} Let $n_1=|P_r^1|$. $P_r^2$ consists of as many copies of $\ml(\ma)$ as possible, i.e.
$$P_r^2 = \left\lfloor \frac{n-n_1}{m!} \right\rfloor \times \ml(\ma)$$
\paragraph{\bf \boldmath Define $P_r^3$.} $P_r^3$ consists of multiple copies of pairs of rankings  defined as follows.
$$P_r^3 = \left(\frac{n-n_1 - |P_r^2|}{2}\right)\times \{[1\succ 2\succ 3\succ 4\succ \others], [2\succ 1\succ 4\succ 3\succ \others]\}$$
It is not hard to verify that $P_r=P_r^1+P_r^2+P_r^3$ share the same properties as $P^*$: $\cor(P_r) = \{1,2\}$;  if $[4\succ 2\succ 1\succ 3\succ\others]$ is removed, then $2$ is the unique winner; and if $[3\succ 1\succ 2 \succ 4\succ\others]$  is removed, then $1$ is the unique winner. This means that $\Par$ is violated at $P_r$.

\paragraph{\bf \boldmath Step 2: define a polyhedron $\ppoly{\cor}$ to represent profiles that satisfy Condition~\ref{cond:Par-MRSE}.} 
%
To define $\ppoly{\cor}$, we recall from Definition~\ref{dfn:MRSE-GISR} that  
for any $a,b$, any $B\subseteq \ma\setminus\{a,b\}$, and any profile $P$, $\scorediff{B,a,b}\cdot \hist(P)$ is the difference between the $\cor_{m-|B|}$ score of $a$ and the $\cor_{m-|B|}$ score of $b$ in $P|_{\ma\setminus B}$.  We are now ready to define $\ppoly{\cor}$ whose $\ba$ matrix has five parts that correspond to  Condition~\ref{cond:Par-MRSE} (1)--(5). Condition~\ref{cond:Par-MRSE} (6) will be incorporated in the $\vbb$ vector of $\ppoly{\cor}$.
\begin{dfn}[\boldmath $\ppoly{\cor}$]
\label{dfn:H-r}
Given $\cor =(\cor_2,\ldots,\cor_{m})$,   we let  
 $\pba{\cor}=\left[\begin{array}{l}\pba{(1)}\\ \pba{(2)}\\ \pba{(3)}\\ \pba{(4)}\\ \pba{(5)}\end{array}\right]$, where 
\begin{itemize}
\item $\pba{(1)}$: for every $1\le i\le m-4$ and every $j\in\ma\setminus\{i+4\}$, $\pba{(1)}$ has a row $\scorediff{\{5,\ldots, i+3\},i+4,j}$.
\item $\pba{(2)}$, $\pba{(3)}$, and $\pba{(4)}$ are defined as follows.
$$\pba{(2)} = \left[\begin{array}{l}\scorediff{\{5,\ldots,m\}, 2,1}\\ \scorediff{\{5,\ldots,m\}, 3,2}\\ \scorediff{\{5,\ldots,m\}, 4,3}\\ \scorediff{\{5,\ldots,m\}, 3,4}\end{array}\right], 
\pba{(3)} = \left[\begin{array}{l}\scorediff{\{3,5,\ldots,m\}, 4,1}\\ \scorediff{\{3,5,\ldots,m\}, 2,4} \\ \scorediff{\{2,3,5,\ldots,m\}, 4,1}\end{array}\right],
\pba{(4)} = \left[\begin{array}{l}\scorediff{\{4,5,\ldots,m\}, 3,2}\\ \scorediff{\{4,5,\ldots,m\}, 1,3}\\ \scorediff{\{1,4,5,\ldots,m\}, 3,2}\end{array}\right]$$
\item $\pba{(5)}$ consists of two rows defined as follows.
$$\pba{(5)} = \left[\begin{array}{l}-\hist(4\succ 2\succ 1\succ 3\succ \others)\\ -\hist(3\succ 1\succ 2\succ 4\succ \others)\end{array}\right]$$
\end{itemize}
 
\begin{align*}
\text{Let \hspace{5mm}}&\pvbb{\cor} = [\underbrace{\pvbb{(1)}}_{\text{for }\pba{(1)}},\underbrace{(s_4^4-s_1^4-1,s_4^4-s_1^4-1,0,0)}_{\text{for }\pba{(2)}},  \underbrace{(s_3^3-s_1^3-1,s_3^3-s_1^3-1,s_2^2-s_1^2-1)}_{\text{for }\pba{(3)}}, \\
&\hspace{20mm}\underbrace{(s_3^3-s_1^3-1,s_3^3-s_1^3-1,s_2^2-s_1^2-1)}_{\text{for }\pba{(4)}}, \underbrace{(-1,-1)}_{\text{for }\pba{(5)}}],
\end{align*}
where for every $1\le i\le m-4$ and every $j\in\ma\setminus\{i+4\}$, $\pvbb{(1)}$ contains a row $s_{m+1-i}^{m+1-i}-s^{m+1-i}_{1}-1$. Let 
\begin{align*}
 \ppoly{\cor} =\left \{\vec x\in {\mathbb R}^{m!}: \pba{\cor}\cdot \invert{\vec x}\le \invert{\pvbb{\cor}}\right \}.
 \end{align*}
\end{dfn}

\paragraph{\bf \boldmath Step 3: Apply Lemma~\ref{lem:sPar-GSR} and~\cite[Theorem~2]{Xia2021:How-Likely}.} We first prove the following properties of $\ppoly{\cor}$.
\begin{claim}[\bf \boldmath Properties of  $\ppoly{\cor}$]
\label{claim:H-MRS}
Given any integer MRSE rule $\cor$, 
\begin{enumerate}[label=(\roman*)]
\item for any integral profile $P$, if $\hist(P)\in \ppoly{\cor}$ then $\sat{\Par}(r,P)=0$;
\item $\piuni\in \ppolyz{\cor}$; 
\item  $\dim(\ppolyz{\cor})= m!-1$.
\end{enumerate}
\end{claim}
\begin{proof}
Part (i) follows after  a similar reasoning as in Step 1 of the proof of Theorem~\ref{thm:sPar-MRSE}. To prove Part (ii), notice that for any $B\subseteq \ma$  and $a,b\in (\ma\setminus B)$, we have $\scorediff{B,a,b}\cdot \vec 1 = 0$. Also notice that for any $R\in\ml(\ma)$ we have $-\hist(R)\cdot \vec 1 =-1<0$. Therefore, $\pba{\cor}\cdot\invert{\vec 1} \le  \invert{\vec 0}$, which means that $\piuni\in \ppolyz{\cor}$. To prove Part (iii), notice that $\pba{\cor}\cdot\invert{\vec x} \le \invert{\vec 0}$ contains one equality in $\pba{(2)}$, i.e.
\begin{equation}
\label{equ:34}
\scorediff{\{5,\ldots,m\},3,4}\cdot\invert{\vec x} = 0
\end{equation}
This means that $\dim(\ppolyz{\cor})\le  m!-1$. Recall that $P_r$ is the $n$-profile defined in Step 1 that satisfies Condition~\ref{cond:Par-MRSE}. Notice that  $\hist(P_r)$ is an inner point of $\ppolyz{\cor}$ in the sense that all  inequalities in $\pba{\cor}\cdot\invert{\vec x}\le \invert{\vec 0}$ except Equation (\ref{equ:34}) are strict, which means that $\dim(\ppolyz{\cor})\ge  m!-1$. This proves Claim~\ref{claim:H-MRS}.
\end{proof}
Because of the existence of $P_r$ defined in Step 1, and Claim~\ref{claim:H-MRS} (i) and (ii), the $1$ case and the VL case of  Lemma~\ref{lem:sPar-GSR} do not hold for any sufficiently large $n$. Therefore, it follows from the L case of Lemma~\ref{lem:sPar-GSR} that $\satmin{ \Par}{\Pi}(r,n )$ is at least $1-O(n^{-0.5})$, because $\ell_n\ge 1$. It remains to show that  $\satmin{ \Par}{\Pi}(r,n )$ is upper-bounded by $1-\Omega(n^{-0.5})$. We have the following calculations.
\begin{align*}
1- \satmin{ \Par}{\Pi}(r,n) = &\sup_{\vec\pi\in\Pi^n}\Pr\nolimits_{P\sim\vec\pi}(\sat{\Par}(r,P)=0)\\
\ge & \sup_{\vec\pi\in\Pi^n}\Pr\nolimits_{P\sim\vec\pi}(\hist(P)\in\ppoly{\cor}) &\text{Claim~\ref{claim:H-MRS} (i)}\\
 = & \Theta(n^{-0.5})&\text{Claim~\ref{claim:H-MRS} (ii), (iii), and~\cite[Theorem~2]{Xia2021:How-Likely}}
\end{align*}
The last equation follows after applying the sup part of~\cite[Theorem~2]{Xia2021:How-Likely} to $\ppoly{\cor}$. More concretely, recall that in Step 1 above we have constructed an $n$-profile $P_r$ for any sufficiently large $n$ and it is not hard to verify that $\hist(P_r)\in \ppoly{\cor}$, which means that $\ppoly{\cor}$ is active at any sufficiently large $n$. Claim~\ref{claim:H-MRS} (ii) implies that the polynomial case of~\cite[Theorem~2]{Xia2021:How-Likely} holds, and Claim~\ref{claim:H-MRS} (iii) implies that $\alpha_n= m!-1$ for $\ppoly{\cor}$.

This proves Theorem~\ref{thm:sPar-MRSE}.
\end{proof}

\subsection{Proof of Theorem~\ref{thm:sPar-Cond-Pos}}
\label{app:proof-thm:sPar-Cond-Pos}

\begin{thm}[\bf Smoothed $\Par$: Condorcetified Integer Positional Scoring Rules]
\label{thm:sPar-Cond-Pos}
{
Given $m\ge 4$, an integer positional irresolute scoring rule $\cor_{\vec s}$, any Condocetified positional scoring rule $\Condorcet{\vec s}$ that is a refinement of $\iCondorcet{ \vec s}$, and any strictly positive and closed $\Pi$ over $\ml(\ma)$ with $\piuni\in \conv(\Pi)$, there exists $N\in\mathbb N$ such that for every $n\ge N$,  
$$\satmin{ \Par}{\Pi}(\Condorcet{\vec s},n ) = 1-  \Theta({\frac{1}{\sqrt n}})$$
}
\end{thm}
\begin{proof}
The proof follows the same logic in the proof of Theorem~\ref{thm:sPar-MRSE}. We first prove the theorem for even $n$ then show how to extend the proof to odd $n$'s.

\paragraph{\bf \boldmath Intuition for $2\mid n$.} Let $\vec s = (s_1,\ldots,s_m)$. We first identify a set of sufficient conditions for $\Par$ to be violated.  
\begin{cond}[\bf Sufficient conditions for the violation of $\Par$]
\label{cond:Par-Cond-Pos}
Given a Condorcetified irresolute integer positional scoring rule $\iCondorcet{ \vec s}$, $P$ satisfies the following conditions.
\begin{enumerate}[label=(\arabic*)]
\item $\iCondorcet{ \vec s}(P) = \{2\}$, and the score of $2$ is higher than the score of any other alternative by at least $s_1-s_m+1$.
\item Alternative $1$ is a weak Condorcet winner, $w_{P}(1,3)=0$, and for every $i\in\ma\setminus\{1,3\}$, $w_{P}(1,i)\ge 2$.
\item $P$ contains at least one vote of $[3\succ 1\succ 2\succ \others]$.
\end{enumerate}
\end{cond}
Recall that $\Condorcet{ \vec s}$ is a refinement of $\iCondorcet{ \vec s}$ and due to Condition~\ref{cond:Par-Cond-Pos} (2), $P$ does not contain a Condorcet winner. Therefore, according to  Condition~\ref{cond:Par-Cond-Pos} (1), we have $\Condorcet{ \vec s}=\{2\}$. Any voter whose preferences are $[3\succ 1\succ 2\succ \others]$ has incentive to abstain from voting, because the voter prefers $1$ to $2$, and $\{1\}$ is the Condorcet winner in $P-[3\succ 1\succ 2\succ \others]$, which means that 
$$\Condorcet{ \vec s}(P-[3\succ 1\succ 2\succ \others])=\{1\}$$
This means that $\sat{\Par}(\Condorcet{ \vec s},P)=0$ for any profile $P$ that satisfies Condition~\ref{cond:Par-Cond-Pos}. The rest of the proof proceeds as follows. 
In Step 1, for any $n$ that is sufficiently large, we construct an $n$-profile $P_{\vec s}$ that satisfies Condition~\ref{cond:Par-Cond-Pos}. Then in Step 2, we formally  define $\ppoly{\iCondorcet{ \vec s}}$ to represent profile that satisfy Condition~\ref{cond:Par-Cond-Pos}.  Finally, in Step 3 we formally prove properties about $\ppoly{\iCondorcet{ \vec s}}$ and apply Lemma~\ref{lem:sPar-GSR} and~\cite[Theorem~2]{Xia2021:How-Likely} to prove Theorem~\ref{thm:sPar-MRSE}.  

\paragraph{\bf \boldmath Step 1  for $2\mid n$: define $P_{\vec s}$.} The construction is similar to the construction in the proof of Claim~\ref{claim:CW-score-winner-diff}, which is done for the following two cases: $\cor_{\vec s}$ is the plurality rule and  $\cor_{\vec s}$ is not the plurality rule.

\begin{itemize}
\item {\bf\boldmath  When $\cor_{\vec s}$ is the plurality rule,} i.e.~$s_2 = s_m$, we let 
\begin{align*}
P_{\vec s} =  \left(\frac n2-6\right)\times [2\succ 1\succ 3\succ \others]+  4\times [2\succ 3\succ 1\succ \others]&\\
+  \left(\frac n2-6\right)\times [3\succ 1\succ 2\succ \others]+ 6\times [1\succ 2\succ 3\succ \others]&
\end{align*}
It is not hard to verify that $P_{\vec s}$ satisfies Condition~\ref{cond:Par-Cond-Pos} for any even number $n\ge 28$.

\item {\bf\boldmath  When $\cor_{\vec s}$ is not the plurality rule,} i.e., $s_2>s_m$, like Step~1 in the proof of Theorem~\ref{thm:sPar-MRSE}, we first construct a profile $P^*$ that consists of a constant number of votes and satisfies Condition~\ref{cond:Par-Cond-Pos}, then extend it to arbitrary odd number $n$. 
Let $2\le k\le m-1$ denote the smallest number such that $s_k>s_{k+1}$. Let $A_1 = [4\succ \cdots\succ k+1]$ and $A_2 = [k+2\succ \cdots\succ m]$, and let $P^*=P_1^*+P_2^*$, where $P_1^*$ is the following $10$-profile that is used to guarantee~Condition~\ref{cond:Par-Cond-Pos} (2) and (3).
\begin{align*}
P_1^* = \{4\times [1\succ 2\succ A_1\succ 3\succ A_2] +3\times [2\succ 3\succ A_1\succ1\succ A_2] &\\
+   2\times [3\succ 1\succ A_1\succ 2\succ A_2] +[2\succ 1 \succ A_1\succ 3\succ A_2] &\}
\end{align*}
And let $P_2^*$ denote the following $36(m-3)!$-profile, which is used to guarantee that $2$ is the unique winner under $P^*$, i.e., Condition~\ref{cond:Par-Cond-Pos} (1).
$$P_2^* = 6\times \{[R_1\succ R_2]:\forall R_1\in \ml(\{1,2,3\}), R_2\in \ml(\{4,\ldots,m\}), \}$$
It is not hard to verify that the following observations hold for $P_1^*$.
\begin{itemize}
\item  $1$ is the Condorcet winner, $w_{P_1^*}(1,3)=0$, and for any $i\in\ma\setminus\{1,3\}$, we have $w_{P_1^*}(1,i)\ge 2$. 
\item The total score of $1$ under $P_1^*$ is $4s_1+3s_2+3s_{k+1}$, the  total score of $2$ under $P_1^*$ is $4s_1+4s_2+2s_{k+1}$, and   the total score of $3$ under $P_1^*$ is $2s_1+3s_2+5s_{k+1}$. Recall that we have assumed that $s_2>s_{k+1}$. Therefore, 
$$4s_1+4s_2+2s_{k+1}>4s_1+3s_2+3s_{k+1}> 2s_1+3s_2+5s_{k+1},$$
which means that the score of $2$ is strictly higher than the scores of $1$ and $3$ in $P_1^*$.
\end{itemize}
Given these observations, it is not hard to verify that $P^*=P_1^*+P_2^*$ satisfies Condition~\ref{cond:Par-Cond-Pos}. Let $P_{\vec s}$ denote as many copies of $P^*$ as possible, plus pairs of rankings $\{[2\succ 1\succ 3\succ \others], [2\succ 3\succ 1\succ \others]\}$. More precisely, let
$$P_{\vec s} = \left\lfloor \frac{n}{|P^*|}\right\rfloor\times P^*  + \left(\frac{n- |P^*|\cdot \lfloor \frac{n}{|P^*|}\rfloor}{2} \right)\times \{[2\succ 1\succ 3\succ \others], [2\succ 3\succ 1\succ \others]\}$$
\end{itemize}
It is not hard to verify that $P_{\vec s}$ satisfies Condition~\ref{cond:Par-Cond-Pos}, which concludes Step 1 for the $2\mid n$ case.

\paragraph{\bf \boldmath Step 2 for $2\mid n$: define a polyhedron $\ppoly{\iCondorcet{\vec s}}$ to represent profiles that satisfy Condition~\ref{cond:Par-Cond-Pos}.} 

\begin{dfn}[\boldmath $\ppoly{\iCondorcet{\vec s}}$]
\label{dfn:H-Cond-scoring}
Given an irresolute integer positional scoring rule $\cor_{\vec s} =(s_1,\ldots,s_m)$,   we let  
 $\pba{\vec s}=\left[\begin{array}{l}\pba{(1)}\\ \pba{(2)}\\ \pba{(3)}\end{array}\right]$, where 
\begin{itemize}
\item $\pba{(1)}$: for every $i\in\ma\setminus\{2\}$, $\pba{(1)}$ contains a row $\score_{i,2}$.
\item $\pba{(2)}$  contains two rows $\pair_{1,3}$ and $\pair_{3,1}$, and for every  $i\in\ma\setminus\{1,3\}$, $\pba{(1)}$ contains a row $\pair_{i,1}$.
\item $\pba{(3)}$ consists of a single row $-\hist(3\succ 1\succ 2\succ \others)$. 
\end{itemize}
 
\begin{align*}
\text{Let \hspace{5mm}}&\pvbb{\vec s} = \left[\underbrace{(s_m-s_1-1)\cdot\vec 1}_{\text{for }\pba{(1)}},\underbrace{(0,0, -2,\ldots,-2)}_{\text{for }\pba{(2)}},  \underbrace{-1}_{\text{for }\pba{(3)}} \right]\\
\text{and \hspace{5mm}}& \ppoly{\vec s} = \left \{\vec x\in {\mathbb R}^{m!}: \pba{\vec s}\cdot \invert{\vec x}\le \invert{\pvbb{\vec s}} \right\}.
\end{align*} 
\end{dfn}

\paragraph{\bf \boldmath Step 3 for $2\mid n$: Apply Lemma~\ref{lem:sPar-GSR} and~\cite[Theorem~2]{Xia2021:How-Likely}.}   We first prove the following properties of $\ppoly{\iCondorcet{\vec s}}$.
\begin{claim}[\bf \boldmath Properties of  $\ppoly{\iCondorcet{\vec s}}$]
\label{claim:H-Condorcetification}
Given any integer positional scoring rule $\vec s$, 
\begin{enumerate}[label=(\roman*)]
\item for any integral profile $P$, if $\hist(P)\in \ppoly{\iCondorcet{\vec s}}$ then $\sat{\Par}(\Condorcet{\vec s},P)=0$;
\item $\piuni\in \ppolyz{\iCondorcet{\vec s}}$; 
\item  $\dim(\ppolyz{\iCondorcet{\vec s}})= m!-1$.
\end{enumerate}
\end{claim}
\begin{proof}
The proof for Part (i) and (ii) are similar to the proof of Claim~\ref{claim:H-MRS}. To prove Part (iii), notice that $\pba{\vec s}\cdot\invert{\vec x} \le \invert{\vec 0}$ contains one equality in $\pba{(2)}$, i.e.
\begin{equation}
\label{equ:13}
\pair_{1,3}\cdot\invert{\vec x} = \invert{0}
\end{equation}
This means that $\dim(\ppolyz{\iCondorcet{\vec s}})\le  m!-1$. Notice that   $\hist(P_{\vec s})$ is an inner point of $\ppolyz{\iCondorcet{\vec s}}$ in the sense that all other inequalities except Equation (\ref{equ:13}) are strict, which means that $\dim(\ppolyz{\iCondorcet{\vec s}})\ge  m!-1$. This proves Claim~\ref{claim:H-Condorcetification}.
\end{proof}
Therefore, we have the following bound.
\begin{align*}
&1- \satmin{ \Par}{\Pi}(\Condorcet{\vec s},n) \\
= &\sup_{\vec\pi\in\Pi^n}\Pr\nolimits_{P\sim\vec\pi}(\sat{\Par}(\Condorcet{\vec s},P)=0)\\
\ge & \sup_{\vec\pi\in\Pi^n}\Pr\nolimits_{P\sim\vec\pi}(\hist(P)\in\ppoly{\iCondorcet{\vec s}}) &\text{Claim~\ref{claim:H-Condorcetification} (i)}\\
 = & \Theta(n^{-0.5})&\text{Claim~\ref{claim:H-Condorcetification} (ii), (iii), and~\cite[Theorem~2]{Xia2021:How-Likely}}
\end{align*}

Consequently, $\satmin{ \Par}{\Pi}(\Condorcet{\vec s},n) = 1-\Omega(n^{-0.5})$. Notice that the $1$ case and VL case Lemma~\ref{lem:sPar-GSR} do not hold because of the existence of $P_{\vec s}$ and Claim~\ref{claim:H-Condorcetification} (ii). Therefore, Theorem~\ref{thm:sPar-Cond-Pos} for the $2\mid n$ case follows after the $1-O(n^{-0.5})$ upper bound proved in Lemma~\ref{lem:sPar-GSR}.

\paragraph{\bf\boldmath Proof for the $2\nmid n$ case.} When $2\nmid n$, we modify the proof as follows. 
\begin{itemize}
\item First, Condition~\ref{cond:Par-Cond-Pos} (2) is replaced by the following condition:

(2$'$):  Alternative $1$ is the Condorcet winner under $P$, $w_{P}(1,3)=1$, and for every $i\in\ma\setminus\{1,3\}$, $w_{P}(1,i)\ge 3$.

\item Second, in Step 1, $P_{\vec s}$ has an additional vote $[2\succ 1\succ 3\succ\others]$.

\item Third, in Step 2 Definition~\ref{dfn:H-Cond-scoring}, the $\pvbb{\vec s}$ components corresponding to $\pba{2}$ is $(1,-1,-3,\ldots,-3)$.
\end{itemize}
A similar claim as Claim~\ref{claim:H-Condorcetification} can be proved for the $2\nmid n$ case. This proves Theorem~\ref{thm:sPar-Cond-Pos}.
\end{proof}

\section{Experimental Results}
\label{app:exp}
We report  satisfaction of $\CC$ and $\Par$ using   simulated data and Preflib linear-order data~\citep{Mattei13:Preflib} under four classes of commonly-used  voting rules studied in this paper, namely positional scoring rules (plurality, Borda,  and veto),  voting rules that satisfy {\sc Condorcet Criterion} (maximin, ranked pairs, Schulze, and Copeland$_{0.5}$),  MRSE (STV), and Condorcetified positional scoring rule (Black's rule).  All experiments were implemented in Python 3 and were run on a MacOS laptop with 3.1 GHz Intel Core i7 CPU and 16 GB memory.

\begin{figure}[htp]
\centering
\begin{tabular}{cc}
\includegraphics[width = 0.45\textwidth]{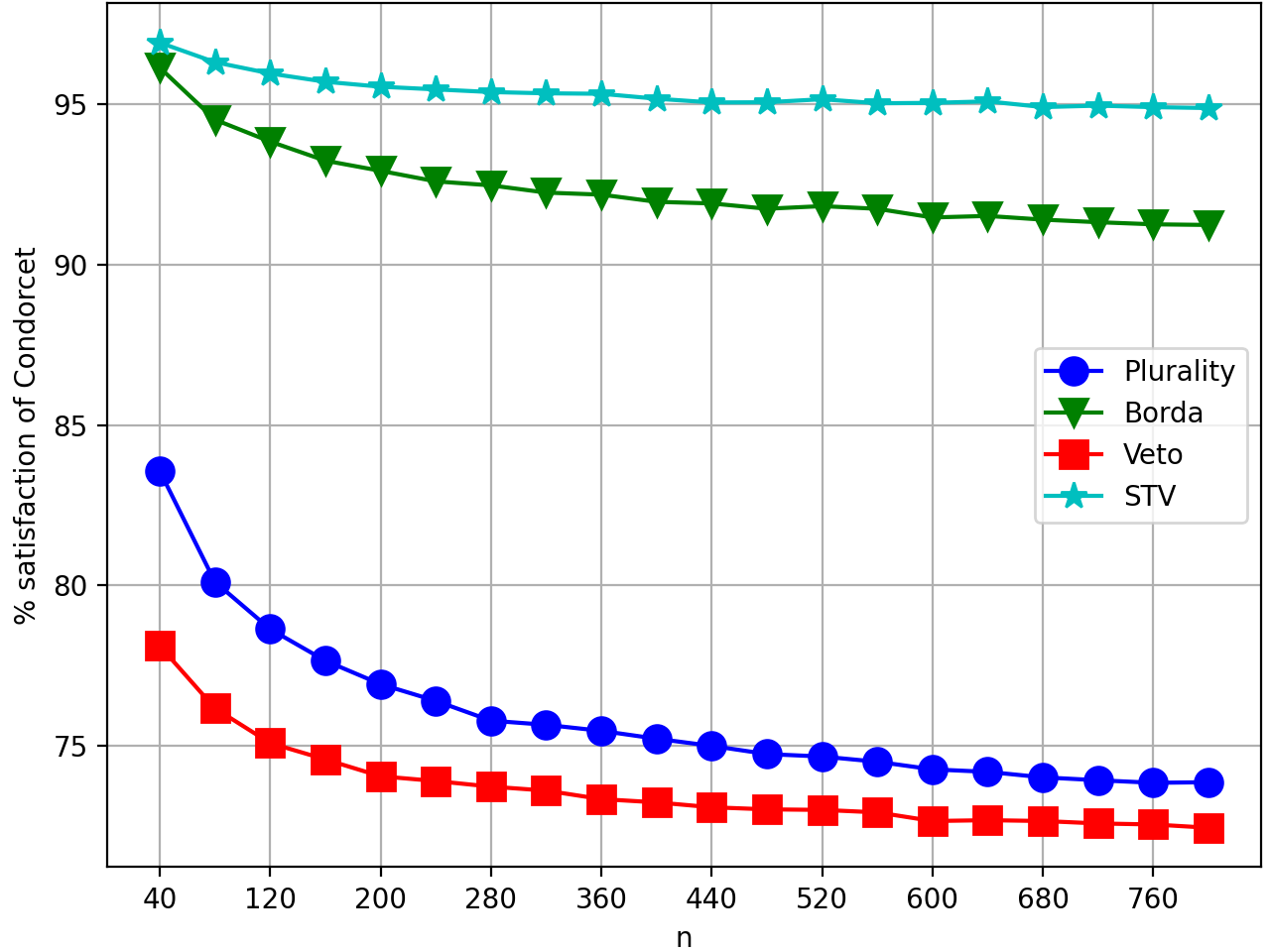}&
\includegraphics[width = 0.45\textwidth]{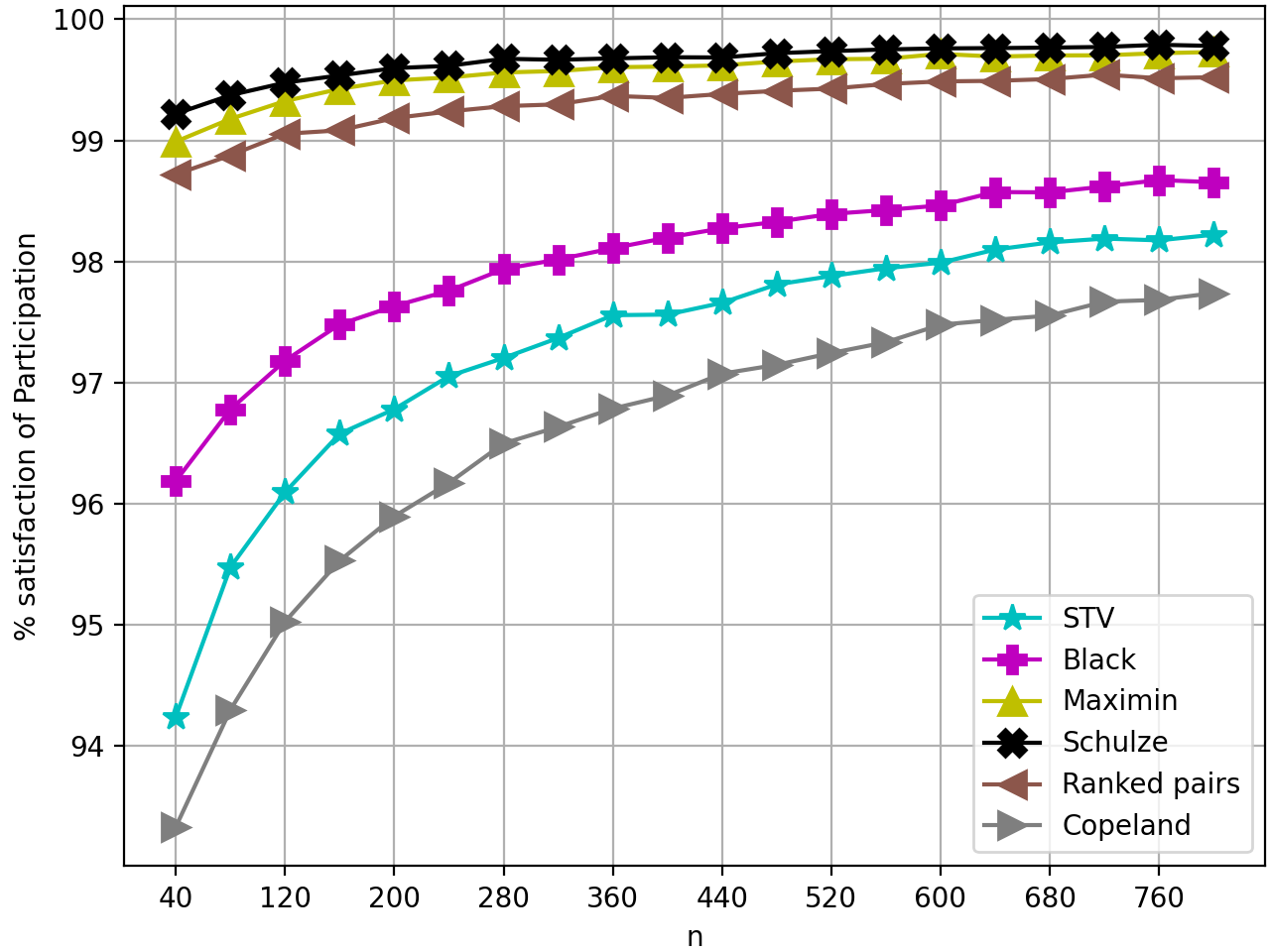}\\
(a) $\CC$. & (b) $\Par$.
\end{tabular}
\caption{\small Satisfaction of $\CC$ and $\Par$  under IC for $m=4$, $n=40$ to $800$, $200000$ trials. \label{fig:exp-sat}}
\end{figure}

\paragraph{\bf Synthetic data.} We  generate profiles of $m=4$ alternatives under IC.\footnote{See~\cite{Brandt2021:Exploring} for theoretical results and extensive simulation studies of $\Par$ under the IAC model.} The number of alternatives $n$ ranges from  $40$ to $800$. In each setting we generate $200000$ profiles. The satisfaction of $\CC$ under plurality, Borda, veto, and STV are presented in Figure~\ref{fig:exp-sat} (a), and the satisfaction of $\Par$ under
STV, maximin, ranked pairs, Schulze, Black, and Copeland$_{0.5}$ are presented in Figure~\ref{fig:exp-sat} (b). Notice that voting rules not  in Figure~\ref{fig:exp-sat} (a) always satisfy $\CC$ and voting rules not  in Figure~\ref{fig:exp-sat} (b) always satisfy $\Par$. 

The  results provide a sanity check for the theoretical results proved in this paper. In particular, Figure~\ref{fig:exp-sat} (a) confirms that the  satisfaction of $\CC$ is $\Theta(1)$ and $1-\Theta(1)$ under positional scoring rules (Theorem~\ref{thm:sCC-scoring}) and STV (Corollary~\ref{Coro:sCC-STV-IC}) w.r.t.~IC. Figure~\ref{fig:exp-sat} (b) confirms that  the satisfaction of $\Par$ is $1-\Theta(n^{-0.5})$ under maximin, ranked pairs, Schulze (Theorem~\ref{thm:sPar-mm-rp-sch}), Copeland$_\alpha$ (Theorem~\ref{thm:sPar-copeland}), STV (Theorem~\ref{thm:sPar-MRSE}),  and Black (Theorem~\ref{thm:sPar-Cond-Pos}).  Figure~\ref{fig:exp-sat-large-n} in Appendix~\ref{app:exp} summarizes results with large  $n$ ($1000$ to $10000$) that further confirm the asymptotic observations described above.

\begin{figure}[htp]
\centering
\begin{tabular}{cc}
\includegraphics[width = 0.45\linewidth]{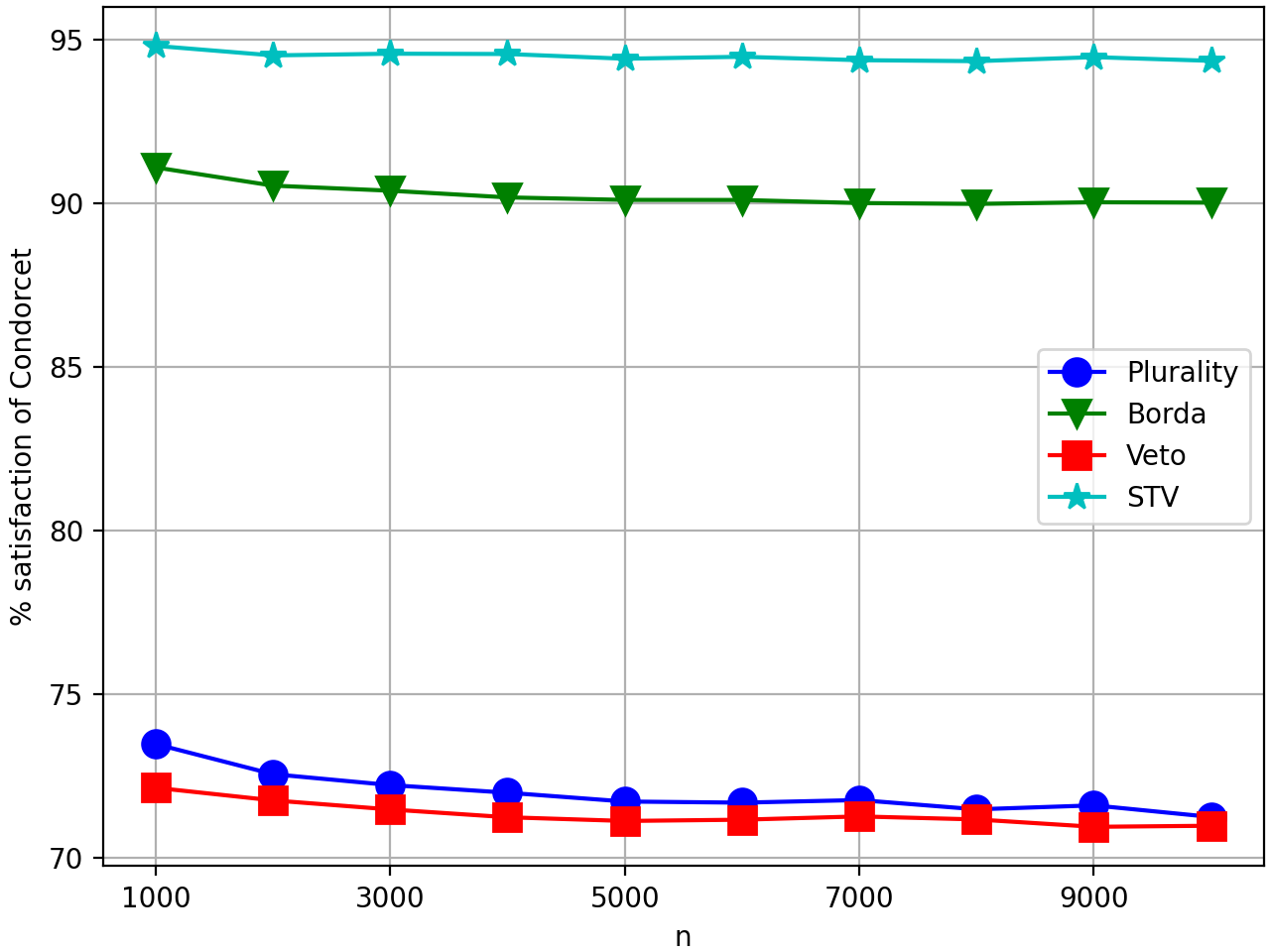}&
\includegraphics[width = 0.45\linewidth]{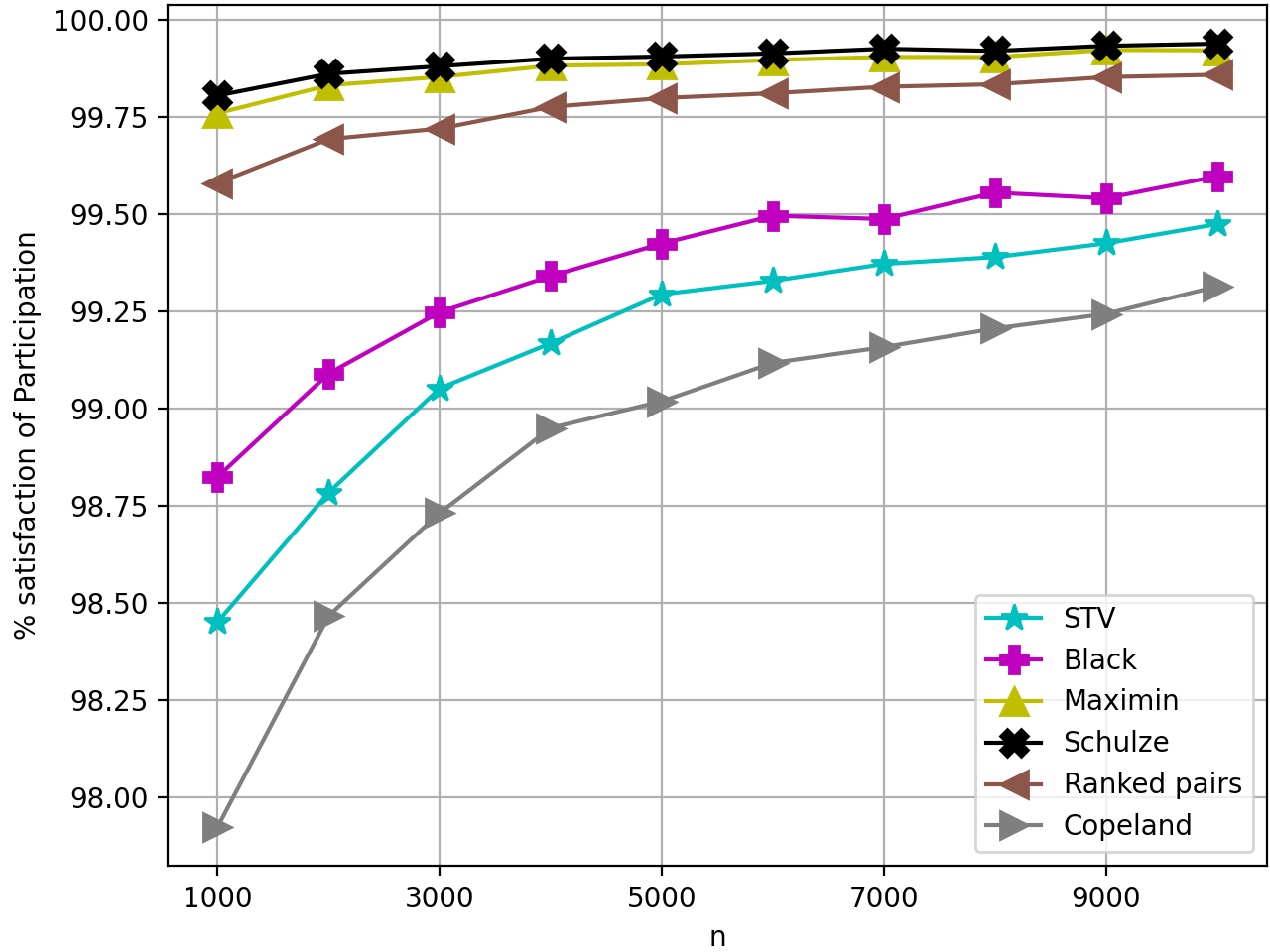}\\
(a) $\CC$. & (b) $\Par$
\end{tabular}
\caption{\small Satisfaction of $\CC$ and $\Par$  under IC for $m=4$, $n=1000$ to $10000$, $200000$ trials. \label{fig:exp-sat-large-n}}
\end{figure}

\paragraph{\bf Preflib data.} We also calculate the satisfaction of $\CC$ and $\Par$ under all voting rules studied in this paper with lexicographic tie-breaking for all 315  Strict Order-Complete Lists (SOC) under election data category  from Preflib~\citep{Mattei13:Preflib}. The results are summarized in Table~\ref{tab:sat-preflib}, which is the bottom part of Table~\ref{tab:summary}. 

\begin{table}[htp]
\centering
\caption{\small Satisfaction of $\CC$ and $\Par$ in 315 Preflib SOC profiles. Some statistics of the data are shown in Figure~\ref{fig:ex-histograms-all}. \label{tab:sat-preflib}}
\resizebox{\textwidth}{!}{
\begin{tabular}{|c|c|c|c|c|c|c|c|c|c|}
\hline
 &\small  Plurality &\small  Borda & \small Veto&\small  STV &\small Black &\small  Maximin&\small  Schulze&\small   Ranked pairs&\small   Copeland$_{0.5}$\\
\hline
$\CC$ & 96.8\% &  92.4\% &  74.2\% & 99.7\% &  100\% &  100\%  &  100\% &  100\% &  100\% \\
\hline $\Par$ & 100\% &  100\% &  100\% & 99.7\% & 99.4\% & 100\%  &  100\% &  100\% &  99.7\% \\
\hline
\end{tabular}
}
\end{table}

\begin{figure}[htp]
\centering
\begin{tabular}{cc}
  \includegraphics[width = 0.5\linewidth]{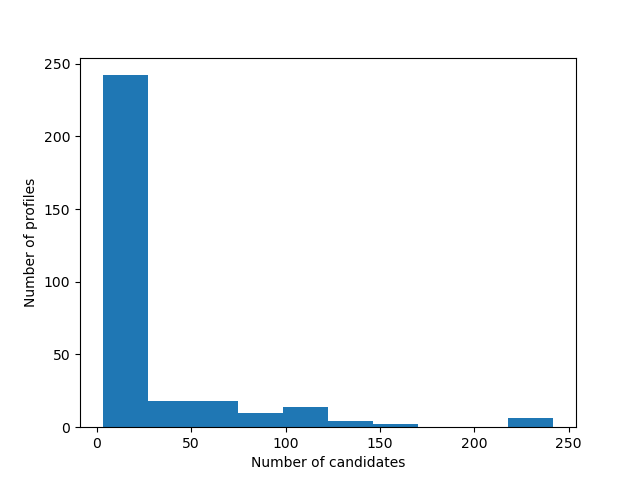} &
    \includegraphics[width = 0.5\linewidth]{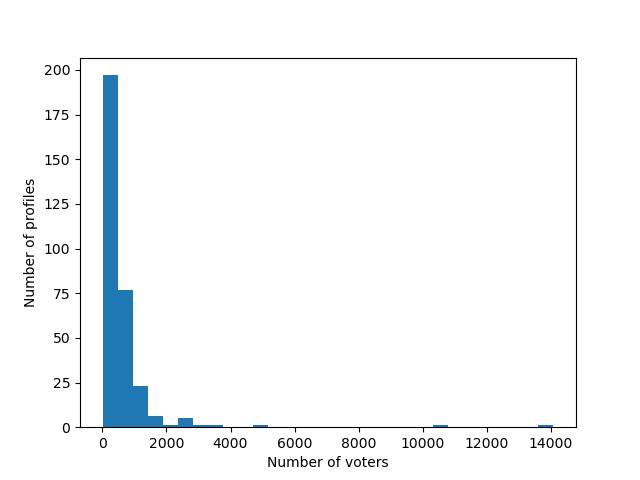} 
\end{tabular}
\caption{\small Histograms of number of candidates and number of voters in the 315 Preflib SOC data studied in this paper. \label{fig:ex-histograms-all}}
\end{figure}

Table~\ref{tab:sat-preflib} delivers the following message, that $\Par$ is less of a concern than $\CC$ in Preflib data---all voting rules have close to $100\%$ satisfaction of $\Par$, while the satisfaction of $\CC$ is much lower for plurality, Borda, and Veto. The most interesting observations are: first, maximin, Schulze, and ranked pairs achieve 100\% satisfaction of $\CC$ and $\Par$ in Preflib data, which is consistent with the belief that Schulze and ranked pairs are superior in satisfying voting axioms, and maximin is doing well in $\Par$ (and indeed, maximin satisfies $\Par$ when $m=3$). Second, STV  does  well in  $\CC$ and $\Par$, though it does not satisfy either in the worst case. Third, veto has poor satisfaction of $\CC$ ($74.2\%$), which is mainly due to the profiles where the number of alternatives is more than the number of voters, so that a Condorcet winner exists and is also a veto co-winner, but loses due to the  tie-breaking mechanism.

\end{document}

%% file: macro.tex
\usepackage[utf8]{inputenc}

\usepackage[numbers]{natbib}
\usepackage{caption}

\usepackage{booktabs} 
\usepackage[ruled]{algorithm2e} 
\usepackage{pifont}
\usepackage{nicefrac}
\usepackage{color}
\usepackage{xcolor,colortbl}
\usepackage{wrapfig}
\usepackage{graphicx}
\usepackage{amsmath,amssymb,amsthm}
\usepackage{mathtools}
 
\usepackage{bm}
\usepackage{rotating}
\usepackage{multicol}
\usepackage{multirow}
\usepackage{titlesec}
\usepackage{latexsym}
\usepackage{enumitem}
\usepackage{hyperref}[backref=page]
\usepackage{float}
\hypersetup{
     colorlinks   = true,
     linkcolor    = red, 
     urlcolor     = blue, 
	 citecolor    = blue 
}

\newtheorem{thm}{Theorem}
\newtheorem{dfn}{Definition}

\newtheorem{lem}{Lemma}
\newtheorem{ex}{Example}

\newtheorem{prop}{Proposition}
\newtheorem{coro}{Corrollary}

\newtheorem{cond}{Condition}
\newtheorem{claim}{Claim}

\newcounter{newct}

\newcommand{\bv}{\begin{array}}

\newcommand{\appLem}[3]{\vspace{1mm}\noindent{\bf Lemma~\ref{#1}.} {\bf (#2). } {\em #3}}
\newcommand{\appThm}[3]{\vspace{1mm}\noindent{\bf Theorem~\ref{#1}.} {\bf (#2). }{\em #3 \vspace{1mm}}}

\newcommand{\ma}{\mathcal A}
\newcommand{\calH}{\mathcal H}

\newcommand{\vH}{{\vec H}}
\newcommand{\mm}{\mathcal M}

\newcommand{\ml}{\mathcal L}
\newcommand{\mV}{\mathcal V}

\newcommand{\calG}{\mathcal G}

\newcommand{\ra}{\rightarrow}

\newcommand{\cor}{{\overline r}} 

\newcommand{\rank}{\text{Rank}}
\newcommand{\sign}{\text{Sign}}
\newcommand{\signH}{\text{Sign}_{\vH}}
\newcommand{\signs}[1]{\text{Sign}_{#1}}

\newcommand{\others}{\text{others}}

\newcommand{\hist}{\text{Hist}}

\newcommand{\fs}{{\cal S}_{\vec H}}
\newcommand{\fsatomic}{ {\cal S}_{\vec H}^{\circ}}
\newcommand{\pfs}[1]{{\cal S}_{#1}}

\newcommand{\sk}{{\cal S}_{K}}

\newcommand{\Omit}[1]{}

\newcommand{\ba}{{\mathbf A}}

\newcommand{\pba}[1]{{\mathbf A}^{#1}}

\newcommand{\vbb}{{\vec{\mathbf b}}}
\newcommand{\pvbb}[1]{{\vec{\mathbf b}}^{#1}}

\newcommand{\md}[2]{{\dim_{#1,n}^{\max}}(#2)}

\newcommand{\vXp}{{\vec X}_{\vec\pi}}

\newcommand{\piuni}{{\pi}_{\text{uni}}}
\newcommand{\avg}[1]{ \text{avg}(#1)}

\newcommand{\conv}{\text{CH}}

\newcommand{\wmg}{\text{WMG}}

\newcommand{\umg}{\text{UMG}}

\newcommand{\pair}{\text{Pair}}
\newcommand{\scorediff}[1]{\text{Score}^{\Delta}_{#1}}
\newcommand{\wcw}{\text{WCW}}

\newcommand{\score}{\text{Score}}
\newcommand{\icopeland}{\overline{\text{Cd}}_\alpha}
\newcommand{\copeland}{\text{Cd}_\alpha}

\newcommand{\corbucklin}{\overline{\text{Bkln}}}
\newcommand{\plu}{\text{Plu}}
\newcommand{\iplu}{\overline{\text{Plu}}}

\newcommand{\maximin}{\text{MM}}
\newcommand{\imaximin}{\overline{\text{MM}}}
\newcommand{\rp}{\text{RP}}
\newcommand{\irp}{\overline{\text{RP}}}
\newcommand{\schulze}{\text{Sch}}
\newcommand{\ischulze}{\overline{\text{Sch}}}
\newcommand{\stv}{\text{STV}}
\newcommand{\istv}{\overline{\text{STV}}}

\newcommand{\eo}{\text{EO}}
\newcommand{\ties}{\text{Ties}}

\newcommand{\minscore}{\text{MS}}

\newcommand{\violation}[3]{\text{Vio}_{#1}^{#2}(#3)}

\newcommand{\Condorcet}[1]{{\text{Cond}_{#1}}}
\newcommand{\iCondorcet}[1]{\overline{\text{Cond}_{#1}}}

\newcommand{\noCW}{{\text{NCW}}}
\newcommand{\CWlose}{{\text{CWL}}}
\newcommand{\CWwin}{{\text{CWW}}}

\newcommand{\condition}[1]{{\text{C}_{#1}}}
\newcommand{\PU}[2]{{\text{PU}_{#1}(#2)}}
\newcommand{\PULR}[2]{{\text{LR}_{#1}(#2)}}

\newcommand{\rev}[1]{{\text{Rev}\left(#1\right)}}

\newcommand{\poly}{{\mathcal H}}

\newcommand{\polyz}{{\mathcal H}_{\leq 0}}

\newcommand{\ppoly}[1]{{\mathcal H}^{#1}}
\newcommand{\ppolyz}[1]{{\mathcal H}_{\leq 0}^{#1}}

\newcommand{\upoly}{{\mathcal C}}

\newcommand{\upolynint}{{\mathcal C}_n^{\mathbb Z}}

\newcommand{\upolyz}{{\mathcal C}_{\leq 0}}

\newcommand{\upolynoCW}{\pupoly{\noCW}}
\newcommand{\upolyznoCW}{\pupoly{\noCW,\le 0}}

\newcommand{\upolyCWwin}{\pupoly{\CWwin}}
\newcommand{\upolyzCWwin}{\pupoly{\CWwin,\le 0}}

\newcommand{\upolyCWlose}{\pupoly{\CWlose}}
\newcommand{\upolyzCWlose}{\pupoly{\CWlose,\le 0}}

\newcommand{\pupoly}[1]{{\mathcal C}_{#1}}

\newcommand{\aupoly}{{\mathcal C^*}}

\newcommand{\aupolyz}{{\mathcal C}_{\leq 0}^*}
\newcommand{\saupoly}{\hat {\mathcal C}}
\newcommand{\saupolyz}{\hat{\mathcal C}_{\leq 0}}

\newcommand{\cpoly}[1]{{\mathcal H}_{#1}}

\newcommand{\cpolynint}[1]{{\mathcal H}_{#1,n}^{\mathbb Z}}
\newcommand{\cpolyz}[1]{{\mathcal H}_{#1,\leq 0}}

\newcommand{\acpoly}[1]{\hat{\mathcal H}_{#1}}
\newcommand{\acpolyz}[1]{\hat {\mathcal H}_{#1,\leq 0}}

\newcommand{\region}[2]{{\mathcal R}_{#1}^{#2}}

\newcommand{\invert}[1]{\left(#1\right)^\top}

\newcommand{\closure}[1]{\text{Closure}(#1)}

\newcommand{\cwinner}{\text{CW}}
\newcommand{\wcwinner}{\text{WCW}}
\newcommand{\almostCW}{\text{ACW}}

\newcommand{\CC}{\text{\sc CC}}
\newcommand{\CL}{\text{CL}}
\newcommand{\Par}{\text{\sc Par}}

\newcommand{\satmin}[2]{\widetilde{#1}_{#2}^{\min}}
\newcommand{\sat}[1]{{#1}}